\documentclass{lmcs}
\pdfoutput=1

\usepackage{lastpage}
\lmcsdoi{15}{1}{12}
\lmcsheading{}{\pageref{LastPage}}{}{}%
{Nov.~29,~2017}{Feb.~13,~2019}{}

\usepackage{pdfsync}
\usepackage{graphicx}
\usepackage{epsfig}
\usepackage{xcolor}
\usepackage{colortbl}
\usepackage{amsmath}
\usepackage{amssymb}
\usepackage{mathtools}
\usepackage{multirow}
\usepackage{makecell}
\usepackage[ruled,vlined,linesnumbered]{algorithm2e}
\usepackage{amsthm}
\usepackage{subfigure}
\usepackage{caption}

\usepackage{pgf}
\usepackage{tikz}

\usetikzlibrary{arrows,%
	shapes,%
	decorations,%
	decorations.pathmorphing,%
	decorations.text,%
	decorations.pathreplacing,%
	automata,%
	arrows,%
	chains,%
	matrix,%
	scopes,%
	shadows,%
	positioning,%
	patterns,%
	fit,%
	calc,%
	through,%
	fadings,%
	backgrounds}

\usepackage{stmaryrd}
\usepackage{color}

\usepackage[textsize=small]{todonotes}

\usepackage{hyperref}

\hypersetup{
    bookmarks=true,         % show bookmarks bar?
    unicode=false,          % non-Latin characters in Acrobat’s bookmarks
    pdftoolbar=true,        % show Acrobat’s toolbar?
    pdfmenubar=true,        % show Acrobat’s menu?
    pdffitwindow=false,     % window fit to page when opened
    pdfstartview={FitH},    % fits the width of the page to the window
    pdftitle={My title},    % title
    pdfauthor={Author},     % author
    pdfsubject={Subject},   % subject of the document
    pdfcreator={Creator},   % creator of the document
    pdfproducer={Producer}, % producer of the document
    pdfkeywords={keyword1} {key2} {key3}, % list of keywords
    pdfnewwindow=true,      % links in new window
    colorlinks=true,        % false: boxed links; true: colored links
    linkcolor=black,        % color of internal links
    citecolor=black,        % color of links to bibliography
    filecolor=magenta,      % color of file links
    urlcolor=cyan           % color of external links
}

\theoremstyle{plain}
\newtheorem{theorem}{Theorem}[section]
\newtheorem{lemma}[theorem]{Lemma}

\newtheorem{corollary}[theorem]{Corollary}
\theoremstyle{definition}
\newtheorem{definition}[theorem]{Definition}

\newtheorem{remark}[theorem]{Remark}

\newcommand{\ignore}[1]{}

\newcommand{\goesto}[1]{\stackrel {#1} \longrightarrow}
\newcommand{\comesfrom}[1]{\stackrel {#1} \longleftarrow}

\newcommand{\symb}{\sigma}
\newcommand{\trans}[3]{{#1} \goesto {#2} {#3}} % {(#1,#2,#3)\in\delta}
\newcommand{\prefix}[2]{#1[0..#2]}
\newcommand{\suffix}[2]{#1[#2..]}

\newcommand{\A}{\mathcal A}
\newcommand{\B}{\mathcal B}
\newcommand{\C}{\mathcal C}
\newcommand{\lang}[1]{\mathcal{L}(#1)}

\newcommand{\xsim}[1]{\sqsubseteq^{#1}}
\newcommand{\xsimrev}[1]{\sqsupseteq^{#1}}
\newcommand{\disim}{\xsim{\mathsf{di}}}
\newcommand{\disimrev}{\xsimrev{\mathsf{di}}}
\newcommand{\desim}{\xsim{\mathsf{de}}}
\newcommand{\fsim}{\xsim{\mathsf f}}
\newcommand{\fsimrev}{\xsimrev{\mathsf f}}
\newcommand{\bwdisim}{\xsim{\mathsf{bw\textrm{-}di}}}
\newcommand{\bwdisimrev}{\xsimrev{\mathsf{bw\textrm{-}di}}}
\newcommand{\medsim}{\xsim{\mathsf{m}}}
\newcommand{\medsimrev}{\xsimrev{\mathsf{m}}}

\newcommand{\xincl}[1]{\subseteq^{#1}}
\newcommand{\xinclrev}[1]{\supseteq^{#1}}
\newcommand{\xstrictincl}[1]{\subset^{#1}}
\newcommand{\directtraceinclusion}{\xincl{\mathsf{di}}}
\newcommand{\directtraceinclusionrev}{\xinclrev{\mathsf{di}}}
\newcommand{\delayedtraceinclusion}{\xincl{\mathsf{de}}}
\newcommand{\delayedtraceinclusionrev}{\xinclrev{\mathsf{de}}}
\newcommand{\fairtraceinclusion}{\xincl{\mathsf f}}
\newcommand{\fairtraceinclusionrev}{\xinclrev{\mathsf f}}
\newcommand{\strictfairtraceinclusion}{\xstrictincl{\mathsf f}}
\newcommand{\bwdirecttraceinclusion}{\xincl{\mathsf{bw\textrm{-}di}}}
\newcommand{\bwdirecttraceinclusionrev}{\xinclrev{\mathsf{bw\textrm{-}di}}}
\newcommand{\accblindbwdirecttraceinclusion}{\bwfinincl}
\newcommand{\countingbwtraceinclusion}{\xincl{\mathsf{bw}\textrm{-}\mathsf{c}}}
\newcommand{\fixedwordsimulation}[1]{\xsim{\mathsf{fx\textrm-}#1}}
\newcommand{\directfixedwordsimulation}{\xsim{\mathsf{fx\textrm-di}}}
\newcommand{\delayedfixedwordsimulation}{\xsim{\mathsf{fx\textrm-de}}}
\newcommand{\delayedfixedwordsimulationrev}{\xsimrev{\mathsf{fx\textrm-de}}}
\newcommand{\fairfixedwordsimulation}{\xsim{\mathsf{fx\textrm-f}}}
\newcommand{\delayedwordsimulation}[1]{\xsim{\mathsf{fx\textrm-de}}_{#1}}
\newcommand{\languageinclusion}{\xincl {}}%{\subseteq^L}
\newcommand{\languageequivalence}{\approx}%{\subseteq^L}
\newcommand{\strictdirecttraceinclusion}{\xstrictincl{\mathsf{di}}}
\newcommand{\strictbwdirecttraceinclusion}{\xstrictincl{\mathsf{bw\textrm{-}di}}}
\newcommand{\strictlanguageinclusion}{\strictfairtraceinclusion}%{\subset^L}

\newcommand{\fwsim}{\xsim{\mathsf{fw}}}
\newcommand{\bwsim}{\xsim{\mathsf{bw}}}
\newcommand{\finincl}{\xincl{\mathsf{fw}}}
\newcommand{\bwfinincl}{\xincl{\mathsf{bw}}}
\newcommand{\strictfinincl}{\xstrictincl{\mathsf{fw}}}
\newcommand{\strictbwfinincl}{\xstrictincl{\mathsf{bw}}}
\newcommand{\fininclrev}{\xinclrev{\mathsf{fw}}}
\newcommand{\bwfininclrev}{\xinclrev{\mathsf{bw}}}

\newcommand{\brel}{\mathrel{R_b}}
\newcommand{\frel}{\mathrel{R_{\!f}}}
\newcommand{\prune}[2]{{\it Prune}({#1},{#2})}
\newcommand{\prunerel}{P}
\newcommand{\makeprunerel}[2]{\mathrel{\prunerel({#1},{#2})}}
\newcommand{\makeprunerelAB}[2]{\mathrel{\prunerel_{\A, \B}({#1},{#2})}}
\newcommand{\makeprunereltransient}[2]{\mathrel{\prunerel_t({#1},{#2})}}
\newcommand{\xprune}[3]{{\it Prune}({#1},{#2},{#3})}

\newcommand{\saturationrel}{S}
\newcommand{\makesaturationrel}[2]{\mathrel{\saturationrel({#1},{#2})}}
\newcommand{\saturate}[2]{{\it Sat}({#1},{#2})}

\newcommand{\id}{{\it id}} % Identity relation

\newcommand{\set}[1]{\{#1\}}
\newcommand{\st}{\; |\; }
\newcommand{\setof}[2]{\set{#1 \st #2}}

\newcommand{\strictxsim}[1]{\sqsubset^{#1}}
\newcommand{\strictxsimrev}[1]{\sqsupset^{#1}}
\newcommand{\strictdisim}{\strictxsim{\mathsf{di}}}
\newcommand{\strictdesim}{\strictxsim{\mathsf{de}}}
\newcommand{\strictdesimrev}{\strictxsimrev{\mathsf{de}}}
\newcommand{\strictfsim}{\strictxsim{\mathsf f}}
\newcommand{\strictbwsim}{\strictxsim{\mathsf{bw}}}
\newcommand{\strictbwdisim}{\strictxsim{\mathsf{bw\textrm-di}}}

\newcommand{\ksim}[1]{\sqsubseteq^{#1}}

\newcommand{\kxsim}[2]{\ksim{#1\textrm - #2}}

\newcommand{\kbwdisim}{\kxsim k {\mathsf{bw\textrm-di}}}

\newcommand{\ksimeq}[1]{\approx^{#1}}

\newcommand{\transksim}[1]{\preceq^{#1}}
\newcommand{\transksimrev}[1]{\succeq^{#1}}
\newcommand{\transksimx}{\preceq^{k\textrm{-}x}}
\newcommand{\transkbwsim}{\preceq^{k\textrm{-}\mathsf{bw}}}
\newcommand{\transkbwsimrev}{\succeq^{k\textrm{-}\mathsf{bw}}}
\newcommand{\transkbwdisim}{\preceq^{k\textrm{-}\mathsf{bw\textrm{-}di}}}
\newcommand{\transkbwdisimrev}{\succeq^{k\textrm{-}\mathsf{bw\textrm{-}di}}}
\newcommand{\accblindkbwsim}{\sqsubseteq^{k\textrm{-}\mathsf{bw}}}
\newcommand{\accblindtranskbwsim}{\preceq^{k\textrm{-}\mathsf{bw}}}
\newcommand{\countingkbwsim}{\sqsubseteq^{k\textrm{-}\mathsf{bw}\textrm{-}\mathsf{c}}}
\newcommand{\countingtranskbwsim}{\preceq^{k\textrm{-}\mathsf{bw}\textrm{-}\mathsf{c}}}
\newcommand{\transkdisim}{\preceq^{k\textrm{-}\mathsf{di}}}
\newcommand{\transkdisimrev}{\succeq^{k\textrm{-}\mathsf{di}}}
\newcommand{\transkdesim}{\preceq^{k\textrm{-}\mathsf{de}}}
\newcommand{\transkdesimrev}{\succeq^{k\textrm{-}\mathsf{de}}}
\newcommand{\transkfsim}{\preceq^{k\textrm{-}\mathsf f}}

\newcommand{\transkdisimnumber}[1]{\preceq^{{#1}\textrm{-}\mathsf{di}}}
\newcommand{\transkdesimnumber}[1]{\preceq^{{#1}\textrm{-}\mathsf{de}}}

\newcommand{\stricttransksimx}{\prec^{k\textrm{-}x}}
\newcommand{\stricttranskbwsim}{\prec^{k\textrm{-}\mathsf{bw}}}
\newcommand{\stricttranskbwdisim}{\prec^{k\textrm{-}\mathsf{bw\textrm{-}di}}}
\newcommand{\stricttranskdisim}{\prec^{k\textrm{-}\mathsf{di}}}

\newcommand{\stricttranskfsim}{\prec^{k\textrm{-}\mathsf f}}

\newcommand{\cprex}[2]{\mathsf{CPre}^{#1}(#2)}
\newcommand{\cpredi}[1]{\cprex {\mathrm{di}} {#1}}

\newcommand{\cprebwdi}[1]{\cprex {\mathrm{bw\textrm{-}di}} {#1}}

\newcommand{\cprelong}[3]{\mathsf{CPre}(#1, #2, #3)}
\newcommand{\cpreone}[2]{\mathsf{CPre}^1(#1, #2)}
\newcommand{\cpretwo}[2]{\mathsf{CPre}^2(#1, #2)}

\newcommand{\tickOK}{\checkmark}
\newcommand{\tickNO}{\times}
\newcommand{\NA}{-}

\newcommand{\defeq}{\;\coloneqq\;}

\tikzset{
	smallstate/.style={state,
		circle, 
		minimum size=4mm, 
		font=\itshape
	},
	initial text=,
	>=stealth',
	>/.style={ultra thick}
}

\title[Efficient Reduction of Automata and Language Inclusion]
{Efficient reduction of nondeterministic automata\\ with application to
  language inclusion testing
}

\author[L.~Clemente]{Lorenzo Clemente\rsuper a}
\address{{\lsuper a}University of Warsaw, 
     Faculty of Mathematics, Informatics and Mechanics,
     Banacha 2,
     02-097 Warszawa,
     Poland
 }

\author[R.~Mayr]{Richard Mayr\rsuper b}
\address{{\lsuper b}University of Edinburgh, School of Informatics, LFCS,
10 Crichton Street, Edinburgh EH89AB, UK}


\begin{document}
%%%%%%%%%%%%%%%%%%%%%%%%%%%%%%%%%%%%%%%%%%%%%%%%%%%%%%%%%%%%%%%%%%%%%%%%%%%%%%%

\newcommand\blfootnote[1]{%
  \begingroup
  \renewcommand\thefootnote{}\footnote{#1}%
  \addtocounter{footnote}{-1}%
  \endgroup
}
\blfootnote{Extended version of results presented at POPL 2013 \cite{CM:POPL2013}.}

\begin{abstract}
\noindent
We present efficient algorithms to reduce the size of nondeterministic 
B\"uchi word automata (NBA) and nondeterministic finite word automata (NFA),
while retaining their languages.
Additionally, we describe methods to solve PSPACE-complete automata problems
like language universality, equivalence, and inclusion for much larger instances 
than was previously possible ($\ge 1000$ states instead of 10-100). 
This can be used to scale up applications of automata in formal verification tools
and decision procedures for logical theories.

The algorithms are based on new techniques for removing transitions (pruning)
and adding transitions (saturation), as well as extensions of classic
quotienting of the state space. 
These techniques use criteria based on combinations of backward and forward trace
inclusions and simulation relations.
Since trace inclusion relations are themselves PSPACE-complete,
we introduce \emph{lookahead simulations} as good polynomial time computable approximations thereof.

Extensive experiments show that the average-case time complexity of our algorithms
scales slightly above quadratically. (The space complexity is worst-case
quadratic.)
The size reduction of the automata depends very much on
the class of instances, but our algorithm consistently reduces the size far
more than all previous techniques.
We tested our algorithms on NBA derived from
LTL-formulae, NBA derived from mutual exclusion protocols and
many classes of random NBA and NFA,
and compared their performance to the well-known automata tool GOAL \cite{GOAL_survey_paper}.
\end{abstract}

\maketitle

\section{Introduction}\label{sec:introduction}

\noindent
Nondeterministic B\"uchi automata (NBA) are a fundamental data structure to represent and manipulate $\omega$-regular languages \cite{Handbook_B}.
They appear in many automata-based formal software verification methods, as well as in
decision procedures for logical theories.
For example, in LTL software model checking 
\cite{Holzmann:Spinbook,optimizing:concur2000}, temporal logic
specifications are converted into NBA. 
In other cases, different versions of a program (obtained by abstraction or
refinement of the original) are translated into automata whose languages are then
compared. Testing the conformance of an implementation with its requirements 
specification thus reduces to a language inclusion problem.
Another application of NBA in software engineering is program termination
analysis by the size-change termination method \cite{Lee:SCT2001,seth:buchi}. 
Via an abstraction of the effect of program operations on
data, the termination problem can often be reduced to a
language inclusion problem between two derived NBA. 

Our goal is to improve the efficiency and scalability of automata-based formal software verification methods.
Our contribution is threefold:
We describe a very effective \emph{automata reduction} algorithm,
which is based on novel, efficiently computable \emph{lookahead simulations},
and we conduct an extensive \emph{experimental evaluation} of our reduction algorithm.

This paper is partly based on results presented at POPL'13 \cite{CM:POPL2013},
but contains several large parts that have not appeared previously.
While \cite{CM:POPL2013} only considered nondeterministic B\"uchi automata
(NBA), we additionally present corresponding results on 
nondeterministic finite automata (NFA).
We also present more extensive empirical results for both NBA and NFA
(cf.~ Sec.~\ref{sec:experiments}).
Moreover, we added a section on the new saturation technique (cf.~ Sec.~\ref{sec:extensions}).
Finally, we added some notes on the implementation (cf.~ Sec.~\ref{sec:implementation}).

\subsection{Automata reduction.}

We propose a novel, efficient, practical, and very effective algorithm to reduce the size of automata,
in terms of both states and transitions.
It is well-known that, in general, there are several non-isomorphic nondeterministic automata of minimal size recognizing a given language,
and even testing the minimality of the number of states of a given automaton is PSPACE-complete \cite{ravikumar:hard:1991}.
Instead, our algorithm produces a smaller automaton recognizing the same language,
though not necessarily one with the absolute minimal possible number of states,
thus avoiding the complexity bottleneck.
The reason to perform reduction is that smaller automata are in general more efficient to handle in a subsequent computation.
Thus, there is an algorithmic tradeoff between the effort for the reduction
and the complexity of the problem later considered for this automaton.
If only computationally easy algorithmic problems are considered,
like reachability or emptiness (which are solvable in NLOGSPACE),
then extensive reduction does not pay off
since in these cases it is faster to solve the initial problem directly.
Instead, the main applications are the following.
\begin{enumerate}

	\item
	PSPACE-complete automata problems like language universality, equivalence, and inclusion \cite{kupfermanvardi:fair_verification}.
	Since exact algorithms are exponential for these problems,
	one should first reduce the automata as much as possible before applying them.

	\item
	LTL model checking \cite{Holzmann:Spinbook},
	where one searches for loops in a graph
	that is the {\em product} of a large system specification with an NBA derived from an LTL-formula.
	Smaller automata often make this easier,
	though in practice it also depends on the degree of nondeterminism \cite{Sebastiani-Tonetta:2003}.
	Our reduction algorithm, based on transition pruning techniques,
	yields automata that are not only smaller, but also sparser
	(fewer transitions per state, on average),
	and thus contain less nondeterministic branching.
	
	\item
	Procedures that combine and modify automata repeatedly.
	Model checking algorithms and automata-based decision procedures for logical theories (cf.~the TaPAS tool \cite{Talence:Presburger})
	compute automata products, unions, complements, projections, etc., and thus the sizes of automata grow rapidly.
	Another example is in the use of automata for the reachability analysis of safe Petri nets \cite{Rathke:RP2014}.
	Thus, it is important to intermittently reduce the automata to keep their size manageable.
	
\end{enumerate}

%other reduction algorithms \cite{optimizing:concur2000,etessami:etal:fairsimulations:05,piterman:generalized06,buchiquotient:ICALP11}

Our reduction algorithm combines the following techniques:
\begin{itemize}
\item
The removal of dead states. These are states that trivially do not contribute to
the language of the automaton, either because they cannot be reached from any
initial state or because no accepting loop in the NBA (resp.~no accepting state
in the NFA) is reachable from them.
\item
Quotienting. Here one finds a suitable equivalence relation on the set of
states and quotients w.r.t.~it, i.e., one merges each equivalence class 
into a single state.
\item
\emph{Transition pruning} (i.e., removing transitions) and \emph{transition saturation}
(i.e., adding transitions), using suitable criteria such that the language of 
the automaton is preserved.
\end{itemize}
The first technique is trivial
and the second one is well-understood \cite{etessami:etal:fairsimulations:05,buchiquotient:ICALP11}.
Here, we investigate thoroughly transition pruning and
transition saturation.

For pruning, the idea is that certain transitions can be removed,
because other `better' transitions remain.
The `better' criterion compares the source and target states of transitions
w.r.t.~certain semantic preorders, e.g., forward and backward simulations and trace inclusions.
We provide a complete picture of which combinations of relations are correct to use for pruning.
Pruning transitions reduces not only the number of transitions, but also, indirectly, the number of states.
By removing transitions, some states may become dead,
%in the sense that they are unreachable from any initial state,
%or that it is impossible to reach any accepting loop (in NBA) or any accepting state (in NFA) from them.
and can thus be removed from the automaton without changing its language.
The reduced automata are generally much sparser than the originals
(i.e., use fewer transitions per state and less nondeterministic branching),
which yields additional performance advantages in subsequent computations.

Dually, for saturation, the idea is that certain transitions can be added,
because other `better' transitions are already present.
Again, the `better' criterion relies on comparing the source/target states 
of the transitions w.r.t.~semantic preorders like forward and backward
simulations and trace inclusions.
We provide a complete picture of which combinations of relations are correct to use for saturation.
Adding transitions does not change the number of states, but it may pave the
way for further quotienting that does.
Moreover, adding some transitions might allow the subsequent pruning of other
transitions, and the final result might even have fewer transitions than
before.
It often happens, however, that there is a tradeoff between the numbers of states
and transitions. 

\begin{figure}

	\begin{tikzpicture}[on grid, node distance=1.5cm and 1.5cm]
		\tikzstyle{vertex} = [smallstate]

		\path node [vertex, initial]		(p0) {$p_0$};
		\path node [vertex]					(p1) [right = of p0] {$p_1$};
		\path node [vertex]					(p2) [right = of p1] {$p_2$};
		\path node []						(dots) [right = of p2] {$\cdots$};
		\path node [vertex]					(pn) [right = of dots] {$p_n$};
		\path node [vertex, accepting]		(p) [right = of pn] {$p$};
		
		\path[->]

			(p0) edge [loop above] node {$a, b$} ()
			(p0) edge node [above] {$a$} (p1)
			(p1) edge node [above] {$a, b$} (p2)
			(p2) edge node [above] {$a, b$} (dots)
			(dots) edge node [above] {$a, b$} (pn)
			(pn) edge node [above] {$\$$} (p)
			(p) edge [loop above] node {$\$$} ();

	\end{tikzpicture}

	\caption{Family of NBA accepting languages $L_n = \set{a, b}^* a \set{a, b}^{n-1} \$^\omega$
	for which the minimal WDBA has size $\Omega(2^n)$.}
	
	\label{fig:WDBA}

\end{figure}
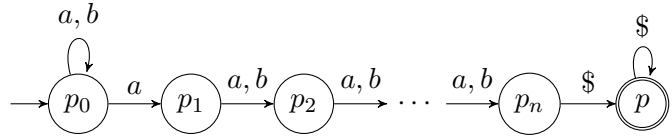
%%% Local Variables:
%%% mode: latex
%%% TeX-master: "ROOT.tex"
%%% End:

Finally, it is worth mentioning that the minimization problem can sometimes be solved efficiently
if one considers minimization \emph{within} a restricted class of languages.
For instance, for the class of \emph{weak deterministic B\"uchi languages}
(a strict subclass of the $\omega$-regular languages)
it is well-known that given a weak deterministic B\"uchi automaton (WDBA) one can find in time $O(n\log n)$
a minimal equivalent automaton in the same class \cite{Loding:IPL:2001}
(essentially by applying Hopcroft's DFA minimization algorithm \cite{Hopcroft:1971:nlogn}).
However, it is possible that a weak deterministic language admits only large WDBA,
but succinct NBA; cf.~ Fig.~\ref{fig:WDBA} (this is similar to what happens for DFA vs. NFA over finite words).
Thus, minimizing a WDBA in the class of WDBA
and minimizing a WDBA in the larger class of NBA are two distinct problems.
Since in this paper we consider size reduction of NBA (and thus WDBA) in the larger class of all NBA,
our method and the one of \cite{Loding:IPL:2001} are incomparable.

\subsection{Lookahead simulations}

Simulation preorders play a central role in automata reduction via pruning, saturation and quotienting,
because they provide PTIME-computable under-approximations of the PSPACE-hard trace inclusions.
However, the quality of the approximation is insufficient in many practical examples.
Multipebble simulations \cite{etessami:hierarchy02} yield better under-approximations of trace inclusions;
while theoretically they can be computed in PTIME for a fixed number of pebbles,
in practice they are not easily computed.

We introduce \emph{lookahead simulations}
as an efficient and practical method to compute good under-approximations of trace inclusions and multipebble simulations.
For a fixed lookahead, lookahead simulations are computable in PTIME,
and it is correct to use them instead of the more expensive trace inclusions and multipebble simulations.
Lookahead itself is a classic concept,
which has been used in parsing and many other areas of computer science,
like in the uniformization problem of regular relations \cite{HoschLandweber:ICALP:1972},
in the composition of e-services
(under the name of lookahead delegators \cite{Gerede:2004:ACE:1035167.1035203,RavikumarSantean:Lookahead:SOFSEM:2007,BrzozowskiSantean:Predictable:TCS:2009,LodingRepke:LookaheadDelegators:FSTTCS13}), % are used as a mechanism to accept words in NFAs.
and in infinite games \cite{HoltmannKaiserThomas,FridmanLodingZimmermann:Lookahead:CSL:2011,KleinZimmermann:Lookahead:ICALP:2015,KleinZimmermann:Lookahead:CSL:2015}.
However, lookahead can be defined in many different variants.
Our contribution is to identify and formally describe the lookahead-variant for simulation preorders
that gives the optimal compromise between efficient computability
and maximizing the sizes of the relations; cf.~Sec.~\ref{sec:lookahead}.
%
% Maybe move this to later
From a practical point of view, we use degrees of lookahead ranging from 4 to
25 steps, depending on the size and shape of the automata.
Our experiments show that even a moderate lookahead often yields much larger
approximations of trace-inclusions and multipebble simulations 
than normal simulation preorder.
Notions very similar to the ones we introduce are discussed in \cite{lange:lookahead:2013}
under the name of \emph{multi-letter simulations}
and \emph{buffered simulations} \cite{LangeBuffered:2014};
% gives complexity results for computing these relations for unbounded lookahead.
cf.~ Remark~\ref{rem:lookahead:related} for a comparison of multi-letter and buffered simulations w.r.t.~lookahead simulations.

\subsection{Experimental results}

We performed an extensive experimental evaluation of our techniques based on lookahead simulations
on tasks of automata reduction and language universality/inclusion testing.
(The raw data of the experiments is stored together with the arXiv version of this
paper \cite{CM:arxiv2018}.)

\subsubsection*{Automata reduction.}

We applied our reduction algorithm on automata of up-to $20000$ states.
These included 1) random automata according to the Tabakov-Vardi model \cite{tabakov:model},
2) automata obtained from LTL formulae, and
3) real-world mutual exclusion protocols.
The empirically determined average-case time complexity on random automata is slightly above quadratic,
while the (never observed) worst-case complexity is $O(n^4)$.
The worst-case space complexity is quadratic. 
Our algorithm reduces the size of automata much more strongly, on average,
than previously available practical methods as implemented in the popular GOAL
automata tool \cite{GOAL_survey_paper}.
However, the exact advantage varies, depending on the type of instances; cf.~Sec.~\ref{sec:experiments}.
For example, consider random automata with 100--1000 states,
binary alphabet and varying transition density ${\it td}$. 
Random automata with ${\it td}=1.4$ cannot be reduced much by {\em any method}.
The only substantial effect is achieved by the trivial removal of dead states which, on
average, yields automata of $78\%$ of the original size.
On the other hand, for ${\it td}=1.8,\dots,2.2$, the best previous reduction methods
yielded automata of $85\%$--$90\%$ of the original size on average,
while our algorithm yielded automata of $3\%$--$15\%$ of the original size on average.

\subsubsection*{Language universality/inclusion.}

Language universality and inclusion of NBA/NFA are PSPACE-complete problems \cite{kupfermanvardi:fair_verification},
but many practically efficient methods have been developed 
\cite{dill:inclusion:1992,doyen:raskin:antichains,Pit06,Abdulla:whensimulation2010,seth:buchi,seth:efficient,abdulla:simulationsubsumption,Rabit_CONCUR2011}.
Still, these all have exponential worst-case time complexity and do not scale well.
Typically they are applied to automata with 15--100 states
(unless the automaton has a particularly simple structure),
and therefore one should first reduce the automata before applying these exact exponential-time methods.

Even better, already the {\em polynomial time} reduction algorithm alone
can solve many instances of the PSPACE-complete universality, equivalence, and inclusion problems.
E.g., an automaton might be reduced to the trivial universal automaton, thus
witnessing language universality,
or when one checks inclusion of two automata,
it may compute a small (polynomial size) certificate for language inclusion in
the form of a (lookahead-)simulation.
Thus, the complete {\em exponential time} methods above need only be invoked in a minority of the cases,
and on much smaller instances.
This allows to scale language inclusion testing to much larger instances
(e.g., automata with $\ge 1000$ states) which are beyond previous methods.

\subsection{Nondeterministic finite automata.}

We present our methods mostly in the framework of nondeterministic B\"uchi
automata (NBA),
but they directly carry over to the simpler case of nondeterministic finite-word automata (NFA).
The main differences are the following:
\begin{itemize}
	\item
          Since NFA accept finite words, it matters in exactly which step an
          accepting state is reached (unlike for NBA where the acceptance
          criterion is to visit accepting states infinitely often).
          Therefore, lookahead-simulations for NFA need to treat accepting states
          in a way which is more restrictive than for NBA.
          Thus, in NFA, one is limited to a smaller range of semantic preorders/equivalences,
	  namely direct and backward simulations 
          (and the corresponding multipebble simulations, lookahead
          simulations and trace inclusions),
          while more relaxed notions (like delayed and fair simulations) can
          be used for NBA.
	\item
          On the other hand, unlike NBA, an NFA can always be transformed into an equivalent
		NFA with just one accepting state without any outgoing transitions
		(unless the language contains the empty word).
		This special form makes it much easier to compute good approximations of direct and backward trace inclusion,
		which greatly helps in the NFA reduction algorithm.
\end{itemize}

{\small
\begin{table}
	\begin{center}% \small
		\begin{tabular}{l|c|cl|cl|cl}
			relations on NBA & complexity
				& \multicolumn{2}{|c|}{quotienting}
				& \multicolumn{2}{|c|}{inclusion}
				& \multicolumn{2}{|c}{pruning$\mbox{}^{(1)}$} \\
		 \hline
			direct simulation $\disim$ & PTIME
				& $\tickOK$ & \cite{somenzi:efficient,optimizing:concur2000}
				& $\tickOK$ & \cite{dill:inclusion:1992}
				& $\tickOK$ & \cite{simulationminimization:03}, Thm.~\ref{thm:bwtrace-fwsim}
			\\
			delayed simulation $\desim$ & PTIME
				& $\tickOK$ & \cite{etessami:etal:fairsimulations:05}
				& $\tickOK$ & \cite{etessami:etal:fairsimulations:05}
				& $\tickNO$ & Fig.~\ref{fig:desim:not:GFP}
                        \\
			fair simulation $\fsim$ & PTIME
				& $\tickNO$ & \cite{etessami:etal:fairsimulations:05}
				& $\tickOK$ & \cite{fairsimulation:02}
				& $\tickNO$ & Fig.~\ref{fig:desim:not:GFP}
			\\
			backward direct sim. $\bwdisim$ & PTIME
				& $\tickOK$ & \cite{somenzi:efficient}
				& $\tickOK$ & Thm.~\ref{lem:bwincl_GFI}
				& $\tickOK$ & Thm.~\ref{thm:bwsim-fwtrace}
			\\
			direct trace inclusion $\directtraceinclusion$ & PSPACE
				& $\tickOK$ & \cite{etessami:hierarchy02}
				& $\tickOK$ & obvious
				& $\tickOK$ & Thm.~\ref{thm:prune_id_strictdirecttraceinclusion}, \ref{thm:bwsim-fwtrace}
			\\
			delayed trace inclusion $\delayedtraceinclusion$ & PSPACE
				& $\tickNO$ & Fig.~\ref{fig:decont:not:GFQ} \cite{buchiquotient:ICALP11}
				& $\tickOK$ & obvious
				& $\tickNO$ & cf.~ Thm.~\ref{thm:prune_transient}
                        \\
			fair trace inclusion $\fairtraceinclusion$ & PSPACE
				& $\tickNO$ & Fig.~\ref{fig:decont:not:GFQ} \cite{buchiquotient:ICALP11}
				& $\tickOK$ & obvious
				& $\tickNO$ & cf.~ Thm.~\ref{thm:prune_transient}
			\\
			direct fixed-word sim. $\directfixedwordsimulation$ & PSPACE
				& $\tickOK$ & Lem.~\ref{lem:GFQ-delayedfixedwordsimulation}\cite{buchiquotient:ICALP11}
				& $\tickOK$ & obvious
				& $\tickOK$ & by $\directtraceinclusion$ % direct trace incl.
			\\
			delayed fixed-word sim. $\delayedfixedwordsimulation$ & PSPACE
				& $\tickOK$ & Lem.~\ref{lem:GFQ-delayedfixedwordsimulation}\cite{buchiquotient:ICALP11}
				& $\tickOK$ & obvious
				& $\tickNO$ & by delayed sim.
			\\
			fair fixed-word sim. $\fairfixedwordsimulation$ & PSPACE
				& $\tickNO$ & by $\fairtraceinclusion$ % fair trace incl.
				& $\tickOK$ & obvious
				& $\tickNO$ & by delayed sim.
			\\
			bwd.~direct trace incl. $\bwdirecttraceinclusion$ & PSPACE 
				& $\tickOK$ & Thm.~\ref{lem:bwincl_GFQ}
				& $\tickOK$ & Thm.~\ref{lem:bwincl_GFI}
				& $\tickOK$ & Thm.~\ref{thm:bwtrace-id}, \ref{thm:bwtrace-fwsim}
			\\
			direct lookahead sim. $\transkdisim$ & PTIME$\mbox{}^{(2)}$
				& $\tickOK$ & Lemma~\ref{lem:lookahead_sim_GFI_GFQ}
				& $\tickOK$ & Lemma~\ref{lem:lookahead_sim_GFI_GFQ}
				& $\tickOK$ & Sec.~\ref{subsec:heavyandlight-Buchi}
			\\
			delayed lookahead sim. $\transkdesim$ & PTIME$\mbox{}^{(2)}$
				& $\tickOK$ & Lemma~\ref{lem:lookahead_sim_GFI_GFQ}
				& $\tickOK$ & Lemma~\ref{lem:lookahead_sim_GFI_GFQ}
				& $\tickOK$ & Sec.~\ref{subsec:heavyandlight-Buchi}
			\\
			fair lookahead sim. $\transkfsim$ & PTIME$\mbox{}^{(2)}$
				& $\tickNO$ & by fair sim.
				& $\tickOK$ & Lemma~\ref{lem:lookahead_sim_GFI_GFQ}
				& $\tickNO$ & Sec.~\ref{subsec:heavyandlight-Buchi}
		 	\\
		 	bwd.~di.~lookahead sim. $\transkbwdisim$ & PTIME$\mbox{}^{(2)}$
				& $\tickOK$ & by $\bwdirecttraceinclusion$
				& $\tickOK$ & by $\bwdirecttraceinclusion$
				& $\tickOK$ & by $\bwdirecttraceinclusion$
			\\
		 	\hline
			\hline
			relations on NFA
			\\
			\hline
			forward direct sim. $\disim$ & PTIME
				& $\tickOK$ & Thm.~\ref{thm:GFQ:NFA}
				& $\tickOK$ & Thm.~\ref{thm:GFI:NFA}
				& $\tickOK$ & Thm.~\ref{thm:bwsim-fwtrace:NFA}, \ref{thm:bwtrace-fwsim:NFA}
			\\
			bwd.~finite-word sim. $\bwsim$ & PTIME
				& $\tickOK$ & Thm.~\ref{thm:GFQ:NFA}
				& $\tickOK$ & Thm.~\ref{thm:GFI:NFA}
				& $\tickOK$ & Thm.~\ref{thm:bwsim-fwtrace:NFA}, \ref{thm:bwtrace-fwsim:NFA}
			\\
			fwd.~finite trace incl. $\finincl$ & PSPACE
				& $\tickOK$ & Thm.~\ref{thm:GFQ:NFA}
				& $\tickOK$ & Thm.~\ref{thm:GFI:NFA}
				& $\tickOK$ & Thm.~\ref{thm:prune_id_strictdirecttraceinclusion:NFA}--\ref{thm:bwtrace-fwsim:NFA}
			\\
			bwd.~finite trace incl. $\bwfinincl$ & PSPACE
				& $\tickOK$ & Thm.~\ref{thm:GFQ:NFA}
				& $\tickOK$ & Thm.~\ref{thm:GFI:NFA}
				& $\tickOK$ & Thm.~\ref{thm:prune_id_strictdirecttraceinclusion:NFA}--\ref{thm:bwtrace-fwsim:NFA}
			\\
			fwd.~di.~lookahead sim. $\transkdisim$ & PTIME$\mbox{}^{(2)}$
				& $\tickOK$ & Sec.~\ref{subsec:heavyandlight-NFA}
				& $\tickOK$ & Sec.~\ref{subsec:heavyandlight-NFA}
				& $\tickOK$ & Sec.~\ref{subsec:heavyandlight-NFA}
			\\
			bwd.~lookahead sim. $\transkbwsim$ & PTIME$\mbox{}^{(2)}$
				& $\tickOK$ & Sec.~\ref{subsec:heavyandlight-NFA}
				& $\tickOK$ & Sec.~\ref{subsec:heavyandlight-NFA}
				& $\tickOK$ & Sec.~\ref{subsec:heavyandlight-NFA}
		\end{tabular}
	\end{center}
	
	\caption{
		Summary of old and new results for simulation-like relations on NBA and NFA.
		$(1)$ For pruning, cf.~ also Table~\ref{fig:GFP_relations} in Sec.~\ref{sec:pruning:NBA} for NBA,
		and Sec.~\ref{sec:pruning:NFA} for NFA.
		$(2)$ PTIME for fixed lookahead.
	}
	
	\label{fig:summary}
\end{table}
}

\subsection*{Outline of the paper.}

A summary of old and new results about simulation-like preorders
as used in inclusion checking, quotienting, and pruning transitions
can be found in Table~\ref{fig:summary}.

The rest of the paper is organized as follows.
In Sec.~\ref{sec:preliminaries}, we define basic notation for automata and
languages.
Sec.~\ref{sec:quotienting} introduces basic semantic preorders and
equivalences between states of automata and considers quotienting methods,
while Sec.~\ref{sec:language_inclusion} shows which preorders witness
language inclusion.
In Sec.~\ref{sec:pruning}, we present the main results on transition
pruning.
Lookahead simulations are introduced in Sec.~\ref{sec:lookahead}
and used in the algorithms for automata reduction and language inclusion
testing in Sections~\ref{sec:heavyandlight} and \ref{sec:inclusion},
respectively.
These algorithms are empirically evaluated in Sec.~\ref{sec:experiments}.
In Sec.~\ref{sec:extensions} we describe and evaluate an extended reduction
algorithm that additionally uses transition saturation methods.
Sec.~\ref{sec:implementation} describes some algorithmic optimizations 
in the implementation, and 
Sec.~\ref{sec:conclusion} contains a summary and directions for future
work.

%%% Local Variables:
%%% mode: latex
%%% TeX-master: "ROOT.tex"
%%% End:

\section{Preliminaries}\label{sec:preliminaries}

A \emph{preorder} $R$ is a reflexive and transitive relation,
a \emph{partial order} is a preorder which is antisymmetric ($xRy \wedge yRx \Rightarrow x=y$),
and a \emph{strict partial order} is an irreflexive ($\neg xRx$), asymmetric ($xRy \Rightarrow \neg yRx$), and transitive relation.
We often denote preorders by $\sqsubseteq$, and when we do so,
with $\sqsubset$ we denote its strict version,
i.e., $x \sqsubset y$ if $x \sqsubseteq y$ and $y \not\sqsubseteq x$;
we follow a similar convention for $\subseteq$.

A \emph{nondeterministic B\"uchi automaton (NBA)} $\A$ is a tuple $(\Sigma, Q, I, F, \delta)$
where $\Sigma$ is a finite alphabet, $Q$ is a finite set of states,
$I \subseteq Q$ is the set of \emph{initial} states,
$F \subseteq Q$ is the set of \emph{accepting} states,
and $\delta \subseteq Q \times \Sigma \times Q$ is the transition relation.
We write $p \goesto \symb q$ for $(p, \symb, q) \in \delta$.
A state of a B\"uchi automaton is \emph{dead} if either it is not reachable from any initial state,
or it cannot reach any accepting loop (i.e., a loop that contains at least one
accepting state).
In our simplification techniques, we always remove dead states, since this
does not affect the language of the automaton.
To simplify the presentation, we assume that automata are \emph{forward and backward complete}, i.e.,
for any state $p \in Q$ and symbol $\symb \in \Sigma$,
there exist states $q_0, q_1 \in Q$ s.t.~$q_0 \goesto \symb p \goesto \symb q_1$.
Every automaton can be converted into an equivalent complete one by adding
at most two states and at most $2 \cdot (|Q| + 2) \cdot |\Sigma|$ transitions.%
\footnote{
For efficiency reasons, our implementation works directly on incomplete automata.
Completeness is only assumed to simplify the technical development.}
A B\"uchi automaton $\A$ describes a set of infinite words (its language), i.e., a subset of $\Sigma^\omega$.
An \emph{infinite trace} of $\A$ on an infinite word $w = \symb_0\symb_1 \cdots \in \Sigma^\omega$ (or \emph{$w$-trace})
\emph{starting} in a state $q_0 \in Q$ is an infinite sequence of transitions 
$\pi = q_0 \goesto {\symb_0} q_1 \goesto {\symb_1} \cdots$.
Similarly, a \emph{finite trace} on a finite word $w = \symb_0\symb_1 \cdots \symb_{m-1} \in \Sigma^*$ (or \emph{$w$-trace})
starting in a state $q_0 \in Q$ and \emph{ending} in a state $q_m \in Q$
is a finite sequence of transitions $\pi = q_0 \goesto {\symb_0} q_1 \goesto {\symb_1} \cdots \goesto {\symb_{m-1}} q_m$.
By convention, a finite trace over the empty word $\varepsilon$ is just a single state $\pi = q$ (where the trace both starts and ends).
For an infinite trace $\pi$ and index $i \geq 0$,
we denote by $\prefix \pi i$ the finite prefix trace $\prefix \pi i = q_0 \goesto {\symb_0} \cdots \goesto {\symb_{i-1}} q_i$,
and by $\suffix \pi i$ the infinite suffix trace $\suffix \pi i = q_i \goesto {\symb_i} q_{i+1} \goesto {\symb_{i+1}} \cdots$.
A finite or infinite trace is \emph{initial} if it starts in an initial state $q_0 \in I$,
and a finite trace is \emph{final} if it ends in an accepting state $q_m \in F$.
A trace is \emph{fair} if it is infinite and $q_i \in F$ for infinitely many $i$'s.
A transition is \emph{transient} if it appears at most once in any trace of the automaton.
%(Equivalently, a transient transition does not belong to a strongly-connected component containing an accepting state reachable from an initial state.)
%
The \emph{language of an NBA $\A$} is $\lang \A = \{w \in \Sigma^\omega \mid \mbox{$\A$ has an initial and fair trace on $w$} \}$.

A \emph{nondeterministic finite automaton (NFA)} $\A = (\Sigma, Q, I, F, \delta)$ has the same syntax as an NBA,
and all definitions from the previous paragraph carry over to NFA.
(Sometimes, accepting states in NBA are called \emph{final} in the context of NFA.)
However, since NFA recognize languages of finite words, their semantics is different.
The language of an NFA $\A$ is thus defined as $\lang \A = \{w \in \Sigma^* \mid \mbox{$\A$ has an initial and final trace on $w$} \}$.

When the distinction between NBA and NFA is not important, we just call $\A$ an automaton.
Given two automata $\A$ and $\B$ we write $\A \languageinclusion \B$ if
$\lang \A \subseteq \lang \B$ and
$\A \languageequivalence \B$ if $\lang \A = \lang \B$.

%%% Local Variables:
%%% mode: latex
%%% TeX-master: "ROOT.tex"
%%% End:

\section{Quotienting reduction techniques}
\label{sec:quotienting}

An interesting problem is how to simplify an automaton while preserving
its semantics, i.e., its language. Generally, one tries to reduce the number
of states/transitions. 
This is useful because the complexity of decision procedures usually depends
on the size of the input automata.
A classical operation for reducing the number of states of an automaton is that of quotienting,
where states of the automaton are identified according to a given equivalence, and transitions are projected accordingly.
Since in practice we obtain quotienting equivalences from suitable preorders,
we directly define quotienting w.r.t.~a preorder.
In the rest of the section, fix an automaton $\A = (\Sigma, Q, I, F, \delta)$,
and let $\sqsubseteq$ be a preorder on $Q$,
with induced equivalence $\approx \defeq (\sqsubseteq \cap \sqsupseteq)$.
Given a state $q \in Q$, we denote by $[q]$ its equivalence class w.r.t.~$\approx$ (which is left implicit for simplicity),
and, for a set of states $P \subseteq Q$, $[P]$ is the set of equivalence classes $[P] = \{ [p] \st p \in P \}$.

\begin{definition}
	The \emph{quotient} of $\A$ by $\sqsubseteq$ is
	$\A/\!\sqsubseteq = (\Sigma, [Q], [I], [F], \delta')$,
	where transitions are induced element-wise as
	$\delta' = \{([q_1],\symb,[q_2]) \st \exists q_1' \in [q_1], q_2' \in [q_2] \cdot (q_1',\symb,q_2') \in \delta\}$.
\end{definition}

\noindent
Clearly, every trace $q_0 \goesto {\symb_0} q_1 \goesto {\symb_1} \cdots$ in $\A$
immediately induces a corresponding trace $[q_0] \goesto {\symb_0} [q_1] \goesto {\symb_1} \cdots$ in $\A/\!\sqsubseteq$,
which is fair/initial/final if the former is fair/initial/final, respectively.
Consequently, $\A\; \languageinclusion\; (\A/\!\sqsubseteq)$ for \emph{any} preorder $\sqsubseteq$.
If, additionally, $(\A/\!\sqsubseteq) \; \languageinclusion\; \A$,
then we say that the preorder $\sqsubseteq$ is \emph{good for quotienting} (GFQ).
\begin{definition}
	A preorder $\sqsubseteq$ is \emph{good for quotienting} (GFQ) if
	$\A \languageequivalence \A/\!\sqsubseteq$.
\end{definition}
%
% Thus, GFQ preorders give a sufficient condition for quotienting to preserve the language.
%
GFQ preorders are downward closed (since a smaller preorder induces a smaller
equivalence, which quotients `less').
We are interested in finding coarse and efficiently computable GFQ preorders for NBA and NFA.
Classical examples are given by forward simulation relations (Sec.~\ref{sec:simulations})
and forward trace inclusions (Sec.~\ref{sec:trace_inclusions}),
which are well known GFQ preorders for NBA.
A less known GFQ preorder for NBA is given by their respective backward variants (Sec.~\ref{sec:backward}).
For completeness, we also consider suitable simulations and trace inclusions for NFA (Sec.~\ref{sec:simulations:finitewords}).
In Sec.~\ref{sec:language_inclusion}, the previous preorders are applied to language inclusion for both NBA and NFA.
In Sec.~\ref{sec:pruning}, we present novel language-preserving transition pruning techniques based on simulations and trace inclusions.
While simulations are efficiently computable, e.g., in PTIME, trace inclusions are PSPACE-complete.
In Sec.~\ref{sec:lookahead}, we present \emph{lookahead simulations},
which are novel efficiently computable GFQ relations coarser than simulations. % and included in the corresponding trace inclusions.

\subsection{Forward simulation relations}
\label{sec:simulations}

Forward simulation \cite{Park:Simulation:1981,Milner:communication:1989}
is a binary relation on the states of $\A$;
it relates states whose behaviors are step-wise related,
which allows one to reason about the internal structure of automaton $\A$---%
i.e., \emph{how} a word is accepted, and not just \emph{whether} it is accepted.
Formally, simulation between two states $p_0$ and $q_0$ can be described in
terms of a game between two players, Spoiler (he) and Duplicator (she),
where the latter wants to prove that $q_0$ can step-wise mimic any behavior of $p_0$, and the former wants to disprove it.
The game starts in the initial configuration $(p_0, q_0)$.
Inductively, given a game configuration $(p_i, q_i)$ at the $i$-th round of the game,
Spoiler chooses a symbol $\symb_i \in \Sigma$ and a transition $\trans{p_i}{\symb_i}{p_{i+1}}$.
Then, Duplicator responds by choosing a matching transition $\trans{q_i}{\symb_i}{q_{i+1}}$,
and the next configuration is $(p_{i+1}, q_{i+1})$.
%If one player cannot move then the other player wins, and Duplicator wins every infinite game. 
%We says that $(p_0,r_0)$ are in simulation preorder iff Duplicator has a winning strategy in the game starting from  $(p_0,r_0)$. 
%While this basic simulation does not consider final states,
%the following variants \cite{etessami:hierarchy02,etessami:etal:fairsimulations:05} have different acceptance conditions.
Since the automaton is assumed to be complete, the game goes on forever,
and the two players build two infinite traces
$\pi_0 = p_0 \goesto {\symb_0} p_1 \goesto {\symb_1} \cdots$ and $\pi_1 = q_0 \goesto {\symb_0} q_1 \goesto {\symb_1} \cdots$.
The winning condition for Duplicator is a predicate on the two traces $\pi_0, \pi_1$,
and it depends on the type of simulation.
%Different simulations have been considered depending on whether one is interested in GFQ or GFI relations.
%
For our purposes, we consider \emph{direct} \cite{dill:inclusion:1992},
\emph{delayed} \cite{etessami:etal:fairsimulations:05}
and \emph{fair simulation} \cite{fairsimulation:02}.
Let $x \in \{\mathrm{di, de, f}\}$.
Duplicator wins the play if $\mathcal C^x(\pi_0, \pi_1)$ holds, where %\cite{etessami:hierarchy02}
\begin{align*}
	\mathcal C^{\mathrm {di}}(\pi_0, \pi_1)	&\quad\iff\quad \forall (i \geq 0) \cdot p_i \in F \implies q_i \in F \\
	\mathcal C^{\mathrm {de}}(\pi_0, \pi_1)	&\quad\iff\quad \forall (i \geq 0) \cdot p_i \in F \implies \exists (j \geq i) \cdot q_j \in F \\
	\mathcal C^{\mathrm f}(\pi_0, \pi_1)	&\quad\iff\quad \textrm{ if $\pi_0$ is fair, then $\pi_1$ is fair }
\end{align*}
Intuitively, direct simulation requires that accepting states are matched immediately (the strongest condition),
while in delayed simulation Duplicator is allowed to accept only after a finite delay.
In fair simulation (the weakest condition),
Duplicator must visit accepting states infinitely often only if Spoiler does so.
Thus, the three conditions are presented in increasing degree of coarseness.
%Thus, $\mathcal C^{\mathrm {di}}(\pi_0, \pi_1)$ implies $\mathcal C^{\mathrm {de}}(\pi_0, \pi_1)$,
%which, in turn, implies $\mathcal C^{\mathrm f}(\pi_0, \pi_1)$.
%
We define $x$-simulation relation $\xsim x \subseteq Q \times Q$, for $x \in  \{\mathrm{di, de, f}\}$,
by stipulating that $p_0 \xsim x q_0$ holds if Duplicator has a winning strategy in the $x$-simulation game,
starting from configuration $(p_0, q_0)$.
Thus, $\disim\, \subseteq\, \desim\, \subseteq\, \fsim$.
Simulation between states in different automata $\A$ and $\B$ can be computed as a simulation on their disjoint union.
%
%All these simulation relations are GFI preorders which can be computed in polynomial time
%\cite{dill:inclusion:1992,HHK:FOCS95,etessami:etal:fairsimulations:05}.
%
%Moreover, direct and delayed simulation are GFQ \cite{etessami:etal:fairsimulations:05},
%while fair simulation is not \cite{fairsimulation:02}. %fairbisimulation:00
%
\begin{lemC}[\cite{dill:inclusion:1992,HHK:FOCS95,fairsimulation:02,etessami:etal:fairsimulations:05}]
	For $x \in \{ \mathrm{di, de, f} \}$, $x$-simulation $\sqsubseteq^x$
        is a PTIME computable preorder.
	For $y \in \{ \mathrm{di, de} \}$, $\sqsubseteq^y$ is GFQ on NBA.
\end{lemC}
\noindent
Notice that fair simulation $\fsim$ is not GFQ.
A simple counterexample can be found in \cite{etessami:etal:fairsimulations:05} (even for fair \emph{bi}simulation);
cf.~ also the automaton from Fig.~\ref{fig:decont:not:GFQ},
where all states are fair bisimulation equivalent,
and thus the quotient automaton would recognize $\Sigma^\omega$.
However, the interest in fair simulation stems from the fact that it is a
PTIME computable under-approximation of fair trace inclusion
(introduced in the next Sec.~\ref{sec:trace_inclusions}).
Trace inclusions between certain states can be used to establish language inclusion between automata,
as discussed in Sec.~\ref{sec:language_inclusion};
this is a part of our inclusion testing presented in Sec.~\ref{sec:inclusion}.

\subsection{Multipebble simulations}
\label{sec:multipebble_simulations}

While simulations are efficiently computable,
their use is often limited by their size,
which can be much smaller than other GFQ preorders.
\emph{Multipebble simulations} \cite{etessami:hierarchy02} offer a generalization of simulations
where Duplicator is given several pebbles that she can use to hedge her bets
and delay the resolution of nondeterminism.
This increased power of Duplicator yields coarser GFQ preorders.

\begin{lemC}[\cite{etessami:hierarchy02}]
	\label{lem:GFQ-multipebbledelayedsimulation}
	Multipebble direct and delayed simulations are GFQ preorders on NBA
	coarser than direct and delayed simulations, respectively.
	They are PTIME computable for a fixed number of pebbles.
\end{lemC}

However, computing multipebble simulations is PSPACE-hard in general \cite{Clemente:PhD},
and in practice it is exponential in the number of pebbles.
For this reason, we study (cf.~Sec.~\ref{sec:lookahead}) lookahead simulations,
which are efficiently computable under-approximations of multipebble simulations,
and, more generally, of trace inclusions, which we introduce next.

\subsection{Forward trace inclusions}
\label{sec:trace_inclusions}

There are other generalizations of simulations (and their multipebble extensions) that are GFQ.
One such example of coarser GFQ preorders is given by \emph{trace inclusions},
which are obtained through the following modification of the simulation game.
In a simulation game, the players build two paths $\pi_0, \pi_1$ by choosing single transitions in an alternating fashion.
That is, Duplicator moves by a single transition by knowing only the next single transition chosen by Spoiler.
We can obtain coarser relations by allowing Duplicator a certain amount of \emph{lookahead} on Spoiler's chosen transitions.
In the extremal case of infinite lookahead, i.e.,
where Spoiler has to reveal his entire path in advance,
we obtain trace inclusions.
Analogously to simulations, we define direct, delayed, and fair trace inclusion, as binary relations on $Q$.
Formally, for $x \in \{\mathrm{di, de, f}\}$, \emph{$x$-trace inclusion} holds between $p$ and $q$, written $p \xincl x q$ if,
for every word $w = \symb_0\symb_1 \cdots \in \Sigma^\omega$,
and for every infinite $w$-trace $\pi_0 = p_0 \goesto {\symb_0} p_1 \goesto {\symb_1} \cdots$ starting at $p_0 = p$, % (as chosen by Spoiler),
there exists an infinite $w$-trace $\pi_1 = q_0 \goesto {\symb_0} q_1 \goesto {\symb_1} \cdots$ starting at $q_0 = q$, % (as chosen by Duplicator),
s.t.~$\mathcal C^x(\pi_0, \pi_1)$ holds.
(Recall the definition of $\mathcal C^x(\pi_0, \pi_1)$ from Sec.~\ref{sec:simulations}).

Like simulations, trace inclusions are preorders.
Clearly, direct trace inclusion $\directtraceinclusion$ is a subset of delayed trace inclusion $\delayedtraceinclusion$,
which, in turn, is a subset of fair trace inclusion $\fairtraceinclusion$.
Moreover, since Duplicator has more power in the trace inclusion game than in the corresponding simulation game,
trace inclusions subsume the corresponding simulation (and even the corresponding multipebble simulation%
\footnote{It turns out that multipebble direct simulation with the maximal
  number of pebbles in fact \emph{coincides} with direct trace inclusion, while the other inclusions are strict for the delayed and fair variants \cite{Clemente:PhD}.}).
In particular, fair trace inclusion $\fairtraceinclusion$ is not GFQ,
since it subsumes fair simulation $\fsim$
which we have already observed not to be GFQ in Sec.~\ref{sec:simulations}.

\begin{figure}
	\begin{tabular}{ccc}
		\begin{tikzpicture}[on grid, node distance= .6cm and 1.3cm]
			\tikzstyle{vertex} = [smallstate]

			\path node [vertex, initial above, accepting] (p) {$p$};
			\path node [vertex] (q) [right = of p] {$q$};
			\path node [vertex, accepting] (r) [right = of q] {$r$};
			\path node [vertex] (s) [right = of r] {$s$};

			\path[->]
				(p) edge node [above] {$a$} (q)
				(q) edge [loop above] node {$a$} ()
				(q) edge node [above] {$a$} (r)
				(r) edge node [above] {$a$} (s)
				(s) edge [loop above] node {$a$} ();

		\end{tikzpicture}
		&
		\quad
		\begin{tabular}{c|c|c|c|c}
			$\delayedtraceinclusion$ & $p$ & $q$ & $r$ & $s$ \\
			\hline
			$p$ & $\tickOK$ & $\tickOK$ & $\tickNO$ & $\tickNO$ \\
			$q$ & $\tickOK$ & $\tickOK$ & $\tickNO$ & $\tickNO$ \\
			$r$ & $\tickOK$ & $\tickOK$ & $\tickOK$ & $\tickNO$ \\
			$s$ & $\tickOK$ & $\tickOK$ & $\tickOK$ & $\tickOK$ \\
			\multicolumn{5}{c}{}
		\end{tabular}
		\quad
		&
		\begin{tikzpicture}[on grid, node distance= .6cm and 1.4cm]
			\tikzstyle{vertex} = [smallstate]

			\path node [vertex, initial, accepting] (p') {$[p]$};
			\path node [vertex, accepting] (r') [right = of p'] {$[r]$};
			\path node [vertex] (s') [right = of r'] {$[s]$};

			\path[->]
				(p') edge [loop above] node {$a$} ()
				(p') edge node [above] {$a$} (r')
				(r') edge node [above] {$a$} (s')
				(s') edge [loop above] node {$a$} ();

		\end{tikzpicture}
		\\
		The original automaton $\A$
		&
		\quad
		Delayed trace inclusion
		\quad
		&
		The quotient automaton $\A/\!\delayedtraceinclusion$

	\end{tabular}
	\caption{Delayed trace inclusion $\delayedtraceinclusion$ is not GFQ.}
	\label{fig:decont:not:GFQ}

\end{figure}

%%% Local Variables:
%%% mode: latex
%%% TeX-master: "ROOT.tex"
%%% End:

We further observe that even the finer delayed trace inclusion $\delayedtraceinclusion$ is not GFQ.
Consider the automaton $\A$ on the left in Fig.~\ref{fig:decont:not:GFQ} (taken from \cite{buchiquotient:ICALP11}).
The states $p$ and $q$ are equivalent w.r.t.~delayed trace inclusion (and are the only two equivalent states),
and thus $[p] = [q]$,
but merging them induces the quotient automaton $\A/\!\delayedtraceinclusion$ on the right in the figure,
which accepts the new word $a^\omega$ that was not previously accepted.

It thus remains to decide whether direct trace inclusion $\directtraceinclusion$ is GFQ.
This is the case, since $\directtraceinclusion$ in fact coincides with multipebble direct simulation,
which is GFQ by Lemma~\ref{lem:GFQ-multipebbledelayedsimulation}.

\begin{lemC}[\cite{etessami:hierarchy02,buchiquotient:ICALP11}]
	Forward trace inclusions are PSPACE computable preorders.
	Moreover, direct trace inclusion $\directtraceinclusion$ is GFQ for NBA,
	while delayed $\delayedtraceinclusion$ and fair $\fairtraceinclusion$ trace inclusions are not.
\end{lemC}

The fact that direct trace inclusion $\directtraceinclusion$ is GFQ
also follows from a more general result presented in the next section,
where we consider a different way to give lookahead to Duplicator.

%\begin{lemma}[\cite{etessami:hierarchy02}]
%	\label{lem:GFQ-directtraceinclusion}
%	Direct trace inclusion $\directtraceinclusion$ is a PSPACE-complete GFQ preorder.
%\end{lemma}

%direct trace inclusion $\directtraceinclusion$ subsumes the GFQ direct simulation $\disim$, and it is itself GFQ.
%We have observed before that fair simulation is not GFQ,
%which entails that the coarser fair trace inclusion $\fairtraceinclusion$ is not GFQ either.
%
%Since delayed simulation $\desim$ is GFQ,
%it remains to establish whether delayed trace inclusion $\delayedtraceinclusion$ is GFQ.
%This question has been answered negatively,
%and a simple counterexample showing that $\delayedtraceinclusion$ is not GFQ can be found in \cite{buchiquotient:ICALP11} (cf.~also Fig.~\ref{fig:decont:not:GFS}).
%
%Notice that delayed simulation $\desim$ and direct trace inclusion $\directtraceinclusion$ are incomparable GFQ preorders.
%This raises the question whether there exists a common GFQ generalization.
%This is indeed the case, and a positive example is given by the so-called \emph{delayed fixed-word simulation} \cite{buchiquotient:ICALP11}.

\subsection{Fixed-word simulations}
\label{sec:fixed-word_simulations}

\emph{Fixed-word simulation} \cite{buchiquotient:ICALP11} is a variant of simulation
where Duplicator has infinite lookahead only on the input word $w$,
but \emph{not} on Spoiler's actual $w$-trace $\pi_0$.
Formally, for $x \in \set{\mathrm{di, de, f}}$,
one considers the family of preorders indexed by infinite words ${\set{\fixedwordsimulation x_w}_{w \in \Sigma^\omega}}$,
where $\fixedwordsimulation x_w$ for a fixed infinite word $w$
is like $x$-simulation, but Spoiler is forced to play the word $w$.
%s.t.~for each fixed word $aw \in \Sigma^\omega$ and states $p, q \in Q$ s.t.~$p \fixedwordsimulation x_{aw} q$,
%whenever $p \goesto a p'$, there exists $q \goesto a q'$ and $p' \fixedwordsimulation x_w q'$;
%the winning condition is determined by $\mathcal C^x$ as before.
%
Then, $x$-fixed-word simulation $\fixedwordsimulation x$ is defined by requiring that Duplicator wins for every infinite word $w$,
that is, $\fixedwordsimulation x = \bigcap_{w \in \Sigma^\omega} \fixedwordsimulation x_w$.
Thus, $x$-fixed-word simulation, by definition, falls between $x$-simulation and $x$-trace inclusion.
What is surprising is that delayed fixed-word simulation $\delayedfixedwordsimulation$
is \emph{coarser} than multipebble delayed simulation
(and thus direct trace inclusion $\directtraceinclusion$,
since this one turns out to coincide with direct multipebble simulation $\disim_n$,
which is included in $\desim_n$ by definition),
and not incomparable as one could have assumed;
this fact is non-trivial \cite{buchiquotient:ICALP11}.
Since delayed fixed-word simulation is GFQ for NBA,
this completes the classification of GFQ preorders for NBA
and makes $\delayedfixedwordsimulation$ the coarsest simulation-like GFQ preorder known to date.
The reader is referred to \cite{buchiquotient:ICALP11} for a more exhaustive
discussion of the results depicted in Fig.~\ref{fig:GFQ_relations}.
%(and even its multipebble variant from \cite{etessami:hierarchy02}),
%
%
%that it is GFQ directly follows from the proof that $\desim$ is GFQ \cite{etessami:etal:fairsimulations:05}.
%
%One can also prove that 
%
%In the rest of the paper we shall not be concerned directly with delayed fixed-word simulation.
%For the purposes of quotienting, we will use only the facts that 
%delayed multipebble simulation and direct trace inclusions are GFQ 
%(Lemma~\ref{lem:GFQ-multipebbledelayedsimulation} and Lemma~\ref{lem:GFQ-directtraceinclusion}).
%
\begin{lemC}[\cite{buchiquotient:ICALP11}]\label{lem:GFQ-delayedfixedwordsimulation}
	Direct/delayed fixed-word simulations $\directfixedwordsimulation,\delayedfixedwordsimulation$
	are PSPACE-complete GFQ preorders.
\end{lemC}

\noindent
The simulations and trace inclusions considered so far explore the state space of the automaton in a forward manner.
Their relationship and GFQ status are summarized in Fig.~\ref{fig:GFQ_relations},
where an arrow means inclusion and a double arrow means equality.
Notice that there is a backward arrow from fixed-word direct simulation to multipebble direct simulation,
and not the other way around as one might expect \cite{buchiquotient:ICALP11}.
In a dual fashion, one can exploit the backward behavior of the automaton to recognize structural relationships
allowing for quotienting states,
which is the topic of the next section.

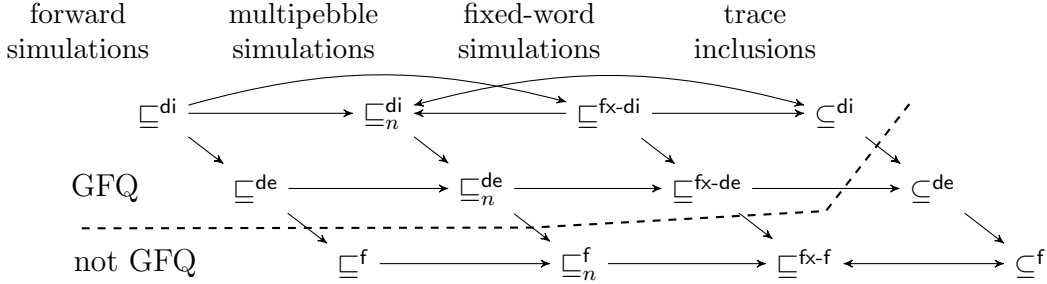
\begin{figure}
	
	\begin{tikzpicture}[on grid, node distance= 1cm and 1.3cm]
%		\tikzstyle{vertex} = [smallstate]

		\path node (disim) {$\disim$};
		\path node (desim) [below right = of disim] {$\desim$};
		\path node (fsim) [below right = of desim] {$\fsim$};
		
		\path node (dimsim) [right = 3cm of disim] {$\disim_n$};
		\path node (demsim) [below right = of dimsim] {$\desim_n$};
		\path node (fmsim) [below right = of demsim] {$\fsim_n$};

		\path node (difxsim) [right = 3cm of dimsim] {$\directfixedwordsimulation$};
		\path node (defxsim) [below right = of difxsim] {$\delayedfixedwordsimulation$};
		\path node (ffxsim) [below right = of defxsim] {$\fairfixedwordsimulation$};

		\path node (dicont) [right = 3cm of difxsim] {$\directtraceinclusion$};
		\path node (decont) [below right = of dicont] {$\delayedtraceinclusion$};
		\path node (fcont) [below right = of decont] {$\fairtraceinclusion$};

		\path (desim) -- (fsim) node [midway,align=center,xshift=-3cm] (start) {};
		\path (demsim) -- (fmsim) node [midway,align=center] (middle0) {}; %{$\qquad\qquad\qquad\qquad$};
		\path (defxsim) -- (decont) node [midway,align=center,yshift=-.3cm] (middle1) {}; %{$\qquad\qquad\qquad\qquad$};
		\path (decont) node [above = 1cm,align=center,xshift=-.3cm] (end) {}; %{$\qquad\qquad\qquad\qquad$};
		
		\path node (sim) [above left = 1.5cm of disim] {$\begin{array}{c}\textrm{forward}\\ \textrm{simulations}\end{array}$};
		\path node (msim) [above left = 1.5cm of dimsim] {$\begin{array}{c}\textrm{multipebble}\\ \textrm{simulations}\end{array}$};
		\path node (msim) [above left = 1.5cm of difxsim] {$\begin{array}{c}\textrm{fixed-word}\\ \textrm{simulations}\end{array}$};
		\path node (msim) [above left = 1.5cm of dicont] {$\begin{array}{c}\textrm{trace}\\ \textrm{inclusions}\end{array}$};
		
		\path[->]

			(disim) edge (desim)
			(disim) edge [->, bend left = 15] node (topnode) {} (difxsim)
			(desim) edge (fsim)
			(disim) edge (dimsim)
			(dimsim) edge (demsim)
			(demsim) edge (fmsim)
			(dimsim) edge [<-] (difxsim)
			(dimsim) edge [<->, bend left = 15] node (topnode) {} (dicont)
			(difxsim) edge (defxsim)
			(defxsim) edge (ffxsim)
			(difxsim) edge (dicont)
			(dicont) edge (decont)
			(decont) edge (fcont)
			(desim) edge (demsim)
			(demsim) edge (defxsim) %[dashed]
			(defxsim) edge (decont)
			(fsim) edge (fmsim)
			(fmsim) edge (ffxsim)
			(ffxsim) edge [<->] (fcont)
			;
%			(r) edge [loop above] node {$a$} ();

		\begin{pgfonlayer}{background}
					
%			\node [fill=blue!20, rectangle, rounded corners = 15pt,
%				fit=(disim) (dimsim) (difxsim) (dicont) (desim) (demsim) (defxsim) (decont) (topnode)] {};
			%\node [fill=blue!20, rectangle, rounded corners = 15pt, fit=(desim) (demsim) (defxsim) (decont)] {};			
%			\node [fill=red!20, rectangle, rounded corners = 15pt, fit=(fsim) (fmsim) (ffxsim) (fcont)] {};

			\newcommand{\thepath}{(start.center) -- (middle0.center) -- (middle1.center) -- (end.center)}
			\newcommand{\therest}{-- (fcont.north east) -- cycle}
			
			\draw [dashed, thick] \thepath;
%			\path [clip] \thepath \therest;
		
%			\node [fill=red!20, rectangle, rounded corners = 15pt,
%				fit=(desim) (demsim) (defxsim) (decont)] {};			

			%\fill [color = red!50, top color = red!100!transparent!50, bottom color = transparent!100, fill opacity = .33] \thepath \therest;

		\end{pgfonlayer}

		\node [above = 0cm of start, left = 2cm of desim] {{\large GFQ}};
		\node [below = 0cm of start, left = 2.9cm of fsim] {{\large not GFQ}};

	\end{tikzpicture}
	
	\caption{Forward-like preorders on NBA}
	\label{fig:GFQ_relations}

\end{figure}

%%% Local Variables:
%%% mode: latex
%%% TeX-master: "ROOT.tex"
%%% End:

\subsection{Backward direct simulation and backward direct trace inclusion}
\label{sec:backward}

Another way of obtaining GFQ preorders is to consider variants of simulation/trace inclusion which go backwards w.r.t.~transitions.
\emph{Backward direct simulation} $\bwdisim$
(called \emph{reverse simulation} in \cite{somenzi:efficient})
is defined like ordinary simulation, except that transitions are taken backwards:
From configuration $(p_i, q_i)$, Spoiler selects a transition $\trans {p_{i+1}} {\symb_i} {p_i}$,
Duplicator replies with a transition $\trans {q_{i+1}} {\symb_i} {q_i}$,
and the next configuration is $(p_{i+1}, q_{i+1})$.
Let $\pi_0 = \cdots \goesto {\symb_1} p_1 \goesto {\symb_0} p_0$
and $\pi_1 = \cdots \goesto {\symb_1} q_1 \goesto {\symb_0} q_0$ be the two infinite backward traces built in this way.
%Duplicator wins the game if, for every position $i$, $p_i \in F \implies q_i \in F$ and $p_i \in I \implies q_i \in I$.
The corresponding winning condition $\mathcal C^\mathrm{bw}_{I, F}$ requires Duplicator to match \emph{both} accepting and initial states:
%
% \begin{align}
\[
	\mathcal C^\mathrm{bw}_{I, F}(\pi_0, \pi_1) \quad \iff \quad \forall (i \geq 0) \cdot
		p_i \in F \implies q_i \in F \textrm{ and } p_i \in I \implies q_i \in I
\]
%\end{align}
%
Then, $p \bwdisim q$ holds if Duplicator has a winning strategy in the backward simulation game starting from $(p, q)$
with winning condition $C^\mathrm{bw}_{I, F}$.
Backward simulation $\bwdisim$ is an efficiently computable GFQ preorder \cite{somenzi:efficient} on NBA incomparable with forward simulations.
%It can be used to establish language inclusion by matching final states of $\A$ with final states of $\B$
%(dually to forward simulations); in this sense, it is GFI.
%
\begin{lemC}[\cite{somenzi:efficient}]
	\label{lem:bwsim:GFQ:GFI}
	Backward simulation $\bwdisim$ is a PTIME computable GFQ preorder on NBA.
\end{lemC}
\noindent
The corresponding notion of \emph{backward direct trace inclusion} $\bwdirecttraceinclusion$ is defined as follows:
$p \bwdirecttraceinclusion q$ if,
for every finite word $w = \symb_0\symb_1 \cdots \symb_{m-1} \in \Sigma^*$,
and for every initial, finite $w$-trace
$\pi_0 = p_0 \goesto {\symb_0} p_1 \goesto {\symb_1} \cdots \goesto {\symb_{m-1}} p_m$ ending in $p_m = p$,
there exists an initial, finite $w$-trace
$\pi_1 = q_0 \goesto {\symb_0} q_1 \goesto {\symb_1} \cdots \goesto {\symb_{m-1}} q_m$ ending in $q_m = q$,
s.t.~$C^\mathrm{bw}_F(\pi_0, \pi_1)$ holds, where
%
% \begin{align}
\[
	\mathcal C^\mathrm{bw}_F(\pi_0, \pi_1) \quad \iff \quad \forall (0 \leq i \leq m) \cdot	p_i \in F \implies q_i \in F
\]
%\end{align}
%
%where the latter condition holds if $p_i \in F$ implies $q_i \in F$ for any $i \geq 0$.
%
Note that backward direct trace inclusion deals with \emph{finite traces} (unlike forward trace inclusions),
which is due to the asymmetry between past and future in $\omega$-automata.

As for their forward counterparts,
backward direct simulation $\bwdisim$ is included in backward direct trace inclusion $\bwdirecttraceinclusion$.
Notice that there is a slight mismatch between the two notions,
since the winning condition of the former is defined over infinite traces,
while the latter is on finite ones.
In any case, inclusion holds thanks to the automaton being backward complete.
Indeed, assume $p \bwdisim q$, and let $\pi_0$ be an initial, finite $w$-trace starting in some $p_0 \in I$ and ending in $p$.
We play the backward direct simulation game from $(p, q)$ by letting Spoiler take transitions according to $\pi_0$
until configuration $(p_0, q_0)$ is reached for some state $q_0$,
and from there we let Spoiler play for ever according to any strategy
(which is possible since the automaton is backward complete).
We obtain a backward infinite path $\pi'_0$ with suffix $\pi_0$ for Spoiler,
and a corresponding $\pi'_1$ with suffix $\pi_1$ for Duplicator s.t.~$\mathcal C^\mathrm{bw}_{I, F}(\pi'_0, \pi'_1)$.
Since $p_0 \in I$, we obtain $q_0 \in I$. Similarly, accepting states are matched all along,
as required in the winning condition for backward direct trace inclusion.
Thus, $p \bwdirecttraceinclusion q$.

In Lemma~\ref{lem:bwsim:GFQ:GFI} we recalled that backward direct simulation $\bwdisim$ is GFQ on NBA.
We now prove that even backward direct trace inclusion $\bwdirecttraceinclusion$ is GFQ on NBA,
thus generalizing the previous result.
\begin{theorem}\label{lem:bwincl_GFQ}
	Backward direct trace inclusion $\bwdirecttraceinclusion$ is a PSPACE-complete GFQ preorder on NBA.%
\end{theorem}
\begin{proof}
	%$n$-pebble delayed simulation is known to be GFQ \cite{etessami:hierarchy02}.
	%For backward trace inclusion, we argue as follows.
	%We first prove that $\bwdirecttraceinclusion$ is GFQ.
	We first show that $\bwdirecttraceinclusion$ is GFQ.
	Let $w = \symb_0 \symb_1 \cdots \in \lang{\A/\!\bwdirecttraceinclusion}$, and we show $w \in \lang \A$.
	There exists an initial and fair $w$-trace
	$\pi = [q_0] \goesto {\symb_0} [q_1] \goesto {\symb_1} \cdots$.
	For $i \geq -1$, let $w_i = \symb_0 \symb_1 \cdots \symb_i$ (with $w_{-1} = \varepsilon$),
	and, for $i \geq 0$, let $\pi[0..i]$ be the $w_{i-1}$-trace prefix of
        $\pi$ ending in $[q_i]$.
	
	For any $i \geq 0$, we build by induction an initial and finite $w_{i-1}$-trace $\pi_i$ of $\A$ ending in $q_i$
	and visiting at least as many accepting states as $\pi[0..i]$ (and at the same time as $\pi[0..i]$ does).
	For $i = 0$, just take the empty $\varepsilon$-trace $\pi_0 = q_0$.
	For $i > 0$, assume that an initial $w_{i-2}$-trace $\pi_{i-1}$ of $\A$ ending in $q_{i-1}$ has already been built.
	We have the transition $\trans {[q_{i-1}]} {\symb_{i-1}} {[q_i]}$ in $\A/\!\bwdirecttraceinclusion$.
	There exist $\hat q \in [q_{i-1}]$ and $\hat q' \in [q_i]$ 
	s.t.~we have a transition $\trans {\hat q} {\symb_{i-1}} {\hat q'}$ in $\A$.
	W.l.o.g.~we can assume that $\hat q' = q_{i}$, since $[\hat q'] = [q_{i}]$.
	By $q_{i-1} \bwdirecttraceinclusion \hat q$, there exists an initial and finite $w_{i-1}$-trace $\pi'$ of $\A$ ending in $\hat q$.
	By the definition of backward direct trace inclusion, $\pi'$ visits at least as many accepting states as $\pi_{i-1}$,
	which, by inductive hypothesis, visits at least as many accepting states as $\pi[0..i-1]$.
	Therefore, $\pi_i := \pi' \goesto {\symb_{i-1}} q_{i}$ is an initial and finite $w_{i-1}$-trace of $\A$ ending in $q_i$.
	Moreover, if $[q_i] \in F' = [F]$, then, since backward direct trace inclusion respects accepting states, $[q_i] \subseteq F$,
	hence $q_i \in F$, and, consequently, $\pi_i$ visits at least as many accepting states as $\pi[0..i]$.
    
	Since $\pi$ is fair, we have thus built a sequence of finite and initial traces
	$\pi_0, \pi_1, \cdots$ visiting unboundedly many accepting states.
	Since $\A$ is finitely branching, by K\"onig's Lemma there exists an initial and fair (infinite) $w$-trace $\pi_\omega$.
	Therefore, $w \in \lang \A$.
	
	Regarding complexity, PSPACE-hardness follows from an immediate reduction from language inclusion of NFA,
	and membership in PSPACE can be shown by reducing to a reachability problem in a finite graph $G$ of exponential size.
	Since reachability in graphs is in NLOGSPACE, we get the desired complexity.
	The finite graph $G = (V, \to)$ is obtained by a product construction combined with a backward determinization construction:
	Vertices are those in
	\begin{align*}
		V = \setof
			{ (p, \hat p) \in Q \times 2^Q }
			{ p \in F \implies \hat p \subseteq F }
	\end{align*}
	and there is an edge $(q, \hat q) \to (p, \hat p)$ if there exists a symbol $\symb \in \Sigma$
	s.t.~$p \goesto \symb q$ and for every $s \in \hat q$ there exists $r \in \hat p$
	s.t.~$r \goesto \symb s$.
	Consider the target set of vertices
	$$T =
		% not needed thanks to the assumption of backward completeness
		%Q \times \set \emptyset \cup
		I \times \setof { \hat p \subseteq Q } {\hat p \cap I = \emptyset}.$$
	We clearly have $p \not\bwdirecttraceinclusion q$ iff
	from vertex $(p, \set q)$ we can reach $T$.
\end{proof}

The results on backward-like simulations established in this section are summarized in Fig.~\ref{fig:backward:GFQ_relations},
where the arrow indicates inclusion.
Notice that backward relations are in general incomparable
with the corresponding forward notions from Fig.~\ref{fig:GFQ_relations}.
In the next section we explore suitable GFQ relations for NFA.

\begin{figure}
	
	\begin{tikzpicture}[on grid, node distance= 1cm and 1.3cm]
%		\tikzstyle{vertex} = [smallstate]

		\path node (bwdisim) {$\bwdisim$};

		\path node (bwdicont) [right = 3cm of bwdisim] {$\bwdirecttraceinclusion$};
		
		\path node (sim) [above = 1cm of bwdisim] {$\begin{array}{c}\textrm{backward direct} \\ \textrm{simulation}\end{array}$};
		\path node (msim) [above = 1cm of bwdicont] {$\begin{array}{c}\textrm{backward direct} \\ \textrm{trace inclusion}\end{array}$};
		
		\path[->]

			(bwdisim) edge (bwdicont)
			;

		\begin{pgfonlayer}{background}

		\end{pgfonlayer}

		\node [above = 0cm of start, left = 2cm of bwdisim] {{\large GFQ}};

	\end{tikzpicture}
	
	\caption{Backward-like preorders on NBA}
	\label{fig:backward:GFQ_relations}

\end{figure}
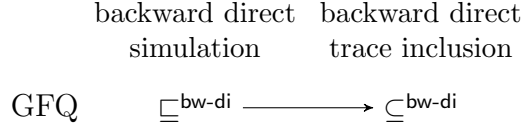
%%% Local Variables:
%%% mode: latex
%%% TeX-master: "ROOT.tex"
%%% End:

\subsection{Simulations and trace inclusions for NFA}
\label{sec:simulations:finitewords}

The preorders presented so far were designed for NBA (i.e., infinite words).
For NFA (i.e., finite words), the picture is much simpler.
Both forward and backward direct simulations $\disim, \bwdisim$ are GFQ also on NFA.
%though it is more commonly called \emph{forward simulation} in the context of finite words.
%
However, over finite words one can consider a backward simulation coarser than $\bwdisim$
where only initial states have to be matched (but not necessarily final ones).
In \emph{backward finite-word simulation} $\bwsim$ the two players play as in backward direct simulation,
except that Duplicator wins the game when the following coarser condition is satisfied
\begin{align*}
	\mathcal C^{\mathrm {bw}}_I(\pi_0, \pi_1)	&\quad\iff\quad \forall (i \geq 0) \cdot p_i \in I \implies q_i \in I
\end{align*}
The corresponding trace inclusions are as follows.
In \emph{forward finite trace inclusion} $\finincl$ Spoiler plays a finite, final trace, and Duplicator has to match it with a final trace.
Dually, in \emph{backward finite trace inclusion} $\bwfinincl$, moves are backward and initial traces must be matched.
Clearly, direct simulation $\disim$ is included in $\finincl$,
and similarly for $\bwsim$ and $\bwfinincl$.
While $\finincl,\bwsim,\bwfinincl$ are not GFQ for NBA
(they are not designed to consider the infinitary acceptance condition of NBA,
which can be shown with trivial examples)
they are for NFA.
The following theorem can be considered as folklore
and its proof is just an adaptation of similar proofs for NBA
in the simpler setting of NFA.
The PSPACE-completeness is an immediate consequence of the fact that
language inclusion for NFA is also PSPACE-complete \cite{MeyerStockmeyer:Equivalence:1972}.
\begin{theorem}
	\label{thm:GFQ:NFA}
	Forward direct simulation $\disim$ and backward finite-word simulation $\bwsim$ are PTIME GFQ preorders on NFA.
	Forward $\finincl$ and backward $\bwfinincl$ finite trace inclusions
	are PSPACE-complete GFQ preorders on NFA.
\end{theorem}

%%% Local Variables:
%%% mode: latex
%%% TeX-master: "ROOT.tex"
%%% End:

\section{Language inclusion}
\label{sec:language_inclusion}

When automata are viewed as finite representations of languages,
it is natural to ask whether two different automata represent the same language,
or, more generally, to compare these languages for inclusion.
Recall that, for two automata $\A$ and $\B$ over the same alphabet $\Sigma$,
%$
we write $\A \languageinclusion \B$ iff $\lang \A \subseteq \lang \B$, and $\A \languageequivalence \B$ iff $\lang \A = \lang \B$.
The \emph{language inclusion/equivalence problem} consists in determining whether $\A \languageinclusion \B$ or $\A \languageequivalence \B$ holds, respectively.
For nondeterministic finite and B\"uchi automata,
language inclusion and equivalence are PSPACE-complete \cite{MeyerStockmeyer:Equivalence:1972,kupfermanvardi:fair_verification}.
This entails that, under standard complexity theoretic assumptions,
there exists no efficient deterministic algorithm for deciding the inclusion/equivalence problem.
Therefore, we consider suitable under-approximations thereof.
% , which are obtained by taking a closer look inside the automata.

\begin{remark}\label{rem:Kurshan}
A partial approach to NBA language inclusion testing has been described by
Kurshan in \cite{Kurshan:Complementing:JCSS:1987}.
Given an NBA $\B$ with $n$ states,
Kurshan's construction builds an NBA $\B'$ with $2n$ states 
such that $\overline{\lang\B} \subseteq \lang{\B'}$, 
i.e., $\B'$ over-approximates the complement of $\B$.
Moreover, if $\B$ is deterministic then $\overline{\lang\B} = \lang{\B'}$.

This yields a sufficient test for inclusion,
since $\lang\A \cap \lang {\B'} = \emptyset$ implies
$\lang \A \subseteq \lang \B$ (though generally not vice-versa).
This condition can be checked in polynomial time.

Of course, for general NBA, this sufficient inclusion test cannot replace
a complete test. Depending on the input automaton $\B$, the
over-approximation $\overline{\lang\B} \subseteq \lang{\B'}$ could be rather coarse. 
\end{remark}

The following definition captures in which sense a preorder on states can be used as a sufficient inclusion test.
\begin{definition}
	Let $\A = (\Sigma, Q_\A, I_\A, F_\A, \delta_\A)$ and $\B = (\Sigma, Q_\B, I_\B, F_\B, \delta_\B)$ be two automata.
	A preorder $\sqsubseteq$ on $Q_\A \times Q_\B$ is \emph{good for inclusion} (GFI) if either one of the following two conditions holds:
	\begin{align*}
		&1. \textrm{ whenever } \forall p\in I_\A \cdot \exists q \in I_\B \cdot p \sqsubseteq q \textrm{, then } \A \languageinclusion \B, \textrm{ or } \\
		&2. \textrm{ whenever } \forall p\in F_\A \cdot \exists q \in F_\B \cdot p \sqsubseteq q \textrm{, then } \A \languageinclusion \B.
	\end{align*}
\end{definition}
\noindent
In other words, GFI preorders give a sufficient condition for inclusion,
by either matching initial states of $\A$ with initial states of $\B$ (case 1),
or by matching accepting states of $\A$ with accepting states of $\B$ (case 2).
However, a GFI preorder is not necessary for inclusion in general%
\footnote{In the presence of multiple initial $\B$ states
it might be the case that inclusion holds
but the language of $\A$ is not included in the language of any of the initial states of $\B$,
only in their ``union''.}.
Usually, forward-like simulations are GFI for case 1, and backward-like simulations are GFI for case 2.
Moreover, if computing a GFI preorder is efficient, then this leads to a sufficient test for inclusion.
Finally, if a preorder is GFI, then all smaller preorders are GFI too, i.e., GFI is $\subseteq$-downward closed.

It is obvious that fair trace inclusion is GFI for NBAs (by matching initial states of $\A$ with initial states of $\B$).
Therefore, all variants of direct, delayed, and fair simulation from Sec.~\ref{sec:simulations},
and the corresponding trace inclusions from Sec.~\ref{sec:trace_inclusions}, are GFI.
We notice here that backward direct trace inclusion $\bwdirecttraceinclusion$ is GFI for NBA
(by matching accepting states of $\A$ with accepting states of $\B$),
which entails that the finer backward direct simulation is GFI as well.

\begin{theorem}\label{lem:bwincl_GFI}
	Backward direct simulation $\bwdisim$ and backward direct trace inclusion $\bwdirecttraceinclusion$
	are GFI preorders for NBA.%
\end{theorem}
\begin{proof}
	Every accepting state in $\A$ is in relation with an accepting state in $\B$.
	Let $w = \symb_0 \symb_1 \cdots \in \lang \A$,
	and let $\pi_0 = p_0 \goesto {\symb_0} p_1 \goesto {\symb_1} \cdots$ be an initial and fair $w$-path in $\A$.
	Since $\pi_0$ visits infinitely many accepting states,
	and since each such state is $\bwdirecttraceinclusion$-related to an accepting state in $\B$,
	by using the definition of $\bwdirecttraceinclusion$
	it is possible to build in $\B$ longer and longer finite, initial traces in $\B$
	visiting unboundedly many accepting states.
	Since $\B$ is finitely branching,
	by K\"onig's Lemma there exists an initial and fair (infinite) $w$-trace $\pi_\omega$ in $\B$.
	Thus, $w \in \lang \B$.
\end{proof}

For NFA, we observe that forward finite trace inclusion $\finincl$ is GFI by matching initial states,
and backward finite trace inclusion $\bwfinincl$ is GFI by matching accepting states.
The proof of the following theorem is immediate.

\begin{theorem}
	\label{thm:GFI:NFA}
	Forward $\finincl$ and backward $\bwfinincl$ finite trace inclusions are GFI preorders on NFA.
\end{theorem}

%%% Local Variables:
%%% mode: latex
%%% TeX-master: "ROOT.tex"
%%% End:

\section{Transition pruning reduction techniques} \label{sec:pruning}

While quotienting-based reduction techniques reduce the number of states by merging them,
we explore an alternative method which prunes (i.e., removes) transitions.
The intuition is that certain transitions can be removed from an automaton without changing its language
when other `better' transitions remain. 
\begin{definition}
	Let $\A = (\Sigma, Q, I, F, \delta)$ be an automaton,
	let $\prunerel \subseteq \delta \times \delta$ be a relation on $\delta$,
	and let $\max \prunerel$ be the set of maximal elements of $\prunerel$, i.e.,
	\begin{align*}
		\max\prunerel = \{(p,\symb,r) \in \delta \st \nexists (p',\symb',r') \in \delta \cdot ((p,\symb,r), (p',\symb',r')) \in \prunerel \}
	\end{align*}
	The \emph{pruned automaton} is defined as $\prune{\A}{\prunerel} \defeq (\Sigma, Q, I, F, \delta')$, where
	%
%	\begin{align*}
		$\delta' = \max \prunerel$. % \\
%		I' &= \{ p \in I \st \exists (p, )
%	\end{align*}
\end{definition}

%\noindent
In most practical cases, $\prunerel$ will be a strict partial order, but this condition
is not absolutely required. 

While the computation of $\prunerel$ depends on $\delta$ in general,
\emph{all subsumed transitions are removed `in parallel'},
and thus $\prunerel$ is {\em not} re-computed even if the removal of a single transition changes $\delta$,
and thus $\prunerel$ itself.
This is important for computational reasons.
Computing $\prunerel$ may be expensive, and
thus it is beneficial to remove at once all transitions that can be witnessed with the $\prunerel$ at hand.
E.g., one might remove thousands of transitions in a single step without
re-computing $\prunerel$. 
On the other hand, removing transitions in parallel makes arguing about correctness much more difficult due to potential mutual dependencies between the involved transitions.

Regarding correctness, note that removing transitions cannot introduce new words in the language,
thus $\prune{\A}{\prunerel} \languageinclusion \A$.
When also the converse inclusion holds (so the language is preserved),
we say that $\prunerel$ is good for pruning (GFP).
\begin{definition}
	A relation $\prunerel \subseteq \delta \times \delta$ is \emph{good for pruning} (GFP)
	if $\prune{\A}{\prunerel} \languageequivalence \A$.
\end{definition}
\noindent
Like GFQ, also GFP is $\subseteq$-downward closed in the space of relations.
We study specific GFP relations obtained by comparing the endpoints of transitions over the same input symbol.
Formally, given two binary state relations
$\brel, \frel\, \subseteq Q \times Q$
for the source and target endpoints, respectively, we define
\begin{align}
	\label{eq:prunerel}
	\makeprunerel{\brel}{\frel} = \{((p,\symb,r),(p',\symb,r')) \in \delta \times \delta \st p \brel p' \textrm{ and } r \frel r' \}.
\end{align}
$\makeprunerel{\cdot}{\cdot}$ is monotone in both arguments.

In the following section,
we explore which state relations $\brel,\frel$ induce GFP relations $\makeprunerel{\brel}{\frel}$ for NBA.
In Sec.~\ref{sec:pruning:NFA}, we present similar GFP relations for NFA.

\subsection{Pruning NBA}
\label{sec:pruning:NBA}

\begin{table}
	
	\begin{center}
		\begin{tabular}{c|ccccc||ccc}
			$\brel\backslash\frel$
								& $\id$		& $\strictdisim$	& $\disim$		& $\strictdirecttraceinclusion$	& $\directtraceinclusion$
																					&	$\strictdesim$	&	$\strictfsim$
																					&	$\strictlanguageinclusion$ \\
			\hline
			$\id$				& $\NA $	& $\tickOK$			& $\NA$			& $\tickOK$				& $\NA$	&	$\tickNO$ &	$\tickNO$ & $\tickNO$		\\
			$\strictbwdisim$	&	$\tickOK$	&	$\tickOK$		&	$\tickOK$&	$\tickOK$&	$\tickOK$	&	$\tickNO$		&	$\tickNO$ &	$\tickNO$		\\
			$\bwdisim$	&	$\NA$	&	$\tickOK$		&	$\NA$&	$\tickNO$&	$\NA$	&	$\tickNO$		&	$\tickNO$ &	$\tickNO$		\\
			
			$\strictbwdirecttraceinclusion$
								&	$\tickOK$	&	$\tickOK$		&	$\tickNO$	&	$\tickNO$&	$\tickNO$&	$\tickNO$		&	$\tickNO$ &	$\tickNO$ \\
			$\bwdirecttraceinclusion$
								&	$\NA$	&	$\tickOK$		&	$\NA$	&	$\tickNO$&	$\NA$&	$\tickNO$		&	$\tickNO$ &	$\tickNO$
								
		\end{tabular}
	\end{center}
	%
	%\vspace{-.5cm}
	%
	%\caption*{ \scriptsize %\hfill
	%\begin{tabular}{rl}
	%\end{tabular}
	
	\caption{GFP relations $\makeprunerel{\brel}{\frel}$ for NBA.
        $\tickOK$ denotes yes, $\tickNO$ denotes no, and $\NA$ denotes the case
        where GFP does not hold for the trivial reason that the relation is not irreflexive.}
	\label{fig:GFP_relations}

\end{table}

%%% Local Variables:
%%% mode: latex
%%% TeX-master: "ROOT.tex"
%%% End:

Our results are summarized in Table~\ref{fig:GFP_relations}.
%
%Intuitively, a GFP relation has the property that a removed transition cannot eventually depend on itself while building runs not using it.
%
It has long been known that $\makeprunerel{\id}{\strictdisim}$ and $\makeprunerel{\strictbwdisim}{\id}$ are GFP
(see \cite{simulationminimization:03}, where the removed transitions are called `little brothers').
However, already slightly relaxing direct simulation to delayed simulation is incorrect,
i.e., $\makeprunerel{\id}{\strictdesim}$ is not GFP.
This is shown in the counterexample in Fig.~\ref{fig:desim:not:GFP},
where $q \strictdesim p$, but removing the dashed transition $p \goesto a q$ (due to $p \goesto a p$) makes the language empty.
The essential problem is that $q \strictdesim p$ holds precisely thanks to the presence of transition $p \goesto a q$,
without which we would have $q \not\strictdesim p$,
thus creating a cyclical dependency.
%even though ${\A} = {\A}/\!\desim$.
Consequently, $\makeprunerel{\id}{\strictfsim}$ and $\makeprunerel{\id}{\strictlanguageinclusion}$ are not GFP either.

\begin{figure}
	\centering
	%\begin{figure}
%	
	\subfigure[
		$\makeprunerel{\id}{\strictdesim}$ is not GFP.
		\label{fig:desim:not:GFP}
	]{
	\begin{tikzpicture}[on grid, node distance= .6cm and 1.7cm]
		\tikzstyle{vertex} = [smallstate]

		\path node [vertex, initial] (p) {$p$};
		\path node (hidden) (x) [below = of p] {}; % just for formatting reasons
		\path node [vertex, accepting] (q) [right = of p] {$q$};

		\path[->]

			(p) edge [dashed] node [above] {$a$} node [below] {$\strictdesimrev$} (q)
			(p) edge [loop above] node {$a,b$} ()
			(q) edge [loop above] node {$a$} ();

	\end{tikzpicture}
	}
%	
%\end{figure}
%%% Local Variables:
%%% mode: latex
%%% TeX-master: "ROOT.tex"
%%% End:
	\qquad\qquad
	\subfigure[
	$\makeprunerel{\id}{\strictdisim} \cup \makeprunerel{\strictbwdisim}{\id}$ is not GFP.
	]
	{
	\label{fig:union:not:GFP}
	\begin{tikzpicture}[on grid, node distance= 1cm and 2.25cm]
		\tikzstyle{vertex} = [smallstate]

		\path node [vertex, initial] (p) {$p$};
		\path node [vertex] (q) [above right = of p] {$q$};
		\path node [vertex] (r) [below right = of p] {$r$};
		\path node [vertex, accepting] (s) [below right = of q] {$s$};

		\path[->]

			(p) edge node [above left] {$a$} (q)
			(q) edge [bend left = 10, dashed] node [above] {$a$} (s)
			(q) edge [bend right = 10] node [below] {$b$} (s)
			(p) edge [bend left = 10, dashed] node [above] {$a$} (r)
			(p) edge [bend right = 10] node [below] {$b$} (r)			
			(r) edge node [below right] {$a$} (s)

			(s) edge [loop right] node {$c$} ();
			
		\path
			(r) -- node [midway, left] {\rotatebox{-90}{$\strictbwdisim$}} (q)
			(r) -- node [midway, right] {\rotatebox{90}{$\strictdisim$}} (q);

	\end{tikzpicture}
}
%%% Local Variables:
%%% mode: latex
%%% TeX-master: "ROOT.tex"
%%% End:
	\caption{Two pruning counterexamples.}
\end{figure}
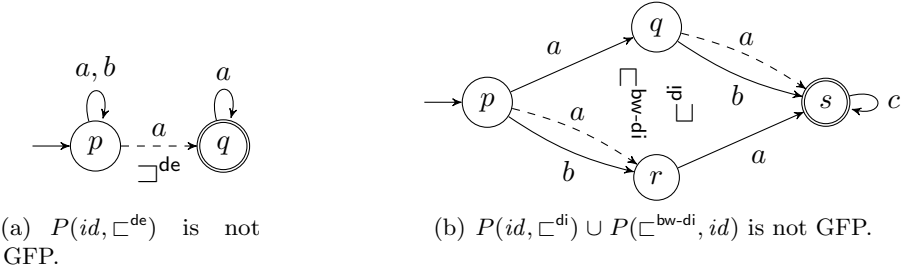

Moreover, while $\makeprunerel{\id}{\strictdisim}$ and $\makeprunerel{\strictbwdisim}{\id}$ are GFP,
their union $\makeprunerel{\id}{\strictdisim} \cup \makeprunerel{\strictbwdisim}{\id}$ (or the transitive closure thereof) is not.
A counterexample is shown in Fig.~\ref{fig:union:not:GFP},
where pruning would remove both the transitions
$p \goesto a r$ (subsumed by $p \goesto a q$ with $r \strictdisim q$)
and $q \goesto a s$ (subsumed by $r \goesto a s$ with $q \strictbwdisim r$),
and $aac^\omega$ would no longer be accepted.
Again, the essential issue is a cyclical dependency:
$r \strictdisim q$ holds only if $q \goesto a s$ is not pruned,
and symmetrically $q \strictbwdisim r$ holds only if $p \goesto a r$ is not pruned.
Therefore, removing any single one of these two transitions is sound, but not removing both.

However, it is possible to relax simulation in
$\makeprunerel{\id}{\strictdisim}$ and $\makeprunerel{\strictbwdisim}{\id}$
to direct trace inclusion $\strictdirecttraceinclusion$, resp., backward trace inclusion $\strictbwdirecttraceinclusion$,
and prove that $\makeprunerel{\id}{\strictdirecttraceinclusion}$ and $\makeprunerel{\strictbwdirecttraceinclusion}{\id}$ are GFP.
This is shown below in Theorems~\ref{thm:prune_id_strictdirecttraceinclusion} and \ref{thm:bwtrace-id}.
\begin{theorem}\label{thm:prune_id_strictdirecttraceinclusion}
	For every strict partial order $R \subseteq \directtraceinclusion$,
	$\makeprunerel{\id}{R}$ is GFP on NBA. In particular, $\makeprunerel{\id}{\strictdirecttraceinclusion}$ is GFP.
\end{theorem}

\begin{proof}
	Let $\B = \prune{\A}{\makeprunerel{\id}{R}}$. 
	We show $\A \languageinclusion \B$.
	If $w = \symb_0\symb_1 \cdots\in \lang{\A}$ then there exists an infinite fair
	initial trace $\hat{\pi}$ on $w$ in $\A$. We show $w \in \lang{\B}$.

	We call a trace $\pi = q_0 \goesto {\symb_0} q_1 \goesto {\symb_1} \cdots$ on
	$w$ in $\A$ $i$-{\em good} if it does not contain any
	transition $q_j \goesto {\symb_j} q_{j+1}$ for $j < i$ s.t.~there exists an $\A$ transition
	$q_j \goesto {\symb_j} q_{j+1}'$ with $q_{j+1} \mathrel{R} q_{j+1}'$ (i.e., no such
	transition is used within the first $i$ steps).
	Since $\A$ is finitely branching, for every state and symbol there exists at
	least one $R$-maximal successor that is still present in
	$\B$, because $R$ is asymmetric and transitive.
	Thus, for every $i$-good trace $\pi$ on
	$w$ there exists an $(i+1)$-good trace $\pi'$ on $w$ 
        s.t.~$\pi$ and $\pi'$ are identical on the first $i$ steps and 
	$\mathcal C^{\mathrm {di}}(\pi, \pi')$, because
	$R \subseteq \directtraceinclusion$.
	Since $\hat{\pi}$ is an infinite fair initial trace on $w$ (which is trivially
	$0$-good), there exists an
	infinite initial trace $\tilde{\pi}$ on $w$ that is
	$i$-good for every $i$ and $\mathcal C^{\mathrm{di}}(\hat{\pi}, \tilde{\pi})$.
	Moreover, $\tilde{\pi}$ is a trace in $\B$.
	Since $\hat{\pi}$ is fair and $\mathcal C^{\mathrm{di}}(\hat{\pi}, \tilde{\pi})$,  
	$\tilde{\pi}$ is an infinite fair initial trace on $w$ that is $i$-good for every $i$.
	Therefore $\tilde{\pi}$ is a fair initial trace on $w$ in $\B$ and thus $w \in \lang{\B}$.
\end{proof}
\begin{theorem}\label{thm:bwtrace-id}
	For every strict partial order $R \subseteq \bwdirecttraceinclusion$, $\makeprunerel{R}{\id}$ is GFP on NBA.
	In particular, $\makeprunerel{\strictbwdirecttraceinclusion}{\id}$ is GFP.
\end{theorem}

\begin{proof}
	Let $\B = \prune{\A}{\makeprunerel{R}{\id}}$. We show $\A \languageinclusion \B$.
	If $w = \symb_0\symb_1 \cdots\in \lang{\A}$ then there exists an infinite fair
	initial trace $\hat{\pi}$ on $w$ in $\A$. We show $w \in \lang{\B}$.

	We call a trace 
	$\pi = q_0 \goesto {\symb_0} q_1 \goesto {\symb_1} \cdots$ on $w$ in $\A$
	$i$-{\em good} if it does not contain any
	transition $q_j \goesto {\symb_j} q_{j+1}$ for $j < i$ s.t.~there exists an $\A$ transition
	$q_j' \goesto {\symb_j} q_{j+1}$ with $q_j \mathrel{R} q_j'$ (i.e., no such
	transition is used within the first $i$ steps).
	We show, by induction on $i$, the following property (P):
	For every $i$ and every initial trace $\pi$ on $w$ in $\A$ there exists an
	initial $i$-good trace $\pi'$ on $w$ in $\A$ 
	s.t.~$\pi$ and $\pi'$ have identical suffixes from step $i$ onwards
	and $\mathcal C^{\mathrm {di}}(\pi, \pi')$.
	The base case $i=0$ is trivial with $\pi'=\pi$.
	For the induction step there are two cases. If $\pi$ is $(i+1)$-good
	then we can take $\pi' = \pi$.
	Otherwise there exists a transition $q_i' \goesto {\symb_i} q_{i+1}$ with
	$q_i \mathrel{R} q_i'$. Without restriction (since $\A$
	is finite and $R$ is asymmetric and transitive) we assume
	that $q_i'$ is $R$-maximal among the
	$\symb_i$-predecessors of $q_{i+1}$.
	In particular, the transition $q_i' \goesto {\symb_i} q_{i+1}$
	is present in $\B$.
	Since $R \subseteq \bwdirecttraceinclusion$, there exists
	an initial trace $\pi''$ on $w$ that has suffix
	$q_i' \goesto {\symb_i} q_{i+1} \goesto{\symb_{i+1}} q_{i+2} \dots$
	and $\mathcal C^{\mathrm {di}}(\pi, \pi'')$.
	Then, by induction hypothesis, there exists an initial $i$-good
	trace $\pi'$ on $w$ that has suffix
	$q_i' \goesto {\symb_i} q_{i+1} \goesto{\symb_{i+1}} q_{i+2} \dots$
	and $\mathcal C^{\mathrm {di}}(\pi'', \pi')$.
	Since $q_i'$ is $R$-maximal among the
	$\symb_i$-predecessors of $q_{i+1}$, we obtain that $\pi'$ is also 
	$(i+1)$-good. Moreover, $\pi'$ and $\pi$ have identical suffixes
	from step $i+1$ onwards. Finally, 
	by $\mathcal C^{\mathrm {di}}(\pi, \pi'')$
	and $\mathcal C^{\mathrm {di}}(\pi'', \pi')$, we obtain
	$\mathcal C^{\mathrm {di}}(\pi, \pi')$.
	% as required.

	Given the infinite fair initial trace $\hat{\pi}$ on $w$ in $\A$,
	it follows from property (P) and K\"onig's Lemma that there
	exists an infinite initial trace $\tilde{\pi}$ on $w$ that is
	$i$-good for every $i$ and $\mathcal C^{\mathrm {di}}(\hat{\pi},
	\tilde{\pi})$.
	Therefore $\tilde{\pi}$ is an infinite fair initial trace 
	on $w$ in $\B$ and thus $w \in \lang{\B}$.
	%, as required.
\end{proof}

One can also compare both endpoints of transitions,
i.e., using relations larger than the identity as in the previous cases.
The following Theorems~\ref{thm:bwsim-fwtrace} and \ref{thm:bwtrace-fwsim} prove that
$\makeprunerel{\strictbwdisim}{\directtraceinclusion}$, resp., $\makeprunerel{\bwdirecttraceinclusion}{\strictdisim}$ are GFP.
Consequently, $\makeprunerel{\strictbwdisim}{\strictdisim}$ is also GFP.
This is already a non-trivial fact.
To witness this, notice that while pruning w.r.t.~$\makeprunerel{\id}{\strictdisim}$/$\makeprunerel{\strictbwdisim}{\id}$
preserves forward/backward simulation, respectively,
pruning w.r.t.~$\makeprunerel{\id}{\strictdisim}$ disrupts backward simulation,
and pruning w.r.t.~$\makeprunerel{\strictbwdisim}{\id}$ disrupts forward simulation.
Therefore, when pruning simultaneously w.r.t.~the coarser $\makeprunerel{\strictbwdisim}{\strictdisim}$
both simulations are disrupted and the structure of the automaton can change radically.

Let us also notice that,
while $\makeprunerel{\strictbwdisim}{\directtraceinclusion}$ and $\makeprunerel{\bwdirecttraceinclusion}{\strictdisim}$ are GFP on NBA,
neither $\makeprunerel{\strictbwdirecttraceinclusion}{\directtraceinclusion}$ subsuming the first one,
nor $\makeprunerel{\bwdirecttraceinclusion}{\strictdirecttraceinclusion}$ subsuming the second one, are GFP on NBA.
Indeed, $\makeprunerel{\strictbwdirecttraceinclusion}{\strictdirecttraceinclusion}$ is already not GFP.
A counter-example is shown in Fig.~\ref{fig:incl:not:GFP},
where removing the dashed transitions $p_0 \goesto a q_0$ (due to $p_1 \goesto a q_1$)
and $r_1 \goesto a s_1$ (due to $r_0 \goesto a s_0$)
causes the word $a^5e^\omega$ to be no longer accepted;
the extra transitions going up from the initial state to the unnamed state and to $r_0$,
and the extra transitions going down from $q_1$ to the unnamed state and to the accepting state
are used in order to ensure that the trace inclusions are strict,
which shows that the two transitions $p_0 \goesto a q_0$ and $r_1 \goesto a s_1$ are even strictly subsumed,
and yet they cannot both be removed, lest the language be altered.
Thus, one cannot use trace inclusions on both endpoints, i.e., at least one endpoint must be a simulation.

\begin{figure}

	\begin{tikzpicture}[on grid, node distance=1.5cm and 1.5cm]
		\tikzstyle{vertex} = [smallstate]

		\path node [vertex, initial] (i) {};
		
		\path node [vertex] (p0) [above right = 1 cm and 1 cm of i] {$p_0$};
		\path node [vertex] (q0) [right = of p0] {$q_0$};
		\path node [vertex] (r0) [right = of q0] {$r_0$};
		\path node [vertex] (s0) [right = of r0] {$s_0$};
		
		\path node [vertex] (p1) [below right = 1 cm and 1 cm of i] {$p_1$};
		\path node [vertex] (q1) [right = of p1] {$q_1$};
		\path node [vertex] (r1) [right = of q1] {$r_1$};
		\path node [vertex] (s1) [right = of r1] {$s_1$};
		
		\path node [vertex, accepting] (f) [below right = .6 and 1 cm of s0] {};
		
		\path node [vertex] (x0) [above = 1.6cm of p0] {};% {$x_0$};
		\path node [vertex] (y0) [right = of x0] {};% {$y_0$};

		\path node [vertex] (x1) [below = 1.6cm of r1] {};% {$x_1$};
		\path node [vertex] (y1) [right = of x1] {};% {$y_1$};

		\path[->]

			(i) edge node [above left = .2cm and -.2 cm] {$a$} (p0)
			(i) edge [bend left = 30] node [above left] {$c$} (x0)
			(i) edge [bend left = 75] node [above left] {$b$} (r0)
			
			(p0) edge [dashed] node [above] {$a$} (q0)
			(q0) edge node [above] {$a$} (r0)
			(r0) edge node [above] {$a$} (s0)
			(s0) edge node [above right] {$a,d$} (f)
			
			(x0) edge node [above] {$a$} (y0)
			(y0) edge [bend left = 30] node [above right] {$a$} (r0)
			
			(i) edge node [below left] {$a,c$} (p1)
			(q1) edge [bend right = 45] node [below left] {$a$} (x1)
			(q1) edge [bend right = 90] node [below left] {$b$} (f)
			
			(p1) edge node [above] {$a$} (q1)
			(q1) edge node [above] {$a$} (r1)
			(r1) edge [dashed] node [above] {$a$} (s1)
			(s1) edge node [above left = .2 cm and -.2 cm] {$a$} (f)
			
			(x1) edge node [below] {$a$} (y1)
			(y1) edge [bend right = 45] node [below right] {$d$} (f)
			
			(f) edge [loop right] node {$e$} ();

		\path
			(p0) -- node [midway] {\rotatebox{-90}{$\strictbwdirecttraceinclusion$}} (p1)
			(q0) -- node [midway] {\rotatebox{-90}{$\strictdirecttraceinclusion$}} (q1)
			(r0) -- node [midway] {\rotatebox{90}{$\strictbwdirecttraceinclusion$}} (r1)
			(s0) -- node [midway] {\rotatebox{90}{$\strictdirecttraceinclusion$}} (s1);

		\begin{pgfonlayer}{background}
			%\node () [color=blue!50, inner sep=4pt, ellipse, fit=(p0) (q)] {};
		\end{pgfonlayer}

	\end{tikzpicture}

	\caption{$\makeprunerel{\strictbwdirecttraceinclusion}{\strictdirecttraceinclusion}$ is not GFP.}
	\label{fig:incl:not:GFP}

\end{figure}

\begin{figure}

	\begin{tikzpicture}[on grid, node distance=1.5cm and 1.5cm]
		\tikzstyle{vertex} = [smallstate]

		\path node [vertex, initial] (p0) {$p_0$};
		\path node [vertex] (q0) [right = of p0] {$q_0$};
		\path node [vertex] (r0) [right = of q0] {$r_0$};
		\path node [vertex] (s0) [right = of r0] {$s_0$};
		
		\path node [vertex] (q1) [below = 2.2cm of q0] {$q_1$};
		\path node [vertex] (r1) [right = of q1] {$r_1$};
		\path node [vertex] (s1) [right = of r1] {$s_1$};
		
		\path node [vertex, accepting] (f) [below right = 1.1 cm and 1 cm of s0] {};
		
		\path node [vertex] (q2) [above = 1cm of q0] {};% {$x_0$};

		\path[->]

			(p0) edge node [above left] {$b,c$} (q2)
			(q2) edge node [above right] {$a$} (r0)
			
			(p0) edge node [above] {$a$} (q0)
			(q0) edge [dashed] node [above] {$a$} (r0)
			(r0) edge node [above] {$a$} (s0)
			(s0) edge node [above right] {$a$} (f)
			
			(p0) edge node [below left] {$a,b$} (q1)
			(q1) edge node [above] {$a$} (r1)
			(r1) edge [dashed] node [above] {$a$} (s1)
			(s1) edge node [below right] {$a$} (f)
			
			(f) edge [loop right] node {$a$} ();

		\path
			(q0) -- node [midway] {\rotatebox{-90}{$\strictbwdirecttraceinclusion$}} (q1)
			(r0) -- node [midway] {$\rotatebox{-90}{\!\!\!\!$\disim$}$ $\rotatebox{90}{$\strictbwdirecttraceinclusion$}$} (r1)
			(s0) -- node [midway] {\rotatebox{90}{$\disim$}} (s1);

	\end{tikzpicture}

	\caption{$\makeprunerel{\strictbwdirecttraceinclusion}{\disim}$ is not GFP.}
	\label{fig:strict:incl:sim:not:GFP}

\end{figure}

%%% Local Variables:
%%% mode: latex
%%% TeX-master: "ROOT.tex"
%%% End:

Moreover, the endpoint using simulation must actually use strict simulation, and not just simulation.
In fact, while $\makeprunerel{\strictbwdisim}{\directtraceinclusion}$ and $\makeprunerel{\bwdirecttraceinclusion}{\strictdisim}$ are GFP on NBA,
neither $\makeprunerel{\bwdisim}{\strictdirecttraceinclusion}$ nor $\makeprunerel{\strictbwdirecttraceinclusion}{\disim}$ is GFP.
A counter-example for the second case is shown in Fig.~\ref{fig:strict:incl:sim:not:GFP} (the first case is symmetric).
If both dashed transitions are removed, the automaton stops recognizing $a^\omega$.

%(Note that $\A = \A/\!\bwdirecttraceinclusion = \A/\!\directtraceinclusion$;
%this example even holds for $\stricttranskbwsim, \stricttranskdisim$ and $k=3$; cf.~Sec.~\ref{sec:lookahead}).

%
\begin{theorem}\label{thm:bwsim-fwtrace}
	%For every partial order $R \subseteq \directtraceinclusion$,
	The relation $\makeprunerel{\strictbwdisim}{\directtraceinclusion}$ is GFP on NBA.
\end{theorem}

\begin{proof}
	%It suffices to prove that $\makeprunerel{\strictbwdisim}{\directtraceinclusion}$ is GFP,
	%since $\strictdisim$ is included in $\directtraceinclusion$.
	%
	Let $\B = \prune{\A}{\makeprunerel{\strictbwdisim}{\directtraceinclusion}}$. We show $\A \languageinclusion \B$.
	Let $w = \symb_0\symb_1 \cdots\in \lang\A$. There exists an infinite fair
	initial trace $\hat{\pi}$ on $w$ in $\A$. We show $w \in \lang{\B}$.

	Let $i \geq 0$. We call a trace $\pi = q_0 \goesto {\symb_0} q_1 \goesto {\symb_1} \cdots$ in $\A$ on $w$ $i$-{\em good}
	if there is no $j \leq i$ 
	s.t.~there exists a state $q_j'$ with $q_j \strictbwdisim q_j'$
	and there exists an infinite trace $\suffix {\pi'} j$ from $q_j'$ on the word $\symb_j\symb_{j+1}\cdots$ with
	$\mathcal C^{\mathrm {di}}(\suffix \pi j, \suffix {\pi'} j)$.
	First, we show that, for every $i \geq 0$, there are $i$-good traces in $\A$.
	For the base case, it suffices to choose the state $q_0$ to be $\strictbwdisim$-maximal
	amongst the starting points of all infinite initial traces $\pi$ on $w$ s.t.~$\mathcal C^{\mathrm {di}}(\pi, \hat\pi)$.
	(This set is non-empty since it contains $\hat\pi$.)
	For the inductive step, let $i \geq 1$ and let $\pi$ be an infinite $(i-1)$-good trace on $w$.
	If $\pi$ is also $i$-good, then we are done.
	Otherwise, $\pi$ is not $i$-good,
	and there exist a state $q_i'$ and an infinite path $\suffix {\pi'} i$ from $q_i'$
	s.t.~$q_i \strictbwdisim q_i'$ and $C^{\mathrm {di}}(\suffix \pi i, \suffix {\pi'} i)$.
	We further choose $q_i'$ to be $\strictbwdisim$-maximal with the former property.
	By the definition of $q_i \strictbwdisim q_i'$,
	there exists an initial path $\prefix {\pi'} i = q_0' \goesto {\symb_0} q_1' \goesto {\symb_1} \cdots \goesto {\symb_{i-1}} q_i'$ ending in $q_i'$
	s.t., for every $j \leq i$, $q_j \bwdisim q_j'$.
	(This last property uses the fact that $\bwdisim$ propagates backward. 
	Backward direct trace inclusion $\bwdirecttraceinclusion$ does not suffice.)
	Thus, $\pi' = q_0' \goesto {\symb_0} q_1' \goesto {\symb_1} \cdots$ is an initial, infinite, and fair trace on $w$.
	Moreover, it is $i$-good: By contradiction, let $q''_j$ be s.t.~$q'_j \strictbwdisim q''_j$ with $j \leq i$.
	It cannot be the case that $j = i$ by the maximality of $q_i'$.
	Since $q_j \bwdisim q'_j$, it also cannot be the case that $j < i$ since $\pi$ is $(i-1)$-good.
	Thus, $\pi'$ is $i$-good.
	
	Second, by K\"onig's Lemma it follows that there exists an initial, infinite, trace $\tilde{\pi} = q_0 \goesto {\symb_0} q_1 \goesto {\symb_1} \cdots $ on $w$
	that is $i$-good for every $i$ and $\mathcal C^{\mathrm {di}}(\hat{\pi}, \tilde{\pi})$.
	In particular, this implies that $\tilde{\pi}$ is fair. 
	We show that such a $\tilde{\pi}$ is also possible in $\B$
	by assuming the opposite and deriving a contradiction.
	Suppose that $\tilde{\pi}$ contains a transition $q_j \goesto {\symb_j} q_{j+1}$ that is not present in $\B$.
	Then there exists a transition $q_j' \goesto {\symb_j} q_{j+1}'$ in $\B$
	s.t.~$q_j \strictbwdisim q_j'$ and 
	$q_{j+1} \directtraceinclusion q_{j+1}'$.
	%
%	We cannot have $j=0$, because in this case $\tilde{\pi}$ would not be $0$-good. So we get $j \ge 1$.
	%
	Since $q_j' \goesto {\symb_j} q_{j+1}'$ and $q_{j+1} \directtraceinclusion q_{j+1}'$,
	there exists an infinite, fair, trace $\suffix {\pi'} j$ from $q_j'$ with
	$\mathcal C^{\mathrm {di}}(\suffix \pi j, \suffix {\pi'} j)$.
	Since $q_j \strictbwdisim q_j'$, this contradicts the fact that $\tilde\pi$ is $j$-good.
	Therefore $\tilde{\pi}$ is a fair initial trace on $w$ in $\B$, and thus $w \in \lang{\B}$.
\end{proof}

\begin{theorem}\label{thm:bwtrace-fwsim}
	The relation $\makeprunerel{\bwdirecttraceinclusion}{\strictdisim}$ is GFP on NBA.
\end{theorem}

\begin{proof}
	Let $\B = \prune{\A}{\makeprunerel{\bwdirecttraceinclusion}{\strictdisim}}$. We show $\A \languageinclusion \B$.
	Let $w = \symb_0\symb_1 \cdots\in \lang{\A}$. 
	There exists an infinite fair initial trace $\hat{\pi}$ on $w$ in $\A$. We show $w \in \lang{\B}$.

	For an index $i \geq 0$,
	we define the following preorder $\preceq_i$ on infinite initial traces on $w$:
	Given two infinite initial traces $\pi, \pi'$ on $w$, with
	$\pi = q_0 \goesto {\symb_0} q_1 \goesto {\symb_1} \cdots$ and
	$\pi' = q_0' \goesto {\symb_0} q_1' \goesto {\symb_1} \cdots$,
	we write $\pi \preceq_i \pi'$ whenever the following condition is satisfied:
	\begin{align*}
		\pi \preceq_i \pi' \quad \iff \quad \mathcal C^{\mathrm {di}}(\pi, \pi') \textrm{ and } \forall j \ge i \cdot q_j \disim q_j',
	\end{align*}
	and we write $\pi \prec_i \pi'$ when, additionally, $q_i \strictdisim q_i'$.
	%or $q_i = q_i'$ and $(q_{i-1} \goesto {\symb_{i-1}} q_i) \, P \, (q_{i-1}' \goesto {\symb_{i-1}} q_i')$.
	%
	%Notice that $\preceq_i \subseteq \preceq_{i+1}$.
	%
	%Notice that $\prec_i$ is acyclic and it satisfies $\prec_i \circ \preceq_i\ \subseteq\ \prec_i$.
	%
	Moreover, we say that $\pi$ is \emph{$i$-good} whenever its first $i$ transitions are also possible in $\B$.
	We show, by induction on $i$, the following property (PP):
	For every infinite initial trace $\pi$ on $w$ and every $i \ge 0$,
	there exists an infinite initial trace $\pi'$ on $w$ s.t.~$\pi \preceq_i \pi'$ and $\pi'$ is $i$-good.
	The base case $i=0$ is trivially true with $\pi' = \pi$.
	For the induction step, consider an infinite initial trace $\pi$ on $w$.
	By induction hypothesis, there exists an infinite initial trace 
	$\pi^1 = q_0^1 \goesto {\symb_0} q_1^1 \goesto {\symb_1} \cdots$ on $w$
	s.t.~$\pi \preceq_i \pi^1$ and $\pi^1$ is $i$-good.
	Consequently, $\pi \preceq_{i+1} \pi^1$ holds,
	and we may additionally assume that $\pi^1$ is \emph{maximal}
	in the sense that there is no other $\pi'$ which is $i$-good and $\pi^1 \prec_{i+1} \pi'$.
	We argue that such a maximal $\pi^1$ is necessarily $(i+1)$-good.
	By contradiction, assume that the transition $q_i^1 \goesto {\symb_i} q_{i+1}^1$ is not in $\B$.
	Then there exists a transition $q_i^2 \goesto {\symb_i} q_{i+1}^2$ in $\B$
	%s.t.~$(q_i^1 \goesto {\symb_i} q_{i+1}^1) \, P \, (q_i^2 \goesto {\symb_i} q_{i+1}^2)$.
	%
	s.t.~$q_i^1 \bwdirecttraceinclusion q_i^2$ and $q_{i+1}^1 \strictdisim q_{i+1}^2$.
	From the definitions of $\bwdirecttraceinclusion$ and $\disim$
	it follows that there exists an infinite initial trace $\pi^2 = q_0^2 \goesto {\symb_0} q_1^2 \goesto {\symb_1} \cdots$ on $w$
	s.t.~$\pi^1 \prec_{i+1} \pi^2$.
	(This last property uses the fact that $\disim$ propagates forward. 
	Direct trace inclusion $\directtraceinclusion$ does not suffice.)
	%
	%Since $(q_i^1 \goesto {\symb_i} q_{i+1}^1) \, P \, (q_i^2 \goesto {\symb_i} q_{i+1}^2)$,
	%we either have $q_{i+1}^1 \strictdisim q_{i+1}^2$ or $q_{i+1}^1 = q_{i+1}^2$, % and $q_i^1 \bwdirecttraceinclusion q_i^2$,
	%and thus $\pi^1 \prec_{i+1} \pi^2$.
	%
	By induction hypothesis,
	there exists an infinite initial trace $\pi^3 = q_0^3 \goesto {\symb_0} q_1^3 \goesto {\symb_1} \cdots$ on $w$
	s.t.~$\pi^2 \preceq_i \pi^3$ (and thus $\pi^2 \preceq_{i+1} \pi^3$) and $\pi^3$ is $i$-good.
	Thus, $\pi^1 \prec_{i+1} \pi^3$, which contradicts the maximality of $\pi^1$.
	Therefore, $\pi^1$ is $(i+1)$-good.

	Given the infinite fair initial trace $\hat{\pi}$ on $w$ in $\A$,
	it follows from property (PP) and K\"onig's Lemma that there
	exists an infinite initial trace $\tilde{\pi}$ on $w$ that is
	$i$-good for every $i$ and $\mathcal C^{\mathrm {di}}(\hat{\pi},
	\tilde{\pi})$.
	Therefore $\tilde{\pi}$ is an infinite fair initial trace 
	on $w$ in $\B$, and thus $w \in \lang{\B}$.
	%, as required.
\end{proof}

Notice that Theorems~\ref{thm:prune_id_strictdirecttraceinclusion} (about $\makeprunerel{\id}{\strictdirecttraceinclusion}$)
and \ref{thm:bwsim-fwtrace} (about $\makeprunerel{\strictbwdisim}{\directtraceinclusion}$) are incomparable:
In the former, we require the source endpoints to be the same (which is forbidden by the latter),
and in the latter we allow the destination endpoints to be the same (which is forbidden by the former),
and it is not clear whether one can find a common GFP generalization.
For the same reason, Theorems \ref{thm:bwtrace-id} (about $\makeprunerel{\strictbwdirecttraceinclusion}{\id}$)
and \ref{thm:bwtrace-fwsim} (about $\makeprunerel{\bwdirecttraceinclusion}{\strictdisim}$) are also incomparable.

\subsubsection*{Pruning w.r.t.~transient transitions.}

Recall that a transition is \emph{transient} when it appears at most once on
every path of the automaton.
(Analogously, one can define transient states.
E.g., \cite{somenzi:efficient} consider variants of direct/backward simulations that
do not care about the accepting status of transient states.)
While at the beginning of the section we observed that $\makeprunerel{\id}{\strictlanguageinclusion}$ is not GFP,
it is correct to remove a transition w.r.t.~$\makeprunerel{\id}{\strictlanguageinclusion}$
when it is subsumed by a transient one \cite[Theorem 3]{somenzi:efficient}.

This motivates us to define the following transient variant of $\makeprunerel{\brel}{\frel}$,
for $\brel, \frel \subseteq Q \times Q$:
\begin{align*}
	\makeprunereltransient{\brel}{\frel} =
	\{((p,\symb,r),(p',\symb,r')) \in \delta \times \delta \st p \brel p', r \frel r', \textrm{ and } (p',\symb,r') \textrm{ is transient} \}.
\end{align*}
The relation $\makeprunereltransient{\id}{\strictlanguageinclusion}$
using the very coarse fair trace inclusion $\strictlanguageinclusion$
is GFP for NBA \cite{somenzi:efficient}.
We note that one cannot relax the source endpoint to go beyond the identity.
In fact, $\makeprunereltransient{\strictbwdisim}{\strictlanguageinclusion}$
---and even $\makeprunereltransient{\strictbwdisim}{\strictdesim}$---is already not GFP.
A counterexample is shown in Fig.~\ref{fig:transient:not:GFP}:
Both transitions $p \goesto a q$ and $q \goesto a r$ are transient,
and $(q, a, r) \makeprunereltransient{\strictbwdisim}{\strictdesim} (p, a, q)$.
%$r \strictlanguageinclusion q$ (even $r \strictdesim q$),
%and $q \strictbwdisim p$.
%
However, removing the smaller transition $q \goesto a r$ changes the language,
since $a^\omega$ is no longer accepted.

However, one can combine pruning w.r.t.~transient transitions using the coarse fair trace inclusion,
and \emph{simultaneously} pruning w.r.t.~all transitions using direct trace inclusion.
Let $R_t \subseteq \delta \times \delta$ be the relation on transitions defined as
$R_t = \makeprunerel{\id}{\strictdirecttraceinclusion} \cup \makeprunereltransient{\id}{\strictlanguageinclusion}$.

%
% \begin{align*}
% 	
% \end{align*}
%
We will use the fact that $R_t$ is GFP when describing our automata reduction algorithm in Sec.~\ref{sec:heavyandlight}.
The following result thus generalizes Theorem~\ref{thm:prune_id_strictdirecttraceinclusion}.
\begin{theorem}
	\label{thm:prune_transient}
	The relation $R_t$ is GFP on NBA.
\end{theorem}
\begin{proof}
        Even though the relation $R_t$ is not transitive in general, it is
        acyclic since $R_t \subseteq \makeprunerel{\id}{\strictlanguageinclusion}$.
        Let $\B = \prune{\A}{R_t}$.
	To show $\A \languageinclusion \B$, let $w = \symb_0\symb_1 \cdots\in \lang{\A}$,
	and let $\hat{\pi}$ be any infinite fair initial trace on $w$ in $\A$.
	We call a trace $\pi = q_0 \goesto {\symb_0} q_1 \goesto {\symb_1} \cdots$ on $w$ in $\A$
	$i$-{\em maximal} if it does not contain any transition $q_j \goesto {\symb_j} q_{j+1}$ for $j < i$
	s.t.~there exists an $\A$ transition $q_j \goesto {\symb_j} q_{j+1}'$
        with $(q_j,\symb_j,q_{j+1}) R_t (q_j,\symb_j,q_{j+1}')$.
        Moreover, let ${\it tt}_i(\pi)$ be the number of transient transitions
        occurring in the first $i$ steps of $\pi$.

	Since $\A$ is finitely branching and $R_t$ is acyclic,
	for every state and symbol there exists at least one $R_t$-maximal successor that is still present in $\B$.
	Thus, for every $i$-maximal fair trace $\pi$ on $w$
	there exists an $(i+1)$-maximal fair trace $\pi'$ on $w$ 
	s.t.~$\pi$ and $\pi'$ are identical on the first $i$ steps.

	Since $\hat{\pi}$ is an infinite fair initial trace on $w$
	(which is trivially $0$-maximal),
	for every $i$ there exists an infinite fair initial trace $\tilde{\pi}_i$ which is $i$-maximal
	and agrees with $\tilde{\pi}_{i-1}$ on the first $i-1$ steps.
        Consider the sequence of these traces $\tilde{\pi}_i$ for increasing
        $i$. We have 
        ${\it tt}_{i-1}(\tilde{\pi}_{i-1}) = {\it tt}_{i-1}(\tilde{\pi}_{i}) \le
        {\it tt}_i(\tilde{\pi}_i)$.

	Since no transient transition can repeat twice in a run,
        the limit $\lim_{i \rightarrow \infty} {\it tt}_i(\tilde{\pi}_i)$
        is bounded from above by the finite number of transient transitions in $\B$.
	Thus there exists a finite number 
        $N = \lim_{i \rightarrow \infty} {\it tt}_i(\tilde{\pi}_i)$.
        Let $N'$ be the smallest number where the limit is reached,
        i.e., $N' := \min\{i\ |\ {\it tt}_i(\tilde{\pi}_i) = N\}$.
        In particular, $N' \ge N$.
        Since, for every $i \geq N'$, the trace $\tilde{\pi}_i$ agrees with
        $\tilde{\pi}_{N'}$ on the first $N'$ steps, it follows that
        $\tilde{\pi}_i[N'..]$ does not contain any transient transition.
        Thus for every $N' \le i \leq j$, $\mathcal C^{\mathrm {di}}(\tilde{\pi}_i[N'..],\tilde{\pi}_j[N'..])$.
        I.e., after $N'$ steps we are effectively pruning w.r.t.\ $\makeprunerel{\id}{\strictdirecttraceinclusion}$
	(and not $\makeprunereltransient{\id}{\strictlanguageinclusion}$),
	and $\directtraceinclusion$ preserves the position of accepting states.
	By arranging the $\tilde{\pi}_i$'s in a finitely-branching tree,
	by K\"onig's lemma there exists a infinite fair initial trace $\tilde{\pi}_\infty$ which is $i$-maximal for every $i$.
	Therefore, $\tilde{\pi}_\infty$ is a trace in $\B$ (by maximality),
	and thus $w \in \lang{\B}$.
\end{proof}

%
%
%In particular, one can show that $R_t(\strictlanguageinclusion)$ is GFP.

\begin{figure}
	
	\begin{tikzpicture}[on grid, node distance= .6cm and 1.4cm]
		\tikzstyle{vertex} = [smallstate]

		\path node [vertex, initial] (p) {$p$};
		\path node [vertex, initial] (q) [right = of p] {$q$};
		\path node [vertex, accepting] (r) [right = of q] {$r$};

		\path[->]

			(p) edge node [above] {$a$} (q)
			(q) edge [bend left = 15, dashed] node [above] {$a$} (r)
			(q) edge [bend right = 15] node [below] {$b$} (r)
			(p) edge [loop above] node {$a,b$} ()
			(r) edge [loop above] node {$a$} ();

	\end{tikzpicture}
	
	%\makeprunereltransient{\strictbwdisim}{\strictlanguageinclusion}$ and even 
	\caption{$\makeprunereltransient{\strictbwdisim}{\strictdesim}$ is not GFP.}	
	\label{fig:transient:not:GFP}
	
\end{figure}

%%% Local Variables:
%%% mode: latex
%%% TeX-master: "ROOT.tex"
%%% End:

\subsection{Pruning NFA}
\label{sec:pruning:NFA}

The proofs of the following theorems are entirely similar to their equivalents for NBA from the previous section---
except for the fact that a simple induction on the length of the word suffices (and thus K\"onig's Lemma is not needed),
and thus they will not be repeated here.
The difference is that forward trace inclusion $\finincl$ needs only to match accepting states at the end of the computation
(and not throughout the computation as in NBA),
and, symmetrically, backward trace inclusion $\bwfinincl$ needs only to match initial states at the beginning of the computation
(and not also accepting states throughout as in NBA).
An analogue of pruning transient transitions for NBA as in Theorem~\ref{thm:prune_transient} is missing for NFA,
since pruning w.r.t.~coarser acceptance conditions like in delayed or fair trace inclusion does not apply to finite words.

\begin{theorem}\label{thm:prune_id_strictdirecttraceinclusion:NFA}
	For every strict partial order $R \mathrel{\subseteq} \finincl$,
	$\makeprunerel{\id}{R}$ is GFP on NFA. In particular, $\makeprunerel{\id}{\strictfinincl}$ is GFP.
\end{theorem}

\begin{theorem}\label{thm:bwtrace-id:NFA}
	For every strict partial order $R \mathrel{\subseteq} \bwfinincl$,
	$\makeprunerel{R}{\id}$ is GFP on NFA.
	In particular, $\makeprunerel{\strictbwfinincl}{\id}$ is GFP.
\end{theorem}

\begin{theorem}\label{thm:bwsim-fwtrace:NFA}
	%For every partial order $R \subseteq \directtraceinclusion$,
	The relation $\makeprunerel{\strictbwsim}{\finincl}$ is GFP on NFA.
\end{theorem}

\begin{theorem}\label{thm:bwtrace-fwsim:NFA}
	The relation $\makeprunerel{\bwfinincl}{\strictdisim}$ is GFP on NFA.
%	Let $P$ be a strict preorder subset of $\makeprunerel \bwdirecttraceinclusion {\id\ \cup \strictdisim}$.
%	Then, $P$ is GFP.
\end{theorem}

%%% Local Variables:
%%% mode: latex
%%% TeX-master: "ROOT.tex"
%%% End:

\section{Lookahead Simulation}
\label{sec:lookahead}

While trace inclusions are theoretically appealing as GFQ/GFI/GFP preorders coarser than simulations,
it is not feasible to use them in practice, because they are too hard to compute (even their membership problem is PSPACE-complete \cite{MeyerStockmeyer:Equivalence:1972,kupfermanvardi:fair_verification}).
Multipebble simulations (\cite{etessami:hierarchy02}; cf.~Sec.~\ref{sec:multipebble_simulations})
constitute sound under-approximations to trace inclusions,
%where Duplicator is allowed to control several pebbles instead of just one,
and by varying the number of pebbles one can achieve a better tradeoff between complexity and size than just computing the full trace inclusion.
%
%For $k > 0$, $k$-pebble simulation is coarser than ordinary simulation and it implies trace inclusion,
%and by increasing $k$, one can control the quality of the approximation to trace inclusion.
%Direct, delayed, fair and backward pebble simulations are not transitive in general,
%but their transitive closures are GFI preorders;
%the direct, delayed and backward variants are also GFQ.
%
However, computing multipebble simulations with $k>0$ pebbles requires solving a game of size $n \cdot n^k$
(where $n$ is the number of states of the automaton),
which is not feasible in practice, even for modest values for $k$.
(Even for $k=2$ one has a cubic best-case complexity, which severely limits
the size of $n$ that can be handled.)
For this reason, we consider a different way to extend Duplicator's power, % in the simulation game,
i.e., by using \emph{lookahead} on the moves of Spoiler.
While lookahead itself is a classic concept,
it can be defined in several ways in the context of adversarial games, like simulation.
We compare different variants for computational efficiency and approximation quality:
\emph{multistep simulation} (Sec.~\ref{sec:multistep:simulation}),
\emph{continuous simulation} (Sec.~\ref{sec:continuous:simulation}),
culminating in \emph{lookahead simulation} (\ref{sec:lookahead:simulation}),
which offers the best compromise,
and it is the main object of study of this section.
We will use lookahead simulation in our automata reduction (Sec.~\ref{sec:heavyandlight})
and inclusion testing algorithms (Sec.~\ref{sec:inclusion}).
In the following, we let $n$ be the number of the states of the automaton.

\subsection{Multistep simulation.}
\label{sec:multistep:simulation}

In \emph{$k$-step simulation} the players select sequences of transitions of length $k > 0$ instead of single transitions.
%Formally, from configuration $(p_i, q_i)$ Spoiler chooses $k$ transitions
%$p_i \goesto {\symb_i} p_{i+1} \goesto {\symb_{i+1}} \dots \goesto {\symb_{i+k-1}} p_{i+k}$,
%and Duplicator responds by choosing a matching sequence
%$q_i \goesto {\symb_i} q_{i+1} \goesto {\symb_{i+1}} \dots \goesto {\symb_{i+k-1}} q_{i+k}$;
%the next configuration is $(p_{i+k}, q_{i+k})$.
%That is, $k$ steps of the simulation game are compressed into one round ($k=1$ corresponds to ordinary simulation):
This gives Duplicator more information,
and thus yields a larger simulation relation.
%
%The size of the game graph for $k$-step simulation is bounded by ${\cal O}(n^2\cdot(d^k+1))$,
%${\cal O}(n^2+k\cdot\log(d))$.
%
%(where $d$ is the maximal out-degree of the automaton),
%while the time is polynomial in $n$ and in $d^k$.
In general, $k$-step simulation %(called \emph{static $k$-letter simulation} in \cite{lange:lookahead:2013})
and $k$-pebble simulation are incomparable,
but $k$-step simulation is strictly contained in $n$-pebble simulation.
However, the rigid use of lookahead in big-steps causes at least two issues:
First, for an NBA with maximal out-degree $d$,
in {\em every round} we have to explore up-to $d^k$ different moves for each player,
which is too large in practice (e.g., $d=4$, $k=12$).
Second, Duplicator's lookahead varies between $1$ and $k$,
depending where she is in her response to Spoiler's long move.
Thus, Duplicator might lack lookahead where it is most needed,
while having a large lookahead in other situations where it is not useful.
In the next notion, we attempt at ameliorating this.

\subsection{Continuous simulation.}
\label{sec:continuous:simulation}

In \emph{$k$-continuous simulation},
Duplicator is continuously kept informed about Spoiler's next $k > 0$ moves,
i.e., she always has lookahead $k$.
Initially, Spoiler makes $k$ moves,
and from this point on they alternate making one move each (and matching the corresponding input symbols).
Thus, $k$-continuous simulation is coarser than $k$-step simulation.
In general, it is incomparable with $k$-pebble simulation for $k < n$,
but it is always contained in $k$-pebble simulation for $k = n$,
and there are examples where the containment is strict.
Note that here the configuration of the game consists not only of the current 
states of Spoiler and Duplicator, but also of the announced $k$ next moves of Spoiler.
While this is arguably the strongest way of giving lookahead to Duplicator,
it requires storing $n^2 \cdot d^{k-1}$ configurations (for branching degree $d$), 
which is infeasible for non-trivial $n$ and $k$ (e.g., $n=10000$, $d=4$, $k=12$).

\subsection{Lookahead simulation.}
\label{sec:lookahead:simulation}

We introduce $k$-lookahead simulation as an optimal compromise between the finer $k$-step and the coarser $k$-continuous simulation.
Intuitively, we put the lookahead under Duplicator's control,
who can choose \emph{at each round} and \emph{depending on Spoiler's move} 
how much lookahead she needs (up to $k$).
Formally, configurations are pairs $(p_i, q_i)$ of states.
From configuration $(p_i, q_i)$, one round of the game is played as follows.
\begin{itemize}
	\item Spoiler chooses a sequence of $k$ consecutive transitions
		$p_i \goesto {\symb_i} {p_{i+1}} \goesto {\symb_{i+1}} \cdots \goesto {\symb_{i+k-1}} p_{i+k}$.
	\item Duplicator chooses a degree of lookahead $m$ such that $1 \le m \le k$.
	\item Duplicator responds with a sequence of $m$ transitions 
		$q_i \goesto {\symb_i} {q_{i+1}} \goesto {\symb_{i+1}} \cdots \goesto {\symb_{i+m-1}} q_{i+m}$.
\end{itemize}
The remaining $k-m$ moves of Spoiler
$p_{i+m} \goesto {\symb_{i+m}} {p_{i+m+1}} \goesto {\symb_{i+m+1}} \cdots \goesto {\symb_{i+k-1}} p_{i+k}$
\emph{are forgotten},
and the next configuration is $(p_{i+m}, q_{i+m})$;
in particular, in the next round Spoiler can chose a different attack from $p_{i+m}$.
In this way, the players build as usual two infinite traces $\pi_0$ from $p_0$ and $\pi_1$ from $q_0$.
Backward simulation is defined similarly with backward transitions.
For any acceptance condition $x \in \{\mathrm{di, de, f}\}$,
Duplicator wins this play if $\mathcal C^x(\pi_0, \pi_1)$ holds,
for $x = \mathrm{bw\textrm-di}$ we require $\mathcal C^\mathrm{bw}_{I, F}(\pi_0, \pi_1)$ (cf.~ Sec.~\ref{sec:backward}),
and for $x = \mathrm{bw}$ we require $\mathcal C^\mathrm{bw}_I(\pi_0, \pi_1)$ (cf.~ Sec.~\ref{sec:simulations:finitewords}).
\begin{definition}\label{def:lookahead-sim}
	Two states $(p_0, q_0)$ are in \emph{$k$-lookahead $x$-simulation},
	written $p_0 \kxsim k x q_0$,
	iff	Duplicator has a winning strategy in the above game.
\end{definition}

In general, greater lookahead yields coarser lookahead relations,
i.e., $\kxsim h x \subseteq \kxsim k x$ whenever $h \leq k$,
and moreover it is not difficult to find examples where the inclusion is actually strict when $h < k$.
A simple such example (not depending on the choice of $x$) for the case $h = 1$ and $k = 2$ can be found in Fig.~\ref{fig:lookhead_non_transitive} (which is also used below to show non-transitivity):
First, we have $p_0 \not\ksim 1 q_0$,
since Duplicator must choose whether to go to $q_1$ (and then Spoiler wins by playing $b$)
or to $q_2$ (and then Spoiler wins by playing $a$).
Moreover, $p_0 \ksim 2 q_0$ holds,
since now with lookahead $k = 2$ we let Duplicator take the transition via $q_1$ or $q_2$
depending on whether Spoiler plays the word $(a+b)a$ or $(a+b)b$, respectively.
%and Fig.~\ref{fig:lookhead_non_preserved}.

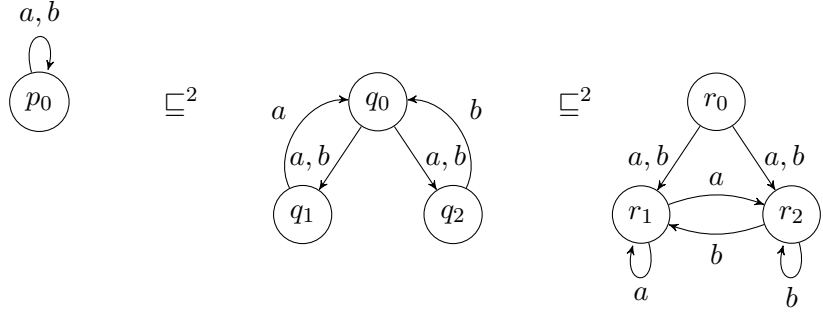
\begin{figure} \centering
	\begin{tikzpicture}[on grid, node distance=2cm and 2cm]
		\tikzstyle{vertex} = [smallstate]
				
		\path node [vertex] (p0) {$p_0$};
		
		\path node [vertex] (q0) [right = 4.5cm of p0] {$q_0$};
		\path node [vertex] (q1) [below left = 1.5cm and 1cm of q0] {$q_1$};
		\path node [vertex] (q2) [below right = 1.5cm and 1cm of q0] {$q_2$};
		
		\path node [vertex] (r0) [right = 4.5cm of q0] {$r_0$};
		\path node [vertex] (r1) [below left = 1.5cm and 1cm of r0] {$r_1$};
		\path node [vertex] (r2) [below right = 1.5cm and 1cm of r0] {$r_2$};
		
		\path[->]
		
			(p0) edge [loop above] node {$a, b$} ()
			
			(q0) edge node [left] {$a,b$} (q1)
			(q0) edge node [right] {$a,b$} (q2)
			(q1) edge [bend left = 60] node [above left] {$a$} (q0)
			(q2) edge [bend right = 60] node [above right] {$b$} (q0)
			
			(r0) edge node [left] {$a,b$} (r1)
			(r0) edge node [right] {$a,b$} (r2)
			(r1) edge [loop below] node {$a$} ()
			(r1) edge [bend left = 20] node [above] {$a$} (r2)
			(r2) edge [loop below] node {$b$} ()
			(r2) edge [bend left = 20] node [below] {$b$} (r1);
			
		\path
			(p0) -- node [pos = .4] {$\ksim 2$} (q0)
			(q0) -- node [pos = .6] {$\ksim 2$} (r0);
		
		\begin{pgfonlayer}{background}
			%\node () [color=blue!50, inner sep=4pt, ellipse, fit=(p0) (q)] {};
		\end{pgfonlayer}
			
	\end{tikzpicture}
	
	\caption{A lookeahead simulation example.}
	
	\label{fig:lookhead_non_transitive}
	
\end{figure}
%%% Local Variables:
%%% mode: latex
%%% TeX-master: "ROOT.tex"
%%% End:

\begin{remark}
	\label{remark:non_transitive}
	$k$-lookahead simulation is not transitive for $k \geq 2$.
	Consider again the example in Fig.~\ref{fig:lookhead_non_transitive}.
	We have $p_0 \ksim k q_0 \ksim k r_0$ (and $k = 2$ suffices),
	but $p_0 \not \ksim k r_0$ for any $k > 0$.
	We have already seen above that $p_0 \ksim k q_0$ holds for $k = 2$.
	Moreover, $q_0 \ksim k r_0$ holds by the following reasoning, with $k = 2$:
		If Spoiler goes to $q_1$ or $q_2$, then Duplicator goes to $r_1$ or $r_2$, respectively.
		Then, it can be shown that $q_1 \ksim k r_1$ holds as follows
		(the case $q_2 \ksim k r_2$ is similar).
		If Spoiler takes transitions $q_1 \goesto a q_0 \goesto a q_1$,
		then Duplicator does $r_1 \goesto a r_1 \goesto a r_1$,
		and	if Spoiler does $q_1 \goesto a q_0 \goesto b q_1$,
		then Duplicator does $r_1 \goesto a r_2 \goesto b r_1$.
		The other cases are similar.
	Finally, $p_0 \not \ksim k r_0$, for any $k > 0$:
	From $r_0$,	Duplicator can play a trace for any word $w$ of length $k > 0$,
	but in order to extend it to a trace of length $k + 1$ for any $w' = wa$ or $wb$,
	she needs to know whether the last $(k + 1)$-th symbol is $a$ or $b$.
	Thus, no finite lookahead suffices for Duplicator.
	%
	%(Incidentally, notice that $r_0$ simulates $p_0$ with $k$-continuous simulation, and $k = 2$ suffices.)
	%
\end{remark}

Non-transitivity of lookahead simulation $\kxsim k x$ (unless $k=1$) is not an obstacle to its applications.
Since we use it to under-approximate suitable preorders,
we consider its transitive closure instead (which is a preorder),
which we denote by $\transksimx$.
Moreover, we denote its asymmetric restriction by $\stricttransksimx = \transksimx \setminus (\transksimx)^{-1}$.
\begin{lemma}
	\label{lem:lookahead_sim_GFI_GFQ}
	For $k > 0$ and $x \in \{\mathrm{di, de, f, bw\textrm-di}\}$,
	the transitive closure of $k$-lookahead $x$-simulation $\transksimx$ is GFI.
	Moreover, it is GFQ for $x \neq \mathrm f$.
\end{lemma}
\begin{proof}
	Being GFQ/GFI follows from the fact that the transitive closure of lookahead simulation is
        included in the corresponding trace inclusion/multipebble simulation.
        Moreover direct/delayed multipebble simulations are included in
        delayed fixed-word simulation; cf.\ Figure~\ref{fig:GFQ_relations}.  
	These are is GFI (cf.~Sec.~\ref{sec:language_inclusion}, and in particular Theorem~\ref{lem:bwincl_GFI} for backward trace inclusion),
	and GFQ for $x \in \set{\mathrm {di, de}}$ by %Lemma~\ref{lem:GFQ-directtraceinclusion},
	Lemma~\ref{lem:GFQ-delayedfixedwordsimulation},
	and for $x = \mathrm {bw\textrm-di}$ by Theorem~\ref{lem:bwincl_GFQ}.
\end{proof}

\begin{figure} \centering
	
	\subfigure[The original automaton $\A$]{
	
		\begin{tikzpicture}[on grid, node distance=2cm and 2cm]
			\tikzstyle{vertex} = [smallstate]
						
			\path node [vertex] (p0) {$p_0$};
		
			\path node [vertex] (q0) [right = 3cm of p0] {$q_0$};
			\path node [vertex] (q1) [below left = 1.5cm and 1cm of q0] {$q_1$};
			\path node [vertex] (q2) [below right = 1.5cm and 1cm of q0] {$q_2$};
		
			\path node [vertex] (r0) [right = 3cm of q0] {$r_0$};
			\path node [vertex] (r1) [below left = 1.5cm and 1cm of r0] {$r_1$};
			\path node [vertex] (r2) [below right = 1.5cm and 1cm of r0] {$r_2$};
		
			\path node [vertex] (q) [above = 1.5cm of q0 ] {$q$};
			\path node [vertex] (r) [above = 1.5cm of r0 ] {$r$};
		
			\path[->]
		
				(p0) edge [loop above] node {$a, b$} ()
			
				(q0) edge node [left] {$a,b$} (q1)
				(q0) edge node [right] {$a,b$} (q2)
				(q1) edge [bend left = 60] node [above left] {$a$} (q0)
				(q2) edge [bend right = 60] node [above right] {$b$} (q0)
			
				(r0) edge node [left] {$a,b$} (r1)
				(r0) edge node [right] {$a,b$} (r2)
				(r1) edge [loop below] node {$a$} ()
				(r1) edge [bend left = 20] node [above] {$a$} (r2)
				(r2) edge [loop below] node {$b$} ()
				(r2) edge [bend left = 20] node [below] {$b$} (r1)
			
				(q) edge node [right] {$a$} (q0)
				(q) edge node [above right] {$a$} (r0)
				(r) edge node [right] {$a$} (r0);
			
			\path
				(p0) -- node [pos = .3] {$\ksimeq 2$} (q0)
				(q0) -- node [pos = .7] {$\ksimeq 2$} (r0);
		
		\end{tikzpicture}
		
	}
	$\qquad$
	\subfigure[The quotient automaton $\A/\!\transksim 2$]{
	
		\begin{tikzpicture}[on grid, node distance=1.5cm and 1.5cm]
			\tikzstyle{vertex} = [smallstate]
					
			\path node [vertex] (qr) {$\set{q, r}$};
			\path node [vertex] (pqr0) [below = 2.5 of qr] {$\set {p_0, q_0, r_0}$};
			\path node [vertex] (qr1) [below left = 2.5cm and 1.5cm of pqr0] {$\set{q_1, r_1}$};
			\path node [vertex] (qr2) [below right = 2.5cm and 1.5cm of pqr0] {$\set{q_2, r_2}$};
				
			\path[->]
	
				(qr) 	edge 					node [left] 		{$a$} (pqr0)
				(pqr0)	edge 					node [above left]	{$a, b$} (qr1)
				(pqr0)	edge 					node [above right]	{$a, b$} (qr2)
				(qr1)	edge [loop below]		node [below]		{$a$} ()
				(qr1)	edge [bend left = 60]	node [above left]	{$a$} (pqr0)
				(qr1)	edge [bend left = 30]	node [above]		{$a$} (qr2)
				(qr2)	edge [loop below]		node [below]		{$b$} ()
				(qr2)	edge [bend right = 60]	node [above right]	{$b$} (pqr0)
				(qr2)	edge [bend left = 30]	node [below]		{$b$} (qr1);
		
			\draw[->] (pqr0) to [out=30,in=60,looseness=8] node [above right]	{$a,b$} (pqr0);
		
		\end{tikzpicture}
	
	}
	
	\caption{Lookeahead simulation is not preserved under quotienting.}
	
	\label{fig:lookhead_non_preserved}
	
\end{figure}
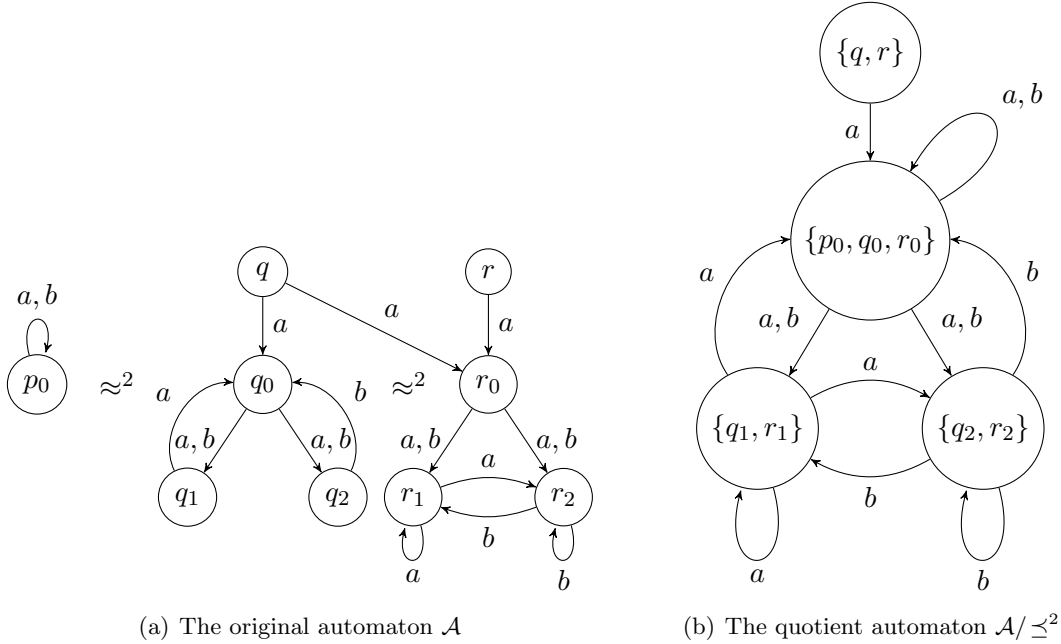
%%% Local Variables:
%%% mode: latex
%%% TeX-master: "ROOT.tex"
%%% End:

\begin{remark}
	Let $\transksim k$ be the transitive closure of $k$-lookahead simulation $\ksim k$.
	While quotienting w.r.t.~ordinary simulation (i.e., lookahead $k = 1$)
	preserves ordinary simulation itself
	in the sense that a quotient class $[p]$ in the quotient automaton $\A/\!\transksim k$
	is simulation equivalent to $p$,
	this is not the case when considering larger lookahead $k > 1$.
	This is a consequence of lack of transitivity; cf.~Fig.~\ref{fig:lookhead_non_preserved},
	which builds on the previous non-transitivity example of Fig.~\ref{fig:lookhead_non_transitive}.
	Here and in the following we define $\ksimeq k$ as $\transksim k \cap \transksimrev k$,
	i.e., the largest equivalence included in $\transksim k$.
	(Notice that $\A/\!\transksim k$ is the same as $\A/\!\ksimeq k$ by definition of quotienting.)
	We have that $p_0 \ksimeq 2 q_0 \ksimeq 2 r_0$, %(since trivially $r_0 \ksimeq 2 p_0$),
	$q_1 \ksimeq 2 r_1$, $q_2 \ksimeq 2 r_2$, and $q \ksimeq 2 r$
	(which follows from the discussion in Remark~\ref{remark:non_transitive}),
	and thus we obtain the quotient automaton $\A/\!\transksim 2$ on the right.
	However, $\set{q, r} \not\ksim 2 r$,
	since Spoiler can play $\set{q, r} \goesto a \set{p_0, q_0, r_0} \goesto a \set{p_0, q_0, r_0}$
	and Duplicator replies with either
		(1) $r \goesto a r_0$, but this is losing since $\set{p_0, q_0, r_0} \not\ksim 2 r_0$ 
        (cf.~the discussion of $p_0 \not\ksim 2 r_0$ in Remark~\ref{remark:non_transitive}), or
		(2) $r \goesto a r_0 \goesto a r_1$,
		but this is losing since Spoiler can then play a $b$ letter (which is not available from $r_1$), or symmetrically
		(3) $r \goesto a r_0 \goesto a r_2$, but this is losing too since Spoiler plays $a$ in this case.
	
	While lookahead simulation is not preserved under quotienting,
	the lemma above shows that the recognized language is nonetheless preserved,
	which is all that we care about for correctness.
\end{remark}

Lookahead simulation offers better reduction under quotienting than ordinary (i.e., $k = 1$) simulation.
We will define a family of automata $\A_n$ of size $O(n^2)$
which is not compressible w.r.t.~ordinary simulation,
but which is compressed to size $O(n)$ w.r.t.~simulation with lookahead $k = 2$.
Therefore, quotienting w.r.t.~lookahead simulation performs better than w.r.t.~ordinary simulation by a linear factor at least.
The construction of $\A_n$ is as follows.
The alphabet is $\Sigma = \set{a, b_1, \dots, b_n}$.
There is a state $p_{\set{i,j}}$ for every unordered pair $\set{i, j} \subseteq \set{1, \dots, n}$,
there is a state $q_i$ for every $i \in \set{1, \dots, n}$,
and finally we have a state $r$.
Transitions are as follows:
$p_{\set{i,j}} \goesto a q_i$, $p_{\set{i,j}} \goesto a q_j$,
and $q_i \goesto {\Sigma\setminus\set{b_i}} r$ for every unordered pair $\set{i, j} \subseteq \set{1, \dots, n}$.
This automaton is incompressible w.r.t.~ordinary simulation
since each two distinct $p_{\set {i,j}}, p_{\set{k,h}}$ are $\ksim 1$-incomparable:
For instance, assume w.l.o.g.~that $i \not\in \set{k, h}$.
Spoiler takes transition $p_{\set {i,j}} \goesto a q_i$,
and now Duplicator takes either transition $p_{\set{k,h}} \goesto a q_k$,
which is losing since Spoiler plays $b_k$ in this case,
or transition $p_{\set{k,h}} \goesto a q_h$,
which is also losing since Spoiler plays $b_h$ in this case.
On the other hand, with lookahead $k = 2$ we can readily see that $p_{\set{i, j}} \ksimeq 2 p_{\set{h, k}}$
(thus falling in the same quotient class),
since now %Spoiler has to announce in the first round whether it will not play $b_i$ or $b_j$ in the second round,
Duplicator can always match Spoiler's choice in the second round because $\Sigma\setminus\set{b_h} \cup \Sigma\setminus\set{b_k} = \Sigma$.

Lookahead simulation offers the optimal tradeoff between $k$-step and $k$-continuous simulation.
Since the lookahead is discarded at each round,
$k$-lookahead simulation is (strictly) included in $k$-continuous simulation
(where the lookahead is never discarded).
However, this has the benefit of only requiring to store $n^2$ configurations,
%(which is needed anyway when computing any binary relation)
which makes computing lookahead simulation space-efficient.
On the other hand, when Duplicator always chooses a maximal reply $m = k$
we recover $k$-step simulation,
which is thus included in $k$-lookahead simulation.
Moreover, thanks to the fact that Duplicator controls the lookahead,
most rounds of the game can be solved without ever reaching the maximal lookahead $k$:
\begin{enumerate*}[label=\arabic*)]
	\item for a fixed attack by Spoiler,
	we only consider Duplicator's responses for small $m = 1, 2, \dots, k$
	until we find a winning one, and
	\item also Spoiler's attacks can be built incrementally	since,
	if he loses for some lookahead, then he also loses for any larger one.
\end{enumerate*}
In practice, this greatly speeds up the computation,
and allows us to use lookaheads in the range $4$-$25$,
depending on the size and structure of the automata;
see Sec.~\ref{sec:experiments} for the experimental evaluation and
benchmark against the GOAL tool \cite{GOAL_survey_paper}.

\begin{remark}\label{rem:lookahead-pebble}
$k$-lookahead simulation can also be expressed as a restriction of $n$-pebble simulation,
where Duplicator is allowed to split pebbles maximally (thus $n$-pebbles),
but after a number $m \le k$ rounds (where $m$ is chosen dynamically by Duplicator)
he has to discard all but one pebble.
Then Duplicator is allowed to split pebbles maximally again, etc.
Thus, $k$-lookahead simulation is contained in $n$-pebble simulation,
though it is generally incomparable with $k$-pebble simulation.
\end{remark}

\begin{remark}\label{rem:lookahead:related}
	In \cite{lange:lookahead:2013,LangeBuffered:2014} very similar lookahead-like simulations are presented.
	In particular, \cite{lange:lookahead:2013} defines two variants of what they call \emph{multi-letter simulations}.
	The \emph{static} variant is the same as multistep simulation from Sec.~\ref{sec:multistep:simulation},
	and the \emph{dynamic} variant corresponds to the case where Duplicator chooses the amount of lookahead at each round,
	\emph{independently of Spoiler's attack}; thus, dynamic multi-letter simulation is included in lookahead simulation,
	since in the latter, Duplicator chooses the amount of lookahead actually used (i.e., the length of the response)
	depending on Spoiler's attack.
	Moreover, \cite{LangeBuffered:2014} introduces what they call \emph{buffered simulations},
	which extend multi-letter simulations by considering unbounded lookahead.
	In particular, what they call \emph{look-ahead buffered simulations}
	correspond to lookahead simulations as presented in Sec.~\ref{sec:lookahead:simulation}
	without a uniform bound on the maximal amount of lookahead that Duplicator can choose at each round,
	and they prove that they are PSPACE-complete to compute.
	Similarly, the more liberal variant that they call \emph{continuous look-ahead buffered simulations}
	corresponds to continuous simulations as presented in Sec.~\ref{sec:continuous:simulation},
	and they show that they are EXPTIME-complete to compute.
	While in principle it might seem that buffered simulations subsume lookahead/continuous simulations,
	in fact from the results of \cite{KleinZimmermann:Lookahead:ICALP:2015}
	it can be established that an exponential amount of lookahead suffices in both cases,
	and thus buffered simulations coincide with lookahead/continuous simulations from this section
	for sufficiently large (but fixed in advance) lookahead.
	%
	%\cite{LangeBuffered:2014}: unbounded lookahead version of our continuous and lookahead fair simulations (just for language inclusion).
	%unbounded lookahead simulation is PSPACE-complete,
	%and unbounded continuous lookahead simulation is EXPTIME-complete (maybe a theoretical hint why they are hard to compute?).
	%Explanation of the adjective "continuous": simulation holds iff there exists a \emph{continuous} function (in the Cantor topology) mapping accepting runs in A to accepting runs in B (over the same word).
	%In particular, from \cite{KleinZimmermann:Lookahead:ICALP:2015} it follows that exponential lookahead suffice for continuous simulation.
\end{remark}

\subsection{Fixpoint logic characterization.}
\label{sec:fixedpoint}

We conclude this section by giving a characterization of lookahead simulation in the modal $\mu$-calculus.
While this characterization could be used as the basis of an algorithm to
compute lookahead simulations symbolically by using fixpoint iteration,
it is more efficient to consider lookahead simulations as a special case of multipebble simulations,
as described in Remark~\ref{rem:lookahead-pebble}.
See Section~\ref{sec:implementation} for details on efficient implementations.

The $\mu$-calculus characterization follows from the following preservation property enjoyed by lookahead simulation:
Let $x \in \{ \mathrm{di, de, f, bw\textrm{-}di} \}$ and $k > 0$.
When Duplicator plays according to a winning strategy,
in any configuration $(p_i, q_i)$ of the resulting play, $p_i \kxsim k x q_i$ holds.
Thus, $k$-lookahead simulation (without acceptance condition) can be characterized as the largest $X \subseteq Q \times Q$
which is closed under a certain monotone predecessor operator.
For convenience, we take the point of view of Spoiler,
and compute the complement relation $W^x = (Q \times Q) \setminus \kxsim k x$ instead.
This is particularly useful for delayed simulation,
since we can avoid recording the obligation bit (see
\cite{etessami:etal:fairsimulations:05})
by using the technique of \cite{piterman:generalized06}.

\subsubsection{Direct and backward simulation.}
Consider the following monotone (w.r.t.~$\subseteq$) predecessor operator $\cpredi X$, for any set $X \subseteq Q \times Q$:
\begin{align*}
	\cpredi X &= \{ (p_0, q_0) \st
		\exists(p_0 \goesto {a_0} p_1 \goesto {a_1} \cdots \goesto {a_{k-1}} p_k)  \\
								& \forall (q_0 \goesto {a_0} q_1 \goesto {a_1} \cdots \goesto {a_{m-1}} q_m)\footnotemark
								  \textrm{ with } 0 < m \leq k, \\
	\textrm{\it either} \quad	& \exists (0 \leq j \leq m) \cdot p_j \in F \textrm{ and } q_j \not\in F, \\
	\textrm{\it or}		\quad 	& (p_m, q_m) \in X \}.
\end{align*}
\footnotetext{Here and in the following,
this is a shorthand for 
``$\forall (q_0 \goesto {b_0} q_1 \goesto {b_1} \cdots \goesto {b_{m-1}} q_m)$
with $a_0 = b_0, \dots, a_{m-1} = b_{m-1}$''.}
Intuitively, $(p,q) \in \cpredi X$ iff, from position $(p, q)$,
in one round of the game Spoiler can either force the game in $X$,
or win immediately by violating the winning condition for direct simulation.
For backward simulation, $\cprebwdi X$ is defined analogously,
except that transitions are reversed and also initial states are taken into account:
\begin{align*}
	\cprebwdi X &= \{ (p_0, q_0) \st
		\exists(p_0 \comesfrom {a_0} p_1 \comesfrom {a_1} \cdots \comesfrom {a_{k-1}} p_k) \\
				& \forall (q_0 \comesfrom {a_0} q_1 \comesfrom {a_1} \cdots \comesfrom {a_{m-1}} q_m) \textrm{ with } 0 < m \leq k, \\
	\textrm{\it either} \quad	& \exists (0 \leq j \leq m) \cdot p_j \in F \textrm{ and } q_j \not\in F, \\
	\textrm{\it or} 	\quad	& \exists (0 \leq j \leq m) \cdot p_j \in I \textrm{ and } q_j \not\in I, \\
	\textrm{\it or}		\quad 	& (p_m, q_m) \in X \}.
\end{align*}
\begin{remark}\label{rem:no_deadlocks}
	The definition of $\cprex x X$ requires that the automaton has no deadlocks;
	otherwise, Spoiler would incorrectly lose if he can only perform at most $k' < k$ transitions.
	%while in the definition of lookahead simulation we only require that Duplicator replies to these $k'$ steps.
	We assume that the automaton is complete to keep the definition simple,
	but our implementation works with general automata.

%        Intuitively, the generalization to incomplete automata works as
%        follows. If Spoiler's move reaches a deadlocked state after $k'$ steps,
%        where $1 \le k' < k$ then Spoiler does not immediately lose. Instead
%        Duplicator needs to reply to this move of length $k'$.
%        In other words, if Spoiler's move ends in a deadlocked state then the
%        lookahead requirements are weakened, because one simply cannot demand
%        any more steps from Spoiler.
\end{remark}
\noindent
For $X = \emptyset$, $\cprex x X$ is the set of states from which Spoiler wins in at most one step.
Thus, Spoiler wins iff he can eventually reach $\cprex x \emptyset$.
Formally, for $x \in \{\mathrm{di, bw\textrm{-}di}\}$, \[ W^x = \mu W \cdot \cprex x W. \]

\subsubsection{Delayed and fair simulation.}

We introduce a more elaborate three-arguments predecessor operator $\cprelong X Y Z$.
Intuitively, a configuration belongs to $\cprelong X Y Z$ iff
Spoiler can make a move s.t., for any Duplicator's reply,
at least one of the following conditions holds:
\begin{enumerate}
	\item Spoiler visits an accepting state, while Duplicator never does so; then, the game goes to $X$.
	\item Duplicator never visits an accepting state; the game goes to $Y$.
	\item The game goes to $Z$ (without any further condition).
\end{enumerate}
%
% Formally, we have the following:
%
\begin{align*}
	\cprelong X Y Z &= \{ (p_0, q_0) \st
		\exists(p_0 \goesto {a_0} p_1 \goesto {a_1} \cdots \goesto {a_{k-1}} p_k) \\
		&\forall (q_0 \goesto {a_0} q_1 \goesto {a_1} \cdots \goesto {a_{m-1}} q_m) \textrm{ with } 0 < m \leq k, \\
		\textrm{\it either} \quad	&	\exists (0 \leq i \leq m) \cdot p_i \in F,
										\forall (i \leq j \leq m) \cdot q_j \not\in F,
										(p_m, q_m) \in X, \\
		\textrm{\it or}		\quad	&	\forall (0 \leq j \leq m) \cdot q_j \not\in F,
										(p_m, q_m) \in Y, \\
		\textrm{\it or}		\quad	&	(p_m, q_m) \in Z \}.
\end{align*}
%
%(Remarks~\ref{rem:no_deadlocks} and \ref{rem:early_stop} also apply to $\cprelong X Y Z$.

For fair simulation, Spoiler wins iff, except for finitely many rounds (least fixpoint $\mu Z$),
he visits accepting states infinitely often
while Duplicator does not visit any accepting state at all (greatest and least fixpoints $\nu X \mu Y$):
\[ W^\mathrm f = \mu Z \cdot \nu X \cdot \mu Y \cdot \cprelong X Y Z. \]
Indeed, for fixed $X$ and $Z$,
a configuration belongs to $\mu Y \cdot \cprelong X Y Z$
if Spoiler can force the game in a finite number of steps to either visit an accepting state and go to $X$ (while Duplicator never visits an accepting state), or go to $Z$ (with the possibility that Duplicator visits an accepting state).
Thus, for fixed $Z$, a configuration belongs to $\nu X \cdot \mu Y \cdot \cprelong X Y Z$
if Spoiler can visit accepting states infinitely often while Duplicator never visits an accepting state,
or go to $Z$.
Finally, a configuration belongs to $W^\mathrm f$ if Spoiler can force the game in a finite number of steps to a configuration
from where he can visit infinitely many accepting states while Duplicator never visits an accepting state,
as required by the fair winning condition for Spoiler.

For delayed simulation, Spoiler wins if, after finitely many rounds,
\begin{enumerate}[1)]
	\item he can visit an accepting state, and from this moment on
	\item he can prevent Duplicator from visiting accepting states in the future.
\end{enumerate}
For condition 1), let $\cpreone X Y := \cprelong X {\emptyset} Y$,
and, for 2), $\cpretwo X Y := \cprelong {\emptyset} X Y$.
From the definition, a configuration belongs to $\cpreone X Y$ if Spoiler can in one step either visit an accepting state (while Duplicator does not do so) and go to $X$, or go to $Y$.
Similarly, a configuration belongs to $\cpretwo X Y$ if Spoiler can in one step either force the game to $X$ while Duplicator does not visit an accepting state, or force the game to $Y$.
Then,
\[ W^\mathrm {de} = \mu W \cdot \cpreone {\nu X \cdot \cpretwo X W} W. \]
Indeed, for any fixed $X$, $\mu W \cdot \cpreone X W$ is the set of configurations
from which Spoiler can force a visit to an accepting state in a finite number of steps
(and Duplicator does not visit an accepting state after Spoiler has done so)
and go to $X$,
and for any fixed $W$, $\nu X \cdot \cpretwo X W$ is the largest set of configurations
from where Spoiler can prevent Duplicator from visiting accepting states, or go to $W$.
Therefore a configuration is in $W^\mathrm {de}$ if Spoiler can force a visit to an accepting state in a finite number of steps,
after which he can prevent Duplicator from visiting accepting states ever after,
as required by the delayed winning condition for Spoiler.

%%% Local Variables:
%%% mode: latex
%%% TeX-master: "ROOT.tex"
%%% End:

\section{The Automata Reduction Algorithm}\label{sec:heavyandlight}

\subsection{Nondeterministic B\"uchi Automata}\label{subsec:heavyandlight-Buchi}

We reduce nondeterministic B\"uchi automata by the quotienting and transition pruning 
techniques from Sections~\ref{sec:quotienting} and \ref{sec:pruning}.
While trace inclusions would be an ideal basis for such techniques,
they (i.e., their membership problems) are PSPACE-complete.
Instead, we use the lookahead simulations from Sec.~\ref{sec:lookahead} 
as efficiently computable under-approximations;
in particular, we use
\begin{itemize}
	\item direct lookahead simulation $\transkdisim$ in place of direct trace inclusion $\directtraceinclusion$,
	\item delayed lookahead simulation $\transkdesim$ in place of delayed fixed-word simulation $\delayedfixedwordsimulation$, % $n$-pebble delayed simulation,
	\item fair lookahead simulation $\transkfsim$ in place of fair trace inclusion $\fairtraceinclusion$, and % (which is GFI).
	\item backward direct lookahead simulation $\transkbwdisim$ in place of
          backward direct trace inclusion $\bwdirecttraceinclusion$.
\end{itemize}
For quotienting, we employ delayed $\transkdesim$, and backward $k$-lookahead $\transkbwdisim$ simulations,
which are GFQ by Lemma~\ref{lem:lookahead_sim_GFI_GFQ}.
For pruning, we apply the results of Sec.~\ref{sec:pruning} and the substitutions above
to obtain the following incomparable GFP relations: % on NBA and NFA: % (if $\A = \A/\!\bwdisim$):
%with $k$-lookahead simulations in practice.
%Since GFP is $\subseteq$-downward closed and
%$\makeprunerel{\cdot}{\cdot}$ is monotone, we obtain that
\begin{itemize}
	\item $\makeprunerel{\id}{\stricttranskdisim}$ (by Theorem~\ref{thm:prune_id_strictdirecttraceinclusion}),
	\item $\makeprunerel{\stricttranskbwdisim}{\id}$ (by Theorem~\ref{thm:bwtrace-id}),
	\item $\makeprunerel{\strictbwdisim}{\transkdisim}$ (by Theorem~\ref{thm:bwsim-fwtrace}),
	\item $\makeprunerel{\transkbwdisim}{\strictdisim}$ (by Theorem~\ref{thm:bwtrace-fwsim}), and
	\item $\makeprunereltransient{\id}{\stricttranskfsim}$ (by Theorem~\ref{thm:prune_transient}).
\end{itemize}

Below we describe two possible ways to combine our simplification techniques: \emph{Heavy-$k$} and \emph{Light-$k$}
(which are parameterized by the lookahead value $k$).

\subsubsection{Heavy-$k$.}

We advocate the following reduction procedure,
which repeatedly applies all the techniques described in this paper
until the automaton cannot be further modified:
\begin{itemize}%[1)]
	\item Remove dead states.
	\item Prune transitions w.r.t.~the GFP relations above (using lookahead $k$).
%$\makeprunerel{\id}{\stricttranskdisim}$,
%$\makeprunerel{\stricttranskbwsim}{\id}$,
%$\makeprunerel{\strictbwdisim}{\transkdisim}$,
%$\makeprunerel{\transkbwsim}{\strictdisim}$
%and $R_t(\stricttranskfsim)$ (see Sec.~\ref{sec:pruning}),
	\item Quotient w.r.t.~$\transkdesim$ and $\transkbwdisim$.
%and quotienting with $\transkdesim$ and $\transkbwsim$.
%(Note that pruning transitions can render states dead, which 
%are then removed.)
\end{itemize}
The resulting simplified automaton cannot be further reduced by any of these techniques.
In this sense, it is a local minimum in the space of automata (w.r.t.~this set
of reduction techniques).
Many different variants are possible where the techniques above are applied in
different orders.
In particular, applying the techniques in a different order might produce a
different local minimum.
In general, there does not exist an optimal order that works best in
every instance. One reason for this is that one needs to decide whether to
first quotient w.r.t.~backward simulation and then to quotient w.r.t.~forward simulation
or vice-versa; cf.~Fig.~\ref{fig_no_opt_order}.

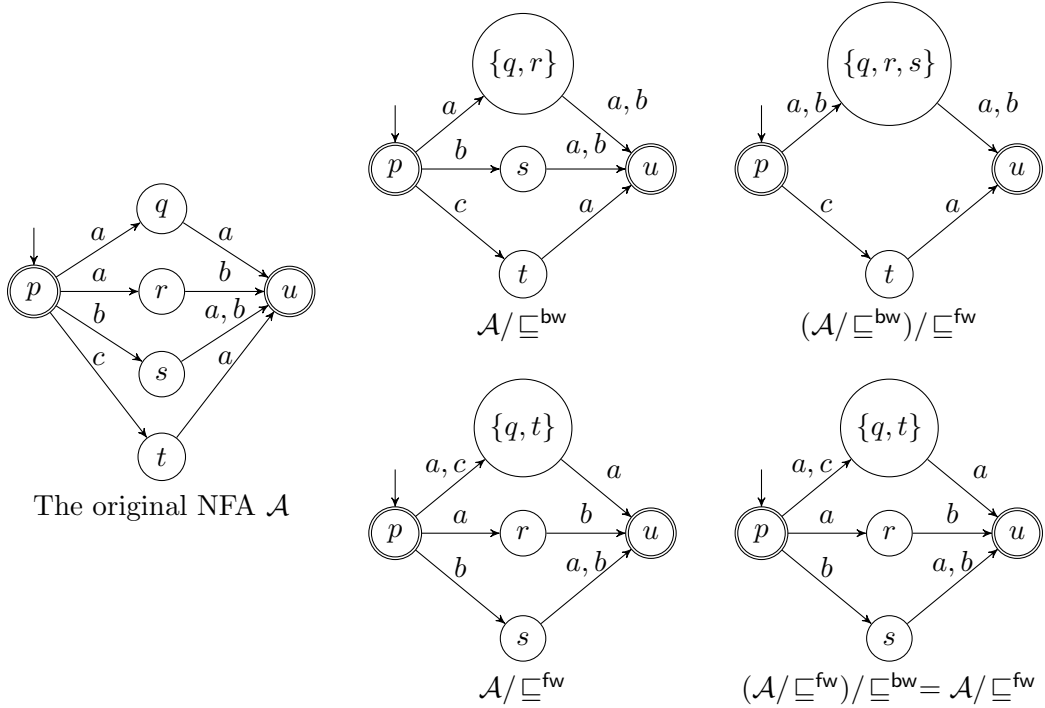
\begin{figure}
	\begin{tabular}{cc}
		\begin{tabular}{c}
			\begin{tikzpicture}[on grid, node distance= 1.1cm and 1.7cm]
				\tikzstyle{vertex} = [smallstate]

				\path node [vertex, initial above, accepting] (p) {$p$};
				\path node [vertex] (q) [above right = of p] {$q$};
		                    \path node [vertex] (r) [right = of p] {$r$};
		                    \path node [vertex] (s) [below right = of p] {$s$};
		                    \path node [vertex] (t) [below = of s] {$t$};
				\path node [vertex, accepting] (u) [right = of r] {$u$};

				\path[->]
					(p) edge node [above] {$a$} (q)
		                            (p) edge node [above] {$a$} (r)
		                            (p) edge node [above] {$b$} (s)
		                            (p) edge node [above] {$c$} (t)
					(q) edge node [above] {$a$} (u)
					(r) edge node [above] {$b$} (u)
					(s) edge node [above] {$a,b$} (u)
		                            (t) edge node [above] {$a$} (u);
			\end{tikzpicture}
			\\
			The original NFA $\A$
		\end{tabular}
		&
		\begin{tabular}{c@{$\qquad$}c}
			\begin{tikzpicture}[on grid, node distance= 1.4cm and 1.7cm]
				\tikzstyle{vertex} = [smallstate]

				\path node [vertex, initial above, accepting] (p) {$p$};
				\path node [vertex] (q) [above right = of p] {$\{q,r\}$};
	                        \path node [vertex] (s) [right = of p] {$s$};
	                        \path node [vertex] (t) [below right = of p] {$t$};
				\path node [vertex, accepting] (u) [right = of r] {$u$};

				\path[->]
					(p) edge node [above] {$a$} (q)
	                                (p) edge node [above] {$b$} (s)
	                                (p) edge node [above] {$c$} (t)
					(q) edge node [above right] {$a,b$} (u)
					(s) edge node [above] {$a,b$} (u)
	                                (t) edge node [above] {$a$} (u);
			\end{tikzpicture}
			&
			\begin{tikzpicture}[on grid, node distance= 1.4cm and 1.7cm]
				\tikzstyle{vertex} = [smallstate]

				\path node [vertex, initial above, accepting] (p) {$p$};
				\path node [vertex] (q) [above right = of p] {$\{q,r,s\}$};
	                        \path node [vertex] (t) [below right = of p] {$t$};
				\path node [vertex, accepting] (u) [right = of r] {$u$};

				\path[->]
					(p) edge node [above] {$a,b\ $} (q)
	                                (p) edge node [above] {$c$} (t)
					(q) edge node [above right] {$a,b$} (u)
	                                (t) edge node [above] {$a$} (u);
			\end{tikzpicture}
			\\
	        $\A/\!\bwsim$ 
			&
			$(\A/\!\bwsim)/\!\fwsim$
			\\
			\\
			\begin{tikzpicture}[on grid, node distance= 1.4cm and 1.7cm]
				\tikzstyle{vertex} = [smallstate]

				\path node [vertex, initial above, accepting] (p) {$p$};
				\path node [vertex] (q) [above right = of p] {$\{q,t\}$};
	                        \path node [vertex] (r) [right = of p] {$r$};
	                        \path node [vertex] (s) [below right = of p] {$s$};
				\path node [vertex, accepting] (u) [right = of r] {$u$};

				\path[->]
					(p) edge node [above] {$a,c\ $} (q)
	                                (p) edge node [above] {$a$} (r)
	                                (p) edge node [above] {$b$} (s)
					(q) edge node [above right] {$a$} (u)
					(r) edge node [above] {$b$} (u)
	                                (s) edge node [above] {$a,b$} (u);
			\end{tikzpicture}
			&
			\begin{tikzpicture}[on grid, node distance= 1.4cm and 1.7cm]
				\tikzstyle{vertex} = [smallstate]

				\path node [vertex, initial above, accepting] (p) {$p$};
				\path node [vertex] (q) [above right = of p] {$\{q,t\}$};
	                        \path node [vertex] (r) [right = of p] {$r$};
	                        \path node [vertex] (s) [below right = of p] {$s$};
				\path node [vertex, accepting] (u) [right = of r] {$u$};

				\path[->]
					(p) edge node [above] {$a,c\ $} (q)
	                                (p) edge node [above] {$a$} (r)
	                                (p) edge node [above] {$b$} (s)
					(q) edge node [above right] {$a$} (u)
					(r) edge node [above] {$b$} (u)
	                                (s) edge node [above] {$a,b$} (u);
			\end{tikzpicture}
			\\
	        $\A/\!\fwsim$ 
			&
			$(\A/\!\fwsim)/\!\bwsim = \A/\!\fwsim$
		\end{tabular}
	\end{tabular}
	\caption{There is no universally optimal order of applying quotienting
        operations. In this example, it is best to first quotient the NFA $\A$
        w.r.t.~backward simulation and then to quotient it w.r.t.~forward
        simulation. Thus one obtains an irrecucible NFA with $4$ states (first
        row above), while
        the reverse order yields an irrecucible NFA with $5$ states (second
        row above).
        To obtain a dual example where it is best
        to first quotient w.r.t.~forward simulation, just reverse the
        direction of all transitions in the original automaton $\A$ and make
        state $u$ initial instead of $p$.
        To obtain a similar example for B\"uchi automata, just add a self-loop
        with action $d$ at state $u$ (resp.~at state $p$ for the dual example).
        }
	\label{fig_no_opt_order}

\end{figure}

%%% Local Variables:
%%% mode: latex
%%% TeX-master: "ROOT.tex"
%%% End:

In practice, the order is determined by efficiency considerations and easily
computable operations are used first. More exactly, our implementation
uses a nested loop, where the inner loop uses only lookahead $1$
(until a fixpoint is reached), while the outer loop uses lookahead $k$. 
In other words, the algorithm uses
expensive operations only when cheap operations have no more effect.
For details about the precise order of the techniques in our implementation,
the reader is referred to \cite{RABIT} (algorithms/Minimization.java). 

\begin{remark}
	Quotienting w.r.t.~simulation is idempotent, since quotienting itself preserves simulation.
	However, in general this is not true for lookahead simulations,
	because these relations are not preserved under quotienting.
	Moreover, quotienting w.r.t.~forward simulations does not preserve backward simulations, and vice-versa.
	Our experiments showed that repeatedly and alternatingly quotienting w.r.t.~$\transkdesim$ and $\transkbwdisim$
	(in addition to our pruning techniques) yields the best reduction effect.
\end{remark}

The Heavy-$k$ procedure {\em strictly subsumes} all simulation-based automata
reduction methods described in the literature 
(removing dead states, quotienting, pruning of `little brother' transitions, 
mediated preorder (see Sec.~\ref{subsec:mediated})), except for the following two:
\begin{enumerate}
	\item 
        The \emph{fair simulation reduction} of \cite{GBS02} 
        is implemented in GOAL \cite{GOAL_survey_paper}, and works by tentatively 
	merging fair simulation equivalent states and then checking if this operation 
	preserved the language. (In general, fair simulation is not
	GFQ.) It potentially subsumes quotienting with $\desim$, provided that
        the chosen merged states are not only fair simulation equivalent, but also delayed
        simulation equivalent.
        However, it does not subsume quotienting with $\transkdesim$.
	We benchmarked our methods against it and found Heavy-$k$ to be much better in
	both effect and efficiency; cf.~Sec.~\ref{sec:experiments}.
	\item The GFQ \emph{jumping-safe preorders} of \cite{buchiquotient:ICALP11,Clemente:PhD} are incomparable to
	the techniques described in this paper. If applied in addition to
        Heavy-$k$ (for quotienting only, since they are GFQ but not GFP),
	they yield a modest extra reduction effect. 
        In our experiments in Sec.~\ref{sec:experiments} we also benchmarked an
        extended version of Heavy-$k$, called \emph{Heavy-$k$-jump}, that
        additionally uses the jumping-safe preorders of
        \cite{buchiquotient:ICALP11,Clemente:PhD}
        for quotienting.
\end{enumerate}

\subsubsection{Light-$k$.}
This reduction procedure is defined purely for comparison reasons.
It demonstrates the effect of the lookahead $k$
in a single quotienting operation and works as follows:
Remove all dead states and then quotient w.r.t.~$\transkdesim$.
Although Light-$k$ achieves much less than Heavy-$k$, it is not necessarily faster.
This is because it uses the more expensive to compute relation $\transkdesim$
directly, while Heavy-$k$ applies other cheaper (pruning) operations
first and only then computes $\transkdesim$ on the resulting smaller
automaton.

\subsection{Nondeterministic Finite Automata}\label{subsec:heavyandlight-NFA}

Most of the techniques from Sec.~\ref{subsec:heavyandlight-Buchi} carry
over to NFA, except for the following differences.
\begin{itemize}
\item
Delayed and fair simulation do not apply to NFA. 
Thus, pruning w.r.t.~$\makeprunereltransient{\id}{\stricttranskfsim}$ is
omitted.
Moreover, instead of quotienting with the transitive closures of lookahead delayed simulation
$\transkdesim$ and lookahead backward direct simulation $\transkbwdisim$,
we quotient NFA with the transitive closures of lookahead forward direct simulation $\transkdisim$
and  lookahead backward simulation $\transkbwsim$.
Those are included in forward $\finincl$ and backward $\bwfinincl$ finite trace inclusion, respectively,
and thus they are GFQ on NFA by Theorem~\ref{thm:GFQ:NFA}.
\item
The transition pruning techniques use $\transkbwsim$ instead of
$\transkbwdisim$ and
$\stricttranskbwsim$ instead of $\stricttranskbwdisim$.
The correctness for NFA follows from the theorems in
Sec.~\ref{sec:pruning:NFA}.
\item
Unlike NBA, every NFA can be transformed into an equivalent one with
just a single accepting state without any outgoing transitions
(unless the language contains the empty word; this case can be handled
separately), as follows:
\begin{enumerate*}[label=\arabic*)]
\item
Add a new accepting state ${\it acc}$.
\item
For every transition $p \goesto a q$ where $q$ is accepting and $q \neq {\it acc}$, add a transition
$p \goesto a {\it acc}$.
\item
Make ${\it acc}$ the only accepting state.
\end{enumerate*}
This transformation adds just one state, but possibly many transitions.
In this new form, the direct forward and backward (lookahead) simulations are significantly
larger, because the acceptance conditions are easier to satisfy.
This greatly increases the effect of the remaining quotienting and pruning reduction
methods, and partly offsets the negative effect caused by the loss of the delayed
and fair simulation based methods.
\item
A variant of the GFQ \emph{jumping-safe preorders} of
\cite{buchiquotient:ICALP11,Clemente:PhD}
can also be applied to NFA. 
Unlike the version for NBA, it does not make use of (jumping) delayed simulation, but
uses (jumping) direct forward and backward simulations.
It is implemented only in the extended Heavy-$k$-jump version of the NFA
reduction algorithm; cf.~Sec.~\ref{sec:experiments}.
\end{itemize}

\subsection{Quotienting w.r.t.~mediated simulation}\label{subsec:mediated}

We show that the quotienting and transition pruning techniques described above
subsume quotienting w.r.t.~{\em mediated preorder} \cite{AbdullaCHV09,AbdullaChenHolikVojnar:TCS:2014}
(but not vice-versa),
in the sense that after applying our reduction algorithm,
quotienting w.r.t.~mediated preorder provably does not yield any further reduction.
Mediated preorder was originally defined for alternating B\"uchi automata
as an attempt at combining backward and forward simulations for automata reduction.
Here, we consider its restriction to nondeterministic B\"uchi automata 
(and the arguments carry over directly to NFA).
\begin{definition}[\cite{AbdullaCHV09,AbdullaChenHolikVojnar:TCS:2014}]
	A relation $M \subseteq Q \times Q$ is a \emph{mediated simulation}%
	\footnote{For two relations $A, B \subseteq Q \times Q$,
	we write $A \circ B$ for the relation $A \circ B \subseteq Q \times Q$
	s.t.~$(x, y) \in A \circ B$ iff there exists $z$ s.t.~$(x, z) \in A$ and $(z, y) \in B$.}
	if
        \begin{enumerate}
        \item
          $M \mathrel{\subseteq} (\disim \circ \bwdisimrev)$, and
        \item
          $(M \circ \disim) \mathrel{\subseteq} M$.
          \end{enumerate}
\end{definition}
It can be shown that mediated simulations are closed under union and composition,
and thus there exists a largest mediated simulation preorder $\medsim$
which is the union of all mediated simulations,
and \cite{AbdullaCHV09,AbdullaChenHolikVojnar:TCS:2014} further shows that $\medsim$ is GFQ.

However, an automaton $\A$ that has
been reduced by the techniques described above cannot be further reduced by
mediated preorder.
First, we have $\A = \A/\!\bwdisim = \A/\!\disim$ by repeated quotienting.
Second, there cannot exist any (distinct) states $p$ and $q$ in $\A$
s.t.~$p \strictdisim q$ and $p \strictbwdisim q$ by the pruning techniques above
(used with simulations as approximations for trace inclusions) and the removal of dead states.
Indeed, if such states $p$ and $q$ exist, then $p$ is removed:
First, every forward transition $p \goesto \sigma p'$ from $p$
is subsumed by a corresponding transition $q \goesto \sigma q'$ from $q'$
s.t.~$p' \disim q'$.
Similarly, every backward transition to $p$ is subsumed by a corresponding transition to $q$.
Therefore, after pruning away all these transitions w.r.t.~$\makeprunerel \strictbwdisim \strictdisim$,
state $p$ becomes dead, and it is thus removed.
Under these conditions, further quotienting with mediated preorder has no effect,
as the following theorem shows.
\begin{lemma}
	Let $\A$ be an automaton s.t.\
	(1) $\disim \cap \disimrev = \id$,
	(2) $\bwdisim \cap \bwdisimrev = \id$, and
	(3) $\disim \cap \bwdisim = \id$.
	Then, $\medsim \cap \medsimrev = \id$, i.e., $\A = \A/\! \medsim$.
\end{lemma}
\begin{proof}
	Let $x \medsim y$ and $y \medsim x$. By definition of $\medsim$,
	there exist mediators $z$ and $w$ s.t.~$x \disim z$ and $y \bwdisim z$,
	and $x \bwdisim w$ and $y \disim w$.
	Since $\medsim \circ \disim \,\subseteq\, \medsim$ we have $x \medsim w$.
	Thus, there exists a mediator $k$ s.t.~$x \disim k$ and $w \bwdisim k$.
	By transitivity of $\bwdisim$, we also have $x \bwdisim k$.
	By (3), we get $x=k$.
	Thus, $x \bwdisim w$ and $w \bwdisim x$.
	By (2), we get $x=w$.
	Thus, $y \disim w=x \disim z$, and, by transitivity, $y \disim z$. 
	Moreover, $y \bwdisim z$ as above.
	By (3) we get $z=y$.
	Thus, $x \disim z = y$ and $y \disim w = x$.
	By (1), we get $x=y$.
\end{proof}

%%% Local Variables:
%%% mode: latex
%%% TeX-master: "ROOT.tex"
%%% End:

\section{Language Inclusion Checking}
\label{sec:preprocessing}
\label{sec:inclusion}

In most of this section we consider the language inclusion problem for NBA. 
For the simpler case of language inclusion on NFA see
Sec.~\ref{subsec:incl-NFA}.

The general language inclusion problem $\A \languageinclusion \B$
is PSPACE-complete \cite{kupfermanvardi:fair_verification};
the complexity reduces to PTIME in certain special instances,
for example when $\B$ is deterministic \cite{Kurshan:Complementing:JCSS:1987}
or, more generally, strongly unambiguous \cite{BousquetLoeding:Unambiguous:LATA:2010}.
It can be solved via complementation of $\B$ \cite{sistla:vardi:wolper:complementation:87,GOAL_survey_paper}
and, more efficiently, by rank-based (cf.~\cite{fogarty_et_al:LIPIcs:2011:3235} and references therein)
or Ramsey-based methods
\cite{seth:buchi,seth:efficient,abdulla:simulationsubsumption,Rabit_CONCUR2011},
or variants of Piterman's construction \cite{Pit06,GOAL_survey_paper};
simulation relations \cite{dill:inclusion:1992}
or succinct pseudo-complementation constructions
\cite{Kurshan:Complementing:JCSS:1987}
(cf.~Remark~\ref{rem:Kurshan})
can provide PTIME under-approximations for this problem,
but do not always manage to prove all cases when inclusion holds.
Since the exact algorithms all have {\em exponential} time complexity, it helps significantly
to first reduce the automata in a preprocessing step.
Better reduction techniques, as described in the previous sections, make it
possible to solve significantly larger instances.
However, our simulation-based techniques can not only be used in
preprocessing to reduce the size of automata, but actually solve most instances of the inclusion problem
{\em directly} by reducing to trivial instances. This is significant, because simulation scales 
{\em polynomially} (almost quadratic average-case complexity; cf.~Sec.~\ref{sec:experiments}).

\subsection{Inclusion-preserving reduction techniques}\label{subsec:incl_preserving_reduction}

Inclusion testing algorithms generally benefit from language-preserving reduction preprocessing
(cf.~Sec.~\ref{sec:heavyandlight}).
However, precisely preserving the languages of $\A$ and $\B$ in the preprocessing is not actually necessary when one is only interested in the answer to the query $\A \languageinclusion \B$.
A preprocessing on $\A,\B$ is said to be \emph{inclusion-preserving}
iff it produces automata $\A',\B'$ s.t.~$\A \languageinclusion \B \iff \A' \languageinclusion \B'$
(regardless of whether $\A \languageequivalence \A'$ or $\B \languageequivalence \B'$).
In the following, we consider two inclusion-preserving preprocessing steps.

\subsubsection{Simplify $\A$.}
\label{sec:simplifyA}

In theory, the problem $\A \languageinclusion \B$ is only hard in $\B$,
but polynomial in the size of $\A$. However, this is only relevant if one actually 
constructs the exponential-size complement of $\B$, which is, of course, to be
avoided. For polynomial simulation-based algorithms it is crucial to also reduce $\A$.
The idea is to remove transitions in $\A$ which are `covered' by better transitions in $\B$.
The development below is similar to the pruning of transitions in Sec.~\ref{sec:pruning},
except that we compare transitions of $\A$ with transitions of $\B$.
%in a spirit similar to the pruning of Sec~\ref{sec:pruning}.

\begin{definition}
	Given $\A = (\Sigma, Q_\A, I_\A, F_\A, \delta_\A)$, 
	$\B = (\Sigma, Q_\B, I_\B, F_\B, \delta_\B)$,
	let $\prunerel \subseteq \delta_\A \times \delta_\B$. % be a relation for comparing transitions in $\A$ and $\B$.
	The pruned version of $\A$ is $\xprune{\A}{\B}{\prunerel} := (\Sigma, Q_\A, I_\A, F_\A, \delta')$
	with \[ \delta' = \{(p,\symb,r) \in \delta_\A \st {\nexists}
	(p',\symb',r') \in \delta_\B.\, (p,\symb,r) \prunerel (p',\symb',r')\} \ . \]
\end{definition}

\noindent
$\A \languageinclusion \B$ implies $\xprune{\A}{\B}{\prunerel} \languageinclusion \B$,
since $\xprune{\A}{\B}{\prunerel} \languageinclusion \A$.
When also the other direction holds (so that pruning is inclusion-preserving),
we say that $\prunerel$ is \emph{good for $\A,\B$-pruning}.
Intuitively, pruning is correct when the removed transitions do not allow $\A$ to accept any word
which is not already accepted by $\B$.
In other words, if there is a counter example to inclusion in $\A$,
then it can even be found in $\xprune{\A}{\B}{\prunerel}$.
\begin{definition}
	A relation $\prunerel \subseteq \delta_\A \times \delta_\B$ is \emph{good for $\A,\B$-pruning} if
	$\A \languageinclusion \B \!\!\iff\!\! \xprune{\A}{\B}{\prunerel} \languageinclusion \B$.
\end{definition}
\noindent
As in Eq.~\ref{eq:prunerel}, we compare transitions by looking at their endpoints:
For state relations $\brel, \frel \subseteq Q_\A \times Q_\B$,
the relation $\makeprunerelAB{\brel}{\frel}$ on transitions is defined as % in Eq.~\ref{eq:prunerel}.
\begin{align*}
%	\label{eq:prunerel:AB}
	\makeprunerelAB{\brel}{\frel} = \{((p,\symb,r),(p',\symb,r')) \in \delta_\A \times \delta_\B \st p \brel p' \textrm{ and } r \frel r' \}.
\end{align*}
Since inclusion-preserving pruning does not need to respect the language,
we can use much coarser relations for comparing endpoints.
Recall that fair trace inclusion $\fairtraceinclusion$ asks to match infinite traces containing infinitely many accepting states
(cf.~Sec.~\ref{sec:trace_inclusions}),
while that backward finite trace inclusion $\accblindbwdirecttraceinclusion$
disregards accepting states entirely and only asks to match finite traces that start in initial states
(cf.~Sec.~\ref{sec:simulations:finitewords}).
%Let $\accblindbwdirecttraceinclusion$ be the variant of $\bwdirecttraceinclusion$
%where accepting states are not taken into consideration,
%i.e., the winning condition of the corresponding simulation game is
%$\mathcal C^\mathrm{bw-}(\pi_0, \pi_1) \iff \forall (i \geq 0) \cdot p_i \in I \implies q_i \in I$.	
%
\begin{theorem}\label{thm:ABpruning}
	$\makeprunerelAB{\accblindbwdirecttraceinclusion}{\fairtraceinclusion}$ is good for $\A,\B$-pruning.
\end{theorem}
\begin{proof}
	Let $\prunerel = \makeprunerelAB{\accblindbwdirecttraceinclusion}{\fairtraceinclusion}$,
	and we want to prove that $\A \languageinclusion \B$ iff $\xprune{\A}{\B}{\prunerel} \languageinclusion \B$.
	The ``only if'' direction is trivial, as remarked above.
	%Since trivially $\xprune{\A}{\B}{\prunerel} \languageinclusion \A$,
	%we obtain 
	%$
	%\A \languageinclusion \B
	%\ \Rightarrow\ \xprune{\A}{\B}{\prunerel} \languageinclusion \B
	%$.
	%Now consider the reverse direction.
	For the ``if'' direction, by contraposition, assume $\xprune{\A}{\B}{\prunerel} \languageinclusion \B$, but
	$\A \not\languageinclusion \B$. There exists a $w \in \lang{\A}$
	s.t.~$w \notin \lang{\B}$. There exists an initial fair trace 
	$\pi = q_0 \goesto {\symb_0} q_1 \goesto {\symb_1} \cdots$ on $w$ in $\A$.
	There are two cases.
	\begin{enumerate}
		\item
		$\pi$ contains a transition
		$q_i \goesto {\symb_i} q_{i+1}$ that is not present in
		$\xprune{\A}{\B}{\prunerel}$.
		Therefore there exists a transition
		$q_i' \goesto {\symb_i} q_{i+1}'$ in $\B$ 
		s.t.~$q_i \accblindbwdirecttraceinclusion q_i'$ and $q_{i+1} \fairtraceinclusion q_{i+1}'$.
		Thus there exists an initial fair trace on $w$ in $\B$ and thus
		$w \in \lang{\B}$. Contradiction.
		\item
		$\pi$ does not contain any transition
		$q_i \goesto {\symb_i} q_{i+1}$ that is not present in
		$\xprune{\A}{\B}{\prunerel}$. Then $\pi$ is also an initial fair trace
		on $w$ in $\xprune{\A}{\B}{\prunerel}$, and thus we obtain
		$w \in \lang{\xprune{\A}{\B}{\prunerel}}$ and
		$w \in \lang{\B}$. Contradiction. \qedhere
	\end{enumerate}
\end{proof}

\noindent
We can approximate $\accblindbwdirecttraceinclusion$ with
the transitive closure $\transkbwsim$
of the corresponding $k$-lookahead simulation $\accblindkbwsim$.
(Recall that $\accblindkbwsim$ is defined like $\kbwdisim$, except that only initial states are considered,
i.e., the winning condition is $\mathcal C^{\mathrm {bw}}_I$ instead of $\mathcal C^{\mathrm {bw}}_{I, F}$
; cf.~Sec.~\ref{sec:simulations:finitewords}.)
%$\mathcal C^\mathrm{bw-}$.
%
Since ``good for $\A,\B$-pruning'' is $\subseteq$-downward closed and $\makeprunerelAB{\cdot}{\cdot}$ is monotone,
we obtain the following corollary of Theorem~\ref{thm:ABpruning}.

\begin{corollary}\label{cor:ABpruning}
	$\makeprunerelAB{\accblindtranskbwsim}{\transkfsim}$ is good for $\A,\B$-pruning.
\end{corollary}

\subsubsection{Simplify $\B$.}
\label{sec:simplifyB}

The following technique is independent of the use of simulation-based reduction,
but it is nonetheless worth mentioning, and moreover we include it in our reduction algorithm.
%A different pruning technique modifies $\B$ instead.
Let $\A \times \B$ be the synchronized product of $\A$ and $\B$.
The idea is to remove states in $\B$ which cannot be reached in $\A \times \B$.
Let $R$ be the set of states in $\A \times \B$ reachable from $I_\A \times I_\B$,
and let $X \subseteq Q_\B$ be the projection of $R$ to the $\B$-component.
%Let $X := \{q \in Q_\B \st \nexists p \in Q_\A.\, (p,q) \in R\}$.
We obtain $\B'$ from $\B$ by removing all states not in $X$ and their associated transitions.
Although $\B' \not\languageequivalence \B$, this operation is clearly inclusion-preserving.

Note that first simplifying $\A$ as in Sec.~\ref{sec:simplifyA}
yields fewer reachable states in $\A \times \B$
and thus increases the effect of the technique for simplifying $\B$.

%%% Local Variables:
%%% mode: latex
%%% TeX-master: "ROOT.tex"
%%% End:

\subsection{Jumping fair simulation as a better GFI relation} \label{sec:jumpsim}

We further generalize the GFI preorder $\transkfsim$ by allowing Duplicator even more freedom.
The idea is to allow Duplicator to take \emph{jumps} during the simulation game (in the spirit of \cite{Clemente:PhD}).
For a preorder $\le$ on $Q$, in the game for \emph{$\le$-jumping $k$-lookahead simulation},
Duplicator is allowed to jump to $\le$-larger states before taking a transition.
Thus, a Duplicator's move is of the form
$q_i \le q_i' \goesto {\symb_i} {q_{i+1}} \le q_{i+1}' \goesto {\symb_{i+1}} \cdots \goesto {\symb_{i+m-1}} q_{i+m}$,
and she eventually builds an infinite $\le$-jumping trace. We say that this trace
is \emph{accepting} at step $i$ iff $\exists q_i'' \in F.\, q_i \le q_i'' \le q_i'$,
and \emph{fair} iff it is accepting infinitely often.
As usual, \emph{$\le$-jumping $k$-lookahead fair simulation} holds
iff Duplicator wins the corresponding game, with the fair winning condition
lifted to jumping traces.
%Given $\le$, jumping $k$-lookahead fair simulation can be computed using the methods of Sec.~\ref{sec:lookahead}.

Not all preorders $\le$ induce GFI jumping simulations.
The preorder $\le$ is called {\em jumping-safe} \cite{Clemente:PhD} if,
for every word $w$, there exists a $\le$-jumping initial fair trace on $w$ 
iff there exists an initial fair non-jumping one.
Thus, jumping-safe preorders allows to convert jumping traces into non-jumping ones.
Consequently, for a jumping-safe preorder $\le$,
$\le$-jumping $k$-lookahead fair simulation is GFI.

One can easily prove that $\bwdirecttraceinclusion$ is jumping-safe, while $\accblindbwdirecttraceinclusion$ is not. 
%Since jumping-safe is $\subseteq$-downward closed, the efficiently computable $\transkbwsim$ is also jumping-safe.
%
We even improve $\bwdirecttraceinclusion$ to a slightly more general jumping-safe relation $\countingbwtraceinclusion$,
by only requiring that Duplicator visits at least as many accepting states as Spoiler does,
but not necessarily at the same time.
Formally, $p_m \countingbwtraceinclusion q_m$ iff,
%for every finite word $w = \symb_0\symb_1 \cdots \symb_{m-1} \in \Sigma^*$, and
for every initial $w$-trace
$\pi_0 = p_0 \goesto {\symb_0} p_1 \goesto {\symb_1} \cdots \goesto {\symb_{m-1}} p_m$, % ending in $p_m = p$,
there exists an initial $w$-trace
$\pi_1 = q_0 \goesto {\symb_0} q_1 \goesto {\symb_1} \cdots \goesto {\symb_{m-1}} q_m$, % ending in $q_m = q$,
s.t.~$|\{i \,|\, p_i \in F\}| \le |\{i \,|\, q_i \in F\}|$.
%, i.e., $\pi_1$ accepts at least as often as $\pi_0$, but not necessarily at the same steps.

\begin{theorem}\label{lem:jumping-fairsim}
	The preorder $\countingbwtraceinclusion$ is jumping-safe.
\end{theorem}
\begin{proof}
	%First, we show that $\countingbwtraceinclusion$ is jumping-safe.
	%
	Since $\countingbwtraceinclusion$ is reflexive, the existence of an initial
	fair trace on $w$ directly implies the existence of a 
	$\countingbwtraceinclusion$-jumping initial fair trace on $w$.
	
	Now, we show the reverse implication.
	Given two initial $\countingbwtraceinclusion$-jumping traces on $w$
	$\pi_0 = p_0 \countingbwtraceinclusion p_0' \goesto {\symb_0} {p_{1}} 
	\countingbwtraceinclusion p_{1}' \goesto {\symb_{1}} \cdots$
	and 
	$\pi_1 = q_0 \countingbwtraceinclusion q_0' \goesto {\symb_0} {q_{1}} 
	\countingbwtraceinclusion q_{1}' \goesto {\symb_{1}} \cdots$
	we define $\mathcal C^c_j(\pi_0, \pi_1)$ iff
	$|\{i \le j\,|\, \exists p_i'' \in F.\,  p_i \countingbwtraceinclusion
	p_i'' \countingbwtraceinclusion p_i'\}| \le 
	|\{i \le j\,|\, \exists q_i'' \in F.\,  q_i \countingbwtraceinclusion
	q_i'' \countingbwtraceinclusion q_i'\}|$.
	We say that an initial $\countingbwtraceinclusion$-jumping trace on $w$
	is {\em $i$-good} iff it does not jump within the first $i$ steps.

	We show, by induction on $i$, the following property (P):
	For every $i$ and every 
	infinite $\countingbwtraceinclusion$-jumping initial trace
	$\pi = p_0 \countingbwtraceinclusion p_0' \goesto {\symb_0} {p_{1}} \countingbwtraceinclusion p_{1}' \goesto {\symb_{1}} \cdots$
	on $w$ there exists 
	an $i$-good $\countingbwtraceinclusion$-jumping initial trace $\pi^i = q_0 \goesto {\symb_0} q_1 \goesto
	{\symb_1} \cdots \goesto {\symb_i} q_i \cdots$ on $w$
	s.t.~$\mathcal C^c_i(\pi, \pi^i)$ and the suffixes of the traces are identical, i.e.,
	$q_i = p_i$ and $\suffix \pi i = \suffix {\pi^i} i$.

	For the case base $i=0$ we take $\pi^0 = \pi$. 
	Now we consider the induction step. 
	By induction hypothesis we get an initial $i$-good trace $\pi^i$
	s.t.~$\mathcal C^c_i(\pi, \pi^i)$ and $q_i = p_i$ and 
        $\suffix \pi i = \suffix {\pi^i} i$.
	If $\pi^i$ is $(i+1)$-good then we can take $\pi^{i+1} = \pi^{i}$.
	Otherwise, $\pi^i$ contains a step 
	$q_i \countingbwtraceinclusion q_i' \goesto {\symb_i} {q_{i+1}}$.
	First we consider the case where there exists a $q_i'' \in F$
	s.t.~$q_i \countingbwtraceinclusion q_i'' \countingbwtraceinclusion q_i'$.
	(Note that the $i$-th step in $\pi^i$ can count as accepting in $\mathcal C^c$
	because $q_i'' \in F$, even if $q_i$ and $q_i'$ are not accepting.)
	By the definition of $\countingbwtraceinclusion$ there exists
	an initial trace $\pi''$ on a prefix of $w$ that ends in $q_i''$
	and visits accepting states at least as often as the non-jumping
	prefix of $\pi^i$ that ends in $q_i$.
	Again by definition of $\countingbwtraceinclusion$ there exists
	an initial trace $\pi'$ on a prefix of $w$ that ends in $q_i'$
	and visits accepting states at least as often as $\pi''$.
	Thus $\pi'$ visits accepting states at least as often as the {\em jumping}
	prefix of $\pi^i$ that ends in $q_i'$ (by the definition of $\mathcal C^c$).
	By composing the traces we get $\pi^{i+1} = \pi' (q_i' \goesto {\symb_i}
	{q_{i+1}}) \suffix {\pi^i} {i+1}$. Thus $\pi^{i+1}$ is an $(i+1)$-good initial trace
	on $w$ and 
        $\suffix \pi {i+1} = \suffix {\pi^i} {i+1} = \suffix {\pi^{i+1}} {i+1}$ and 
	$\mathcal C^c_{i+1}(\pi^i, \pi^{i+1})$ and $\mathcal C^c_{i+1}(\pi, \pi^{i+1})$.
	The other case where there is no $q_i'' \in F$
	s.t.~$q_i \countingbwtraceinclusion q_i'' \countingbwtraceinclusion q_i'$ is
	similar, but simpler.

	Let $\pi$ be an initial $\countingbwtraceinclusion$-jumping fair trace on $w$.
	By property (P) and K\"onig's Lemma there 
	exists an infinite initial non-jumping fair trace $\pi'$ on $w$.
	Thus $\countingbwtraceinclusion$ is jumping-safe.
\end{proof}

\noindent
As a direct consequence, $\countingbwtraceinclusion$-jumping $k$-lookahead fair simulation is GFI.
Since $\countingbwtraceinclusion$ is difficult to compute,
we approximate it by a corresponding lookahead-simulation $\countingkbwsim$ which, in the same spirit,
counts and compares the number of visits to accepting states in every round of the $k$-lookahead backward simulation game.
Let $\countingtranskbwsim$ be the transitive closure of $\countingkbwsim$.

\begin{corollary}
	$\countingtranskbwsim$-jumping $k$-lookahead fair simulation is GFI.
\end{corollary}

Fig.~\ref{fig_jumping_example} shows how the option to jump
w.r.t.\ $\countingbwtraceinclusion$
(resp.\ $\countingtranskbwsim$)
benefits Duplicator, making jumping simulation larger than lookahead simulation.
First, we have $p_0 \not\transkfsim p_1$ for every finite $k$.
If Spoiler plays $p_0 \goesto{a^k} p_0$ (thus revealing his first $k$ steps),
then Duplicator can only respond with either $p_1 \goesto{a^{k'}} q_1$
or $p_1 \goesto{a^{k'}} r_1$ for some $k'$ with $1 \le k' \le k$.
In the former (resp.\ latter) case, Spoiler wins by playing $p_0 \goesto{ac} t_0$
(resp.\ $p_0 \goesto{ab} s_0$) to which Duplicator has no response.
However, $\countingbwtraceinclusion$-jumping $k$-lookahead fair simulation contains $(p_0,p_1)$
(as well as $(p_0,q_1)$, $(p_0,r_1)$, $(q_0,r_1)$ and $(r_0,q_1)$)
even for $k=1$. Since $q_1$ and $r_1$ are equivalent w.r.t.~$\countingbwtraceinclusion$,
Duplicator can jump between then as needed before making a required
$b$ (resp.~$c$) step to $s_1$ (resp.~$t_1$).

\begin{figure}[htbp]
	\begin{tabular}{cc}
		\begin{tikzpicture}[on grid, node distance= 1.3cm and 1.7cm]
			\tikzstyle{vertex} = [smallstate]

			\path node [vertex, initial left] (p) {$p_0$};
			\path node [vertex] (q) [above right = of p] {$q_0$};
                        \path node [vertex] (r) [below right = of p] {$r_0$};
                        \path node [vertex, accepting] (s) [right = of q] {$s_0$};
                        \path node [vertex, accepting] (t) [right = of r] {$t_0$};

			\path[->]
				(p) edge node [above] {$a$} (q)
                                (p) edge node [above] {$a$} (r)
                                (q) edge node [above] {$b$} (s)
                                (r) edge node [above] {$c$} (t)
                        
				(p) edge [loop above] node {$a$} ()
                                (s) edge [loop above] node {$b$} ()
			        (t) edge [loop above] node {$c$} ();
		\end{tikzpicture}
		&
		\quad
                \begin{tikzpicture}[on grid, node distance= 1.3cm and 1.7cm]
			\tikzstyle{vertex} = [smallstate]

			\path node [vertex, initial left] (p) {$p_1$};
			\path node [vertex] (q) [above right = of p] {$q_1$};
                        \path node [vertex] (r) [below right = of p] {$r_1$};
                        \path node [vertex, accepting] (s) [right = of q] {$s_1$};
                        \path node [vertex, accepting] (t) [right = of r] {$t_1$};

			\path[->]
				(p) edge node [above] {$a$} (q)
                                (p) edge node [above] {$a$} (r)
                                (q) edge node [above] {$b$} (s)
                                (r) edge node [above] {$c$} (t)
                        
			        (q) edge [loop above] node {$a$} ()
                                (r) edge [loop above] node {$a$} ()
                                (s) edge [loop above] node {$b$} ()
			        (t) edge [loop above] node {$c$} ();
		\end{tikzpicture}
              \end{tabular}
	\caption{Jumping simulation can be strictly larger than lookahead simulation.}
	\label{fig_jumping_example}
\end{figure}
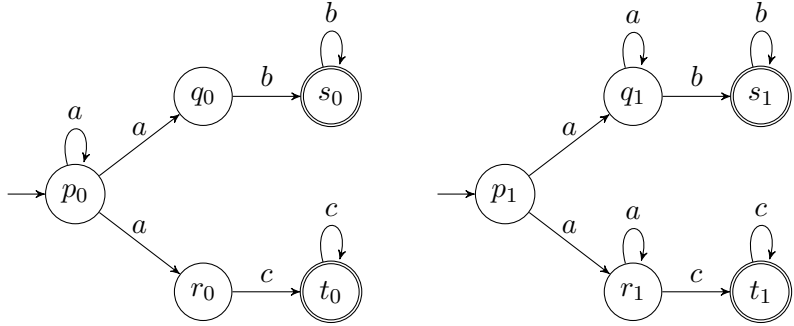

%%% Local Variables:
%%% mode: latex
%%% TeX-master: "ROOT.tex"
%%% End:

An orthogonal alternative to $\countingtranskbwsim$ is also implemented in \cite{RABIT}.
One can use a jumping-safe preorder (called \emph{segmented jumping}) that is
defined directly w.r.t.~$k$-lookahead backward simulations. 
Here Duplicator must visit at least one accepting state in each of her long moves,
regardless of whether Spoiler visited any accepting states, i.e., in each
round of the game Duplicator must accept at least once but possibly less often than
Spoiler. However, combining
segmented jumping with $\countingtranskbwsim$ (i.e., taking their union) 
would not be jumping-safe any more. First, since their union is not
necessarily transitive, one would need to consider the transitive closure of
the union to obtain a preorder. Moreover, the size and structure of the
segments in the segmented jumping relation are not fixed a-priori but chosen
dynamically in the computation of the $k$-lookahead backward simulation.
Thus the transitive closure of the union would allow a scenario where first the
segmented relation decreases the number of visits to accepting states while
preserving at least one visit per segment. Then $\countingtranskbwsim$
could shift the location of these visits to accepting states to positions
earlier in the run while preserving their number. Then 
the segmented relation could again decrease the number of visits to accepting
states, since they are now in different segments. Repeating this alternation
of counting and segmented relations could yield a situation where only one
visit to an accepting state remains in the entire run, which is not
jumping-safe any more. 

A further generalization of the jumping-simulation method has also been implemented in
\cite{RABIT} (activated by using option \texttt{-jf2} instead of the basic option \texttt{-jf}).
Given some jumping-safe preorder $R$, Duplicator is not only allowed
to jump to states that are $R$-larger than Duplicator's current state,
but also to states that are $R$-larger than \emph{Spoiler's} current state.
Note that Duplicator may only jump to states in her own automaton
$\B$, and not to states in Spoiler's automaton $\A$.
It is easy to see that this more liberal use of jumping still yields
(potentially larger) GFI
relations. However, in practice it rarely gives any advantage, and it is
sometimes considerably slower to compute, due to the higher degree of
branching in Duplicator's moves.

%%% Local Variables:
%%% mode: latex
%%% TeX-master: "ROOT.tex"
%%% End:

\subsection{The inclusion testing algorithm} \label{sec:inclalg}

Given these techniques, we propose the following algorithm for testing inclusion $\A \languageinclusion \B$.

\begin{enumerate}
	
	\item[(1)]
		Use the Heavy-$k$ procedure to perform language-preserving reduction to $\A$ and $\B$ separately,
		and additionally apply the inclusion-preserving reduction
                techniques from 
                Sec.~\ref{subsec:incl_preserving_reduction}
				simultaneously to $\A$ (discussed in Sec.~\ref{sec:simplifyA})
				and to $\B$ (discussed in Sec.~\ref{sec:simplifyB}).
		Lookahead simulations are computed not only on $\A$ and $\B$, but
		also {\em between} them (i.e., on their disjoint union).
		Since they are GFI, we check whether they already witness
                inclusion. Since many simulations are computed between partly
                reduced versions of $\A$ and $\B$, this witnesses
                inclusion much more often than checking fair simulation
                between the original versions. 
		This step either stops showing inclusion,
		or produces smaller inclusion-equivalent automata $\A', \B'$.
		%$\A',B'$ s.t.~$\lang{\A} \subseteq \lang{\B}\ \Leftrightarrow\ \lang{\A'} \subseteq \lang{\B'}$.
		
	\item[(2)]
		Check the GFI $\countingtranskbwsim$-jumping $k$-lookahead fair simulation from Sec.~\ref{sec:jumpsim} between $\A'$ and $\B'$,
		and stop if the answer is yes.
		
	\item[(3)]
	If inclusion was not established in steps (1) or (2) then try to find a counterexample
	to inclusion. This is best done by a Ramsey-based method (optionally using
	simulation-based subsumption techniques), e.g.,
	\cite{Rabit_CONCUR2011,RABIT}. 
	Use a small timeout value, since in most
	non-included instances there exists a very short counterexample.
	Stop if a counterexample is found.
	
	\item[(4)]
	If steps (1)-(3) failed (rare in practice), use any complete method,
	(e.g., Rank-based, Ramsey-based or Piterman's construction) to test $\A' \languageinclusion \B'$.
	At least, it will benefit from working on the smaller instance $\A', \B'$
	produced by step (1).
	
\end{enumerate}

\noindent
Note that steps (1)-(3) take polynomial time, while step (4) takes exponential time.
(For the latter, we recommend the improved Ramsey method of \cite{Rabit_CONCUR2011,RABIT}
and the on-the-fly variant of Piterman's construction \cite{Pit06} implemented in GOAL \cite{GOAL_survey_paper}.) 
This algorithm allows to solve much larger instances of the
inclusion problem than previous methods
\cite{sistla:vardi:wolper:complementation:87,GOAL_survey_paper,fogarty_et_al:LIPIcs:2011:3235,seth:buchi,seth:efficient,abdulla:simulationsubsumption,Rabit_CONCUR2011,Pit06},
i.e., automata with 1000-20000 states instead of 10-100 states; 
cf.~Sec.~\ref{sec:experiments}.

The currently implemented version of the above algorithm 
(\cite{RABIT}; RABIT v. $\ge 2.3$) uses some additional tricks. 
E.g., it hedges its bets in order to be fast on both the included and
non-included instances, by adding an initial step (0) in the algorithm.
In step (0) it performs a quick lightweight reduction and searches for 
short counterexamples, in order to quickly catch easy instances where
inclusion does not hold.
Moreover, it can run steps (2), (3) and (4) concurrently in
parallel threads (if invoked with option \texttt{-par}), 
and stops as soon as an answer is found.

%%% Local Variables:
%%% mode: latex
%%% TeX-master: "ROOT.tex"
%%% End:

\subsection{Language Inclusion Testing for NFA}\label{subsec:incl-NFA}

Just like in the reduction algorithm of Sec.~\ref{sec:heavyandlight},
one can also adapt the language inclusion checking algorithm to NFA.
The differences can be summarized as follows:
\begin{enumerate}
\item
We use the modified Heavy-k reduction algorithm for NFA with the changes
described in Sec.~\ref{subsec:heavyandlight-NFA}.
In particular, the NFA are transformed into the form with only one accepting
state, and delayed and fair (lookahead) simulations are not used.
Still, the direct forward and backward (lookahead) simulations are GFI and can
witness language inclusion, as a consequence of Theorem~\ref{thm:GFI:NFA}.
 
The inclusion-preserving reduction of $\A$ from Sec.~\ref{sec:simplifyA}
needs to be adapted to use direct trace inclusion $\directtraceinclusion$ 
(approximated by direct lookahead simulation $\transkdisim$)
instead of fair trace inclusion $\fairtraceinclusion$.
I.e., we use $\makeprunerelAB{\accblindtranskbwsim}{\transkdisim}$ for
$\A,\B$-pruning.
The inclusion-preserving reduction of $\B$ from Sec.~\ref{sec:simplifyB} carries over directly to NFA.

\item
The GFI jumping simulations of Sec.~\ref{sec:jumpsim} can also be adapted
to NFA. For the forward direction we use direct lookahead simulation, instead
of fair lookahead simulation. 
For the jumping-safe relation we can use the larger acceptance-blind backward trace
inclusion $\accblindbwdirecttraceinclusion$ (approximated by the
transitive closure of the corresponding $k$-lookahead simulation
$\accblindtranskbwsim$), instead of the counting backward trace inclusion
$\countingbwtraceinclusion$ (and its approximation $\countingtranskbwsim$).
\item
If the steps above did not witness inclusion, then one can apply 
a complete method to test inclusion $\A' \languageinclusion \B'$
on the derived smaller instance $\A', \B'$.
One type of complete methods are basic antichain-based methods
\cite{Wulf:antichains2006} that use subsumption techniques to reduce the
search space in the search for a counterexample.
More recent methods \cite{Abdulla:whensimulation2010,tool:libvata}
use stronger subsumption techniques in the search for a counterexample, 
which rely on simulation preorder (or similar approximations of language
inclusion).
Another complete method to check NFA inclusion is the {\em bisimulation modulo
 congruence} technique of \cite{Bonchi:bisimCongr2013}.
It can, roughly, be understood as collective subsumption, instead of
the individual one-on-one subsumption of
\cite{Wulf:antichains2006,Abdulla:whensimulation2010,tool:libvata}.
An element of the search space may be discarded because a set of other 
elements (instead of just one other element) makes it redundant.
This potentially allows to reduce the size of the search space even more.
The higher computational effort to check this collective subsumption yields a higher 
worst-case complexity than methods based on one-on-one subsumption,
but for typical practical instances it is often much faster.
\end{enumerate}

Unlike for B\"uchi automata (where our inclusion algorithm has a significant
advantage over previous ones; cf.~Sec.~\ref{subsec:experiments_Buchi_incl}),
the version for NFA is not necessarily always faster than the pure
antichain (resp.~congruence) based ones in 
\cite{Wulf:antichains2006,Abdulla:whensimulation2010,tool:libvata}
(resp.~\cite{Bonchi:bisimCongr2013}).
For NFA, the search space for counterexamples has a simpler structure than for
NBA. Thus the disadvantages of antichain-based methods are less relevant for
NFA. Moreover, NFA allow the construction of congruence bases as in
\cite{Bonchi:bisimCongr2013}. It is open whether a similar kind of
congruences can be established for NBA.
On many instances of NFA inclusion, the antichain-based tool of
\cite{tool:libvata} and the congruence-based tool of
\cite{Bonchi:bisimCongr2013} outperform our implementation \cite{RABIT},
though it can still be faster on some instances where the antichain 
(resp.~congruence base) happen to be very large.

%%% Local Variables:
%%% mode: latex
%%% TeX-master: "ROOT.tex"
%%% End:

%%% Local Variables:
%%% mode: latex
%%% TeX-master: "ROOT.tex"
%%% End:

\section{Experiments}\label{sec:experiments}

We test the effectiveness of Heavy-k reduction 
on Tabakov-Vardi random B\"uchi automata \cite{tabakov:model}, 
on automata derived from LTL formulae, and on 
automata derived from mutual exclusion protocols,
and compare it to the best previously available 
techniques implemented in GOAL \cite{GOAL_survey_paper}.
A scalability test shows that Heavy-k has almost quadratic
average-case complexity and it is vastly more efficient than GOAL.
We also test our methods for language inclusion
on large instances and compare their performance to
previous techniques.
Moreover, we also test the NFA version of Heavy-k reduction
on random NFA. 
Unless otherwise stated, the experiments were run with 
GOAL \cite{GOAL_survey_paper} version 2012-05-02 with Java 6
and RABIT/Reduce \cite{RABIT} version 2.4.0
with Java 7 on Intel Xeon X5550 2.67GHz and 14GB of memory.
(The raw data of the experiments is included in the arXiv version of this
paper \cite{CM:arxiv2018}.)

\subsection{B\"uchi automata}

\subsubsection{Reduction of random NBA}\label{subsec:experiments:minbuchi}

The Tabakov-Vardi model \cite{tabakov:model} generates random automata
according to the following parameters:
\begin{itemize}
\item
The number of states $n$.
\item
The size of the alphabet $|\Sigma|$.
\item
The transition density ${\it td}$. 
It determines the number of transitions in the
automaton as follows. For every symbol in $\Sigma$ there are  
$\lfloor n \cdot {\it td}\rfloor$ transitions labeled with this symbol.
\item
The acceptance density ${\it ad}$. 
This is the percentage of states that are accepting.
\end{itemize}
Apart from these parameters, Tabakov-Vardi random automata do not have any
special structure that could be exploited to make the reduction problem
or the language inclusion problem easier.
Random automata provide general reproducible test cases, on average.
Moreover, they are not biased towards any particular method, since they
do not come from any particular application domain.
A general purpose tool aught to perform well even on these hard test cases.

\begin{figure}[htp]
\begin{center}
\includegraphics[scale=0.6]{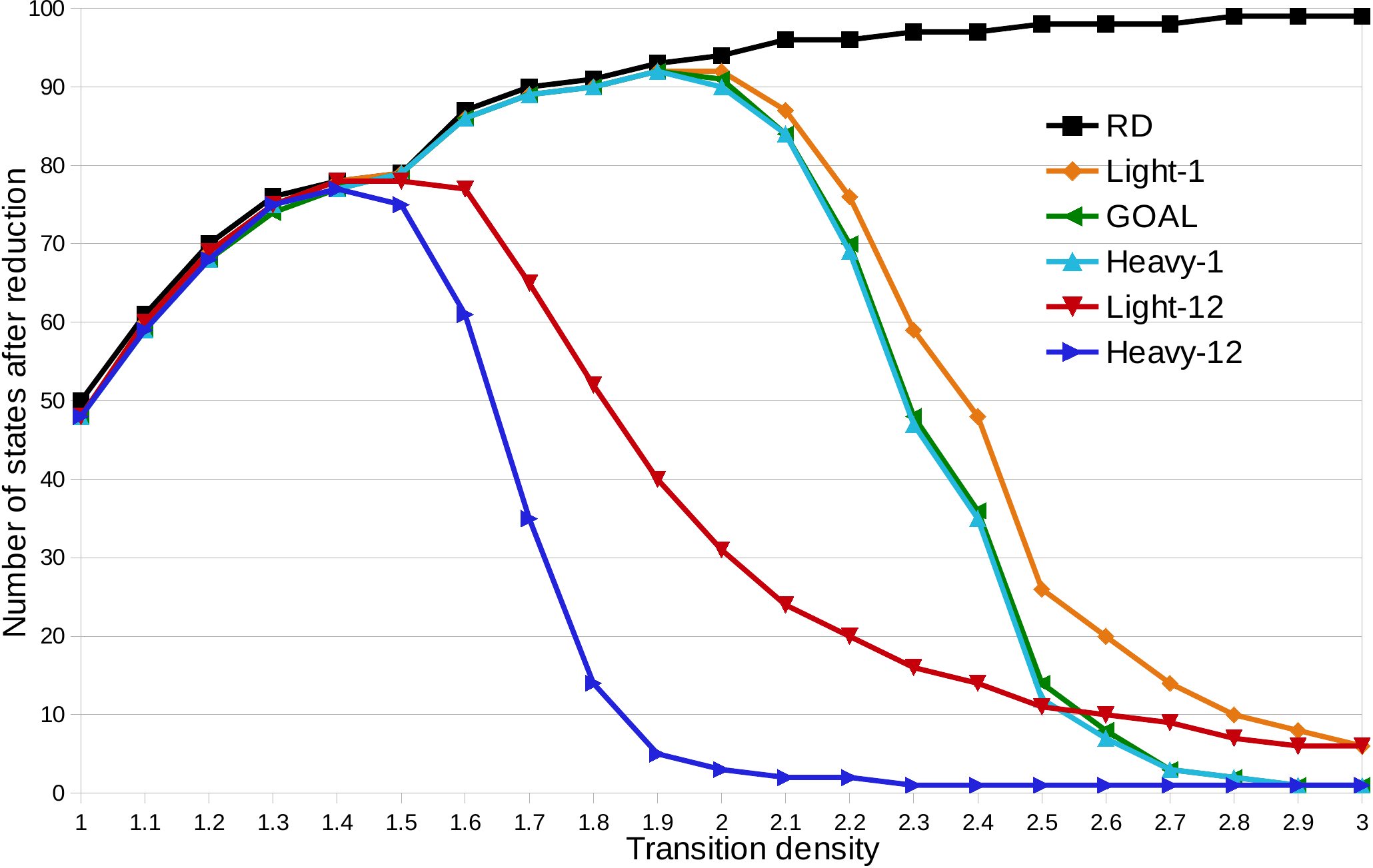}
\end{center}
\caption{
We consider Tabakov-Vardi B\"uchi automata with $n=100$, $|\Sigma|=2$, ${\it ad}=0.5$ and the range of
${\it td}=1.0, 1.1, \dots, 3.0$.
Each curve represents a different method, and we plot the number of states
after reduction. Each data point is the average of
$300$ random automata.
}\label{fig:buchimin}

\begin{center}
\includegraphics[scale=0.6]{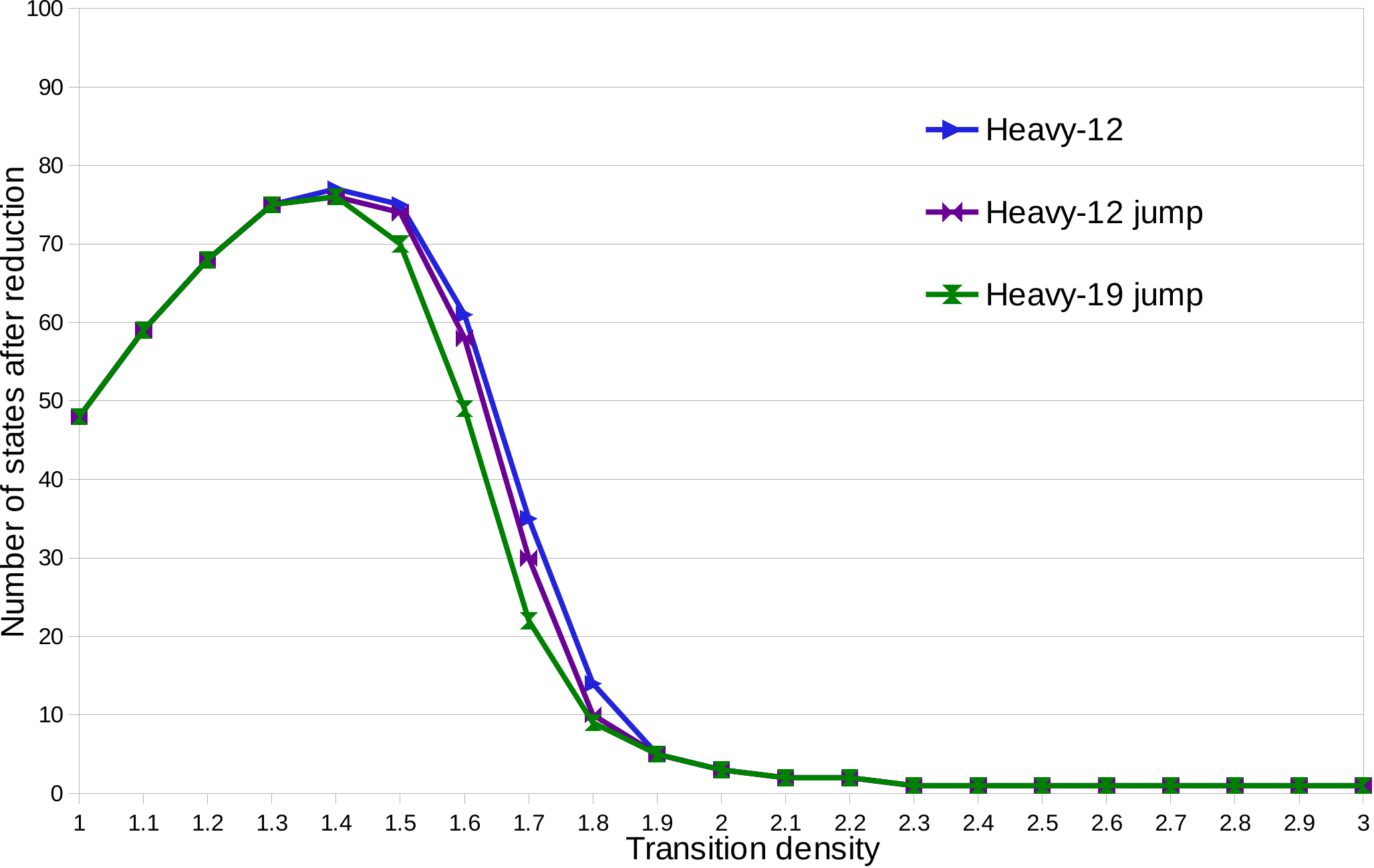}
\end{center}
\caption{
Moreover, we plotted the effects of two methods that augment our
Heavy reduction algorithm by quotienting with (a variant of) 
the jumping-safe preorders of \cite{buchiquotient:ICALP11,Clemente:PhD}:
Heavy-12 jump and Heavy-19 jump. These yield another slight improvement
in reduction, but are slower to compute.
}\label{fig:buchimin2}
\end{figure}

The inherent difficulty of the reduction problem,
and thus also the effectiveness of reduction methods, depends strongly on the class of
random automata, i.e., on the parameters listed above. Thus, one needs
to compare the methods over the whole range, not just for one example.
Variations in the acceptance density ${\it ad}$ do not affect Heavy-k much, 
but very small values make reduction harder for the other methods. 
By far the most important parameter 
is the transition density ${\it td}$, and thus we compare different techniques 
across different values of ${\it td}$. 
Fig.~\ref{fig:buchimin} and \ref{fig:buchimin2} compare the effect of different techniques.
Each curve represents a different method: RD (just remove dead states),
Light-1, Light-12, Heavy-1, Heavy-12, and GOAL. The GOAL curve
shows the best effort of all previous techniques (as implemented in GOAL),
which include RD, quotienting with backward and forward simulation, pruning of
little brother transitions and the fair simulation reduction of
\cite{GBS02}.

Sparse automata with low ${\it td}$ have more dead
states. Thus the RD method achieves a certain reduction at low ${\it td}$,
but this effect vanishes as ${\it td}$ gets higher.
For ${\it td} \le 1.4$ the effect of RD dominates. In that range, the other
techniques have only a very small effect. 
In the range $1.5 \le {\it td} \le 2.0$, GOAL still has hardly any effect
(apart from that of RD), but Light-12 and Heavy-12 achieve a significant
reduction. For ${\it td} \ge 2.0$, GOAL begins to have an effect, but it
is much smaller than that of our best techniques.

Generally, GOAL reduces just slightly worse than Heavy-1, but it is no match for
our best techniques like Heavy-12.
Heavy-12 vastly outperforms all other previous techniques, particularly in the
interesting range between ${\it td}=1.4$ and ${\it td}=2.5$. 
Moreover, the reduction of GOAL (in particular the fair simulation reduction of
\cite{GBS02}) is very slow. For GOAL, the average reduction time per automaton varies
between 39s (at ${\it td}=1.0$) and 612s (maximal at ${\it td}=2.9$).
In contrast, for Heavy-12, the average reduction time per automaton
varies between 0.012s (at ${\it td}=1.0$) and 1.482s (max. at ${\it td}=1.7$).
So Heavy-12 reduces not only much better, but also at least $400$
times faster than GOAL (see also the scalability tests below).

The computation time of Heavy-k depends both 
on the density ${\it td}$ and on the lookahead $k$.
Fig.~\ref{fig:3d} shows the average computation time of Heavy-k on automata
with $100$ states, ${\it ad}=0.5$, $|\Sigma|=2$ and varying transition density ${\it td}$ and 
lookahead $k$. The most difficult cases are those where
size reduction is possible (and thus the algorithm does not give up quickly), but 
where the size of the instance is not massively reduced.
(If some step in the algorithm greatly reduced the size of an instance, then 
subsequent computations on the now smaller automaton would be much faster.) 
For Heavy-k, the peak of the average computation time is
around ${\it td}=1.6,1.7$ (like in the scalability test; see below).

\begin{figure}[htb]
\begin{center}
\includegraphics[scale=0.7]{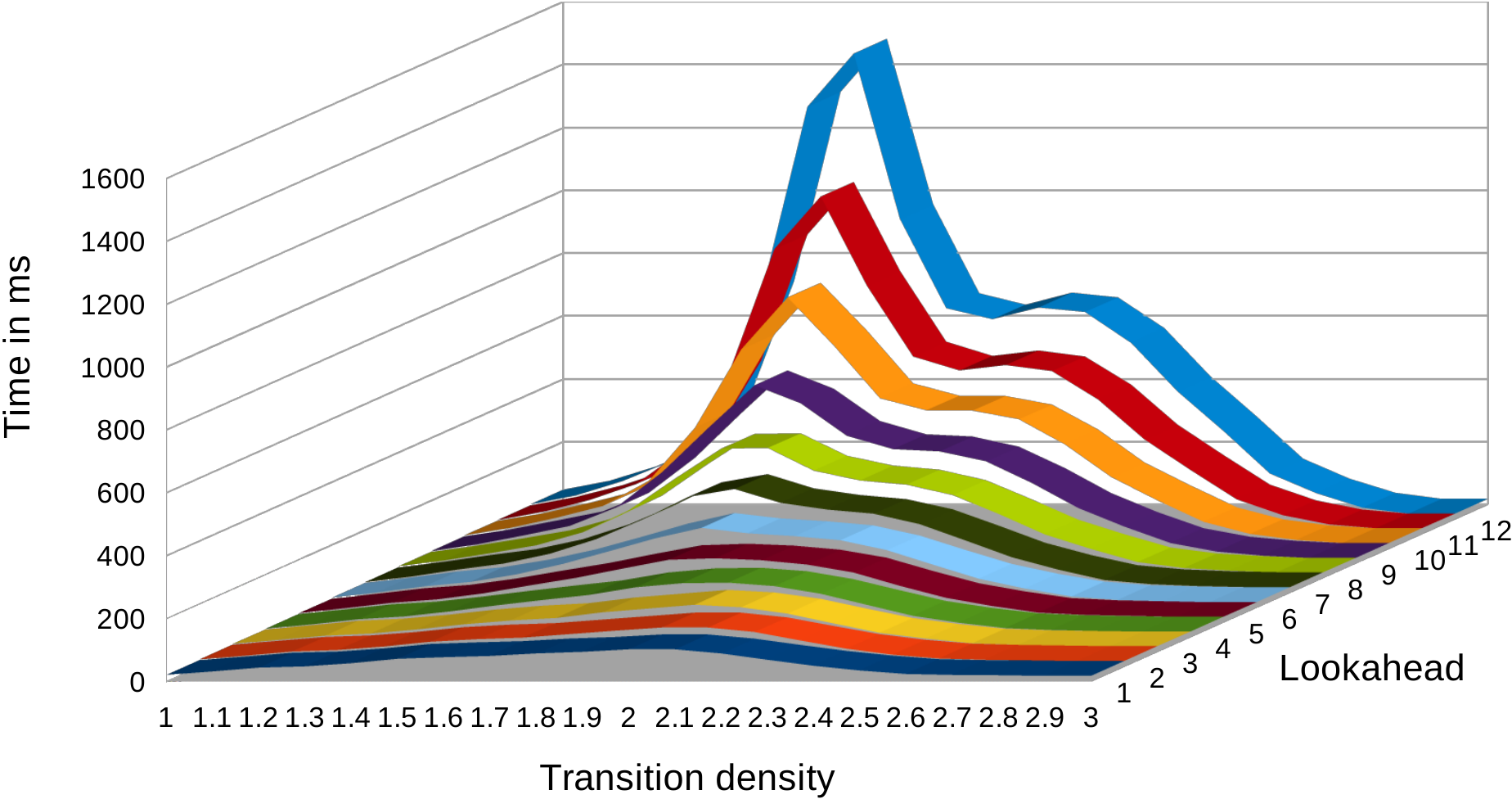}
\end{center}
\caption{Average computation time for reduction with Heavy-k on
  Tabakov-Vardi random B\"uchi automata with $n=100$ states, $|\Sigma|=2$, 
 ${\it ad}=0.5$ and varying transition density ${\it td}$ and lookahead $k$.}\label{fig:3d}  
\end{figure}

For ${\it td} \ge 2.0$, Heavy-12 yields very small automata. Many of these are
even universal, i.e., with just one state and a universal loop.
However, this frequent universality is {\em not} due to trivial reasons 
(otherwise simpler techniques like Light-1 and GOAL would also recognize this).
For example, we argue that in the tested interval of parameters $n$, $|\Sigma|$ and ${\it td}$,
there are not sufficiently many transitions to alone explain that the automaton is universal%
---and thus there are more interesting non-local structural reasons which make the automata universal.
Given Tabakov-Vardi random automata with
parameters $n$, $|\Sigma|$ and ${\it td}$,
let $U(n,|\Sigma|,{\it td})$ be the probability
that every state has at least one outgoing transition for every symbol in $\Sigma$.
Such an automaton would be trivially universal if ${\it ad}=1$.
(Note that ${n \choose k}=0$ for $k>n$.)

\begin{theorem}
	We have
	$U(n,|\Sigma|,{\it td}) = (\alpha(n,T)/\beta(n,T))^{|\Sigma|}$,
	where 
	$T=\lfloor n\cdot {\it td} \rfloor$,
	$\beta(n,T) = {{n^2} \choose T}$,
	and $\alpha(n,T)=\sum_{m=n}^{n^2} {{m-n} \choose {T-n}} \sum_{i=0}^{n} (-1)^i
	{n \choose i} {{m-in-1} \choose {n-1}}$.
\end{theorem}
\begin{proof}
For each symbol in $\Sigma$ there are 
$T=\lfloor n\cdot {\it td}\rfloor$ transitions and $n^2$ possible places for transitions,
described as a grid.
$\alpha(n,T)$ is the number of ways $T$ items can be placed onto an 
$n\times n$ grid s.t.~every row contains $\ge 1$ item, i.e., every state has an
outgoing transition. $\beta(n,T)$ is the number of possibilities without this
restriction, which is trivially ${{n^2} \choose T}$.
Since the Tabakov-Vardi model chooses transitions for different symbols
independently, we have $U(n,|\Sigma|,{\it td}) = (\alpha(n,T)/\beta(n,T))^{|\Sigma|}$.
It remains to compute $\alpha(n,T)$.
For the $i$-th row let $x_i \in \{1,\dots,n\}$ be the maximal column
containing an item. The remaining $T-n$ items can only be distributed to lower columns. 
Thus $\alpha(n,T) = \sum_{x_1,\dots,x_n} {{(\sum x_i)-n} \choose {T-n}}$.
With $m=\sum x_i$ and a standard dice-sum problem \cite{Niven:1965}
the result follows.
\end{proof}

For $n=100$, $|\Sigma|=2$ we obtain the following values for
$U(n,|\Sigma|,{\it td})$:
$10^{-15}$ for ${\it td}=2.0$, $2.9\cdot 10^{-5}$ for ${\it td}=3.0$,
$0.03$ for ${\it td}=4.0$, $0.3$ for ${\it td}=5.0$, 
$0.67$ for ${\it td}=6.0$, and $0.95$ for ${\it td}=8.0$.
So this transition saturation effect is negligible in our tested
range with ${\it td} \le 3.0$. 

While Heavy-12 performs very well, an even smaller lookahead can already
be sufficient for a good reduction. However, this depends very much on the
density ${\it td}$ of the automata. Fig.~\ref{fig:lookahead} shows the effect of
the lookahead by comparing Heavy-k for varying $k$ on different classes of
random automata with different density.

\begin{figure}[htb]
\begin{center}
\includegraphics[scale=0.7]{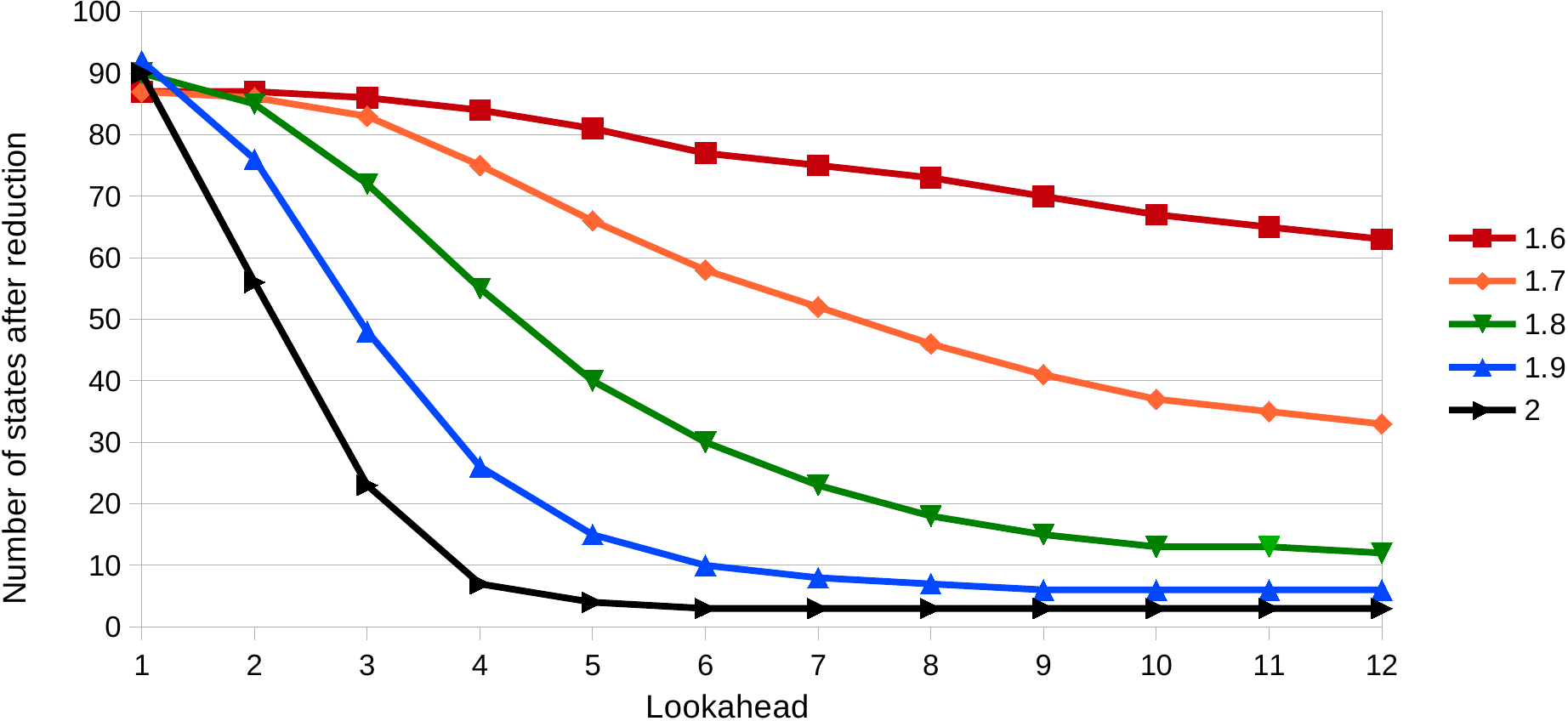}
\end{center}
\caption{The effect of the lookahead on Tabakov-Vardi automata. 
We set $n=100$, $|\Sigma|=2$, and ${\it ad}=0.5$, and vary the transition density 
${\it td}=1.6, 1.7, 1.8, 1.9, 2.0$ and the lookahead from $1,\dots,12$.
Every point is the average of the Heavy-k minimization of $1000$ random automata.
While a lower lookahead suffices for denser automata, more is needed for sparser
  instances.}\label{fig:lookahead}
\end{figure}

\subsubsection{Density of simulations on NBA}\label{subsec:density}

The big advantage of Heavy-12 over Light-12 is due to the pruning techniques.
However, these only reach their full potential at higher lookaheads 
(thus the smaller difference between Heavy-1 and Light-1).
Indeed, the simulation relations get much denser with higher lookahead $k$,
as Fig.~\ref{fig:density} shows.

\begin{figure}[htbp]
\begin{center}
\includegraphics[scale=0.8]{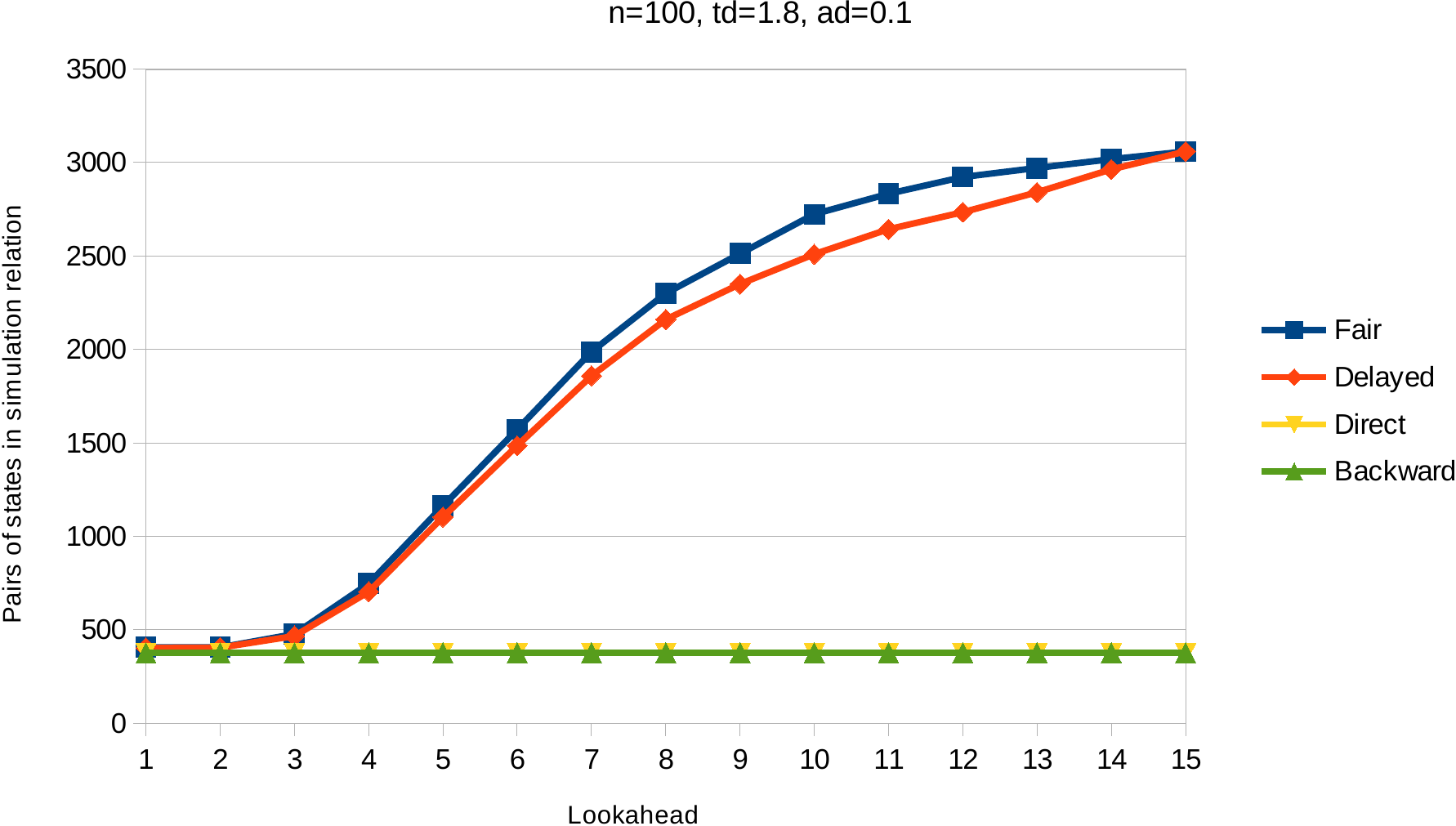}
\end{center}

\vspace*{1cm}
\begin{center}
\includegraphics[scale=0.8]{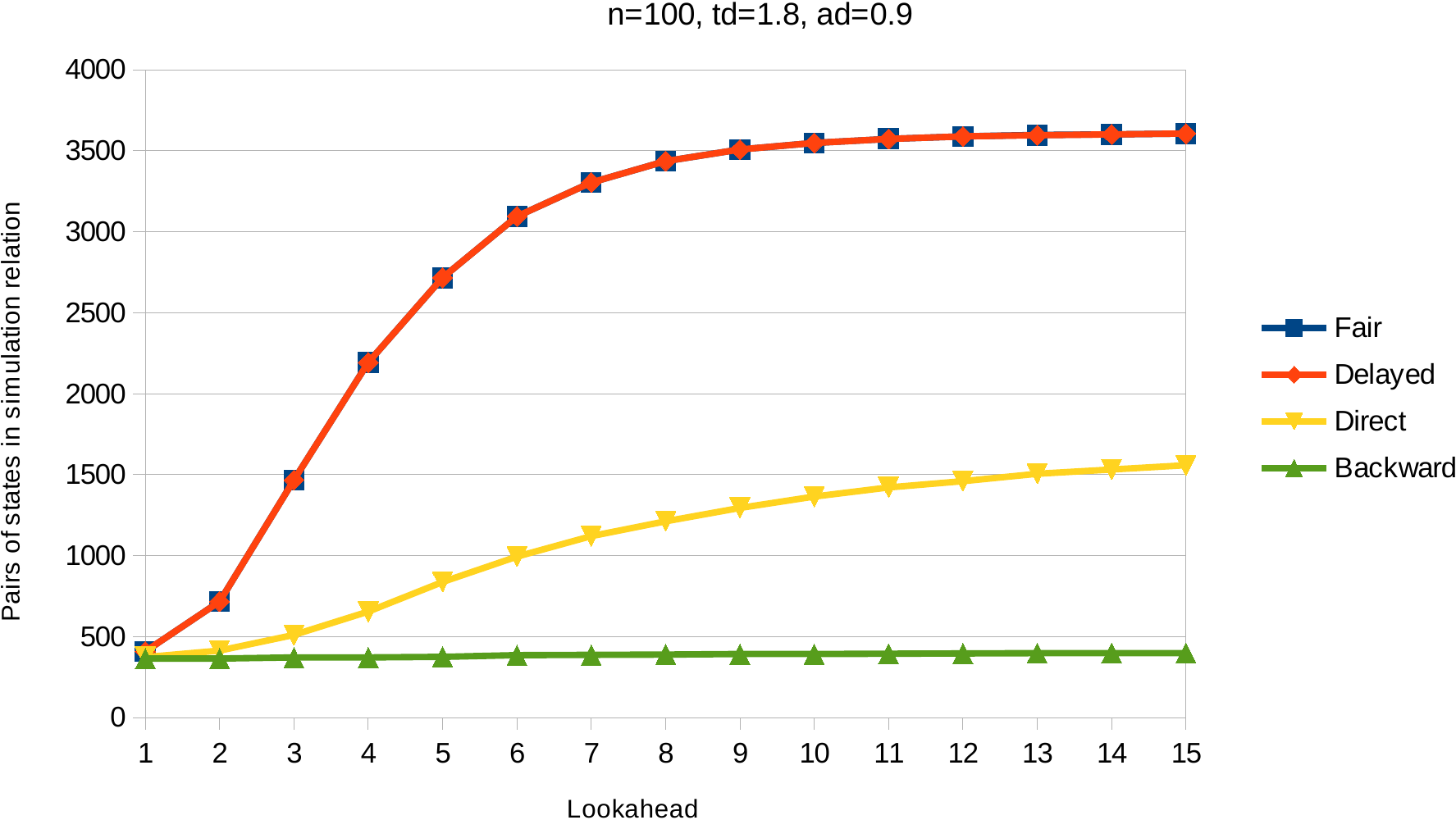}
\end{center}
\caption{
We consider Tabakov-Vardi random B\"uchi automata with $n=100$, $|\Sigma|=2$ and
${\it td}=1.8$ (a non-trivial case; larger ${\it td}$ yield larger
simulations). We let ${\it ad}=0.1$ (resp.~${\it ad}=0.9$), and plot the size of fair, delayed,
direct, and backward $k$-lookahead simulation as $k$ increases from 
$1$ to $15$. Every point is the average of $1000$ automata.
}\label{fig:density}
\end{figure}

Fair and delayed simulation relations are not much larger than direct
simulation for $k=1$, but they benefit strongly from higher $k$.
Backward simulation increases only slightly (e.g., from $363$ pairs for
lookahead $1$ to $397$ pairs for lookahead $15$ in the case of
${\it ad}=0.9$).
Initially, it seems as if backward (resp.\ direct) simulation does not benefit from higher $k$ if 
${\it ad}$ is small (on random automata), but this is wrong.
Even random automata get less random during the Heavy-k reduction process,
making lookahead more effective for backward (resp.\ direct) simulation.
Consider the case of $n=300$, ${\it td}=1.8$ and ${\it ad}=0.1$. 
Initially, the average ratio $|\transkdisimnumber{12}|/|\transkdisimnumber{1}|$
is $1.00036$, but after quotienting with $\transkdesimnumber{12}$
this ratio is $1.103$.

\subsubsection{Sparseness of the reduced NBA}\label{subsec:sparseness}

The number of states of a nondeterministic automaton is not the only measure
of its complexity. The amount of nondeterministic branching is also highly
relevant in many applications, e.g., in model checking \cite{Sebastiani-Tonetta:2003},
as well as the actual position of accepting states
when one analyzes the behavior of specific emptiness checking algorithms \cite{Blahoudek:SPIN:2014}. 
Automata with a high transition density (i.e., a large number of transitions,
relative to the number of states and symbols) have more nondeterministic
branching. Conversely, automata with a low transition density have less
nondeterministic branching. We call the latter type {\em sparse automata}.
A priori, a method that reduces the number of states of automata might
influence its transition density in either direction. In particular, the
density might become higher---e.g., there might be
a tradeoff to describe the same language with fewer states but more
transitions (per state). 
However, we show in Fig.~\ref{fig:sparseness1} that our Heavy-12 reduction method does not incur this tradeoff.
Indeed, it yields automata that are not only smaller, \emph{but also sparser}.

\begin{figure}[htb]
\begin{center}
\includegraphics[scale=0.7]{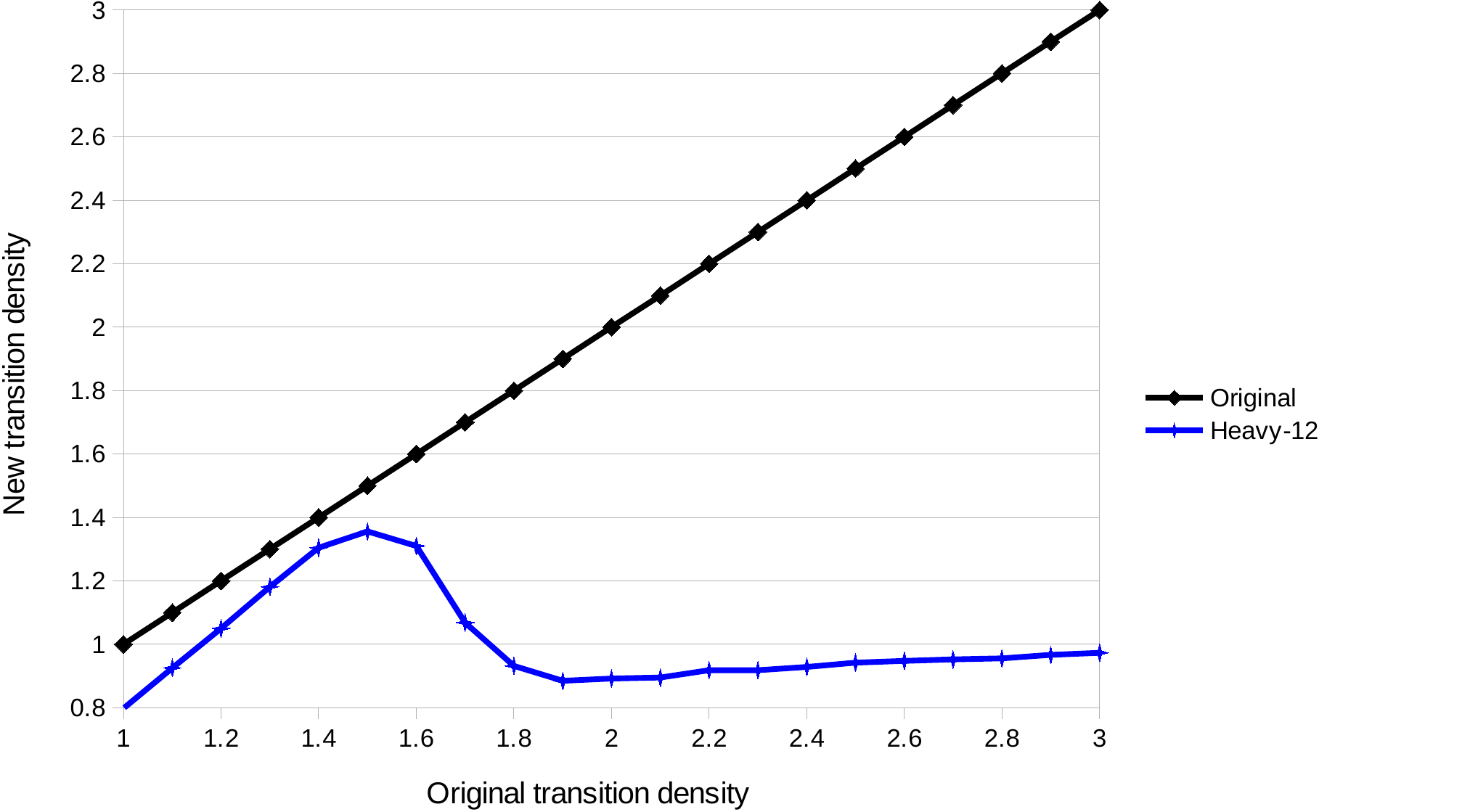}
\end{center}
\caption{
Heavy-12 produces sparse automata.
We consider Tabakov-Vardi random B\"uchi automata with $n=100$, $|\Sigma|=2$,
${\it ad}=0.5$ and ${\it td}=1.0, \dots, 3.0$.
The x-axis is the transition density of the original automata while
the y-axis is the average density of the reduced automata.
The two curves show the average transition density of the original automata
(this is just the identity function) and the average transition density of the
Heavy-12 reduced automata.
Every point is the average of $1000$ automata.
}\label{fig:sparseness1}
\end{figure}

\subsubsection{Reducing NBA derived from LTL}\label{subsec:LTL}

For model checking \cite{Holzmann:Spinbook}, LTL-formulae are converted into B\"uchi automata. This
conversion has been extensively studied and there are many different
algorithms which try to construct the smallest possible automaton
for a given formula; cf.~references in \cite{GOAL_survey_paper}
and \cite{BKRS:TACAS2012} and \cite{Spot}.

It should be noted however, that LTL is designed for human readability and 
does not cover the full class of $\omega$-regular languages. Moreover,
the website and database B\"uchi Store~\cite{buechistore:2011,buechistore:2013} contains handcrafted automata for
almost every human-readable LTL-formula, and almost all of these automata have 
$\le 10$ states.

Moreover, new LTL to B\"uchi automata converters are being developed 
every year~\cite{GOAL_survey_paper,BKRS:TACAS2012,Spot}, and it is not
in the scope of this paper to benchmark all converters.

For the scope of this paper, B\"uchi automata generated from 
random LTL formulae are simply yet another class of test cases
for our size reduction algorithm.
In particular, they are different from the Tabakov-Vardi random automata.

In order to get interesting test cases, we used random LTL formulae that are
larger than typical human-readable ones and obtained larger automata 
on average (see below).

Moreover, for LTL model checking, the size of the automata 
is not the only criterion~\cite{Sebastiani-Tonetta:2003}, since 
more nondeterminism also makes the problem harder. However, our 
transition pruning techniques reduce the amount of nondeterministic branching,
and yield automata that are not only smaller but also sparser (i.e., `less
nondeterministic'); cf.~our results below, and also Sec.~\ref{subsec:sparseness}.

Using a function of GOAL, we created 300 random LTL-formulae of non-trivial
size: length 70, 4 predicates and probability weights 1 for boolean and 2 for
future operators. We then converted these formulae to B\"uchi automata and 
reduced them with GOAL. Of the 14 different converters implemented in GOAL
we chose LTL2BA \cite{GastinOddoux2001} (as implemented in GOAL, 
which behaves slightly differently from the stand-alone LTL2BA tool)
since it was the only one (in GOAL) which could handle such large formulae.
(The second best was COUVREUR \cite{Couvreur:FM:1999} which succeeded on 90\% of the
instances, but produced much larger automata than LTL2BA. The other converters 
ran out of time (4h) or memory (14GB) on most instances.)
We thus obtained 300 automata and reduced them with GOAL. 
The resulting automata vary significantly in size from 1 state to 1722 states.

Then we tested how much {\em further} these automata could be reduced in size by our
Heavy-12 method. In summary,
82\% of the automata could be further reduced in size. 
The average number of states declined from 138 to 78, and the average number
of transitions from 3102 to 1270. Since larger automata have a 
disproportionate effect on averages, we also computed the average reduction
ratio per automaton, i.e., 
$(1/300) \sum_{i=1}^{300} {\it newsize}_i/{\it oldsize}_i$.
(Note the difference between the average ratio and the ratio of averages.)
The average ratio was $0.76$ for states and $0.68$ for transitions.
The computation times for reduction vary a lot due to different automata sizes
(average 4.1s), but were always less than the time used by the LTL to automata translation. 
If one only considers the 150 automata above median size (30 states)
then the results are even stronger. 100\% of these automata could be further
reduced in size.
The average number of states declined from 267 to 149, and the average number
of transitions from 6068 to 2435. 
The average reduction ratio was $0.65$ for states and $0.54$ for transitions.

\subsubsection{Reducing NBA derived from mutual exclusion protocols}\label{sec:protocols}

In Table~\ref{tab:protocol} we consider automata derived from mutual exclusion protocols.
The protocols were described in a language of guarded commands and
automatically translated into B\"uchi automata, whose size is given in the
column `Original'. 
% By row, the protocols are Bakery.1, Bakery.2, Fischer.3.1, Fischer.3.2, 
% Fischer.2, Phils.1.1, Phils.2 and Mcs.1.2.
We reduce these automata with GOAL and with our Heavy-12 method 
and describe the sizes of the resulting automata and the runtime in 
subsequent columns. 

\begin{table}[ht]
\begin{center}% \small
	\begin{tabular}{|l|c|c|c|c|c|c|c|c|}
	  \hline
          Automaton name &
	\multicolumn{2}{|c|}{Original}
	  & \multicolumn{2}{c|}{GOAL} & Time & \multicolumn{2}{c|}{Heavy-12} & Time\\
	  \cline{2-5}\cline{7-8}
	  & Trans. & States & Tr. & St. & GOAL & Tr. & St. &
	  Heavy-12\\
	  \hline
	  bakery.1.c.ba & 2597 & 1506 & N/A & N/A & $>2h$ & 696 & 477 & 5.3s\\
	  \hline
	  bakery.2.c.ba & 2085 & 1146 & N/A & N/A & $>2h$ & 927 & 643 & 7.6s\\
	  \hline
	  fischer.3.1.c.ba & 1401 & 638 & 14 & 10 & 15.38s & 14 & 10 & 0.86s\\
	  \hline
	  fischer.3.2.c.ba & 3856 & 1536 & 212 & 140 & 4529s & 96 & 70 & 3.4\\
	  \hline
	  fischer.2.c.ba & 67590 & 21733 & N/A & N/A & oom(14GB) & 316 & 192 & 253.5s\\
	  \hline
	  phils.1.1.c.ba & 464  & 161 & 362 & 134 & 540.3s & 359 & 134 & 1.5s\\
	  \hline
	  phils.2.c.ba & 2350 & 581 & 284 & 100 & 164.2s & 225 & 97 & 1.8s\\
	  \hline
	  mcs.1.2.c.ba & 21509 & 7968 & 108 & 69 & 2606.7s & 95 & 62 & 42.9s\\
	  \hline
	\end{tabular}
\end{center}
\caption{Reduction of NBA derived from mutual exclusion protocols, comparing
GOAL \cite{GOAL_survey_paper} (version 2012-05-02 on Java 6) 
and RABIT/Reduce method Heavy-12 (version 2.4.0 on Java 7) 
using an Intel i7-740, 1.73 GHz.
In some instances GOAL ran out of time (2h) or memory (14GB).
}\label{tab:protocol}
\end{table}

\subsubsection{Scalability of NBA reduction}

We tested the scalability of Heavy-12 reduction by applying it to
Tabakov-Vardi random automata of increasing size $n$ but fixed ${\it td}$, 
${\it ad}$ and $\Sigma$. We ran four separate tests with ${\it td}=1.4, 1.6, 1.8$ and
$2.0$. In each test we fixed ${\it ad}=0.5$, $|\Sigma|=2$ and increased the number of states
from $n=50$ to $n=1000$ in increments of 50. For each parameter point 
we created $300$ random automata and reduced them with
Heavy-12. We analyze the average size of the reduced automata in percent of
the original size $n$, and how the average computation time increases with
$n$.

For ${\it td}=1.4$ the average size of the reduced automata stays around 
$77\%$ of the original size, regardless of $n$.
For ${\it td}=1.6$ it stays around $65\%$.
For ${\it td}=1.8$ it {\em decreases} from
$28\%$ at $n=50$ to $2\%$ at $n=1000$.
For ${\it td}=2.0$ it {\em decreases} from
$8\%$ at $n=50$ to $<1\%$ at $n=1000$.
See Fig.~\ref{fig:scalability_size}.

\begin{figure}[htbp]
\begin{center}
\includegraphics[scale=0.7]{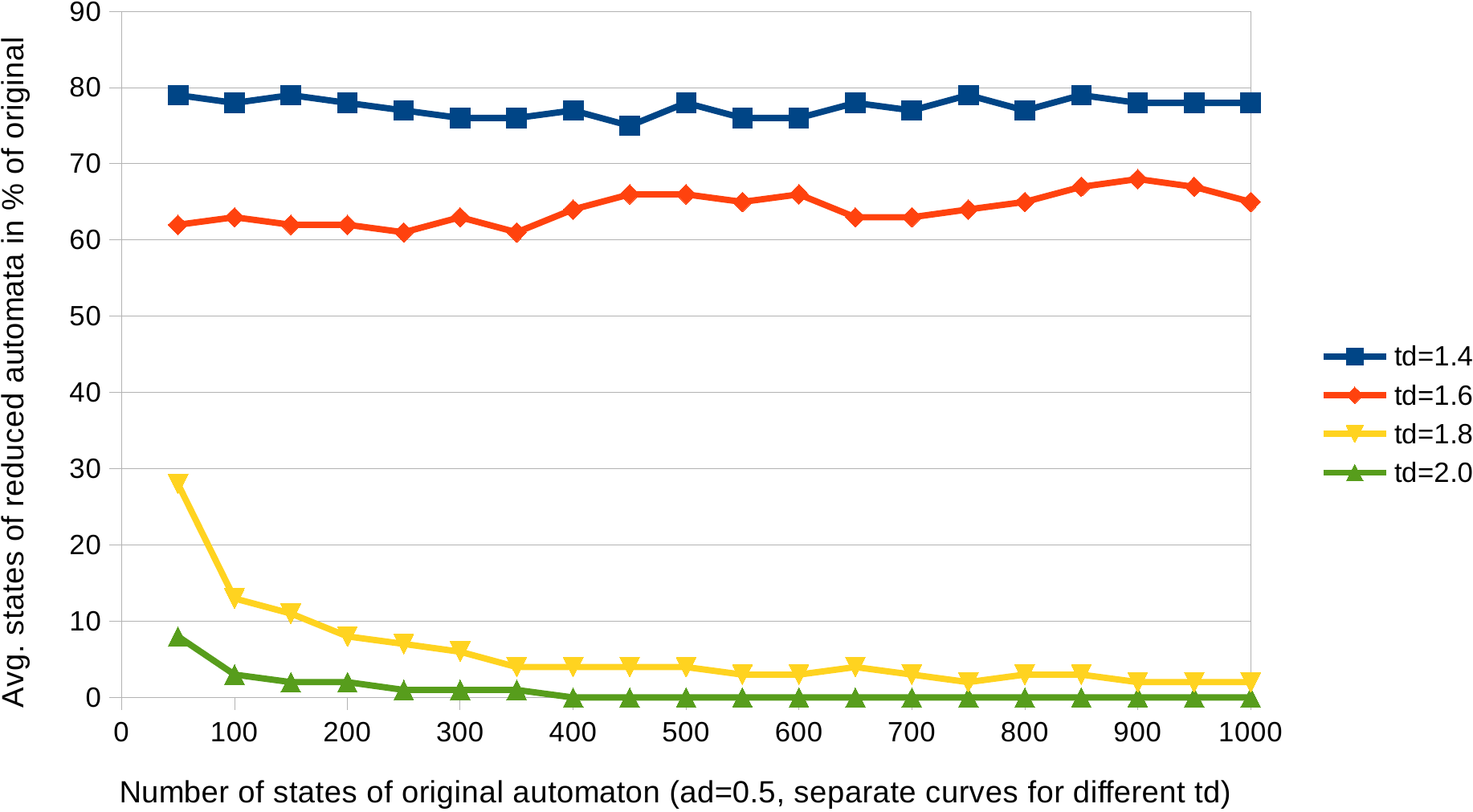}
\end{center}
\caption{Reduction of Tabakov-Vardi random B\"uchi automata with ${\it ad}=0.5$,
$|\Sigma|=2$, and increasing $n=50,100,\dots,1000$.
Different curves for different ${\it td}$.
We plot the average size of the Heavy-12 reduced automata, in percent of
their original size. Every data point is the average of $300$ automata.
}\label{fig:scalability_size}
\end{figure}

Note that the lookahead of 12 did {\em not change} with $n$.
Surprisingly, larger automata do not require larger lookahead for a good reduction.

In Fig.~\ref{fig:scalability_buchi} we plot the average computation time (measured in ms) in $n$ and then compute
the optimal fit of the function ${\it time} = a \cdot n^b$ to the data by the least-squares
method, i.e., this computes the parameters $a$ and $b$ of the function that
most closely fits the experimental data. The important parameter is the
exponent $b$. For ${\it td}=1.4, 1.6, 1.8, 2.0$ we obtain
$0.0036 \cdot n^{2.26}$, $0.012 \cdot n^{2.41}$, $0.02 \cdot n^{2.16}$ and
$0.0046 \cdot n^{2.37}$, respectively.
We also measured the median time used for reduction, and it was always very
close to the average time.

Thus, the average-case complexity of Heavy-12 scales slightly above 
quadratically (with exponents between $2.16$ and $2.41$). 
This is especially surprising given that Heavy-12 does not only
compute one simulation relation but potentially many simulations until 
the repeated reduction reaches a fixpoint.
Quadratic complexity is the very best one can hope for in any method that
explicitly compares states/transitions by simulation relations, since the
relations themselves are of quadratic size. 
Lower complexity is only possible with pure partition refinement techniques 
(e.g., bisimulation, which is $O(n\log n)$ for graphs with a fixed
out-degree), but these achieve even less
reduction than quotienting with direct simulation (i.e., almost nothing 
on hard instances).

\begin{figure}[htbp]
\begin{center}
\includegraphics[scale=0.8]{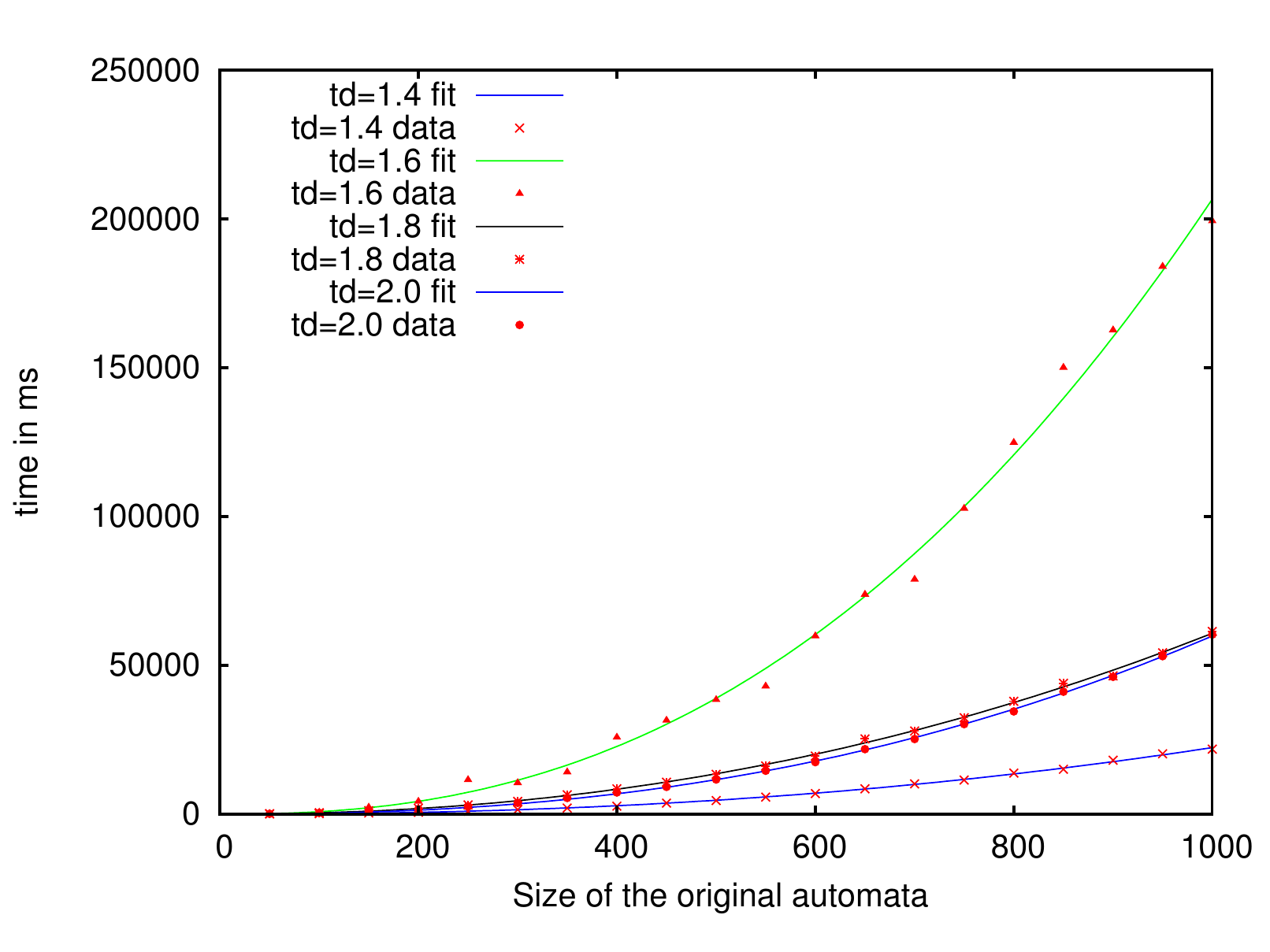}
\end{center}
\caption{Average computation time for Heavy-12 on Tabakov-Vardi B\"uchi
  automata with ${\it ad}=0.5$, $|\Sigma|=2$ and 
  ${\it td}=1.4, 1.6, 1.8, 2.0$, with a least squares fit of the function 
  $y=a \cdot x^b$.
  The x-axis shows the number of states of the original automata and the
  y-axis shows the average runtime in ms.
  Each data point is the average of 300 automata.}\label{fig:scalability_buchi}
\end{figure}

\subsubsection{Language inclusion testing for NBA}\label{subsec:experiments_Buchi_incl}

We evaluated the language inclusion testing algorithm of
Sec.~\ref{sec:inclalg} (with lookahead up-to $15$)
on non-trivial instances and compared its 
performance to previous techniques like ordinary fair simulation checking
and the best effort of GOAL (which uses simulation-based reduction 
followed by checking fair simulation preorder and an on-the-fly variant of Piterman's construction
\cite{Pit06,GOAL_survey_paper}).

We considered pairs of Tabakov-Vardi random automata with $1000$ states each,
$|\Sigma|=2$ and ${\it ad}=0.5$. For each separate case of 
${\it td}=1.6, 1.8$ and $2.0$, we created $300$ such automata pairs and tested
if language inclusion holds. (For ${\it td}<1.6$ inclusion rarely holds,
except trivially if one automaton has the empty language. For ${\it td}>2$
inclusion holds very often and is relatively easy to prove, since the
languages of the automata are often almost universal.)

For ${\it td}=1.6$ our algorithm solved $297$ of $300$ instances (i.e., 99\%):
$45$ included, $252$ not included, and $3$ timeouts (30min).
Of the $45$ included cases, $16$ were shown during the reduction/preprocessing
(step 1), $29$ were shown by jumping fair simulation, using lookaheads
between 9 and 15 (step 2), and none of the included cases were shown by the Ramsey
method (step 4).
(Step 3 can only prove non-inclusion.) 
The average computation time for the included cases was $192.6$ seconds.
Of the $252$ non-included cases, most were shown very quickly by short
counterexamples. The average computation time for the non-included cases 
was $33$ seconds.
In contrast, ordinary fair simulation solved only $13$ included instances.
GOAL (timeout 30min, 14GB memory) solved only $13$ included
instances (the same $13$ as fair simulation) and $155$ non-included instances,
i.e., a success rate of just 56\%. (The results were the same if the timeout
for GOAL was increased to 60min.)

For ${\it td}=1.8$ our algorithm solved $300$ of $300$ instances (i.e., 100\%):
$103$ included, and $197$ non-included.
Of the $103$ included cases, all were shown during reduction/preprocessing
(step 1), and none by steps 2, 3, 4.
The average computation time for the included cases was $118$ seconds.
The average computation time for the $197$ non-included cases 
was $6.6$ seconds.
Ordinary fair simulation solved only 5 included instances.
GOAL (timeout 30min, 14GB memory) solved only 5 included
instances (the same 5 as fair simulation) and $115$ non-included instances,
i.e., a success rate of just 40\%.

For ${\it td}=2.0$ our algorithm solved $300$ of $300$ instances (i.e., 100\%):
$143$ included, and $157$ non-included.
Of the $143$ included cases, all were shown during reduction/preprocessing
(step 1) and none by steps 2, 3, 4.
The average computation time for the included cases was $127$ seconds.
The average computation time for the $157$ non-included cases 
was $5.4$ seconds.
Ordinary fair simulation solved only 1 of the $143$ included instances.
GOAL (timeout 30min, 14GB memory) solved only 1 of $143$ included
instances (the same one as fair simulation) and $106$ of $157$ non-included
instances, i.e., a success rate of just 35.7\%.

\subsection{Finite automata}

\subsubsection{Reduction of random NFA}\label{subsec:NFA_reduction}

Like in Sec.~\ref{subsec:experiments:minbuchi},
we consider Tabakov-Vardi random automata.
However, here we interpret these automata as NFA instead of B\"uchi automata,
and reduce them such that the finite-word language is preserved.

Generally, random NFA are harder to reduce in size than random NBA,
because for NFA it matters when (i.e., in exactly which step) one 
encounters an accepting state. In contrast, for NBA it only matters whether
one encounters accepting states infinitely often.
Thus random NFA have somewhat more complex languages than random NBA
to begin with.

The generated Tabakov-Vardi random automata normally have many accepting states. However,
before applying the reduction methods, we first transform the NFA into equivalent ones with
just a single accepting state without any outgoing transitions. 
(Unless the empty word is in the language, in which case the initial state is
accepting too.)
Note that the same cannot be done for B\"uchi automata. 
This transformation of NFA
makes direct (and backward) simulations significantly
larger, and thus increases the effect of the reduction methods.
The reason is that direct/backward simulations need to match accepting states 
immediately, regardless of whether the input word has already been fully read to the
end or not. This makes it very hard for Duplicator to win the simulation game,
and thus yields very small direct/backward simulations. 
However, if an NFA has just a single accepting state without any outgoing
transitions then this state needs to be matched at most once in a simulation
game, which is much easier.
Note that this transformation does not actually make an NFA more complex, in
the following sense. While one gets some additional transitions 
(albeit of a special type, all going to the one accepting state),
the description of the set of accepting states becomes correspondingly simpler,
since it just consists of a single element.
Fig.~\ref{fig:Nosingle} shows that doing this transformation is very
important. Without it, even the best methods, like Heavy-12, perform very
poorly (see the graph for Heavy-12, multi acc states).
For random NFA with transition density $\le 1.5$, all methods do not achieve 
much more than remove dead states. For such automata, the transformation into
the form with one accepting state does not make much difference, except for
the tradeoff between the number of transitions and the complexity of
describing the set of accepting states. (See also the results
on the transition density in Fig.~\ref{fig:sparseness2}.)

Fig.~\ref{fig:NFAmin} shows the effect of different reduction methods
(that all use the trick of transforming NFA into a form with a single
accepting state).
Like for B\"uchi automata, our Heavy method reduces the size far more than any
simpler technique.

\begin{figure}[htbp]
\begin{center}
\includegraphics[scale=0.7]{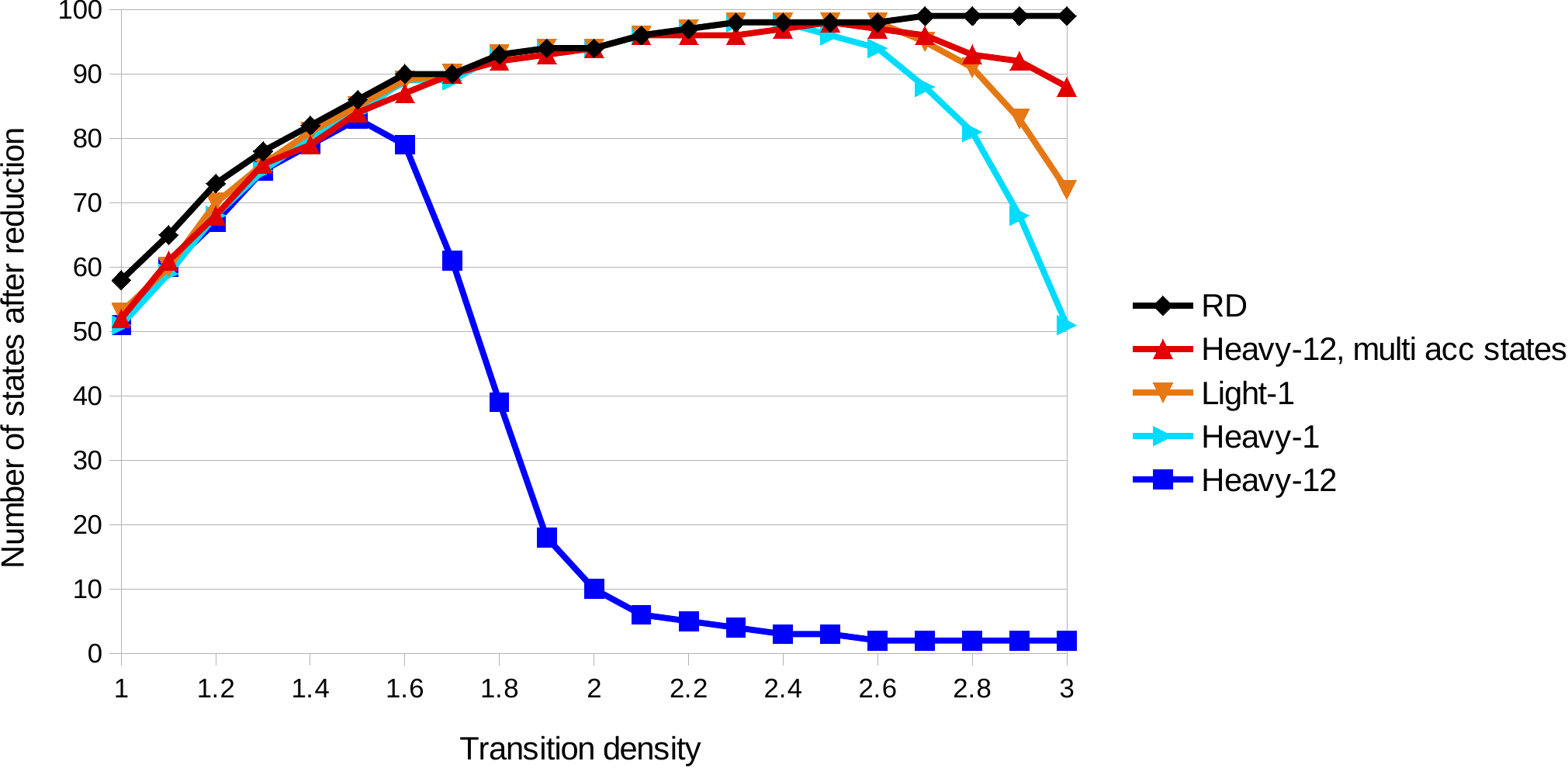}
\end{center}
\caption{Tabakov-Vardi random NFA with $n=100$, $|\Sigma|=2$, ${\it ad}=0.5$ and the range of
${\it td}=1.0, 1.1, \dots, 3.0$. 
Each curve represents a different reduction method: RD (just remove dead
states).
Heavy-12, multi acc states: Like Heavy-12 but \emph{without} the transformation into
a form with a single accepting state.
Light-1, Heavy-1 and Heavy-12 all use the transformation into
a form with a single accepting state.
Each data point is the average of $1000$ random automata.}\label{fig:Nosingle}
\end{figure}

\begin{figure}[htbp]
\begin{center}
\includegraphics[scale=0.7]{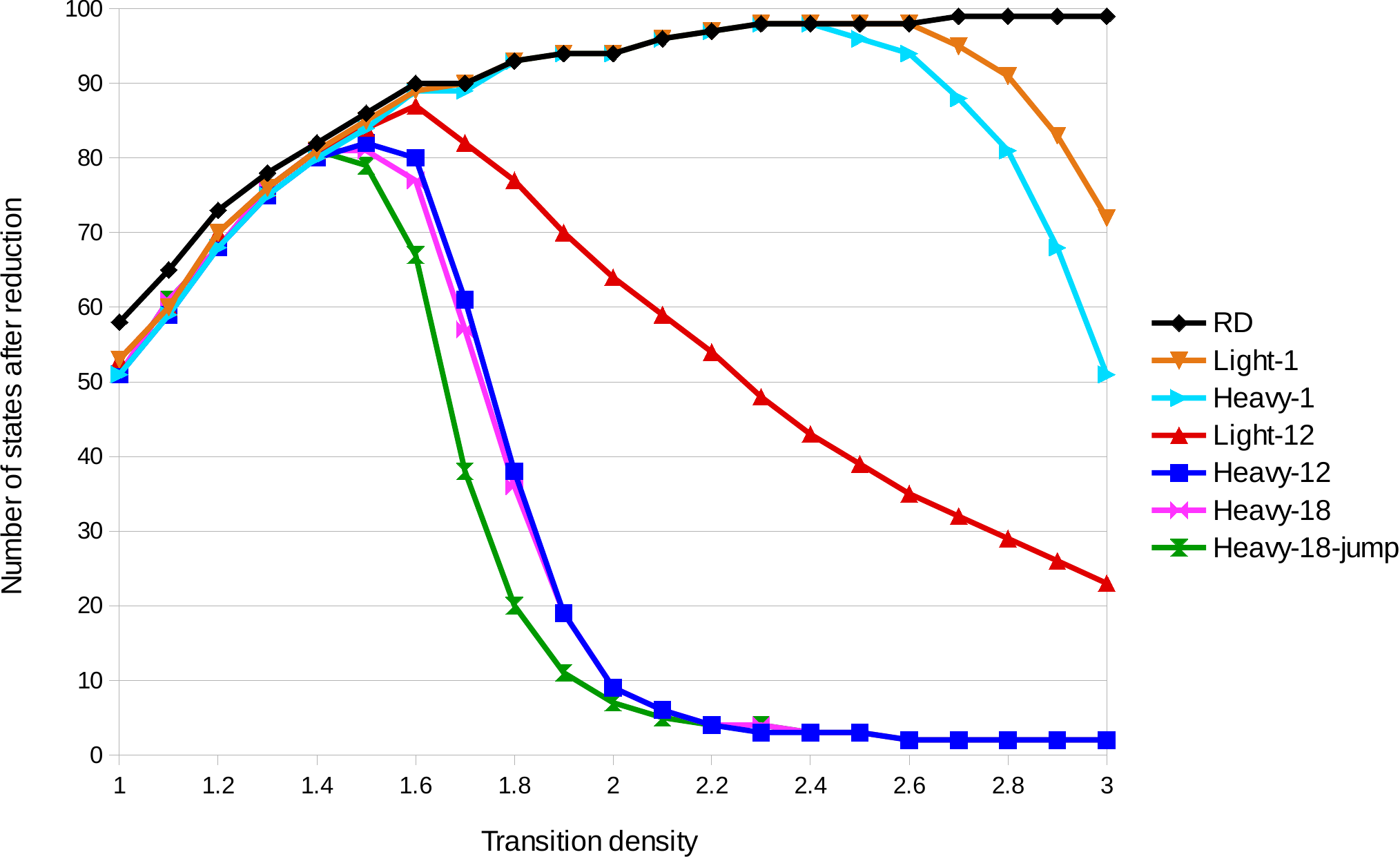}
\end{center}
\caption{Tabakov-Vardi random NFA with $n=100$, $|\Sigma|=2$, ${\it ad}=0.5$ and the range of
${\it td}=1.0, 1.1, \dots, 3.0$. 
Each curve represents a different reduction method: RD (just remove dead states),
Light-1, Heavy-1, Light-12, Heavy-12, Heavy-18 and Heavy-18 jump (which is
Heavy-18 augmented by quotienting with (a variant of) 
the jumping-safe preorders of \cite{buchiquotient:ICALP11,Clemente:PhD}).
Each data point is the average of $1000$ random automata.}\label{fig:NFAmin}
\end{figure}

\subsubsection{Density of simulations on NFA}\label{subsec:density_NFA}

In Fig.~\ref{fig:NFAsimdensity} we measure the density of direct simulation and backward simulation on
Tabakov-Vardi random NFA. We take $n=100$, $|\Sigma|=2$ and
${\it td}=1.8$ (a non-trivial case; larger ${\it td}$ yield larger
simulations). To show the effect of the acceptance density, we consider two
cases: ${\it ad}=0.1$ and ${\it ad}=0.9$.
Like in Sec.~\ref{subsec:NFA_reduction}, these 
NFA had been transformed into equivalent ones with
just a single accepting state. This makes direct simulation on NFA significantly
larger than on B\"uchi automata, particularly if ${\it ad}$ is high.
We plot the size of direct and backward simulation as the lookahead $k$ increases from 
$1$ to $15$. 

\begin{figure}[htp]
\begin{center}
\includegraphics[scale=0.8]{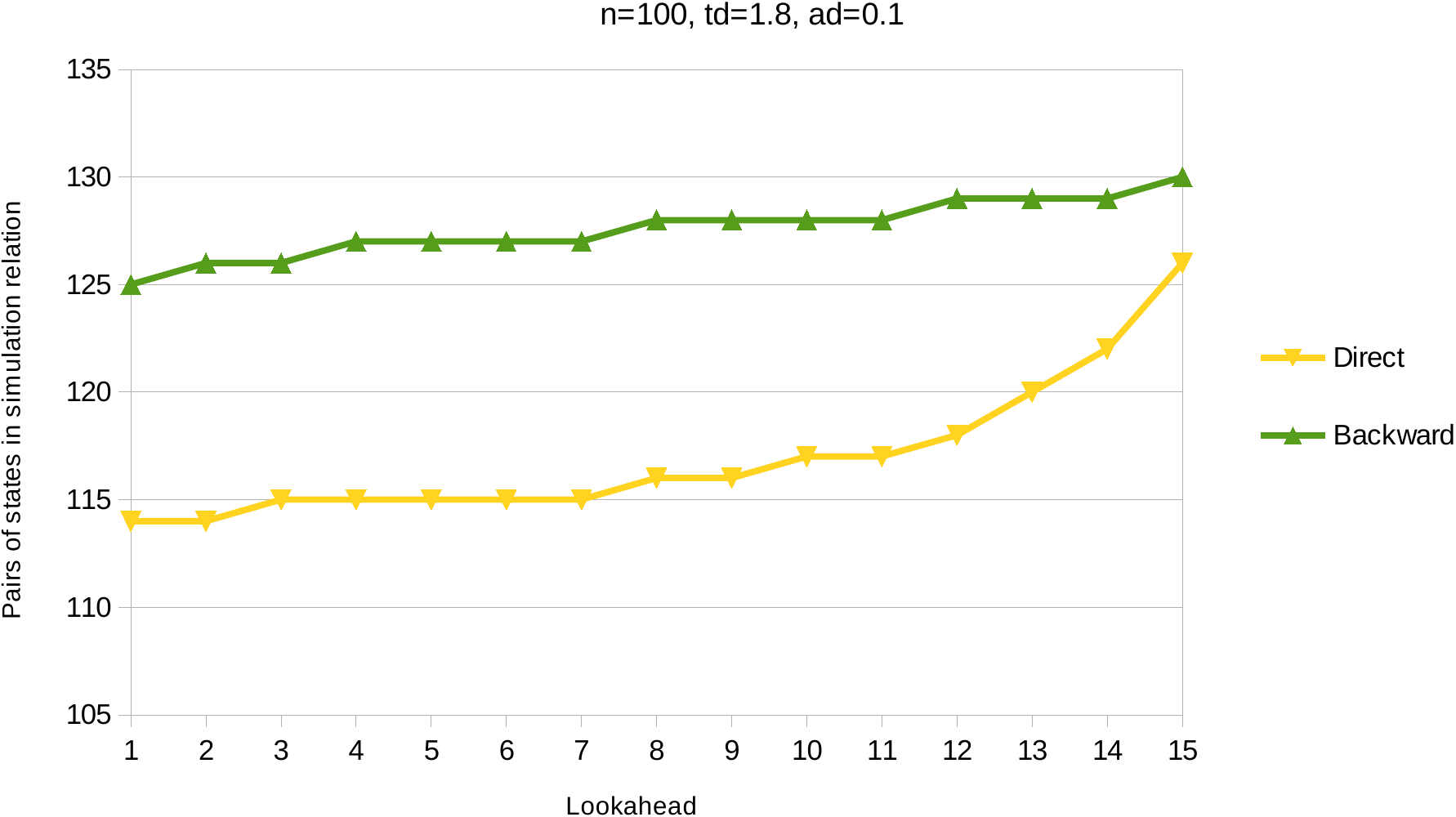}
\end{center}

\vspace*{1cm}
\begin{center}
\includegraphics[scale=0.8]{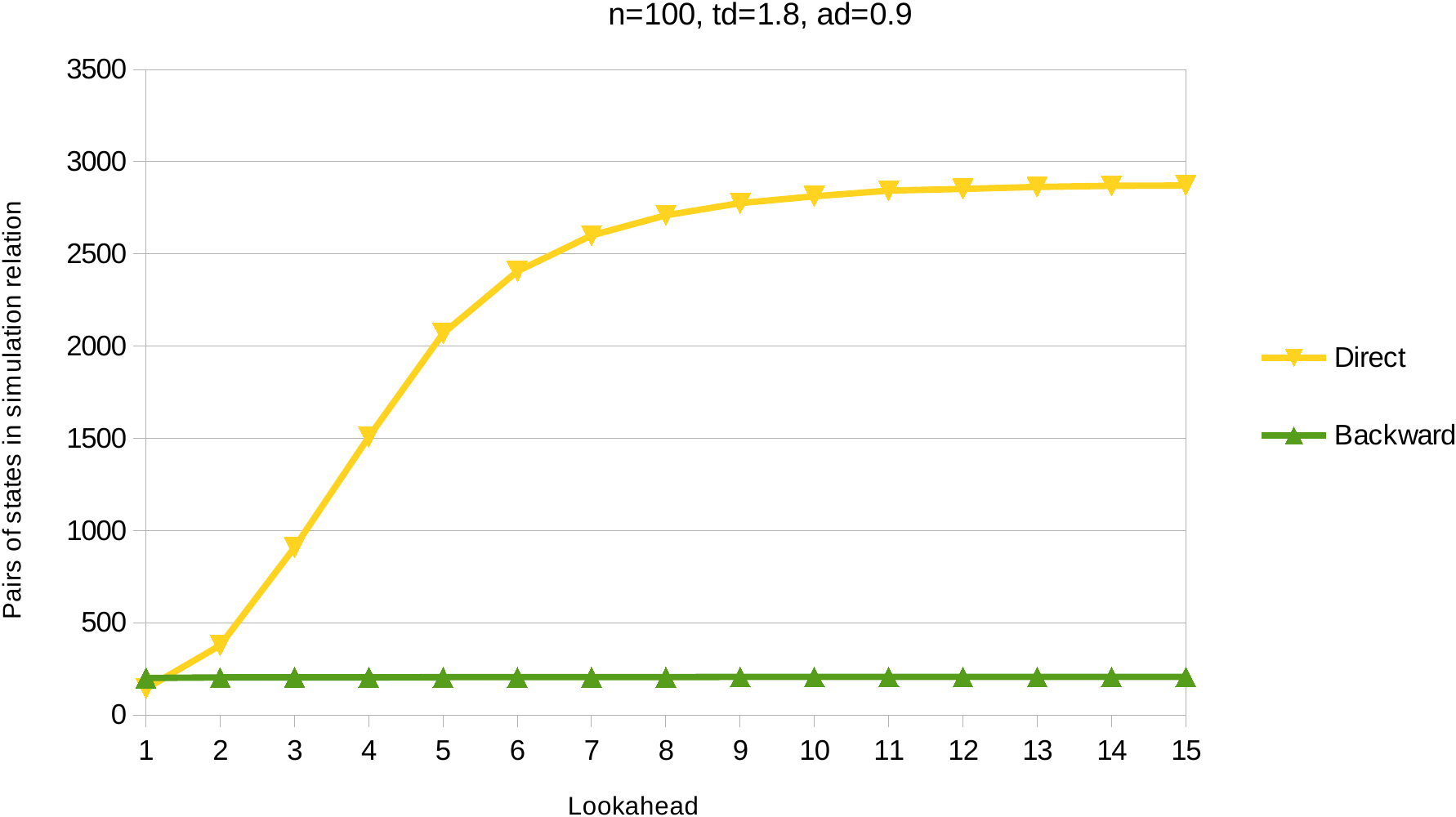}
\end{center}
\caption{Density of direct simulation and backward simulation on 
Tabakov-Vardi random NFA with $n=100$, $|\Sigma|=2$, ${\it td}=1.8$ and ${\it
  ad}=0.1$ (top) and
${\it ad}=0.9$ (bottom), respectively.
On the x-axis, the lookahead increases from $1$ to $15$.
On the y-axis, we measure the size of the simulation relations.
Every data point is the average of $1000$ random automata.
}
\label{fig:NFAsimdensity}
\end{figure}

\subsubsection{Sparseness of reduced NFA}\label{subsec:sparseness_NFA}

Like for B\"uchi automata in Sec.~\ref{subsec:sparseness},
we measure the average transition density of the Heavy-12 reduced random NFA.
Our algorithm first transforms the NFA into a form 
with only one accepting state. This adds a significant number of transitions
and thus increases the transition density. However, 
for NFA with transition density $> 1.5$, 
the Heavy-12 procedure 
then decreases the transition density again.
In Fig.~\ref{fig:sparseness2} we thus plot the original transition density,
the density after the transformation into the form with one accepting state
and the density of the Heavy-12 reduced automata.

\begin{figure}[htb]
\begin{center}
\includegraphics[scale=0.7]{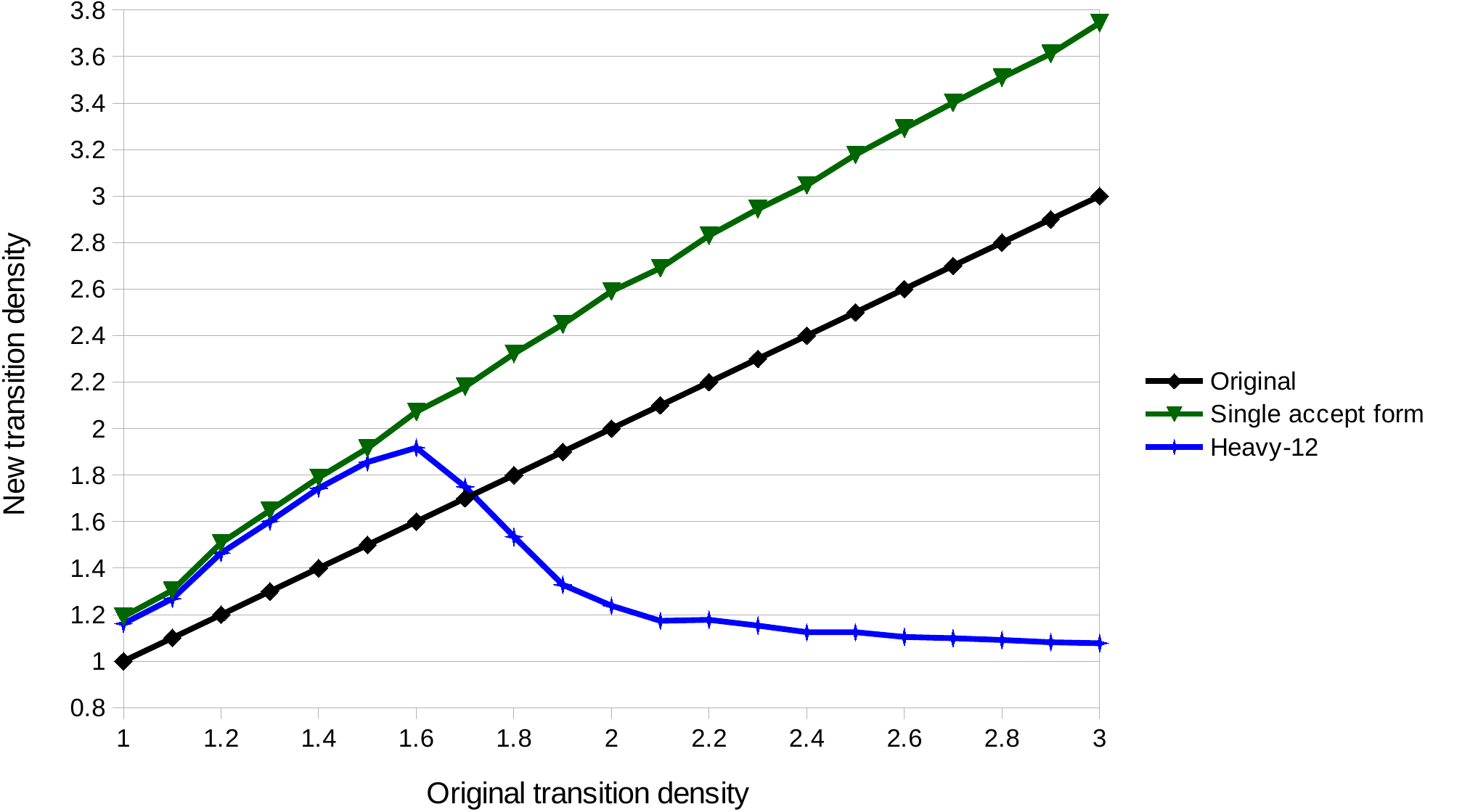}
\end{center}
\caption{
Heavy-12 produces sparse automata.
We consider Tabakov-Vardi random NFA with $n=100$, $|\Sigma|=2$,
${\it ad}=0.5$ and ${\it td}=1.0, \dots, 3.0$.
The x-axis is the transition density of the original automata while
the y-axis is the average density of the new automata.
The three curves show the average transition density of the original automata
(this is just the identity function),
the density after the transformation into the form with one accepting state,
and the average transition density of the Heavy-12 reduced automata.
Every point is the average of $1000$ automata.
}\label{fig:sparseness2}
\end{figure}

\subsubsection{Scalability of NFA reduction}

We test the scalability of reducing NFA with Heavy-12 by testing
Tabakov-Vardi random automata of increasing size $n$ but fixed ${\it td}$, 
${\it ad}$ and $\Sigma$. We ran four separate tests with ${\it td}=1.4, 1.6, 1.8$ and
$2.0$. In each test we fixed ${\it ad}=0.5$, $|\Sigma|=2$ and increased the number of states
from $n=50$ to $n=600$ in increments of $50$. For each parameter point 
we created $300$ random automata and reduced them with
Heavy-12. We analyze the average size of the reduced automata in percent of
the original size $n$, and how the average computation time increases with
$n$.

For ${\it td}=1.4$ the average size of the reduced automata stays around 
$77\%$ of the original size, regardless of $n$.
For ${\it td}=1.6$ it stays around $81\%$.
For ${\it td}=1.8$ it {\em decreases} from
$53\%$ at $n=50$ to $9\%$ at $n=600$.
For ${\it td}=2.0$ it {\em decreases} from
$23\%$ at $n=50$ to $1\%$ at $n=600$.
See Fig.~\ref{fig:NFA_scalability_size}.

\begin{figure}[htbp]
\begin{center}
\includegraphics[scale=0.7]{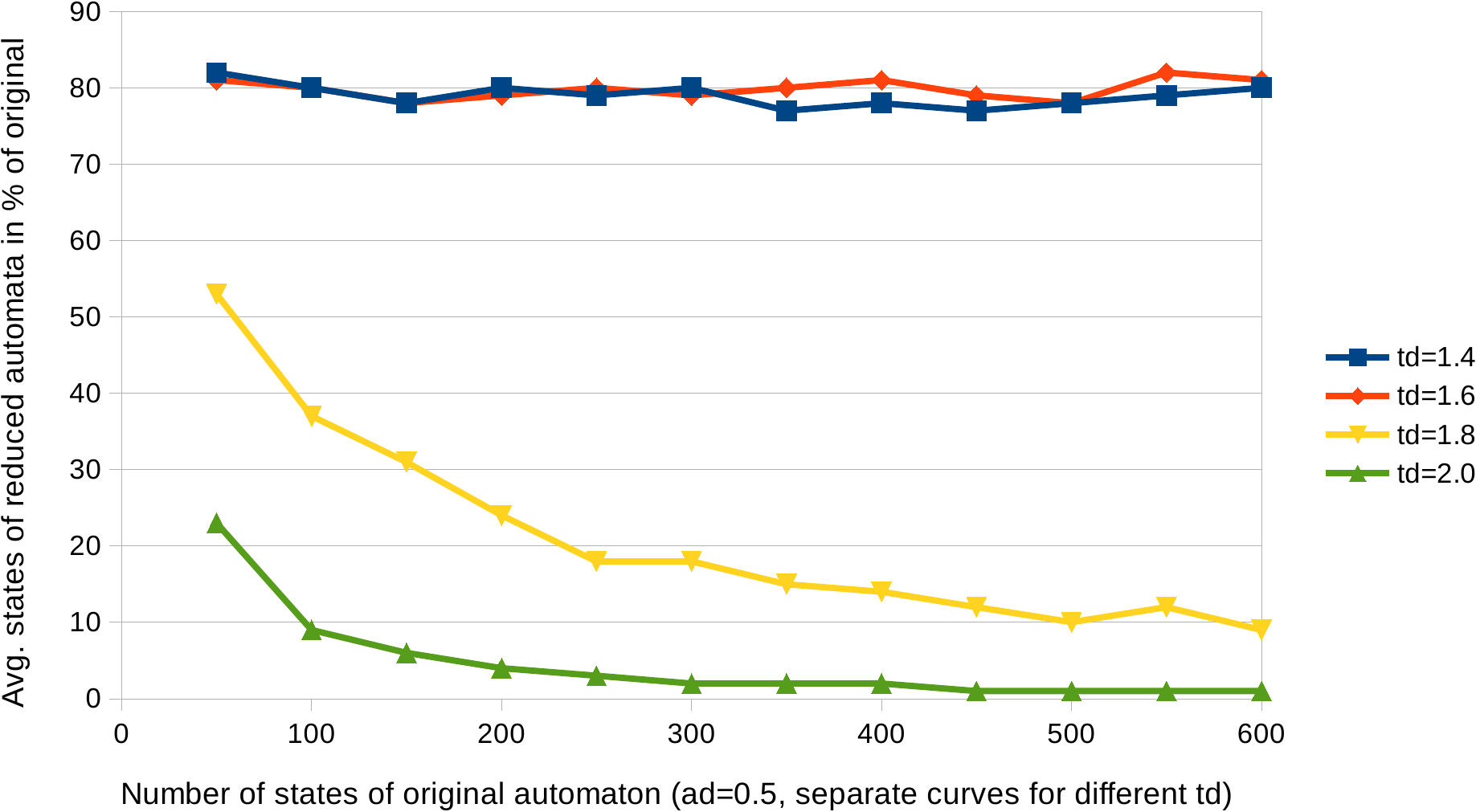}
\end{center}
\caption{Reduction of Tabakov-Vardi random NFA with ${\it ad}=0.5$,
$|\Sigma|=2$, and increasing $n=50,100,\dots,600$.
Different curves for different ${\it td}$.
We plot the average size of the Heavy-12 reduced automata, in percent of
their original size. Every data point is the average of $300$ automata.
}\label{fig:NFA_scalability_size}
\end{figure}

Note that the lookahead of 12 did {\em not change} with $n$.
Surprisingly, larger automata do not require larger lookahead for a good reduction.

Unlike in B\"uchi automata reduction, the average time to reduce NFA was
much higher than the median time (for transition densities $1.6$ and $1.8$).
For example, the average time to reduce a random NFA with $600$ states and
${\it td}=1.6$ was 199s, while the median time was 43s.
For ${\it td}=1.8$ the average and median times were 1316s and 20s, respectively.
Apparently, a few random NFA are very hard instances which increase the
average reduction time.
Therefore, we analyze both the average and the median reduction time for
NFA below.

In Fig.~\ref{fig:scalability_NFA_size_median},
we plot the median computation time (measured in ms) in $n$ and then compute
the optimal fit of the function ${\it time} = a \cdot n^b$ to the data by the least-squares
method as above. For ${\it td}=1.4, 1.6, 1.8, 2.0$ we obtain
$0.0058 \cdot n^{2.22}$, $0.011 \cdot n^{2.37}$, $0.0048 \cdot n^{2.38}$ and
$0.0048 \cdot n^{2.26}$, respectively.
So the median computation times scale slightly above quadratically.

\begin{figure}[htbp]
\begin{center}
\includegraphics[scale=0.7]{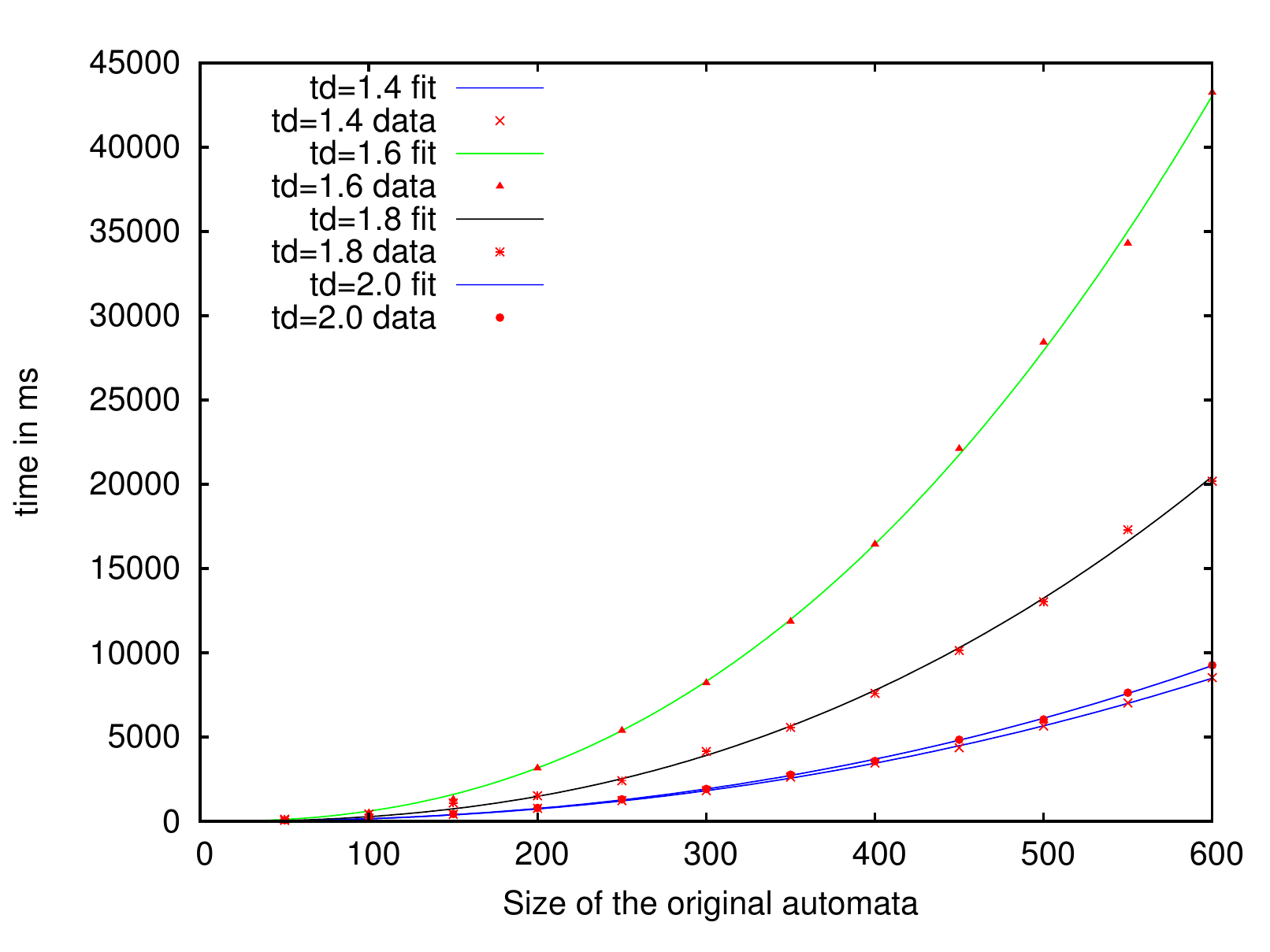}
\end{center}
\captionof{figure}{Median computation time for Heavy-12 on Tabakov-Vardi NFA
  with ${\it ad}=0.5$, $|\Sigma|=2$ and 
  ${\it td}=1.4, 1.6, 1.8, 2.0$, with a least squares fit of the function 
  $y=a \cdot x^b$.
  The x-axis shows the number of states of the original automata and the
  y-axis shows the median runtime in ms.}\label{fig:scalability_NFA_size_median}
\end{figure}

In Fig.~\ref{fig:scalability_NFA_size} 
we plot the average computation time (measured in ms) in $n$ and then compute
the optimal fit of the function ${\it time} = a \cdot n^b$ to the data by the least-squares
method, i.e., this computes the parameters $a$ and $b$ of the function that
most closely fits the experimental data. The important parameter is the
exponent $b$. For ${\it td}=1.4, 1.6, 1.8, 2.0$ we obtain
$0.0033 \cdot n^{2.33}$, $1.15 \cdot 10^{-4} \cdot n^{3.33}$, $3 \cdot 10^{-27} \cdot n^{11.7}$
and $0.008 \cdot n^{2.20}$, respectively.
Clearly, for ${\it td}=1.8$, the curve fits the experimental data extremely poorly
(with exponent $b=11.7$ and scale factor $a=3 \cdot 10^{-27}$).
Apparently, for ${\it td}=1.8$, a few very hard instances create outliers
that distort the averages (unlike the median). 
In Fig.~\ref{fig:bad_scalability_NFA_size} 
(resp.~Fig.~\ref{fig:scalability_NFA_size}) we plot the averages including
(resp.~excluding) the case of ${\it td}=1.8$.

\begin{figure}[htbp]
\begin{center}
\includegraphics[scale=0.7]{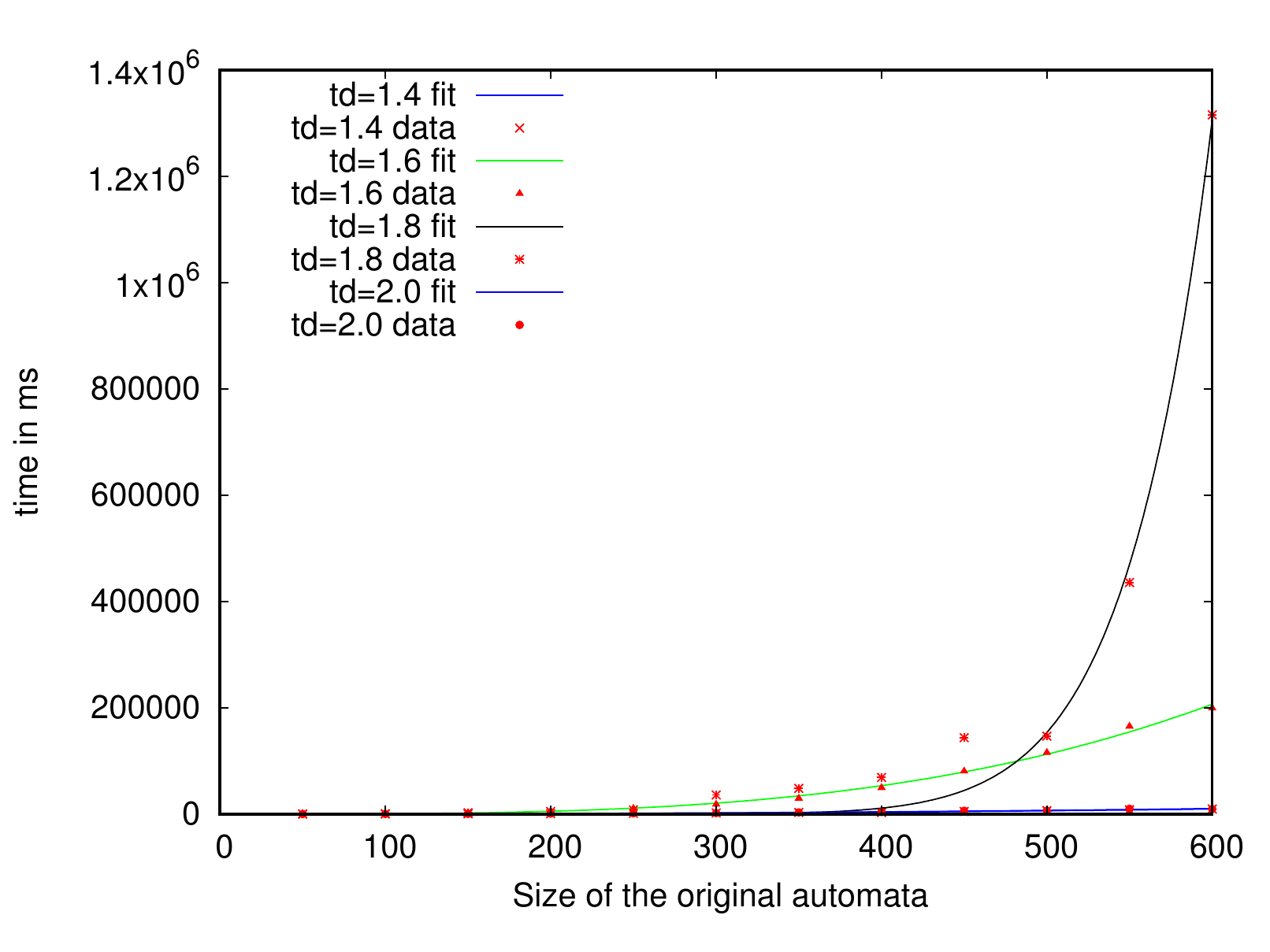}
\end{center}
\captionof{figure}{Average computation time for Heavy-12 on Tabakov-Vardi NFA
  with ${\it ad}=0.5$, $|\Sigma|=2$ and 
  ${\it td}=1.4, 1.6, 1.8, 2.0$, with a least squares fit of the function 
  $y=a \cdot x^b$.
  The x-axis shows the number of states of the original automata and the
  y-axis shows the average runtime in ms.
  Note the poor fit of the curve for ${\it td}=1.8$.
}\label{fig:bad_scalability_NFA_size}

\begin{center}
\includegraphics[scale=0.7]{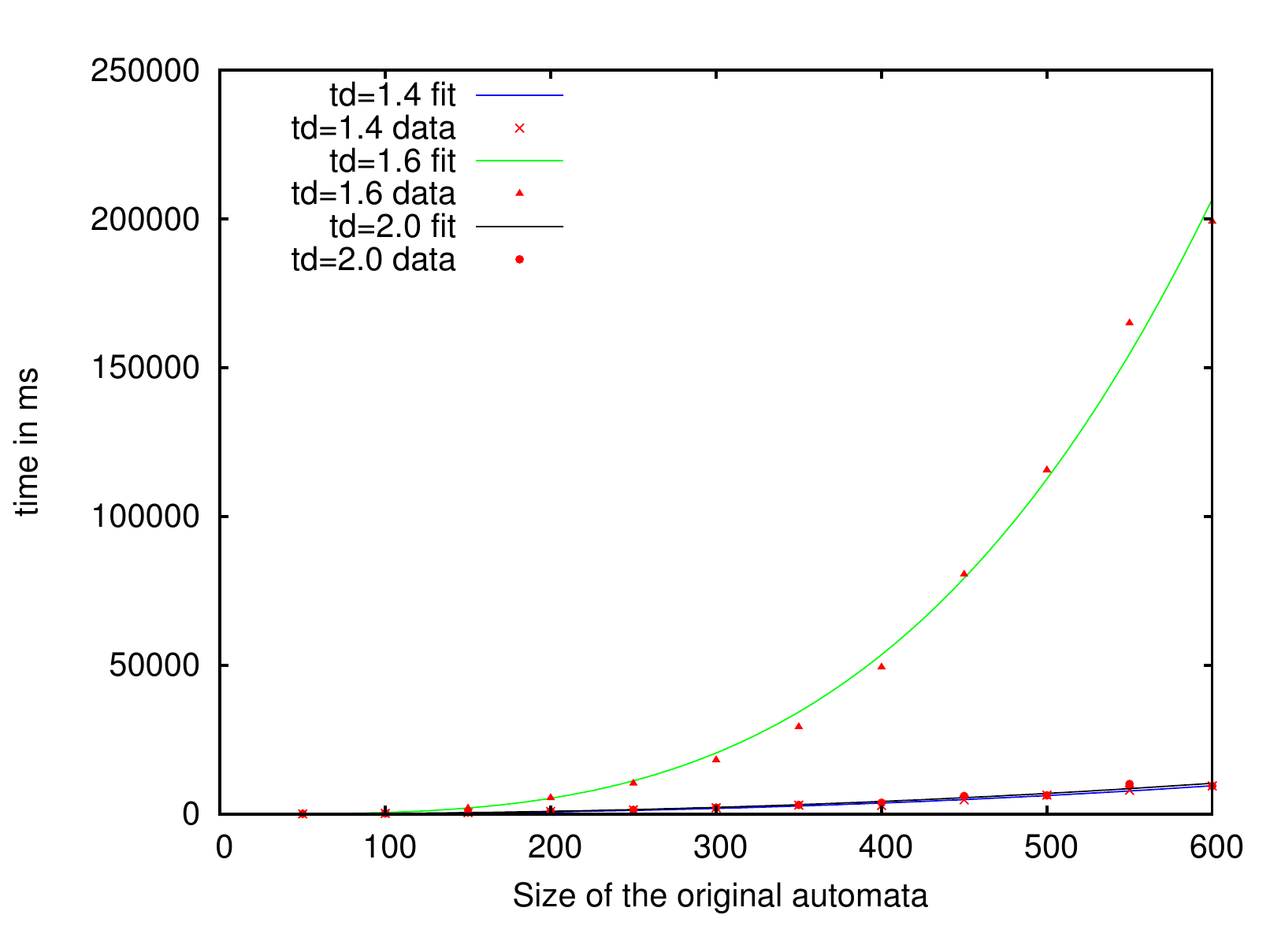}
\end{center}
\captionof{figure}{Average computation time for Heavy-12 on Tabakov-Vardi NFA
  with ${\it ad}=0.5$, $|\Sigma|=2$ and 
  ${\it td}=1.4, 1.6, 2.0$, with a least squares fit of the function 
  $y=a \cdot x^b$.
  The x-axis shows the number of states of the original automata and the
  y-axis shows the average runtime in ms.}\label{fig:scalability_NFA_size}
\end{figure}

%%% Local Variables:
%%% mode: latex
%%% TeX-master: "ROOT.tex"
%%% End:

\newpage
\section{Extensions: Adding Transitions}\label{sec:extensions}

In this section we describe a technique, called \emph{saturation},
that reduces the number of states of automata by adding more transitions.
The idea is that certain transitions may be added to an automaton 
without changing its language when other better transitions are already
present.
Conceptually, this is dual to the transition pruning techniques of Sec.~\ref{sec:pruning}.
This technique is implemented in our tool \cite{RABIT} and can have a
significant effect on some instances. However, it is not part
of the default Heavy-k algorithm (described and tested in Sections
\ref{sec:heavyandlight} and \ref{sec:experiments}), 
due to low efficiency and a tradeoff between the numbers of states
and transitions.
Note that adding transitions itself does not change the number of states.
However, the changed automaton might allow the application of further
quotienting that then reduces the number of states. Moreover, the changed
automaton might be treated with the pruning techniques from
Sec.~\ref{sec:pruning}, and this might remove some transitions other than
the recently added ones. This modification might pave the way for further
quotienting, etc., which finally results in an automaton with fewer states,
and possibly even fewer transitions, than the one produced by the default
Heavy-k method.
 
One downside is that, even if the number of states eventually decreases, not
all the added transitions can be removed again. So one might obtain an
automaton with fewer states but more transitions than the one produced
by Heavy-k. This tradeoff effect between the numbers of states and
transitions exists in practice, though it is not very strong; see the
experiments in Sec.~\ref{sec:extensions:experiments}.

\begin{definition}
	Let $\A = (\Sigma, Q, I, F, \delta)$ be an automaton,
        $\Delta = Q \times \Sigma \times Q$ the set of all possible
        transitions between states in $\A$,
        $\saturationrel \subseteq \Delta \times \Delta$ a reflexive
        binary relation on $\Delta$,
		and, for a set of transitions $\Gamma \subseteq \Delta$,
	\begin{align*}
		S^{-1}(\Gamma) = \{(p',\symb',r') \in \Delta \st
			\exists (p,\symb,r) \in \Gamma \cdot (p',\symb',r')\; \saturationrel \; (p,\symb,r)\}.
	\end{align*}
	The \emph{$\saturationrel$-saturated automaton} is defined as
	$\saturate{\A}{\saturationrel} \defeq (\Sigma, Q, I, F, \delta')$,
	where
%	\begin{align*}
		$\delta' = S^{-1}(\delta)$.
%		I' &= \{ p' \in Q \st \exists (p', \symb', q') \in S^{-1}(\delta \cap I \times \Sigma \times Q) \} \\
%		F' &= \{ q' \in Q \st \exists (p', \symb', q') \in S^{-1}(\delta \cap Q \times \Sigma \times F) \}
%	\end{align*}
	%
\end{definition}
The intuition is that more transitions can be added without changing the
language if better (i.e., $\saturationrel$-larger) transitions already exist.
Since $\saturationrel$ is reflexive, saturation only adds transitions, % initial, and final states
%(i.e., $\delta \subseteq \delta'$, $I \subseteq I'$, and $F \subseteq F'$),
and thus $\A \languageinclusion \saturate{\A}{\saturationrel}$.
When the converse inclusion also holds, we say that $\saturationrel$ is good for saturation.
\begin{definition}
	A relation $\saturationrel \subseteq \Delta \times \Delta$ is \emph{good for saturation} (GFS)
	if $\saturate{\A}{\saturationrel} \languageequivalence \A$.
%	and it is \emph{good for backward saturation} (\GFSbw) if $\bwsaturate{\A}{R} \languageequivalence \A$,
\end{definition}

The GFS property is downward closed in the space of reflexive relations,
i.e., if $\saturationrel$ is GFS and 
$\id \subseteq \saturationrel' \subseteq \saturationrel$, then
$\saturationrel'$ is also GFS.

We study GFS relations which add transitions to already existing states,
and they do so by comparing the endpoints of such transitions over the same input symbol.
(This is similar to our pruning technique from Sec.~\ref{sec:pruning}.)
Formally, given two binary state relations $\brel, \frel \subseteq Q \times Q$
for the source and target endpoints, respectively, we define
\begin{align*}
%	\label{eq:saturationrel}
	\makesaturationrel{\brel}{\frel} = \{((p,\symb,r),(p',\symb,r')) \in \Delta \times \Delta \st p \mathrel{\brel} p' \textrm{ and } r \mathrel{\frel} r' \}.
\end{align*}
$\makesaturationrel{\cdot}{\cdot}$ is monotone in both arguments.

Given an automaton $\A$ and relations $\brel,\frel$ on the states of $\A$,
we will construct a new automaton $\B = \saturate{\A}{\makesaturationrel{\brel}{\frel}}$.
When reasoning about whether $\makesaturationrel{\brel}{\frel}$ is GFS
(i.e., whether $\B \languageequivalence \A$),  
it is important to keep in mind that the relations $\brel,\frel$
are valid only w.r.t.~$\A$, but not necessarily w.r.t.~$\B$.

\subsection{Saturation of NBA}
\label{sec:extensions:theory:NBA}

\begin{table}
	
	\begin{center}
		\begin{tabular}{c|cccccc||cc}
$\brel\backslash\frel$ 	& $\id$	& $\disim$ & $\directtraceinclusion$ &
                                    $\delayedfixedwordsimulation$ &
                  $\delayedtraceinclusion$ & $\fsim$ & $\bwdisimrev$ & $\bwdirecttraceinclusionrev$ \\
			\hline
			$\id$ & $\tickOK$	& $\tickOK$
                        & $\tickOK$			& $\tickOK$
                        & $\tickNO$	&	$\tickNO$ &	$\tickOK$ &
                        $\tickOK$		\\
$\disimrev$ & $\tickOK$ & $\tickOK$ & $\tickOK$ & $\tickOK$ & $\tickNO$ &
                        $\tickNO$ & $\tickNO$ & $\tickNO$ \\
$\directtraceinclusionrev$ & $\tickOK$ & $\tickOK$ & $\tickOK$ & $\tickOK$ & $\tickNO$ &
                        $\tickNO$ & $\tickNO$ & $\tickNO$ \\
$\delayedfixedwordsimulationrev$ & $\tickOK$ & $\tickOK$ & $\tickOK$ & $\tickOK$ & $\tickNO$ &
                        $\tickNO$ & $\tickNO$ & $\tickNO$ \\
$\delayedtraceinclusionrev$ & $\tickNO$ & $\tickNO$ & $\tickNO$ & $\tickNO$ & $\tickNO$ & $\tickNO$ & $\tickNO$ & $\tickNO$ \\
$\fsimrev$ & $\tickNO$ & $\tickNO$ & $\tickNO$ & $\tickNO$ & $\tickNO$ &
                        $\tickNO$ & $\tickNO$ & $\tickNO$ \\
$\bwdisim$ & $\tickOK$ & $\tickNO$ & $\tickNO$ & $\tickNO$ & $\tickNO$ &
                        $\tickNO$ & $\tickOK$ & $\tickOK$ \\
$\bwdirecttraceinclusion$ & $\tickOK$ & $\tickNO$ & $\tickNO$ & $\tickNO$ & $\tickNO$ &
                        $\tickNO$ & $\tickOK$ & $\tickOK$ 							
		\end{tabular}
	\end{center}
	
	\caption{GFS relations $\makeprunerel{\brel}{\frel}$ for NBA.
        $\tickOK$ denotes yes and $\tickNO$ denotes no.}
	\label{fig:GFS_relations}

\end{table}

%%% Local Variables:
%%% mode: latex
%%% TeX-master: "ROOT.tex"
%%% End:

We study which semantic preorders induce GFS relations on NBA.
Our results are summarized in Table~\ref{fig:GFS_relations}.

In the transition pruning techniques of Sec.~\ref{sec:pruning},
the source states of transitions were compared w.r.t.~\emph{backward} simulation
(resp.~trace inclusion), while the target states were compared 
w.r.t.~(various types of) \emph{forward} simulation (resp.~trace inclusion).
However, for saturation this would be incorrect as the counterexample
from Fig.~\ref{fig:bw_fw:not:GFS} shows. 
In this automaton $\A$ (without the dashed transition) 
we have that $(r,a,r) \makesaturationrel \bwdisim
  \disim (p,a,q)$, but adding the dashed transition $(r,a,r)$
changes the language, since $a^\omega$ is now accepted,
i.e., $a^\omega \in \lang{\saturate{\A}{\makesaturationrel \bwdisim
    \disim}} \setminus \lang{\A}$.
Thus, $\makesaturationrel \bwdisim \disim$ is not GFS.
A similar example in Fig.~\ref{fig:fw_bw:not:GFS}
shows that $\makesaturationrel \disimrev \bwdisimrev$ is not GFS either:
Here %$r \disim p$ and $s \bwdisim q$
$(p, a, q) \makesaturationrel \disimrev \bwdisimrev (r, a, s)$,
and thus the dashed transition $p \goesto a q$ is added,
which causes the new word $a^\omega$ to be accepted---%
while this was not the case in the original automaton.
While $\makesaturationrel \bwdisim \disim$ and $\makesaturationrel \disimrev \bwdisimrev$ are not GFS,
one does obtain GFS relations by replacing either $\bwdisim$ or $\disim$ 
by the identity, which immediately follows from the more general results below.

\begin{figure}
	
	\begin{tikzpicture}[on grid, node distance= .6cm and 1.6cm]
		\tikzstyle{vertex} = [smallstate]

		\path node [vertex, initial, accepting] (p) {$p$};
		%\path node (hidden) (x) [below = of p] {}; % just for formatting reasons
		\path node [vertex, accepting] (q) [right = of p] {$q$};
                \path node [vertex, initial, accepting] (r) [right = of q] {$r$};

		\path[->]
			(p) edge node [above] {$a$} (q)
			(q) edge [loop above] node {$b,c$} ()
			(p) edge [loop above] node {$b,c$} ()
                        (r) edge [loop above] node {$c$} ()
                        (r) edge [dashed, loop below] node {$a$} ();
	\end{tikzpicture}
	
	\caption{$\makesaturationrel \bwdisim \disim$ is not GFS.}
	\label{fig:bw_fw:not:GFS}

\end{figure}

\begin{figure}
	
	\begin{tikzpicture}[on grid, node distance= .6cm and 1.6cm]
		\tikzstyle{vertex} = [smallstate]

		\path node [vertex, initial] (p) {$p$};
		\path node [vertex, accepting, right = of p] (q) {$q$};
		\path node [vertex, right = of q] (r) {$r$};
		\path node [vertex, right = of r] (s) {$s$};
		
		\path[->]
			(p) edge [loop above] node {$a$} ()
            (p) edge [dashed] node [above] {$a$} (q)
			(q) edge [loop above] node {$a$} ()
			(r) edge node [above] {$a$} (s)
			;
	\end{tikzpicture}
	
	\caption{$\makesaturationrel \disimrev \bwdisimrev$ is not GFS.}
	\label{fig:fw_bw:not:GFS}

\end{figure}

%%% Local Variables:
%%% mode: latex
%%% TeX-master: "ROOT.tex"
%%% End:

Comparing both source and target states w.r.t.~forward relations
can yield GFS relations, as the following theorem shows.

\begin{theorem}\label{thm:GFS:delayedfixedwordsimulation}
	The relation $\makesaturationrel \delayedfixedwordsimulationrev
        \delayedfixedwordsimulation$ using fixed-word delayed simulation is
        GFS on NBA.
\end{theorem}
\begin{proof}
	Let $\B = \saturate{\A}{\makesaturationrel \delayedfixedwordsimulationrev \delayedfixedwordsimulation}$.
        (Note that the relations $\delayedfixedwordsimulationrev, \delayedfixedwordsimulation$
        are valid w.r.t.~$\A$, but
        not necessarily w.r.t.~$\B$.)
	We only need to prove the non-trivial inclusion $\B \languageinclusion \A$.
	Let $w = \symb_0 \symb_1 \cdots \in \lang \B$.
	Then there exists an initial fair trace
	$\pi = p_0 \goesto {\symb_0} p_1 \goesto {\symb_1} \cdots$ in $\B$,
	which is not necessarily a trace in $\A$, 
        since it might use new transitions introduced by the saturation procedure.
	However, for every transition $p_i \goesto {\symb_{i}} p_{i+1}$ in $\B$
	there exists a transition $q \goesto {\symb_{i}} q'$ for some states
        $q,q'$ in $\A$
	s.t.~$q \delayedfixedwordsimulation p_i$ and $p_{i+1} \delayedfixedwordsimulation q'$.
	In particular,
	\begin{align*}
		q \delayedwordsimulation {\suffix w i} p_i \qquad \textrm{ and } \qquad
		p_{i+1} \delayedwordsimulation {\suffix w {i+1}} q'.
	\end{align*}
	We inductively construct an initial fair trace $\rho = r_0 \goesto {\symb_0} r_1 \goesto {\symb_1} \cdots$ in $\A$
	s.t.~$p_i \delayedwordsimulation {\suffix w i} r_i$ for every $i \geq 0$. 
	For the base case $i = 0$, we just take $r_0 = p_0$ (thus $\rho$ is initial).
	For the inductive step,
	assume $r_0 \goesto {\symb_0} \cdots \goesto {\symb_{i-1}} r_i$ has already been constructed.
	Since $p_i \goesto {\symb_{i}} p_{i+1}$ in $\B$, 
        there exists a transition $q \goesto {\symb_{i}} q'$ for some states
        $q,q'$ in $\A$
	s.t.~$q \delayedwordsimulation {\suffix w i} p_i$ and
        $p_{i+1} \delayedwordsimulation {\suffix w {i+1}} q'$.
	By inductive assumption, $p_i \delayedwordsimulation {\suffix w i} r_i$,
	and thus $q \delayedwordsimulation {\suffix w i} r_i$ by transitivity.
        Since $q \goesto {\symb_{i}} q'$ in $\A$,
	there exists a transition $r_i \goesto {\symb_{i}} r$ in $\A$
	s.t.~$p_{i+1} \delayedwordsimulation {\suffix w {i+1}} q' \delayedwordsimulation {\suffix w {i+1}} r$.
	Let $r_{i+1} = r$, which again establishes the inductive invariant $p_{i+1} \delayedwordsimulation {\suffix w {i+1}} r_{i+1}$.
	Clearly, $\rho$ is infinite.
	Moreover, since $\pi$ is fair and since by the delayed winning condition
	each occurrence of an accepting state in $\pi$ is eventually followed by an accepting state in $\rho$,
	we have that $\rho$ is fair as well.
	This shows $w \in \lang \A$.
\end{proof}

As a corollary of Theorem~\ref{thm:GFS:delayedfixedwordsimulation},
using any relation included in fixed-word delayed simulation
results in a GFS relation (cf.~the taxonomy of GFQ relations of Fig.~\ref{fig:GFQ_relations}),
such as direct and delayed simulations, together with their multipebble and lookahead variants.
In short, every GFQ relation induces a GFS relation.
This is not an accident, as shown in the following result.

\begin{lemma}
	Let $\equiv \subseteq Q \times Q$ be an equivalence between states.
	Then, $\makesaturationrel \equiv \equiv$ is GFS iff $\equiv$ is GFQ.
\end{lemma}
\begin{proof}
	%Let $\approx$ be the equivalence induced by $\sqsubseteq$.
	%If $\makesaturationrel \id \sqsubseteq$ is GFS,
	%so it is $\makesaturationrel \id \approx$.
	%
	Consider the saturated automaton
	$\B = \saturate{\A}{\makesaturationrel \equiv \equiv}$
	and the quotient automaton $\C = \A/\!\equiv$.
	We show that $\B \approx \C$.
	%
%	Assume $\makesaturationrel \approx \approx$ is GFS.
	%
%	Then, $\A \languageequivalence \B$,
%	and by the definition of quotienting $\A, \B \languageinclusion \C$.
	%
	Take an initial fair run $\pi = [p_0] \goesto {\symb_0} [p_1] \goesto {\symb_1} \cdots$ in $\C$,
	where $[p_i]$ denotes the equivalence class of state $p_i$ w.r.t.~$\equiv$.
	Without loss of generality, let $p_0$ be initial,
	and let $p_i$ be accepting if $[p_i]$ contains an accepting state.
	We build an initial fair run $\pi' = p_0 \goesto {\symb_0} p_1 \goesto {\symb_1} \cdots$ in $\B$.
	By the definition of quotienting, each transition $[p_i] \goesto {\symb_i} [p_{i+1}]$ in $\C$
	originates from a concrete transition $\hat p_i \goesto {\symb_i} \hat p_{i+1}$ in $\A$
	for some $\hat p_i \in [p_i]$ and $\hat p_{i+1} \in [p_{i+1}]$.
	Since $\hat p_i \equiv p_i$ and $\hat p_{i+1} \equiv p_{i+1}$,
	by the definition of saturation
	there exists a transition $p_i \goesto {\symb_i} p_{i+1}$ in $\B$.
	This shows $\C \languageinclusion \B$.
	
	For the other inclusion, consider an initial fair run $\pi = p_0 \goesto {\symb_0} p_1 \goesto {\symb_1} \cdots$ in $\B$.
	By the definition of saturation, each transition $p_i \goesto {\symb_i} p_{i+1}$ in $\B$
	originates from a concrete transition $\hat p_i \goesto {\symb_i} \hat p_{i+1}$ in $\A$
	for some $\hat p_i \equiv p_i$ and $\hat p_{i+1} \equiv p_{i+1}$.
	Thus, by the definition of quotienting,
	$\pi = [p_0] \goesto {\symb_0} [p_1] \goesto {\symb_1} \cdots$ is an initial fair run in $\C$,
	which shows $\B \languageinclusion \C$ and thus concludes the proof.
\end{proof}
On the other hand,
$\makesaturationrel \id \delayedtraceinclusion$ and $\makesaturationrel \delayedtraceinclusionrev \id$,
using the coarser delayed trace inclusion, are not GFS.
(The same phenomenon happens w.r.t.~GFQ relations; cf.~Fig.~\ref{fig:GFQ_relations}.)
In order to see that $\makesaturationrel \id \delayedtraceinclusion$ is not GFS,
consider the automaton $\A$ from Fig.~\ref{fig:decont:not:GFS} (without the dashed transition).
We have $p \delayedtraceinclusion q$:
If Spoiler plays $pq^\omega$,
then Duplicator replies with $qrs^\omega$,
if Spoiler plays $pq^nrs^\omega$ for $n \geq 1$,
then Duplicator replies with $q^{n+1}rs^\omega$,
and in both cases the delayed acceptance condition is satisfied.
Since there is a transition $q \goesto a q$,
the saturated automaton $\saturate{\A}{\makesaturationrel \id \delayedtraceinclusion}$
has the additional dashed transition $q \goesto a p$,
and now it accepts the new word $a^\omega$ not previously accepted.
Similarly, in order to see that $\makesaturationrel \delayedtraceinclusionrev \id$ is not GFS,
consider the automaton $\A$ from Fig.~\ref{fig:decont:not:GFS:two} (without the dashed transition).
We have $q \delayedtraceinclusion r$,
and the transition $q \goesto a q$
induces the additional dashed transition $r \goesto a q$
in the saturated automaton $\saturate{\A}{\makesaturationrel \delayedtraceinclusionrev \id}$,
and again the new word $a^\omega$ is suddenly accepted.

\begin{figure}
	
	\begin{tikzpicture}[on grid, node distance= .6cm and 1.6cm]
		\tikzstyle{vertex} = [smallstate]

		\path node [vertex, initial, accepting] (p) {$p$};
		\path node [vertex] (q) [right = of p] {$q$};
		\path node [vertex, accepting] (r) [right = of q] {$r$};
		\path node [vertex] (s) [right = of r] {$s$};

		\path[->]
			(p) edge node [above] {$a$} (q)
			(q) edge [bend left = 60, dashed] node [below] {$a$} (p)
			(q) edge [loop above] node {$a$} ()
			(q) edge node [above] {$a$} (r)
			(r) edge node [above] {$a$} (s)
			(s) edge [loop above] node {$a$} ();

	\end{tikzpicture}
	
	\caption{$\makesaturationrel \id \delayedtraceinclusion$ using delayed trace inclusion is not GFS.}
	\label{fig:decont:not:GFS}

\end{figure}

\begin{figure}
	
	\begin{tikzpicture}[on grid, node distance= .6cm and 1.6cm]
		\tikzstyle{vertex} = [smallstate]

		\path node [vertex, initial] (q) [right = of p] {$q$};
		\path node [vertex, accepting] (r) [right = of q] {$r$};
		\path node [vertex] (s) [right = of r] {$s$};
		\path node [vertex, accepting] (t) [right = of s] {$t$};
		\path node [vertex] (u) [right = of t] {$u$};

		\path[->]
			(q) edge [loop above] node {$a$} ()
			(q) edge node [above] {$a$} (r)
			(r) edge node [above] {$a$} (s)
			(r) edge [bend left = 60, dashed] node [below] {$a$} (q)
			(s) edge [loop above] node {$a$} ()
			(s) edge node [above] {$a$} (t)
			(t) edge node [above] {$a$} (u)
			(u) edge [loop above] node {$a$} ();

	\end{tikzpicture}
	
	\caption{$\makesaturationrel \delayedtraceinclusionrev \id$	using delayed trace inclusion is not GFS.}
	\label{fig:decont:not:GFS:two}

\end{figure}

%%% Local Variables:
%%% mode: latex
%%% TeX-master: "ROOT.tex"
%%% End:

Also $\makesaturationrel \id \fsim$ and $\makesaturationrel \fsimrev \id$ using fair simulation are not GFS.
For a simple counterexample, consider the automaton $\A$ in Fig.~\ref{fig:fsim:not:GFS} (without the dashed transition).
We have $p \fsim q \fsim r$ (and in fact, the three states are fair simulation equivalent).
Since $p \goesto a q$, 
the saturated automaton $\saturate{\A}{\makesaturationrel \fsimrev \id}$
has the additional dashed transition $q \goesto a q$,
and since $q \goesto a r$, 
the saturated automaton $\saturate{\A}{\makesaturationrel \id \fsim}$
has the additional dashed transition $q \goesto a q$.
In both cases, the saturated automaton accept the new word $a^\omega$ not previously accepted.

\begin{figure}
	
	\begin{tikzpicture}[on grid, node distance= .6cm and 1.6cm]
		\tikzstyle{vertex} = [smallstate]

		\path node [vertex, initial] (p) {$p$};
		\path node [vertex, accepting, right = of p] (q) {$q$};
		\path node [vertex] (r) [right = of q] {$r$};

		\path[->]
			(p) edge [loop above] node {$a$} ()
			(p) edge node [above] {$a$} (q)
			(q) edge node [above] {$a$} (r)
			(q) edge [dashed, loop above] node {$a$} ()
			(r) edge [loop above] node {$a$} ()
			;

	\end{tikzpicture}
	
	\caption{$\makesaturationrel \id \fsim$ and $\makesaturationrel \fsimrev \id$ using fair simulation are not GFS.}
	\label{fig:fsim:not:GFS}

\end{figure}
%%% Local Variables:
%%% mode: latex
%%% TeX-master: "ROOT.tex"
%%% End:

These counterexamples do not apply in the special case 
where the newly added transitions are transient in the saturated automaton.

\begin{theorem}\label{thm:sat_transient}
	The relation $\makesaturationrel \fairtraceinclusionrev
        \fairtraceinclusion$ is GFS on NBA,
	provided that the newly added transitions are transient in the saturated automaton.
\end{theorem}
\begin{proof}
	Let $\B = \saturate{\A}{\makesaturationrel \fairtraceinclusionrev \fairtraceinclusion}$,
	and we assume that the new transitions in $\B$ which are not in $\A$ are transient in $\B$.
	Thus, for a word $w = \symb_0 \symb_1 \cdots$,
	an initial and fair trace $\pi = p_0 \goesto {\symb_0} p_1 \goesto {\symb_1} \cdots$ in $\B$
	ultimately does not contain any transition which is not already in $\A$,
	i.e., there exists a $k$ s.t.~$\suffix \pi k$ is a fair trace in $\A$.
	For every $i < k$ and for every transition $p_i \goesto {\symb_{i}} p_{i+1}$ in $\B$,
	there exists a transition $q \goesto {\symb_{i}} q'$ in $\A$ 
	s.t.~$q \fairtraceinclusion p_i$ and $p_{i+1} \fairtraceinclusion q'$.
	We proceed backwards and we build a sequence $\pi_k, \pi_{k-1}, \dots, \pi_0$
	s.t.~$\pi_i$ for $i\leq k$ is a fair trace in $\A$ starting in $p_i$ and reading the suffix $\suffix w i$.
	Then, $\pi_0$ is an initial fair trace witnessing $w \in \lang \A$.
	Assume $\pi_{i+1}$ starting in $p_{i+1}$ is already constructed.
	Since $p_{i+1} \fairtraceinclusion q'$,
	there exists a fair trace $\pi'$ from $q'$ in $\A$ reading $\suffix w {i+1}$,
	and since $q \goesto {\symb_{i}} q'$,
	there exists a fair trace $\pi'$ from $q$ in $\A$ reading $\suffix w {i}$.
	Since $q \fairtraceinclusion p_i$, we deduce the existence of the fair trace $\pi_i$ from $p_i$ in $\A$ reading $\suffix w {i}$.
\end{proof}

Note that the criterion in Theorem~\ref{thm:sat_transient} 
requires that the added transitions are transient in the new saturated automaton
rather than in the original one.
This is different from the transition pruning criterion 
in Theorem~\ref{thm:prune_transient} that requires certain transitions to be
transient in the original automaton, and thus also in the new pruned automaton.
This makes it difficult to apply Theorem~\ref{thm:sat_transient} in practice, 
since adding some transition might cause another added transition to become
non-transient and vice-versa, i.e., there is not always a unique maximal
solution.

Dually to Theorem~\ref{thm:GFS:delayedfixedwordsimulation}, we obtain GFS
relations if both source and target states are compared w.r.t.~backward
relations (but note that the directions of the relations are inverted here).

\begin{theorem}\label{thm:GFS:bwdirecttraceinclusion}
	The relation $\makesaturationrel \bwdirecttraceinclusion
        \bwdirecttraceinclusionrev$ using backward direct trace inclusion is
        GFS on NBA.
\end{theorem}
\begin{proof}
	Let $\B = \saturate{\A}{\makesaturationrel \bwdirecttraceinclusion \bwdirecttraceinclusionrev}$.
	We prove the non-trivial inclusion	$\B \languageinclusion \A$.
	Let $w = \symb_0 \symb_1 \cdots \in \lang\B$.
	There exists an initial fair trace $\pi = p_0 \goesto {\symb_0} p_1 \goesto {\symb_1} \cdots$ in $\B$,
	which is not necessarily a trace in $\A$, since it might use new transitions introduced by the saturation procedure.
	However, for every transition $p_i \goesto {\symb_{i}} p_{i+1}$ in $\B$
	there exists a transition $q \goesto {\symb_{i}} q'$ in $\A$ 
	s.t.~$p_i \bwdirecttraceinclusion q$ and $q' \bwdirecttraceinclusion p_{i+1}$.
	By the definition of $\bwdirecttraceinclusion$,
	we construct inductively a sequence $\pi_0, \pi_1, \dots$ of finite traces in $\A$
	s.t.~each $\pi_i$ is initial, ends in $p_i$, and contains at least
        as many accepting states as does $\prefix{\pi} i$.
        The base case of $i=0$ is trivial.
	For the induction step we assume 
        that such a $\pi_i$ is already constructed, and consider the transition $q \goesto {\symb_{i}} q'$.
	Since $p_i \bwdirecttraceinclusion q$, there exists an initial trace
        in $\A$ ending in $q$.
	We extend this trace by the transition 
        $q \goesto {\symb_{i}} q'$ in $\A$ above 
	and use $q' \bwdirecttraceinclusion p_{i+1}$ to extract the
        required initial trace $\pi_{i+1}$ 
        in $\A$ ending in $p_{i+1}$.
	A routine application of K\"onig's Lemma shows the existence of an initial and fair trace $\pi_\infty$ in $\A$,
	thus showing $w \in \lang \A$.
\end{proof}

%The GFS status of the examined saturation relations is summarized in Fig.~\ref{fig:GFS_relations}.
%\input{GFS_table}

\subsection{Saturation of NFA}
\label{sec:extensions:theory:NFA}

The full picture of GFS preorders is much simpler for finite words than for infinite ones.
Criteria based on delayed and fair simulation (resp.~trace-inclusion) cannot be used for saturation of NFA, of course.
However, one can use forward and backward trace inclusion over finite words,
yielding Theorems~\ref{thm:GFS:NFA:forward} and \ref{thm:GFS:NFA:backward} below.
Their proofs are straightforward adaptions of the proofs of 
Theorems~\ref{thm:GFS:delayedfixedwordsimulation} and \ref{thm:GFS:bwdirecttraceinclusion}, respectively,
with the difference that over finite words one can use induction on the length of the accepted word,
and thus avoid K\"onig's Lemma.
Finally, there is no analogue of Theorem~\ref{thm:sat_transient} about adding transient transitions,
since this is only useful when states are compared w.r.t.~fair trace inclusion,
a notion that does not apply to NFA.

\begin{theorem}\label{thm:GFS:NFA:forward}
	The relation $\makesaturationrel \fininclrev \finincl$ using forward trace inclusion is GFS on NFA.
\end{theorem}
\begin{proof}
	Let $\B = \saturate{\A}{\makesaturationrel \fininclrev \finincl}$.
	We prove the non-trivial inclusion	$\B \languageinclusion \A$.
        We show, by induction on $n$, that for every word $w$ of length $|w|=n$ and every finite
	final trace $\pi_0$ on $w$ in $\B$ that starts from some state $p_0$, there exists a 
	corresponding finite final trace 
        $\pi_1$ on $w$ in $\A$ of length $n$ from $p_0$.
        The base case of $n=0$ is trivial.
        For the induction step, let $w=\symb_0 w'$ with $|w'|=n-1$ and
        let $\pi_0 = p_0 \goesto {\symb_0} p_1\ \pi_0'$ be the trace in $\B$.
        There exists a transition $q_0 \goesto {\symb_0} q_1$ in $\A$ 
        s.t.~$q_0 \finincl p_0$ and $p_1 \finincl q_1$.
        By the induction hypothesis, we know that there exists a final trace $\pi_1'$
        on $w'$ of length $n-1$ from $p_1$ in $\A$.
        Since $p_1 \finincl q_1$, there also exists a final trace $\pi_1''$ on
	$w'$ of length $n-1$ from $q_1$ in $\A$.
        Thus we have a final trace $\pi_1''' = q_0 \goesto {\symb_0} q_1\ \pi_1''$ 
        on $w$ of length $n$ from $q_0$ in $\A$.
        Since $q_0 \finincl p_0$, there also exists a final trace
        $\pi_1$ on $w$ of length $n$ from $p_0$ in $\A$.
\end{proof}

\begin{theorem}\label{thm:GFS:NFA:backward}
	The relation $\makesaturationrel \bwfinincl \bwfininclrev$ using backward trace inclusion is GFS on NFA.
\end{theorem}
\begin{proof}
	Let $\B = \saturate{\A}{\makesaturationrel \bwfinincl \bwfininclrev}$.
	We prove the non-trivial inclusion	$\B \languageinclusion \A$.
        We show, by induction on $n$, that for every word $w$ with $|w|=n$ and every 
	finite initial trace $\pi_0$ on $w$ in $\B$ that ends in some state $p_n$, there exists a 
	corresponding finite initial trace 
        $\pi_1$ on $w$ in $\A$ that ends in $p_n$.
        The base case of $n=0$ is trivial.
        For the induction step, let $w=w'\symb_n$ with $|w'|=n-1$ and
        let $\pi_0 = \pi_0'\ p_{n-1} \goesto {\symb_n} p_n$ be the trace in $\B$.
        There exists a transition $q_{n-1} \goesto {\symb_n} q_n$ in $\A$ 
        s.t.~$p_{n-1} \bwfinincl q_{n-1}$ and $q_n \bwfinincl p_n$.
        By the induction hypothesis (applied to word $w'$, trace $\pi_0'$ and state
	$p_{n-1}$), we know that there exists an initial trace $\pi_1'$
        on $w'$ in $\A$ that ends in $p_{n-1}$.
        Since $p_{n-1} \bwfinincl q_{n-1}$, there also exists an initial trace $\pi_1''$ on
	$w'$ in $\A$ that ends in $q_{n-1}$.
        Thus we have an initial trace $\pi_1''' = \pi_1''\ q_{n-1} \goesto {\symb_n} q_n$ 
        on $w$ in $\A$ that ends in $q_n$.
        Since $q_n \bwfinincl p_n$, there also exists an initial trace
        $\pi_1$ on $w$ in $\A$ that ends in $p_n$.
\end{proof}

Fig.~\ref{fig_saturation_example} shows a worked example where a previously
irreducible 6-state NFA is transformed into an equivalent 5-state NFA by
applying saturation and transition pruning. (A corresponding example for
B\"uchi automata can be obtained by adding a self-loop at state $u$.)

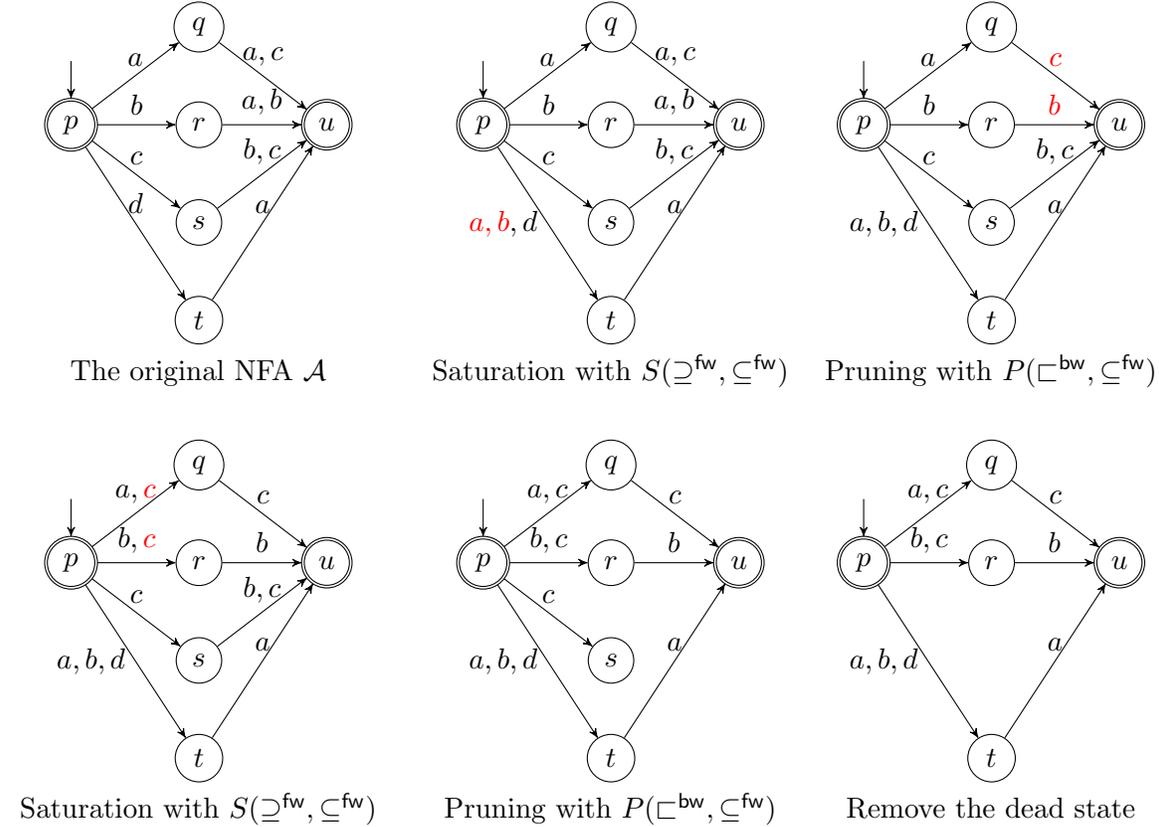
\begin{figure}[htbp]
	\begin{tabular}{ccc}
		\begin{tikzpicture}[on grid, node distance= 1.3cm and 1.7cm]
			\tikzstyle{vertex} = [smallstate]

			\path node [vertex, initial above, accepting] (p) {$p$};
			\path node [vertex] (q) [above right = of p] {$q$};
                        \path node [vertex] (r) [right = of p] {$r$};
                        \path node [vertex] (s) [below right = of p] {$s$};
                        \path node [vertex] (t) [below = of s] {$t$};
			\path node [vertex, accepting] (u) [right = of r] {$u$};

			\path[->]
				(p) edge node [above] {$a$} (q)
                                (p) edge node [above] {$b$} (r)
                                (p) edge node [above] {$c$} (s)
                                (p) edge node [above] {$d$} (t)
				(q) edge node [above] {$a,c$} (u)
				(r) edge node [above] {$a,b$} (u)
				(s) edge node [above] {$b,c$} (u)
                                (t) edge node [above] {$a$} (u);
		\end{tikzpicture}
		&
		\quad
\begin{tikzpicture}[on grid, node distance= 1.3cm and 1.7cm]
			\tikzstyle{vertex} = [smallstate]

			\path node [vertex, initial above, accepting] (p) {$p$};
			\path node [vertex] (q) [above right = of p] {$q$};
                        \path node [vertex] (r) [right = of p] {$r$};
                        \path node [vertex] (s) [below right = of p] {$s$};
                        \path node [vertex] (t) [below = of s] {$t$};
			\path node [vertex, accepting] (u) [right = of r] {$u$};

			\path[->]
				(p) edge node [above] {$a$} (q)
                                (p) edge node [above] {$b$} (r)
                                (p) edge node [above] {$c$} (s)
                                (p) edge node [left] {${\color{red}a,b},d$} (t)
				(q) edge node [above] {$a,c$} (u)
				(r) edge node [above] {$a,b$} (u)
				(s) edge node [above] {$b,c$} (u)
                                (t) edge node [above] {$a$} (u);
		\end{tikzpicture}
		\quad
		&
\begin{tikzpicture}[on grid, node distance= 1.3cm and 1.7cm]
			\tikzstyle{vertex} = [smallstate]

			\path node [vertex, initial above, accepting] (p) {$p$};
			\path node [vertex] (q) [above right = of p] {$q$};
                        \path node [vertex] (r) [right = of p] {$r$};
                        \path node [vertex] (s) [below right = of p] {$s$};
                        \path node [vertex] (t) [below = of s] {$t$};
			\path node [vertex, accepting] (u) [right = of r] {$u$};

			\path[->]
				(p) edge node [above] {$a$} (q)
                                (p) edge node [above] {$b$} (r)
                                (p) edge node [above] {$c$} (s)
                                (p) edge node [left] {$a,b,d$} (t)
				(q) edge node [above] {${\color{red}c}$} (u)
				(r) edge node [above] {${\color{red}b}$} (u)
				(s) edge node [above] {$b,c$} (u)
                                (t) edge node [above] {$a$} (u);
		\end{tikzpicture}
		\\
		The original NFA $\A$
		&
		\quad
                Saturation with $\makesaturationrel \fininclrev \finincl$
		\quad
		&
		Pruning with $\makeprunerel{\strictbwsim}{\finincl}$
                \\
                \vspace{2mm}
\\
\begin{tikzpicture}[on grid, node distance= 1.3cm and 1.7cm]
			\tikzstyle{vertex} = [smallstate]

			\path node [vertex, initial above, accepting] (p) {$p$};
			\path node [vertex] (q) [above right = of p] {$q$};
                        \path node [vertex] (r) [right = of p] {$r$};
                        \path node [vertex] (s) [below right = of p] {$s$};
                        \path node [vertex] (t) [below = of s] {$t$};
			\path node [vertex, accepting] (u) [right = of r] {$u$};

			\path[->]
				(p) edge node [above] {$a,{\color{red}c}$} (q)
                                (p) edge node [above] {$b,{\color{red}c}$} (r)
                                (p) edge node [above] {$c$} (s)
                                (p) edge node [left] {$a,b,d$} (t)
				(q) edge node [above] {$c$} (u)
				(r) edge node [above] {$b$} (u)
				(s) edge node [above] {$b,c$} (u)
                                (t) edge node [above] {$a$} (u);
		\end{tikzpicture}
		&
		\quad
\begin{tikzpicture}[on grid, node distance= 1.3cm and 1.7cm]
			\tikzstyle{vertex} = [smallstate]

			\path node [vertex, initial above, accepting] (p) {$p$};
			\path node [vertex] (q) [above right = of p] {$q$};
                        \path node [vertex] (r) [right = of p] {$r$};
                        \path node [vertex] (s) [below right = of p] {$s$};
                        \path node [vertex] (t) [below = of s] {$t$};
			\path node [vertex, accepting] (u) [right = of r] {$u$};

			\path[->]
				(p) edge node [above] {$a,c$} (q)
                                (p) edge node [above] {$b,c$} (r)
                                (p) edge node [above] {$c$} (s)
                                (p) edge node [left] {$a,b,d$} (t)
				(q) edge node [above] {$c$} (u)
				(r) edge node [above] {$b$} (u)
                                (t) edge node [above] {$a$} (u);
		\end{tikzpicture}
		\quad
		&
\begin{tikzpicture}[on grid, node distance= 1.3cm and 1.7cm]
			\tikzstyle{vertex} = [smallstate]

			\path node [vertex, initial above, accepting] (p) {$p$};
			\path node [vertex] (q) [above right = of p] {$q$};
                        \path node [vertex] (r) [right = of p] {$r$};
                        \path node [vertex] (t) [below = of s] {$t$};
			\path node [vertex, accepting] (u) [right = of r] {$u$};

			\path[->]
				(p) edge node [above] {$a,c$} (q)
                                (p) edge node [above] {$b,c$} (r)
                                (p) edge node [left] {$a,b,d$} (t)
				(q) edge node [above] {$c$} (u)
				(r) edge node [above] {$b$} (u)
                                (t) edge node [above] {$a$} (u);
		\end{tikzpicture}
		\\
		Saturation with $\makesaturationrel \fininclrev \finincl$
		&
		\quad
                Pruning with $\makeprunerel{\strictbwsim}{\finincl}$
		\quad
		&
		Remove the dead state
	\end{tabular}
	\caption{A worked example with application of saturation and
          pruning. The initial NFA $\A$ cannot be reduced any more by just
          quotienting and pruning. However, a repeated application of
          saturation and pruning (e.g., invoke the Reduce tool \cite{RABIT}
          with option \texttt{-sat2}) yields a smaller automaton with fewer
          states (5 instead of 6) \emph{and} fewer transitions (10 instead of 11). 
        }
	\label{fig_saturation_example}

\end{figure}

%%% Local Variables:
%%% mode: latex
%%% TeX-master: "ROOT.tex"
%%% End:

\subsection{Experimental Evaluation}\label{sec:extensions:experiments}

We implemented an automaton reduction method called \emph{Heavy-k-jump-sat}
that extends the method Heavy-k-jump of Sec.~\ref{sec:heavyandlight} by an
extra outer loop that adds as many extra transitions as possible, based on the
criteria described in Sec.~\ref{sec:extensions:theory:NBA}.
We call this
\emph{transition saturation}, thus the suffix -sat in the name of the method.
As usual, we use $\transkdisim$ to approximate $\directtraceinclusion$
(which approximates $\finincl$ for NFA),
$\transkdesim$ to approximate $\delayedfixedwordsimulation$,
$\transkbwdisim$ to approximate $\bwdirecttraceinclusion$ and
$\transkbwsim$ to approximate $\bwfinincl$.

For an input B\"uchi automaton $\A_{\it init}$, Heavy-k-jump-sat works as follows:
\begin{enumerate}
\item
Reduce $\A_{\it init}$ with Heavy-k-jump and obtain $\A'$.
\item
Let $\A$ be a copy of the current automaton, i.e., $\A \defeq \A'$.
Saturate $\A'$ with transitions 
w.r.t.~$\saturationrel \defeq \makesaturationrel \transkdesimrev \transkdesim$, 
i.e., $\A' \defeq \saturate{\A'}{\saturationrel}$.
\item
Quotient $\A'$ w.r.t.~$\transkbwdisim$.
\item
Saturate $\A'$ with transitions w.r.t.~$\saturationrel \defeq \makesaturationrel \transkbwdisim
\transkbwdisimrev$, i.e., $\A' \defeq \saturate{\A'}{\saturationrel}$.
\item
Reduce $\A'$ with Heavy-k-jump.
\item
If the current automaton $\A'$ has fewer states, or the same number of states and
fewer transitions, than the automaton $\A$ (the one last seen before executing step 2.), then
goto step 2. Otherwise terminate and return $\A$, the smallest automaton seen
so far. (Note that $\A'$ might have more transitions than $\A$.)
\end{enumerate}

For NFA we use the saturation criteria from
Sec.~\ref{sec:extensions:theory:NFA}.
Heavy-k-jump-sat for NFA works as described above, except that
at Step (2)\ we saturate 
w.r.t.~$\makesaturationrel \transkdisimrev \transkdisim$ 
instead of
$\makesaturationrel \transkdesimrev \transkdesim$,
at Step (3)\ we quotient with $\transkbwsim$
instead of 
$\transkbwdisim$,
and at Step (4)\ we saturate with
$\makesaturationrel \transkbwsim \transkbwsimrev$
instead of
$\makesaturationrel \transkbwdisim \transkbwdisimrev$.

The correctness of Heavy-k-jump-sat follows from
Theorem~\ref{thm:GFS:delayedfixedwordsimulation}
and Theorem~\ref{thm:GFS:bwdirecttraceinclusion}
(resp.~Theorems~\ref{thm:GFS:NFA:forward} and \ref{thm:GFS:NFA:backward}
for NFA) and the correctness of Heavy-k-jump.

Note that the algorithm above is not optimal, in the sense that a more
aggressive application of the saturation techniques might sometimes yield an
even smaller automaton. While the number of states can never increase, the
number of transitions might fluctuate (go up and down) many times if
the steps (2)--(5) were applied repeatedly, before the number of states finally decreases
again. This is because Heavy-k-jump does not necessarily remove the same
transitions that the saturation methods have added.
The termination criterion in step (6) is more strict, since it stops 
immediately if no progress is seen, even though a continuation might possibly
yield an even smaller result.
The version above has been chosen for pragmatic reasons of balancing speed and
effectiveness. Alternatively, one might stop only when a loop is detected---% 
i.e., if the same automaton is seen twice, that is,
if $\A$ and $\A'$ are isomorphic at step (6). However, 
this could take a very long time if the number of transitions fluctuates, and
it rarely yields any significant advantage. 
On Tabakov-Vardi random NBA/NFA, the more aggressive version Heavy-k-jump-sat2
produced a different result (compared to that produced by Heavy-k-jump-sat) 
in only $<1\%$ of the test cases.

We now compare the behavior of Heavy-k-jump and Heavy-k-jump-sat. 
Given some input automaton $\A_{\it init}$, let $\A$ be the reduced automaton 
produced by Heavy-k-jump and $\A_s$ be the result of Heavy-k-jump-sat.
It follows directly from the definitions above that one of the following two
cases holds.
\begin{itemize}
\item
$\A_s$ has strictly less states than $\A$.
In this case there is no restriction on the number of transitions of $\A_s$. 
It can be lower, equal or higher than the number of transitions in $\A$.
\item
$\A_s$ has exactly the same number of states as $\A$.
In this case the number of transitions of $\A_s$ is lower than or equal to
the number of transitions in $\A$.
\end{itemize}
Thus Heavy-k-jump-sat prioritizes reducing the number of states over reducing
the number of transitions. In other words, there can be a tradeoff 
where Heavy-k-jump-sat produces an automaton with fewer states but more
transitions, compared to the one produced by Heavy-k-jump.
(Recall the empirical result from Sec.~\ref{subsec:sparseness} that
Heavy-k-jump, on average, produces automata that are not only smaller but also 
sparser than the original.)

For example, some of the NBA derived from mutual exclusion protocols
considered in Sec.~\ref{sec:protocols} can be reduced even further.
The automaton fischer.2.c.ba was reduced to 192 states and 316 transitions by
Heavy-12, to 190 states and 314 transitions by
Heavy-12-jump,
and to 177 states and 392 transitions by Heavy-12-jump-sat.
The automaton fischer.3.2.c.ba was reduced to 70 states and 96 transitions by
Heavy-12 and Heavy-12-jump, and to 27 states and 53 transitions by
Heavy-12-jump-sat.
In the first automaton we had a tradeoff between states and transitions, while
in the second automaton both were reduced.

However, empirically, on most automata this tradeoff effect between states and transitions is
not very strong. Our tests on Tabakov-Vardi random automata show that 
Heavy-k-jump-sat very often produces automata with \emph{both} fewer states and
fewer transitions, when compared to Heavy-k-jump.

Fig.~\ref{fig:BuchiH12JS} shows that the extra
effect of the saturation methods (i.e., the difference between
Heavy-k-jump and Heavy-k-jump-sat)
is very modest for B\"uchi automata, when we use our standard lookahead of $k=12$.
For transition densities ${\it td} \le 1.4$ the number of states is marginally
reduced at the expense of having a slightly higher number of transitions.
For $1.5 \le {\it td} \le 1.8$ both states and transitions are slightly
reduced. For ${\it td} \ge 1.9$ there is no difference, because the automata
produced by Heavy-12-jump are already very small.
In the interesting region of $1.5 \le {\it td} \le 1.8$, Heavy-12-jump-sat
yields automata with fewer states than Heavy-12-jump in about 10\%--25\% of
the cases, while the number of transitions is only larger in 5\%-8\% of the
cases.

\begin{figure}[htbp]
\begin{center}
\includegraphics[scale=0.7]{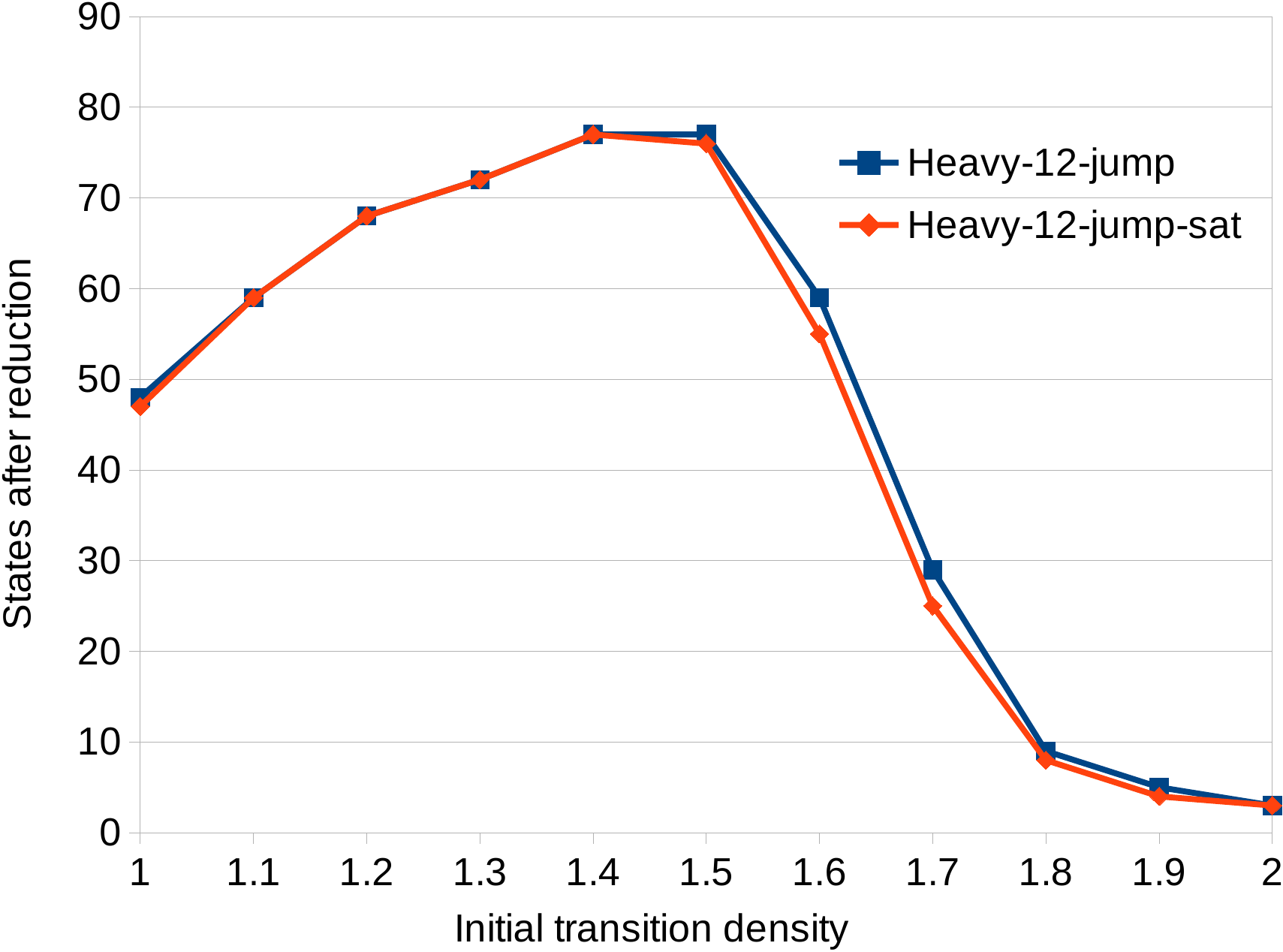}

\vspace{3mm}
\includegraphics[scale=0.7]{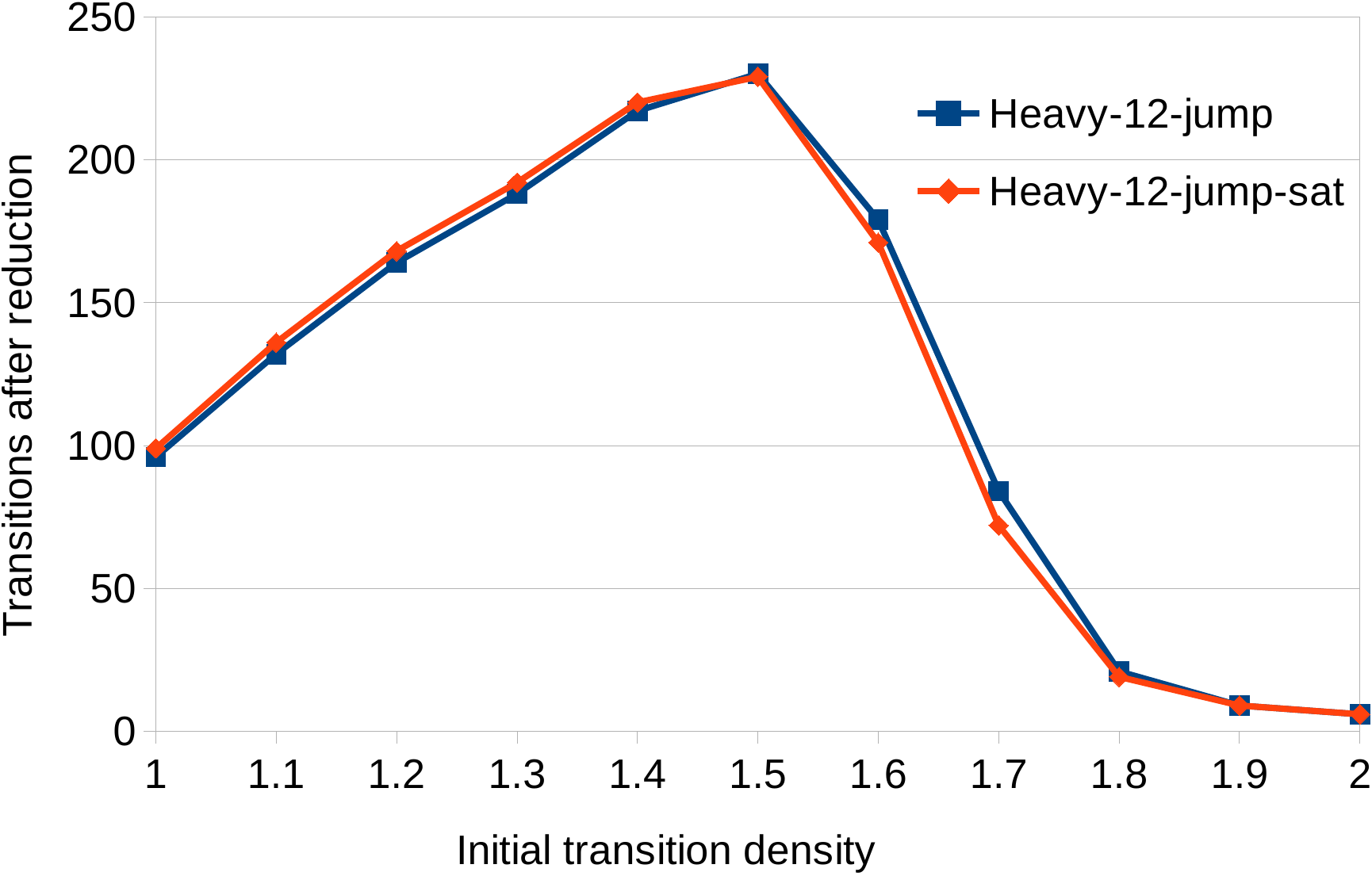}
\end{center}
\caption{We consider Tabakov-Vardi random B\"uchi automata with $n=100$, $|\Sigma|=2$,
${\it ad}=0.5$ and ${\it td}=1.0, \dots, 2.0$.
The x-axis is the transition density of the original automata.
In the upper/lower picture the y-axis is the average number of
states/transitions
of the reduced automata after
applying Heavy-12-jump and Heavy-12-jump-sat, respectively.
There is hardly any difference between the methods for ${\it td} < 1.4$ or 
${\it td} > 2.0$.
Every data point is the average of $1000$ automata.
}\label{fig:BuchiH12JS}
\end{figure}

In contrast, the saturation methods have a significant effect for NFA,
as shown in 
Fig.~\ref{fig:FiniteH12JS}.
For transition densities ${\it td} \le 1.4$ the number of states is marginally
reduced at the expense of having a moderately higher number of transitions.
For $1.5 \le {\it td} \le 1.9$ both states and transitions are significantly
reduced. For ${\it td} \ge 2.0$ there is no difference, because the automata
produced by Heavy-12-jump are already very small.
In the interesting region of $1.5 \le {\it td} \le 1.9$, Heavy-12-jump-sat
yields automata with fewer states than Heavy-12-jump in about 30\%--60\% of
the cases, while the number of transitions is only larger in 8\%-10\% of the
cases.
  
\begin{figure}[htbp]
\begin{center}
\includegraphics[scale=0.6]{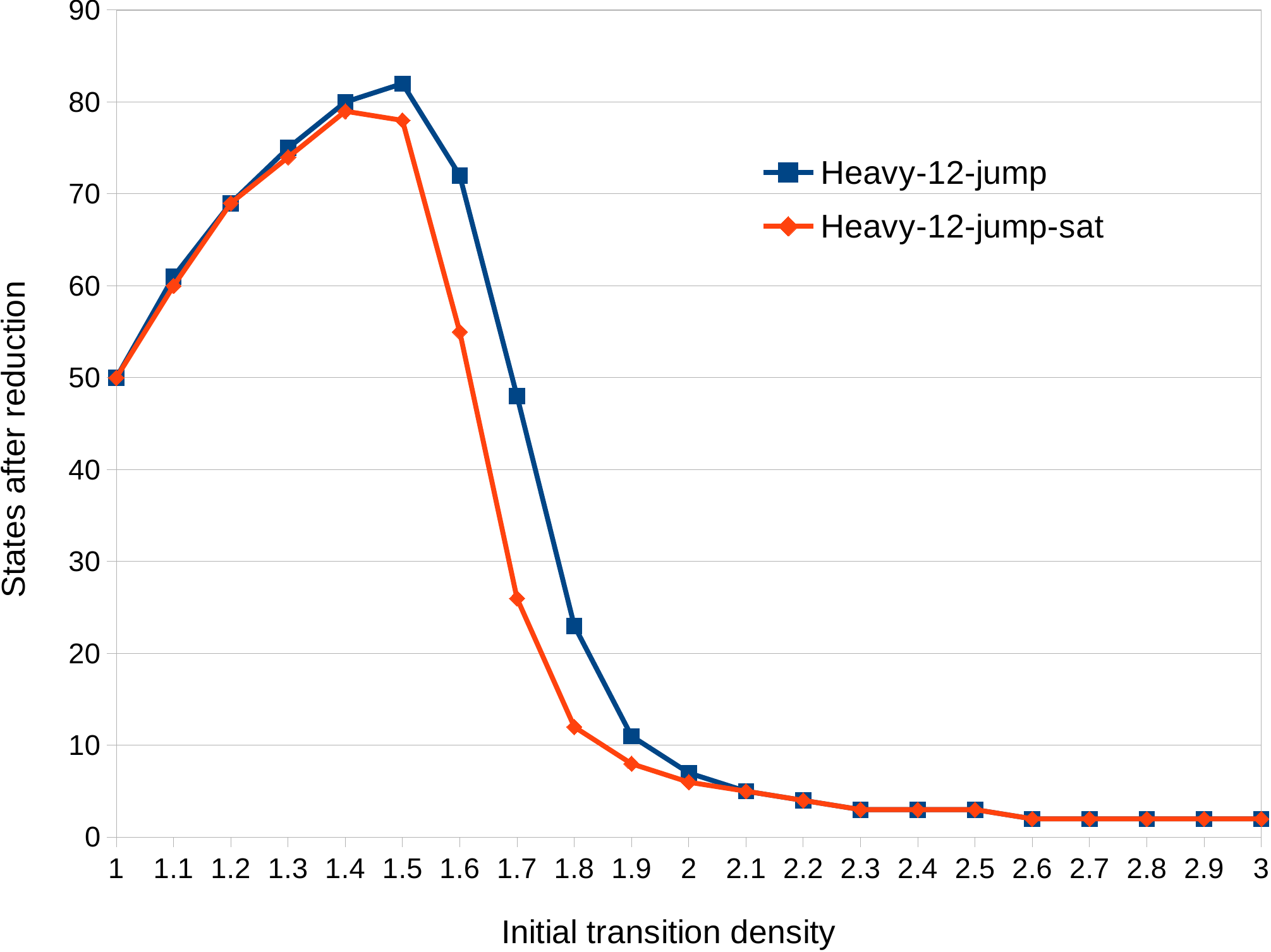}

\vspace{3mm}
\includegraphics[scale=0.6]{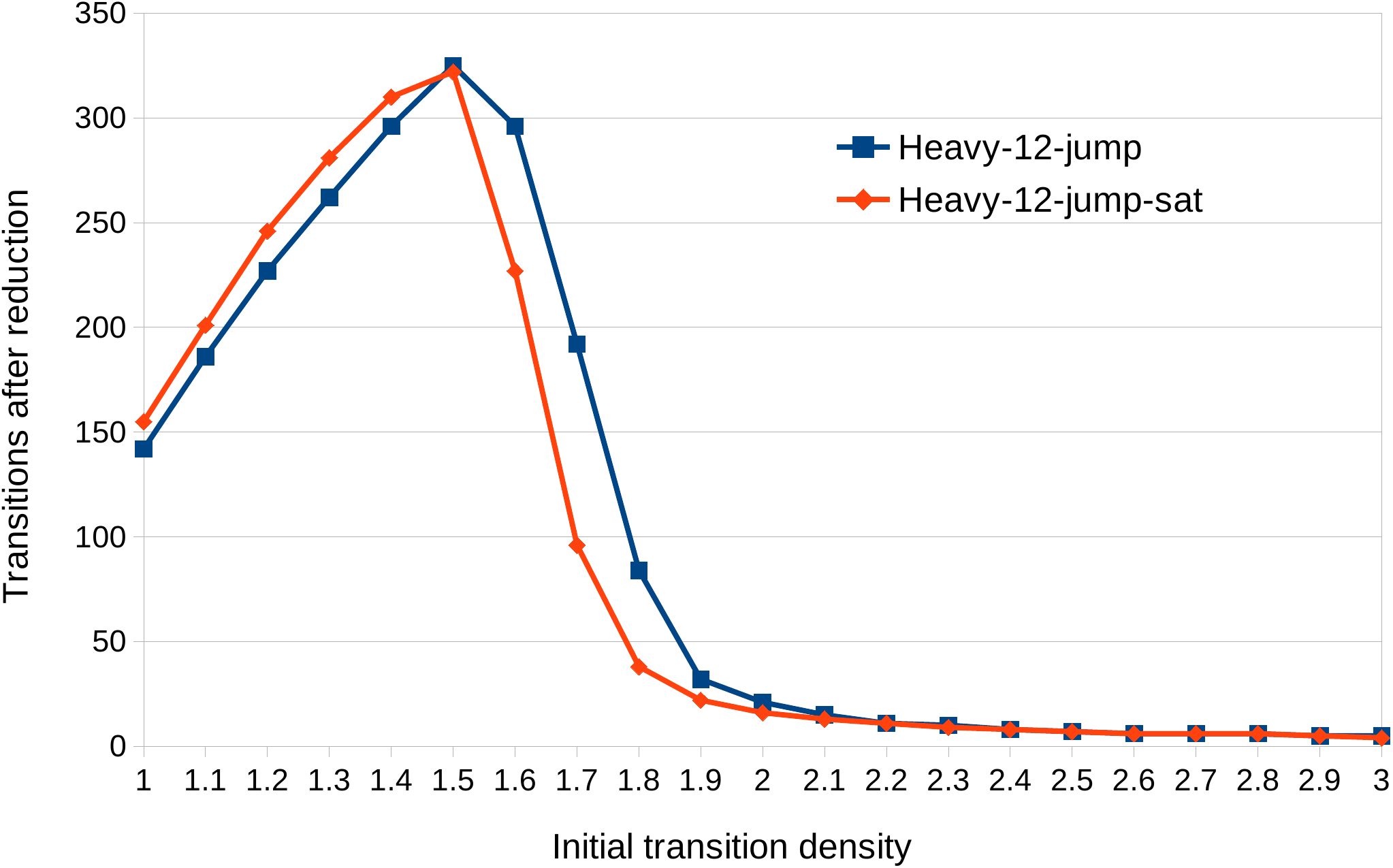}
\end{center}
\caption{We consider Tabakov-Vardi random NFA with $n=100$, $|\Sigma|=2$,
${\it ad}=0.5$ and ${\it td}=1.0, \dots, 3.0$.
The x-axis is the transition density of the original automata.
In the upper/lower picture the y-axis is the average number of
states/transitions
of the reduced automata after
applying Heavy-12-jump and Heavy-12-jump-sat, respectively.
Every data point is the average of $1000$ automata.
}\label{fig:FiniteH12JS}
\end{figure}

In Fig.~\ref{fig:HJS_buchi_speed}
we compare the speed of Heavy-12-jump and Heavy-12-jump-sat
on B\"uchi automata and NFA, respectively.
The results heavily depend on the transition density of the input automata,
but in the interesting region of $1.5 \le {\it td} \le 1.8$,
Heavy-12-jump-sat is about 2-4 times slower.

\begin{figure}[htbp]
{\bf\noindent\large B\"uchi automata}
\begin{center}
\includegraphics[scale=0.7]{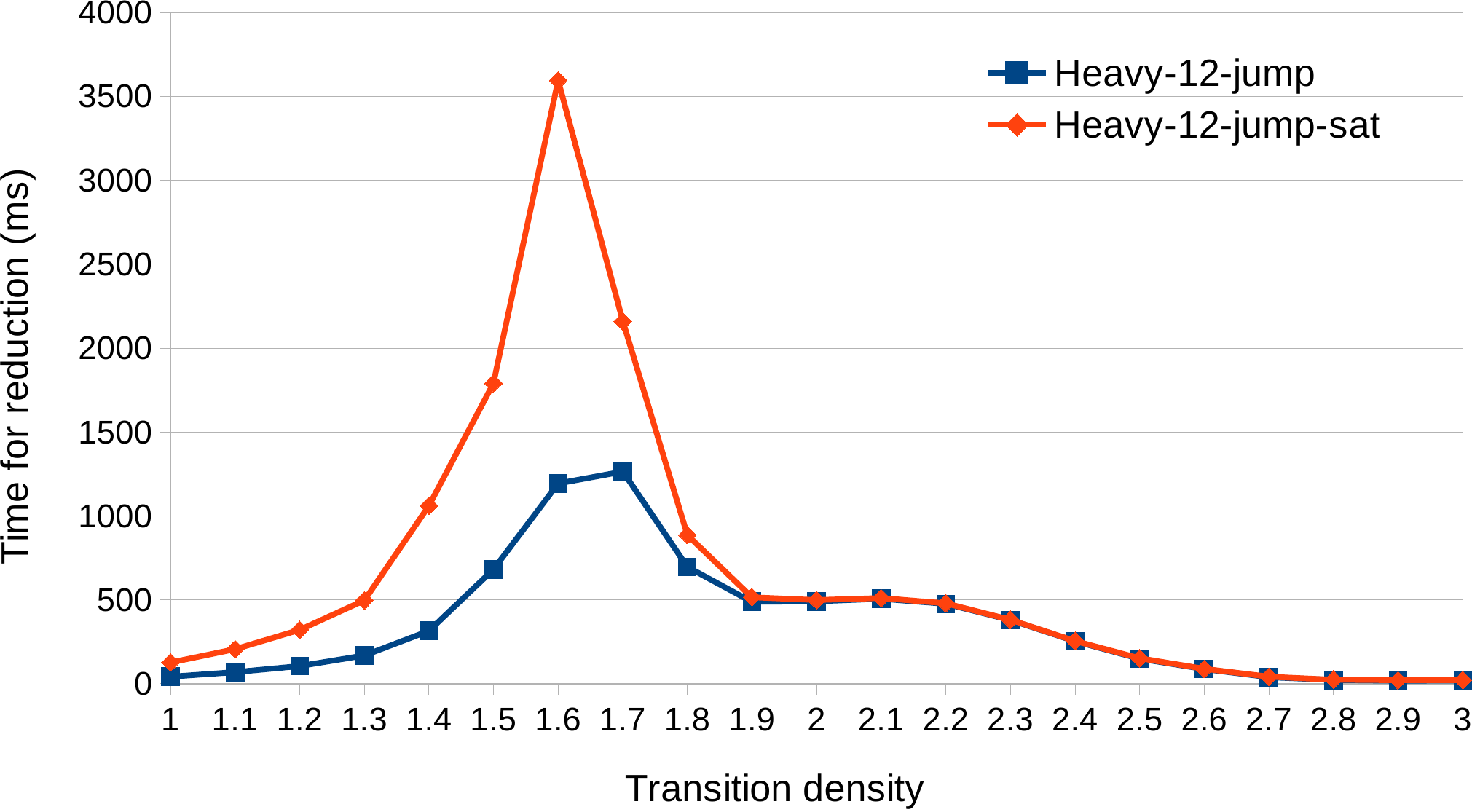}
\end{center}

\vspace{3mm}
{\bf\noindent\large Finite automata}
\begin{center}
\includegraphics[scale=0.7]{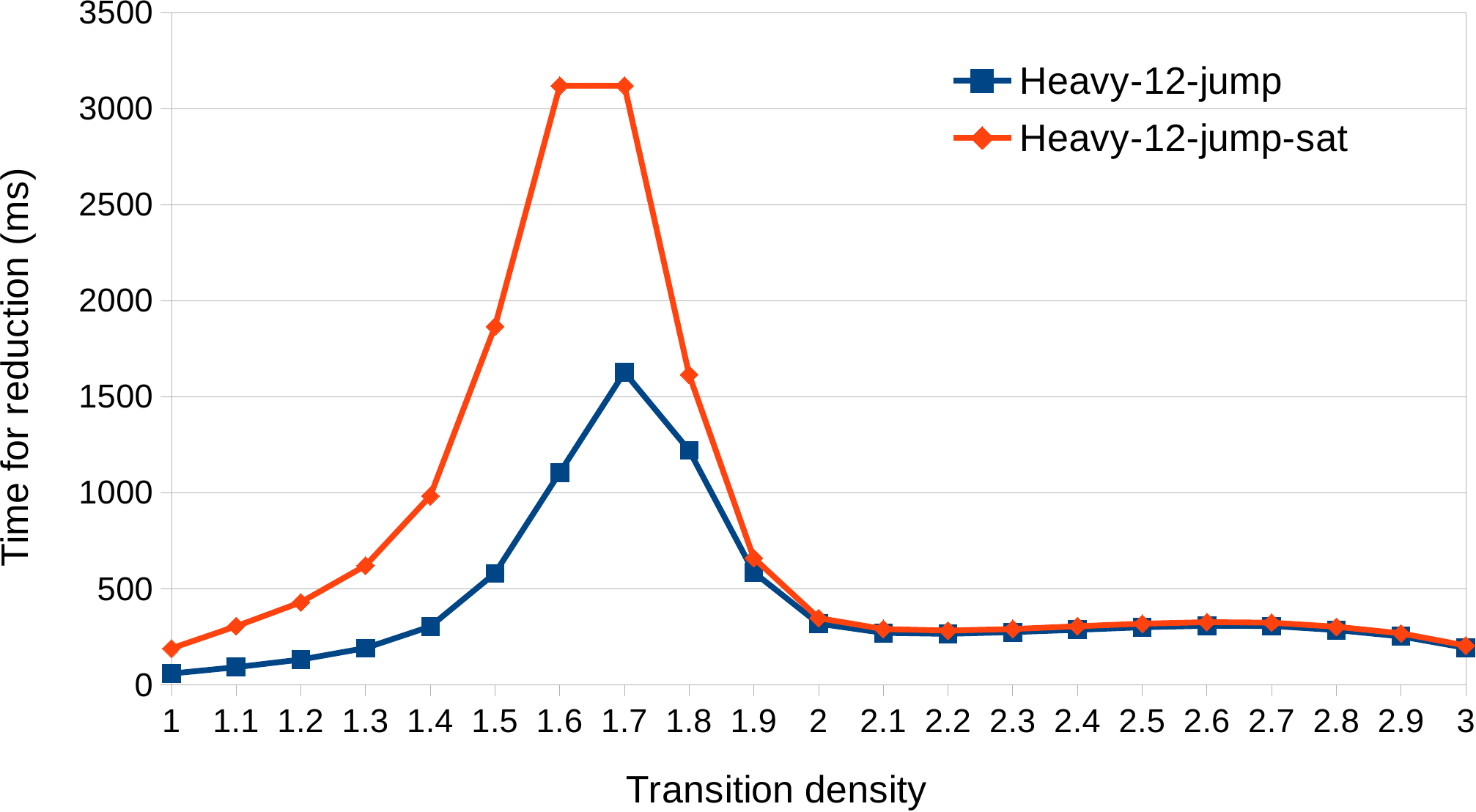}
\end{center}
\caption{We consider Tabakov-Vardi random B\"uchi automata (upper picture) and
Tabakov-Vardi random NFA (lower picture) with $n=100$, $|\Sigma|=2$,
${\it ad}=0.5$ and ${\it td}=1.0, \dots, 3.0$.
The x-axis is the transition density of the original automata while
the y-axis is the average time (in ms) to reduce the automata with
Heavy-12-jump and Heavy-12-jump-sat, respectively.
Every data point is the average of $1000$ automata.
Java 7 on Intel 2 Q8300, 2.50GHz.
}\label{fig:HJS_buchi_speed}
\end{figure}

%%% Local Variables:
%%% mode: latex
%%% TeX-master: "ROOT.tex"
%%% End:

\newpage
\section{Notes on the Implementation}\label{sec:implementation}

\subsection*{The tools}
Our tools Reduce and RABIT \cite{RABIT} implement the reduction algorithm
of Sections~\ref{sec:heavyandlight} and \ref{sec:extensions} and
the inclusion testing algorithm of Sec.~\ref{sec:inclusion}, respectively.
They are distributed in one package, written in Java (version $\ge 7$) 
and are licensed under the GPLv2.

The Reduce and RABIT tools currently support only the \texttt{.ba} format to describe
automata. This format is also supported by GOAL \cite{GOAL_survey_paper}.
However, the package contains a utility to convert NFA in \texttt{.ba} format
into the \texttt{.timbuk} format used by Libvata \cite{tool:libvata}.
The \texttt{.ba} format is very basic. First it gives the unique initial
state.
Then comes a list of labeled transitions (one per line) and
then the list of accepting states (one per line).
Example:
\begin{verbatim}
[1]
a,[1]->[2]
b,[2]->[1]
c,[1]->[3]
[2]
[3]
\end{verbatim}

The default reduction algorithm in Reduce is Heavy-k-jump. 
Other versions like Heavy-k, Heavy-k-jump-sat and a more aggressive version
Heavy-k-jump-sat2 can be invoked with options \texttt{-nojump},
\texttt{-sat} and \texttt{-sat2}, respectively.
By default it assumes that the input is an NBA. It switches to NFA with the
option \texttt{-finite}. The lookahead $k\ge 1$ is given as a parameter.
Example: \texttt{java -jar Reduce.jar example.ba 12 -sat} reduces the NBA 
example.ba with Heavy-12-jump-sat.

The tool RABIT tests inclusion between NBA (or NFA with option
\texttt{-finite}). Since it implements many optimizations, it should be
invoked with option \texttt{-fast} for best performance.
(The default is no optimization, which is very slow.)
The lookahead may be specified as a parameter, but by default RABIT
uses a heuristic that depends on the sizes and shapes of the input automata.
Example: \texttt{java -jar RABIT.jar A.ba B.ba -fast -jf} tests inclusion between
the NBA A.ba and B.ba. The additional option \texttt{-jf} has the effect that,
after the automata reduction, a new thread is created that runs in parallel to
the Ramsey-based antichain method. This thread alternatingly checks the GFI 
$\countingtranskbwsim$-jumping $k$-lookahead fair simulation 
and the segmented jumping $k$-lookahead fair simulation
from
Sec.~\ref{sec:jumpsim} for an ever increasing $k=1,2,3,\dots$. 
(Both threads stop as soon as one of them finds a solution.)

\subsection*{Algorithmic issues}
The most critical part of the code is the computation of the $k$-lookahead
simulations, since this takes the majority of the time on non-trivial cases.
It becomes computationally feasible by using several optimizations described
informally below. For details the reader is referred to \cite{RABIT}
(algorithms/Simulation.java and algorithms/ParallelSimulation.java).
\begin{itemize}
\item
We represent binary relations between states by boolean matrices,
i.e., $(p,q)$ is in the relation iff the matrix element $(p,q)$ is true.
By using the $\mu$-calculus characterization of lookahead simulations in
Sec.~\ref{sec:fixedpoint}, they can be computed by a
fixpoint iteration that converges to the lookahead simulation.
E.g., for direct simulation one starts with all matrix elements set to true
and refines downward. (It is more complex for delayed- and fair simulations.)
This takes place \emph{in situ}, i.e., there is only one copy of the matrix.
Thus changing an element in the matrix possibly affects tests and
changes of other matrix elements already in the same round of the iteration,
instead of the next round. This reduces the number of iterations
significantly. (Achieving this in situ effect might be a problem for certain
types of symbolic representations, e.g., BDDs.) 
\item
Consider one round of the fixpoint refinement for $k$-lookahead direct simulation. 
In every such refinement round one needs to check, for every matrix
element $(p,q)$ that is still true, whether it should remain true.
As explained in Sec.~\ref{sec:lookahead:simulation}, for every check of 
a pair of states $(p,q)$, Spoiler builds his 
attacks incrementally, i.e., he explores the tree of all possible attacks 
starting at $p$ by depth-first search up-to depth $k$.
For every partial Spoiler attack in this tree, Duplicator searches for a
possible defense. Whenever Duplicator can defend (from $q$) against a branch of
non-maximal depth, deeper exploration of this branch is omitted.
Thus, most of these checks of $(p,q)$ are resolved without exploring the full
tree of all attacks from $p$ up-to maximal depth.
The following item elaborates different ways how Duplicator can search for
a valid defense.
\item
Given an attack by Spoiler from some state $p$ of some depth $k' \le k$, 
Duplicator needs to check whether there is a defense from state $q$.
A basic algorithm would explore the tree of Duplicator's moves up-to depth
$k'$ (and stop once a defense is found). 
This basic version is implemented in \cite{RABIT}, but not normally used
(except in cases of very small lookahead), because other versions are often
more efficient.
The basic version is wasteful (for higher lookahead), because
many of Spoiler's attacks share common prefixes.
A more efficient variant (also implemented in \cite{RABIT})
views lookahead simulation as a special case of
multipebble simulation, as explained in Remark~\ref{rem:lookahead-pebble}.
When checking a pair of states $(p,q)$,
Duplicator maintains and propagates several pebbles (starting with just one
pebble on state $q$) that encode all her possible moves
against Spoiler's current attack from $p$ (up-to depth $k$).
Unlike in general multipebble simulation, this use of pebbles is local to the
current round of the game, since Duplicator needs to commit to just one pebble
after at most $k$ steps. 

This version needs to efficiently handle many sets
of pebbles, i.e., subsets of states of the automaton. In an automaton with 
$n$ states, the states are represented by integers in the set $\{0,\dots,n-1\}$, and 
sets of pebbles as subsets thereof. The only needed set operations are to 
add elements and to iterate through all elements of a set, but not to
explicitly check membership. Java generic sets are
not optimal here, due to their internal overhead 
(here the elements are just integers and typically $n \le 30000$). 
Our implementation uses a combination of integer arrays
(to describe a list of size $|S|$; for sets $S$ with $|S|^2 \le 4n$) 
and boolean arrays (of length $n$; otherwise) to
represent these sets. In the former case, adding a new element is
$\mathcal{O}(|S|)$ (since duplicates must be avoided), 
but iterating through the set is optimal. Thus it is used
only for sets that are small, relative to $n$.
In the latter case, adding a new element is $\mathcal{O}(1)$, but iterating
through the set takes $\mathcal{O}(n)$ steps (instead of $\mathcal{O}(|S|)$ steps), 
which is inefficient if $|S| \ll n$. So this representation is used only for larger
sets. Since new sets are derived from previous sets by the propagation of the 
pebbles in the automaton, it is possible to predict (roughly, but reasonably
accurately) whether a new set will be small or large in the sense described above. 
\item
Even before the main fixpoint refinement loop starts, one can do a 
quick pre-processing step that sets many matrix elements to false.
If there exists a short word $w$ that can be read from state $p$
but not from state $q$ then it is trivial that $q$ cannot simulate $p$.
Our implementation checks this condition for all words up-to a short length,
typically 4-8 (depending on the size of the alphabet). 
The parallel version (see below) does this operation in a separate
thread with words $w$ of increasing length (up-to a certain maximum
depending on memory requirements).
\item
Finally, the iterations over the boolean matrix can be parallelized. 
The matrix is split into many small parts, and each part is handled by 
a worker-thread from the pool of available worker threads.
This uses Java 7 fork-join concurrency and is invoked by the option
\texttt{-par}.
While each worker-thread writes only to a small part of the matrix, it
still needs read access to the whole matrix. Thus the computation cannot
easily be distributed, and shared memory access is still a bottleneck.
Due to the monotonicity of the fixpoint refinement algorithm, missed
updates are not a problem in the parallel version.
The parallel version can be several times faster than the single-threaded one,
depending on the hardware and on the input instance.
However, all our benchmarks were done with the single-threaded version.
\end{itemize}

\noindent
An optimized algorithm to compute ordinary simulations (i.e., with lookahead
$k=1$) was described in \cite{HHK:FOCS95,HS:MEMICS2009}. 
We explain its main idea for the case of forward direct simulation.
When some matrix elements are changed (from true to false) in an iteration
of the fixpoint refinement, then only some elements may change in the next
round, while other elements cannot possibly change (yet) because they
are not directly affected by the recently changed elements.
By keeping track of these dependencies, one can avoid redundant tests of 
elements that will not change (or at least not yet in this iteration).
Our implementation uses this technique only in the case of lookahead $k=1$,
but not for higher lookaheads, because the dependency information becomes more complex
and keeping track of it is not cost-effective.
Worse yet, at higher lookaheads the computation time is not evenly distributed
over all matrix elements (i.e., pairs of states). Instead the distribution is
highly skewed towards a minority of hard cases. Typically, these are
pairs of states where lookahead simulation does not hold, but where
non-simulation is only established very late, i.e., after many iterations.
In many earlier iterations Duplicator still wins easily (in small time)
and the element stays true. Avoiding these redundant tests would yield only a small
benefit, but still incur the required overhead.
The element is only set to false in a late iteration where Duplicator finally
admits defeat after having vainly searched through the entire tree of all
possible candidates for a defense (against a certain Spoiler attack) up-to the maximal allowed lookahead.
This last test is thus a costly operation that cannot be avoided.

%%% Local Variables:
%%% mode: latex
%%% TeX-master: "ROOT.tex"
%%% End:

\section{Conclusion and Future Work}\label{sec:conclusion}

Our automata reduction technique Heavy-k, and its extensions
Heavy-k-jump, Heavy-k-jump-sat and Heavy-k-jump-sat2,
perform significantly better than previous methods implemented in GOAL 
\cite{GOAL_survey_paper}.
In particular, they can be applied to solve PSPACE-complete
automata problems like language inclusion for much larger instances than
before.
Our tools Reduce and RABIT \cite{RABIT} implement these
algorithms in Java (version $\ge 7$) 
and are licensed under the GPLv2.

Future work includes more efficient algorithms for computing lookahead
simulations by using symbolic representations of the relations/automata by BDDs or
related formalisms.
It would also be very useful to have a practically feasible way to compute
under-approximations of language inclusion for NBA/NFA that are
\emph{orthogonal} to multipebble/lookahead-simulations. Then one could
obtain an even better approximation by considering the transitive closure of
the union of all approximations.

Quotienting techniques for alternating automata over infinite words are well-understood~
\cite{wilke:fritz:simulations:05,FritzWilke:DLT:2006,multipebble:concur10}
(cf.~also \cite{AbdullaCHV09,AbdullaChenHolikVojnar:TCS:2014}).
An interesting research direction is to extend our transition pruning and saturation methods
from nondeterministic to alternating finite, B\"uchi, generalized B\"uchi, and parity automata.

Finally, transition pruning techniques have more recently been applied to reduce automata on
finite trees \cite{AHM:TACAS2016}
and infinite trees \cite{Almeida:PhD2017}.

%%% Local Variables:
%%% mode: latex
%%% TeX-master: "ROOT.tex"
%%% End:

\subsection*{Acknowledgments.}
We acknowledge the anonymous reviewers for their extensive comments on earlier drafts of this article.

% \newpage
\bibliographystyle{abbrv}
\bibliography{bibliography}

\end{document}